\pdfoutput=1
\documentclass[oneside,12pt]{IIScthesisPSnPDF}

\usepackage{cancel}
\usepackage{mathrsfs}

\usepackage[compact]{titlesec}
\usepackage{mdwlist}

\usepackage{varwidth}
\usepackage{pdfpages}

\usepackage{bm}
\usepackage{capt-of}

\usepackage[export]{adjustbox} 
\usepackage{amsfonts}
\usepackage{amsmath}
\usepackage{amssymb}
\usepackage{booktabs}
\usepackage[vertfit]{breakurl}
\usepackage{caption}
\usepackage{cleveref}
\usepackage{datetime}
\usepackage{enumerate}
\usepackage{enumitem}
\usepackage{fancyhdr}
\usepackage{float}
\hypersetup{
	colorlinks   = true,
	urlcolor = darkgreen,
	linkcolor = darkblue
}
\usepackage{latexsym}
\usepackage{makecell}
\usepackage{mathtools}
\usepackage{multirow}
\usepackage{multicol}
\usepackage{nicematrix}
\usepackage{pgf}
\usepackage{pgfplots}
\usepackage{pgfplotstable}
\usepackage{pifont}
\usepackage{ragged2e}
\usepackage{stmaryrd}
\usepackage{stfloats}
\usepackage{subcaption}
\usepackage{tabu}
\usepackage{tabularx}
\usepackage{threeparttable}
\usepackage{tikz}
\usepackage{todonotes}
\usepackage{url}
\usepackage{verbatim}
\usepackage{wasysym}
\usepackage{wrapfig}
\usepackage{xcolor}
\usepackage{xspace}

\usetikzlibrary{shapes,arrows, trees}

\usepackage[english]{babel}
\usepackage [autostyle, english = american]{csquotes}
\MakeOuterQuote{"}
\usepackage[framemethod=tikz]{mdframed} 
\usepackage{algpseudocode}
\usepackage{nccmath}
\pgfplotsset{compat=1.17}

\newtheorem{theorem}{Theorem}[chapter]
\newtheorem{lemma}{Lemma}[chapter]

\newtheorem{definition}{Definition}[chapter]

\newcommand{\remove}[1]{}

\newenvironment{proof}{\noindent{\bf Proof:} \hspace*{1mm}}{\hfill $\Box$ }
\newcommand{\notes}[1]{}

\newcommand{\first}[1]{$1^{\mathrm{st}}$}
\newcommand{\second}[1]{$2^{\mathrm{nd}}$}

\newcommand{\squishlisttwo}{
\begin{list}{$\blacktriangleright$}
{ \setlength{\itemsep}{0.5pt}
\setlength{\parsep}{0pt}
\setlength{\topsep}{0pt}
\setlength{\partopsep}{0.5pt}
\setlength{\leftmargin}{1em}
\setlength{\labelwidth}{1em}
\setlength{\labelsep}{0.5em} } }

\newcommand{\squishend}{
\end{list} }
\allowdisplaybreaks[1]

\newdateformat{monthyeardate}{ \monthname[\THEMONTH], \THEYEAR}

\newcommand{\blankpage}{
	\newpage
	\thispagestyle{empty}
	\mbox{}
	\newpage
}

\crefname{observation}{observation}{observations}
\crefname{algorithm}{algorithm}{algorithms}
\crefname{align}{equation}{equations}
\crefname{eqnarray}{equation}{equations}

\newcommand{\papertitle}{MPCLeague: Robust MPC Platform for Privacy‐Preserving Machine Learning}

\newcommand{\thisW}{{\bf Our Work}}
\newcommand{\thisP}{{\bf MPCLeague}}

\newcommand{\published}[2] {\href{#1}{\textcolor{darkgreen}{#2}}}
\newcommand{\publishedW}[2] {\href{#1}{\textcolor{UniRot}{#2}}}
\newcommand{\undersub} {\textcolor{darkred}{Under Submission}}

\newcommand{\ptitle}[1] {\textcolor{black} {\em #1}}
\newcommand{\CORE}[1] {{\href{https://www.core.edu.au}{{(CORE #1)}}}}

\definecolor{UniBlau}{cmyk}{1,0.7,0,0}
\definecolor{UniGruen}{cmyk}{0.6,0,1,0}
\definecolor{UniOrange}{cmyk}{0,0.3,1,0}
\definecolor{UniRot}{cmyk}{0.4,1,0,0}
\definecolor{darkred}{rgb}{.6,0,0}
\definecolor{darkgreen}{rgb}{0,.4,0}
\definecolor{darkblue}{rgb}{0,0,.6}

\newif\iffullversion
\fullversionfalse
\fullversiontrue

\iffullversion
    \newcommand{\detail}[1]{{#1}}
\else
    \newcommand{\detail}[1]{}
\fi

\newif\ifsubmission
\submissionfalse

\ifsubmission
	\newcommand{\EXTRALINES}[1] {}

	\newcommand{\OLD}[1]{}
	
\else	
	\newcommand{\EXTRALINES}[1] {\textcolor{darkblue} {#1}}

		\newcommand{\OLD}[1]{{\leavevmode\color{UniBlau}{(OLD CONTENT): #1}}}
	
\fi

\newcommand{\TSthis}{\textsf{ASTRA}}
\newcommand{\TSthisT}{\textsf{ASTRA}\textsubscript{\sf T}}
\newcommand{\TSthisC}{\textsf{ASTRA}\textsubscript{\sf C}}

\newcommand{\Tthis}{\textsf{SWIFT}}

\newcommand{\TthisT}{\textsf{SWIFT}\textsubscript{\sf T}}
\newcommand{\TthisC}{\textsf{SWIFT}\textsubscript{\sf C}}

\newcommand{\Fthis}{\textsf{Tetrad}}

\newcommand{\FthisA}{\textsf{Tetrad-R}\textsuperscript{\sf I}}
\newcommand{\FthisB}{\textsf{Tetrad-R}\textsuperscript{\sf II}}
\newcommand{\FthisT}{\textsf{Tetrad}\textsubscript{\sf T}}
\newcommand{\FthisC}{\textsf{Tetrad}\textsubscript{\sf C}}

\newcommand{\TWthis}{\textsf{ABY2.0}}
\newcommand{\TWthisT}{\textsf{ABY2.0}\textsubscript{\sf T}}
\newcommand{\TWthisC}{\textsf{ABY2.0}\textsubscript{\sf C}}

\newcommand{\negl}{\ensuremath{\mathsf{negl}}}
\newcommand{\csec}{\kappa}

\newcommand{\abort}{\ensuremath{\mathtt{abort}}}
\newcommand{\aborts}{\ensuremath{\mathtt{aborts}}}
\newcommand{\continue}{\ensuremath{\mathtt{continue}}}
\newcommand{\flag}{\ensuremath{\mathsf{flag}}}

\newcommand{\Partyset}{\ensuremath{\mathcal{P}}}

\newcommand{\PartysetO}{\ensuremath{\mathcal{P}_{\sf on}}}
\newcommand{\PartysetV}{\ensuremath{\mathcal{P}_{\sf ver}}}
\newcommand{\bitb}{\ensuremath{\mathsf{b}}} 
\newcommand{\isTr}{\ensuremath{\mathsf{isTr}}}
\newcommand{\ckt}{\ensuremath{\mathsf{ckt}}}

\newcommand{\Hash}{\ensuremath{\mathsf{H}}}
\newcommand{\commit}{\ensuremath{\mathsf{Com}}}
\newcommand{\Commit}[1]{\ensuremath{\commit(#1)}}

\newcommand{\NOTBool}{\textsf{NOT}}

\newcommand{\SELECT}{\ensuremath{\mathsf{select}}}
\newcommand{\INPUT}{\ensuremath{\mathsf{Input}}}
\newcommand{\OUTPUT}{\ensuremath{\mathsf{Output}}}

\newcommand{\Adv}{\ensuremath{\mathcal{A}}}
\newcommand{\Sim}{\ensuremath{\mathcal{S}}}

\newtheorem{notation}[theorem]{Notation}

\newcommand{\xor}{\oplus}

\newcommand{\band}{\odot}

\newcommand{\Order}{\mathcal{O}}

\newcommand{\cmark}{\ding{51}}
\newcommand{\xmark}{\ding{55}}

\newcommand{\iseq}{\ensuremath{\stackrel{?}{=} }}

\newcommand{\bigominus}{\mathop{\raisebox{-.05em}{\large\boldmath$\ominus$}}}
\newcommand{\MatMul}{\bigodot}

\newcommand{\bitset}{\{0,1\}}
\newcommand{\Prob}{\ensuremath{\mathsf{Pr}}}
\newcommand{\msb}{\ensuremath{\mathsf{msb}}}
\newcommand{\lsb}{\ensuremath{\mathsf{lsb}}}

\newcommand{\Z}[1]{\ensuremath{\mathbb{Z}}_{2^{#1}}}
\newcommand{\F}{\ensuremath{\mathbb{F}}}
\newcommand{\NN}{\ensuremath{\mathbb{N}}}

\newcommand{\Triple}{\ensuremath{\mathsf{Triple}}}

\newcommand{\jsend}{\ensuremath{\mathsf{jsnd}}}
\newcommand{\Sh}{\ensuremath{\mathsf{Sh}}}
\newcommand{\JSh}{\ensuremath{\mathsf{JSh}}}
\newcommand{\Rec}{\ensuremath{\mathsf{Rec}}}
\newcommand{\Mult}{\ensuremath{\mathsf{Mult}}}
\newcommand{\MultT}{\ensuremath{\mathsf{Mult3}}}
\newcommand{\MultF}{\ensuremath{\mathsf{Mult4}}}
\newcommand{\MultPre}{\ensuremath{\mathsf{MultPre}}}
\newcommand{\MultS}{\ensuremath{\mathsf{MultS}}}

\newcommand{\dotp}{\ensuremath{\mathsf{dotp}}}
\newcommand{\dotpPre}{\ensuremath{\mathsf{dotpPre}}}
\newcommand{\bitA}{\ensuremath{\mathsf{bit2A}}}
\newcommand{\dbitA}{\ensuremath{\mathsf{dbit2A}}}
\newcommand{\bitinj}{\ensuremath{\mathsf{bitInj}}}
\newcommand{\dbitinj}{\ensuremath{\mathsf{dbitInj}}}
\newcommand{\bitinjS}{\ensuremath{\mathsf{bitInjS}}}
\newcommand{\piecewise}{\ensuremath{\mathsf{piecewise}}}
\newcommand{\bitext}{\ensuremath{\mathsf{bitext}}}
\newcommand{\relu}{\ensuremath{\mathsf{ReLU}}}
\newcommand{\drelu}{\ensuremath{\mathsf{dReLU}}}
\newcommand{\sig}{\ensuremath{\mathsf{Sig}}}
\newcommand{\sftmx}{\ensuremath{\mathsf{softmax}}}
\newcommand{\cv}{\ensuremath{\mathsf{Conv}}}
\newcommand{\eql}{\ensuremath{\mathsf{eq}}}
\newcommand{\obv}{\ensuremath{\mathsf{obv}}}

\newcommand{\maxTW}{\ensuremath{\mathsf{max2}}}
\newcommand{\minTW}{\ensuremath{\mathsf{min2}}}
\newcommand{\maxT}{\ensuremath{\mathsf{max3}}}
\newcommand{\minT}{\ensuremath{\mathsf{min3}}}

\newcommand{\agmin}{\ensuremath{\mathsf{argmin}}}
\newcommand{\agmax}{\ensuremath{\mathsf{argmax}}}

\newcommand{\trgen}{\ensuremath{\mathsf{trgen}}}

\newcommand{\prot}[1]{\ensuremath{\Pi_{#1}}}
\newcommand{\protB}[1]{\ensuremath{\Pi_{#1}^{\bf B}}}

\newcommand{\piMultO}{\ensuremath{\Pi_{\Mult}^{\mathsf{NoPre}}}}

\newcommand{\sz}{\ensuremath{\zeta}} 
\newcommand{\nm}{\ensuremath{m}} 
\newcommand{\cc}{\ensuremath{c}} 
\newcommand{\cg}{\ensuremath{g}} 
\newcommand{\cG}{\ensuremath{G}} 
\newcommand{\nL}{\ensuremath{L}} 
\newcommand{\nM}{\ensuremath{M}} 
\newcommand{\ic}{\ensuremath{u}}

\newcommand{\rt}{\ensuremath{\theta}}
\newcommand{\re}{\ensuremath{\eta}}

\setlist[description]{style=unboxed,leftmargin=0cm}

\newenvironment{myitemize}{
	\begin{list}{{$\bullet$}}{
			\setlength\partopsep{0pt}
			\setlength\parskip{0pt}
			\setlength\parsep{0pt}
			\setlength\topsep{2pt}
			\setlength\itemsep{2pt}
			\setlength{\itemindent}{8pt}
			\setlength{\leftmargin}{1pt}
		}
	}{
	\end{list}
}

\newcommand{\tabref}[1]{Table~\protect\ref{tab:#1}}
\newcommand{\secref}[1]{~\S\protect\ref{sec:#1}}

\newcommand{\figref}[1]{Figure~\ref{fig:#1}}
\newcommand{\figlab}[1]{\label{fig:#1}}

\newenvironment{boxfig*}[2]{
	\begin{figure*}[h!]		
		\fontsize{5}{5}\selectfont
		\newcommand{\FigCaption}{#1}
		\newcommand{\FigLabel}{#2}
		\vspace{-.05cm}
		\begin{center}
			\begin{small}			 
				\begin{adjustbox}{max width=\textwidth}
					\begin{tabular}{@{}|@{~~}l@{~~}|@{}}
						\hline
						\rule[-1ex]{0pt}{1ex}\begin{minipage}[b]{.95\linewidth}
							\vspace{1ex}	
						}{%
						\end{minipage}\\
						\hline
					\end{tabular}	
				\end{adjustbox}		
			\end{small}
			\vspace{-0.1cm}
			\caption{\FigCaption}
			\figlab{\FigLabel}
		\end{center}
		\vspace{-.38cm}
	\end{figure*}
}

\newenvironment{myboxfig*}[2]{
	\begin{figure*}[!htb]		
		\fontsize{5}{5}\selectfont
		\newcommand{\FigCaption}{#1}
		\newcommand{\FigLabel}{#2}
		\vspace{-.10cm}
		\begin{center}
			\caption{\FigCaption}
			\begin{small}			 
				\begin{adjustbox}{max width=\textwidth}
					\begin{tabular}{@{}|@{~~}l@{~~}|@{}}
						\hline
						\rule[-1ex]{0pt}{1ex}\begin{minipage}[b]{.95\linewidth}
							\vspace{1ex}	
						}{%
						\end{minipage}\\
						\hline
					\end{tabular}	
				\end{adjustbox}		
			\end{small}
			\vspace{-0.25cm}
			\figlab{\FigLabel}
		\end{center}
		\vspace{-.38cm}
	\end{figure*}
}

\newcommand{\boxref}[1]{Fig.~\ref{#1}}

\RequirePackage{mdframed}

\newenvironment{titlebox}[5]
{\mdfsetup{
		style=#2,
		innertopmargin=1.1\baselineskip,
		skipabove={\dimexpr0.2\baselineskip+\topskip\relax},
		skipbelow={1em},needspace=3\baselineskip,
		singleextra={\node[#3,right=10pt,overlay] at (P-|O){~{\sffamily\bfseries #1 }};},%
		firstextra={\node[#3,right=10pt,overlay] at (P-|O) {~{\sffamily\bfseries #1 }};},
		frametitleaboveskip=9em,
		innerrightmargin=5pt
	}
	\newcommand{\TitleCaption}{#4}
	\newcommand{\TitleLabel}{#5}
	\begin{mdframed}[font=\small]
		\setlist[itemize]{leftmargin=13pt}\setlist[enumerate]{leftmargin=13pt}\raggedright%
	}
	{\end{mdframed}
	\vspace{-1.6em}
	{\captionof{figure}{\small \TitleCaption}\label{\TitleLabel}}
	\medskip
}

\tikzstyle{normal} = [thick, fill=white, text=black, draw, rounded corners, rectangle, minimum height=.7cm, inner sep=3pt]
\tikzstyle{gray} = [thick, fill=gray!90, text=white, rounded corners, rectangle, minimum height=.7cm, inner sep=3pt]

\mdfdefinestyle{commonbox}{%
	align=center, middlelinewidth=1.1pt,userdefinedwidth=\linewidth
	innerrightmargin=8pt,innerleftmargin=8pt,innertopmargin=5pt,
	splittopskip=7pt,splitbottomskip=3pt
}

\mdfdefinestyle{roundbox}{style=commonbox,roundcorner=5pt,userdefinedwidth=\linewidth}

\newenvironment{systembox}[3]
{\vspace{\baselineskip}\begin{titlebox}{Functionality \normalfont #1}{roundbox}{normal}{#2}{#3}}
	{\end{titlebox}}

\newenvironment{protocolbox}[3]
{\begin{titlebox}{Protocol \normalfont #1}{commonbox}{normal}{#2}{#3}}
	{\end{titlebox}}

\newenvironment{simulatorbox}[3]
{\begin{titlebox}{Simulator \normalfont #1}{commonbox}{normal}{#2}{#3}}
	{\end{titlebox}}

\newenvironment{splittitlebox}[5]
{\mdfsetup{
		style=#2,
		innertopmargin=1.1\baselineskip,
		skipabove={\dimexpr0.2\baselineskip+\topskip\relax},
		skipbelow={1em},needspace=3\baselineskip,
		singleextra={\node[#3,right=10pt,overlay] at (P-|O){~{\sffamily\bfseries #1 }};},%
		firstextra={\node[#3,right=10pt,overlay] at (P-|O) {~{\sffamily\bfseries #1 }};},
		frametitleaboveskip=9em,
		innerrightmargin=5pt
	}
	\newcommand{\TitleCaption}{#4}
	\newcommand{\TitleLabel}{#5}
	\begin{mdframed}[font=\small]
		\setlist[itemize]{leftmargin=13pt}\setlist[enumerate]{leftmargin=13pt}\raggedright%
	}
	{\end{mdframed}
	\vspace{-0.6em}
	{\captionof{figure}{\small \TitleCaption}\label{\TitleLabel}}
	\medskip
}

\mdfdefinestyle{commonsplitbox}{%
	align=center, middlelinewidth=1.1pt,userdefinedwidth=\linewidth
	innerrightmargin=8pt,innerleftmargin=8pt,innertopmargin=5pt,
	splittopskip=7pt,splitbottomskip=3pt
}

\mdfdefinestyle{roundsplitbox}{style=commonsplitbox,roundcorner=5pt,userdefinedwidth=\linewidth}

\newenvironment{protocolsplitbox}[3]
{\begin{splittitlebox}{Protocol \normalfont #1}{commonsplitbox}{normal}{#2}{#3}}
	{\end{splittitlebox}}

\newenvironment{simulatorsplitbox}[3]
{\begin{splittitlebox}{Simulator \normalfont #1}{commonsplitbox}{normal}{#2}{#3}}
	{\end{splittitlebox}}

\newenvironment{systembox*}[3]
{\begin{strip}
		\vspace{\baselineskip}\begin{titlebox}{Functionality \normalfont #1}{roundbox}{normal}{#2}{#3}}
		{\end{titlebox}
\end{strip}}

\newenvironment{gsystembox*}[3]
{\begin{strip}
		\vspace{\baselineskip}\begin{titlebox}{Global Functionality \normalfont #1}{roundbox}{normal}{#2}{#3}}
		{\end{titlebox}
\end{strip}}

\newenvironment{protocolbox*}[3]
{\begin{strip}
		\begin{titlebox}{Protocol \normalfont #1}{commonbox}{normal}{#2}{#3}}
		{\end{titlebox}
\end{strip}}

\newenvironment{algobox*}[3]
{\begin{strip}
		\begin{titlebox}{Algorithm \normalfont #1}{commonbox}{normal}{#2}{#3}}
		{\end{titlebox}
\end{strip}}

\newenvironment{reductionbox*}[3]
{\begin{strip}
		\begin{titlebox}{Reduction \normalfont #1}{commonbox}{normal}{#2}{#3}}
		{\end{titlebox}
\end{strip}}

\newenvironment{gamebox*}[3]
{\begin{strip}
		\begin{titlebox}{Game \normalfont #1}{commonbox}{gray}{#2}{#3}}
		{\end{titlebox}
\end{strip}}

\newenvironment{simulatorbox*}[3]
{\begin{strip}
		\begin{titlebox}{Simulator \normalfont #1}{commonbox}{normal}{#2}{#3}}
		{\end{titlebox}
\end{strip}}

\mdfdefinestyle{mycommonbox}{%
	align=center, middlelinewidth=1.1pt,userdefinedwidth=\linewidth
	innerrightmargin=0pt,innerleftmargin=2pt,innertopmargin=3pt,
	splittopskip=7pt,splitbottomskip=3pt
}
\mdfdefinestyle{myroundbox}{style=mycommonbox,roundcorner=5pt,userdefinedwidth=\linewidth}

\newenvironment{titlebox*}[5]
{\mdfsetup{
		style=#2,
		innertopmargin=0.3\baselineskip,
		skipabove={0.4em},
		skipbelow={1em},needspace=3\baselineskip,
		frametitleaboveskip=5em,
		innerrightmargin=5pt
	}
	\newcommand{\TitleCaption}{#4}
	\newcommand{\TitleLabel}{#5}
	\begin{mdframed}[font=\small]
		\setlist[itemize]{leftmargin=13pt}\setlist[enumerate]{leftmargin=13pt}\raggedright%
	}
	{\end{mdframed}
	\vspace{-2em}
	{\captionof{figure}{\normalfont \TitleCaption}\label{\TitleLabel}}
	\medskip
}

\newenvironment{mysystembox*}[3]
{\begin{strip}
		\vspace{\baselineskip}\begin{titlebox*}{Functionality \normalfont #1}{myroundbox}{normal}{#2}{#3}}
		{\end{titlebox*}
\end{strip}}

\newenvironment{mygsystembox*}[3]
{\begin{strip}
		\vspace{\baselineskip}\begin{titlebox*}{Global Functionality \normalfont #1}{myroundbox}{normal}{#2}{#3}}
		{\end{titlebox*}
\end{strip}}

\newenvironment{myprotocolbox*}[3]
{\begin{strip}
		\begin{titlebox*}{Protocol \normalfont #1}{mycommonbox}{normal}{#2}{#3}}
		{\end{titlebox*}
\end{strip}}

\newenvironment{myalgobox*}[3]
{\begin{strip}
		\begin{titlebox*}{Algorithm \normalfont #1}{mycommonbox}{normal}{#2}{#3}}
		{\end{titlebox*}
\end{strip}}

\newenvironment{myreductionbox*}[3]
{\begin{strip}
		\begin{titlebox*}{Reduction \normalfont #1}{mycommonbox}{normal}{#2}{#3}}
		{\end{titlebox*}
\end{strip}}

\newenvironment{mygamebox*}[3]
{\begin{strip}
		\begin{titlebox*}{Game \normalfont #1}{mycommonbox}{gray}{#2}{#3}}
		{\end{titlebox*}
\end{strip}}

\newenvironment{mysimulatorbox*}[3]
{\begin{strip}
		\begin{titlebox*}{Simulator \normalfont #1}{mycommonbox}{normal}{#2}{#3}}
		{\end{titlebox*}
\end{strip}}

\newcommand{\algoHead}[1]{\vspace{0.2em} \underline{\textbf{#1}} \vspace{0.1em}}

\makeatletter
\algnewcommand{\ExtendedState}[1]{\State
	\parbox[t]{\dimexpr\linewidth-\ALG@thistlm}{\hangindent=\algorithmicindent\strut\hangafter=3#1\strut}}
\makeatother

\algnewcommand\algorithmicinput{\textbf{Input:}}
\algnewcommand\Input{\item[\algorithmicinput]}

\algrenewcommand{\algorithmiccomment}[1]{{\color{gray}// #1}}

\newcommand{\xmath}[1]{\ensuremath{#1}\xspace}

\newcommand{\Func}[1][\relax]{\xmath{\mathcal{F}_{\textsc{#1}}}}

\DeclarePairedDelimiterX{\dbrackets}[1]{\lbrack}{\rbrack}{
	\nhphantom{\lbrack}{#1} \delimsize\lbrack \mathopen{} #1 \mathclose{} \delimsize\rbrack \nhphantom{\rbrack}{#1}
}
\DeclarePairedDelimiterX{\dbraces}[1]{\lbrace}{\rbrace}{
	\nhphantom{\lbrace}{#1} \delimsize\lbrace \mathopen{} #1 \mathclose{} \delimsize\rbrace \nhphantom{\rbrace}{#1}
}
\DeclarePairedDelimiterX{\dparens}[1]{\lparen}{\rparen}{
	\nhphantom{\lparen}{#1} \delimsize\lparen \mathopen{} #1 \mathclose{} \delimsize\rparen \nhphantom{\rparen}{#1}
}

\newcommand{\nhphantom}[2]{\sbox0{$\left#1\vphantom{#2}\right.$}\hspace{-0.58\wd0}}

\newcommand{\myparagraph}[1]{\medskip\noindent{}\textbf{#1.}}

\newcommand{\vl}[1]{\ensuremath{\mathsf{#1}}}

\newcommand{\pd}[1]{\ensuremath{\mathsf{\lambda}_{#1}}}
\newcommand{\padR}[1]{\ensuremath{\mathsf{\lambda}^{\sf R}_{#1}}}
\newcommand{\pad}[2]{\ensuremath{\mathsf{\lambda}_{#1}^{#2}}} 
\newcommand{\mk}[1]{\ensuremath{\mathsf{m}_{#1}}}

\newcommand{\gm}[2]{\ensuremath{\gamma_{#1}^{#2}}}

\newcommand{\vct}[1]{\ensuremath{\vec{\mathbf{#1}}}}
\newcommand{\Mat}[1]{\ensuremath{\mathbf{#1}}}

\newcommand{\sqr}[1]{\ensuremath{\left[#1\right]}}
\newcommand{\spr}[1]{\ensuremath{\dparens*{#1}}}
\newcommand{\sgr}[1]{\ensuremath{\langle #1 \rangle}}
\newcommand{\shr}[1]{\ensuremath{\llbracket #1 \rrbracket}}

\newcommand{\shrB}[1]{\ensuremath{{\llbracket #1 \rrbracket}^{\bf B}}}
\newcommand{\shrG}[1]{\ensuremath{{\llbracket #1 \rrbracket}^{\bf G}}}

\newcommand{\TTP}{\ensuremath{\mathsf{TTP}}}
\newcommand{\ttp}{\mathsf{ttp}}

\newcommand{\VrfyP}{\ensuremath{\mathsf{VrfyP0}}}

\newcommand{\piMultR}{\ensuremath{\Pi_{\Mult}^{\mathsf{R}}}}

\newcommand{\piMultRT}{\ensuremath{\Pi_{\MultT}^{\mathsf{R}}}}

\newcommand{\piVrfyP}{\ensuremath{\Pi_{\VrfyP}}}

\newcommand{\FSETUP}{\ensuremath{\mathcal{F}_{\mathsf{setup}}}} 
\newcommand{\FZero}{\ensuremath{\mathcal{F}_{\mathsf{zero}}}} 

\newcommand{\SetD}{\ensuremath{\mathcal{D}}}
\newcommand{\SetE}{\ensuremath{\mathcal{E}}}

\newcommand{\GS}{\ensuremath{\mathcal{G}}}
\newcommand{\Gb}{\ensuremath{\mathsf{Gb}}}
\newcommand{\En}{\ensuremath{\mathsf{En}}}
\newcommand{\Ev}{\ensuremath{\mathsf{Ev}}}
\newcommand{\De}{\ensuremath{\mathsf{De}}}

\newcommand{\rtt}{\ensuremath{\mathsf{rtt}}}
\newcommand{\TP}{\ensuremath{\mathsf{TP}}}

\newcommand{\PlSet}[1]{\ensuremath{\Phi_{#1}}}
\newcommand{\shrC}[1]{\ensuremath{{\llbracket #1 \rrbracket}^{\bf C}}}
\newcommand{\av}[1]{\ensuremath{\mathsf{\alpha}_{\mathsf{#1}}}}
\newcommand{\key}[2]{\ensuremath{\mathsf{K}_{#1}^{#2}}}
\newcommand{\pigsh}{\ensuremath{\mathrm{\Pi}_{\mathsf{Sh}}^{\bf G}}}
\newcommand{\GC}{\ensuremath{\mathsf{GC}}}
\newcommand{\h}{\ensuremath{\mathcal{H}}}
\newcommand{\pigrec}{\ensuremath{\mathrm{\Pi}_{\mathsf{Rec}}^{\bf G}}}
\newcommand{\pigfrec}{\ensuremath{\mathrm{\Pi}_{\mathsf{fRec}}^{\bf G}}}
\newcommand{\piab}{\ensuremath{\mathrm{\Pi}_{\mathsf{A2B}}}}
\newcommand{\arval}[1]{\ensuremath{#1^{\sf R}}} 
\newcommand{\piba}{\ensuremath{\mathrm{\Pi}_{\mathsf{B2A}}}}
\newcommand{\Y}{\ensuremath{\mathbf{Y}}}
\newcommand{\X}{\ensuremath{\mathbf{X}}}
\newcommand{\Ckt}{\ensuremath{\mathsf{Ckt}}}
\newcommand{\onesec}{\ensuremath{1^\kappa}}
\newcommand{\poly}{\ensuremath{\mathsf{poly}}}
\newcommand{\ppt}{\textsf{PPT}}
\newcommand{\priv}{\ensuremath{\mathsf{\bf priv}}}
\newcommand{\Real}{\ensuremath{\textsc{real}}}
\newcommand{\Ideal}{\ensuremath{\textsc{ideal}}}
\newcommand{\Grb}[1]{\ensuremath{\mathsf{G}^{\sf #1}}}
\newcommand{\GrbD}[1]{\ensuremath{\hat{\mathsf{G}}^{\sf #1}}}
\newcommand{\Size}[1]{\ensuremath{|#1|}}

\newcommand{\Key}[1]{\ensuremath{k_{#1}}}
 
\newcommand{\msg}{\mathsf{msg}}

\newcommand{\COT}[2]{\ensuremath{\mathsf{cOT}^{#1}_{#2}}} 
 
\newcommand{\OTN}[3]{\ensuremath{\mathsf{#3\text{-}OT}^{#1}_{#2}}} 

\fancypagestyle{mainmatter-pages}{%
	\fancyhf{}
	\fancyfoot[L]{\hyperref[ToC-first-page]{\footnotesize Jump to Contents}}
	\fancyfoot[C]{\thepage}
}
\let\oldtableofcontents\tableofcontents
\renewcommand{\tableofcontents}{%
	\cleardoublepage
	\phantomsection
	\label{ToC-first-page}
	\oldtableofcontents
}
\let\oldmainmatter\mainmatter
\renewcommand{\mainmatter}{%
	\cleardoublepage
	\oldmainmatter
	\pagestyle{mainmatter-pages}%
}

\begin{document}
\title{\papertitle} 

\submitdate{July, 2021} 
\phd
\dept{Computer Science and Automation}
\faculty{Faculty of Engineering}
\author{Ajith Suresh}

\maketitle


\begin{center}
	\LARGE{\underline{\textbf{Declaration of Originality}}}
\end{center}
\noindent I, \textbf{Ajith Suresh}, with SR No. \textbf{04-04-00-10-12-17-1-14980} hereby declare that
the material presented in the thesis titled

\begin{center}
	\textbf{\papertitle}
\end{center}

\noindent represents original work carried out by me in the \textbf{Department of Computer Science and Automation} at \textbf{Indian Institute of Science} during the years \textbf{2017-2021}.

\noindent With my signature, I certify that:
\begin{itemize}
	\item I have not manipulated any of the data or results.
	\item I have not committed any plagiarism of intellectual
	property.
	I have clearly indicated and referenced the contributions of
	others.
	\item I have explicitly acknowledged all collaborative research
	and discussions.
	\item I have understood that any false claim will result in severe
	disciplinary action.
	\item I have understood that the work may be screened for any form
	of academic misconduct.
\end{itemize}

\vspace{15mm}

\noindent {\footnotesize{Date: 28$^{th}$ July, 2021	\hfill	Student Signature}} \qquad

\vspace{20mm}

\noindent In my capacity as supervisor of the above-mentioned work, I certify
that the above statements are true to the best of my knowledge, and 
I have carried out due diligence to ensure the originality of the
report.

\vspace{15mm}

\noindent  {\footnotesize{Advisor Name: Arpita Patra \hfill Advisor Signature}} \qquad

\blankpage

\vspace*{\fill}
\begin{center}
	\large\bf \textcopyright \ Ajith Suresh\\
	\large\bf July, 2021\\
	\large\bf All rights reserved
\end{center}
\vspace*{\fill}
\thispagestyle{empty}

\blankpage

\vspace*{\fill}
\begin{center}
	\Large DEDICATION \\[3em]
	\Large\it I dedicate my sincere efforts to my sweet and loving\\[1em]
	\Large\it \textbf{Achan,\ Amma\ \&\ Devu}\\[1em]
	\Large\it whose love, affection and words of encouragement guided me in achieving success and honor,\\[2em]
	\Large\it Along with all my beloved\\
	\Large\it \textbf{Teachers}\\
	\Large\it for being a great source of inspiration\\
\end{center}
\vspace*{\fill}
\thispagestyle{empty}

\setcounter{secnumdepth}{3}
\setcounter{tocdepth}{2}

\frontmatter 
\pagenumbering{roman}

\prefacesection{Acknowledgements}
I begin by thanking the Almighty for blessing me with the strength to fulfil my commitment to research during my PhD.

It gives me immense pleasure to express my deep sense of gratitude and respect to my research supervisor, Professor Arpita Patra, to accept me into her family. I joined this family in 2014 as an M.Tech (Research) student, and she has been a constant source of inspiration ever since. I would love to address her as my elder sister in place of my advisor. It has been a great experience to work under the supervision of Prof. Arpita, who treats her students as her family, bridging the gap between faculty and student. Her support and guidance helped me all along during my research and in completing my thesis. 

I am greatly indebted to Ashish Choudhury, Professor at IIIT Bangalore, for being a significant contributor in kick-starting my research. It is one of the questions posed by him during my initial journey, which led to the line of research that constitutes the contributions of this thesis. His genuine kindness, warmth and strong work ethic have left a lasting impression in my life, which I will cherish forever. 

The CrIS lab has been like my second home, and I am grateful to all my lab mates with whom I have had a great time. I would like to extend warm gratitude to Divya Ravi and Megha Byali for being my strong support system on both academic and personal fronts. Also, a special mention to Nishat Koti, with whom I have shared a lot of time working on research problems.  

During my studies, I got an opportunity to visit the ENCRYPTO group at the Technical University of Darmstadt. I owe my deepest gratitude to Prof. Thomas Schneider for hosting me and giving me a chance to interact with his group. These internships have played a key role in boosting my confidence, and my interactions with them have helped me develop new insights and expand my research horizon. I would like to use this opportunity to thank Rahul Rachuri and Hossein Yalame for co-authoring some of my works. 

I am thankful to Ms Padmavathi and Ms Kushael at the CSA office for their constant help and support on the administrative front. I am also thankful to Prof. Deepak D'Souza and Prof. Bhavana Kanukurthi of CSA, IISc, for their support during my initial PhD journey. 

I am grateful for the Google PhD fellowship for supporting my research. My sincere acknowledgements to IACR, NDSS, Alan Turing Institute, and the institute GARP sponsorship for supporting my travel to international conferences. 

My heartfelt thanks to all the teachers who have taught me since first grade. I'm also grateful to my undergraduate, graduate and post-graduate family. Finally, I would love to express profound gratitude to my parents and my wife for encouraging me with unfailing support and advice throughout my research. Their unwavering love and affection mean the world to me, and I dedicate this thesis to them.


\prefacesection{Abstract}
In the modern era of computing, machine learning tools have demonstrated their potential in vital sectors, such as healthcare and finance, to derive proper inferences. The sensitive and confidential nature of the data in such sectors raises genuine concerns for data privacy. This motivated the area of Privacy-preserving Machine Learning~(PPML), where privacy of data is guaranteed. Typically, machine learning techniques require significant computing power, which leads clients with limited infrastructure to rely on the method of Secure Outsourced Computation~(SOC). In the SOC setting, the computation is outsourced to a set of specialized and powerful cloud servers and the service is availed on a pay-per-use basis.  In this thesis, we design an efficient platform, MPCLeague, for PPML in the SOC setting using Secure Multi-party Computation~(MPC) techniques.

MPC, the holy-grail problem of secure distributed computing, enables a set of $n$ mutually distrusting parties to perform joint computation on their private inputs in a way that no coalition of $t$ parties can learn more information than the output (privacy) or affect the true output of the computation (correctness). While MPC, in general, has been a subject of extensive research, the area of MPC with a small number of parties has drawn popularity of late mainly due to its application to real-time scenarios, efficiency and simplicity. This thesis focuses on designing efficient MPC frameworks for 2, 3 and 4 parties, with at most one corruption and supports ring structures. 

Our platform aims at achieving the most substantial security notion of robustness, where the honest parties are guaranteed to obtain the output irrespective of the behaviour of the corrupt parties. A robust protocol prevents the corrupt parties from repeatedly causing the computations to rerun, thereby upholding the trust in the system. While on the roadmap to attain robustness, our frameworks also demonstrate constructions with improved performance that achieve relaxed notions of security: security with abort and fairness. A fair protocol enforces the restriction that either all parties or none of them receive the output. On the other hand, honest parties may not receive the output while corrupt parties do for the case of security with abort.

The general structure of the computation involves the execution of the protocol steps once the participating parties have supplied their inputs. Finally, the output is distributed to all the parties. However, to enhance practical efficiency, many recent works resort to the preprocessing paradigm, which splits the computation into two phases; a preprocessing phase where input-independent (but function-dependent), computationally heavy tasks can be computed, followed by a fast online phase. Since the same functions in ML are evaluated several times, this paradigm naturally fits the case of PPML, where the ML algorithm is known beforehand.

At the heart of this thesis are four frameworks -- $\TSthis, \Tthis, \Fthis, \TWthis$ - catered to different settings.

\begin{itemize}
	\item[--] $\TSthis$: We begin with the setting of 3 parties~(3PC), which forms the base case for honest majority. If a majority of the participating parties are honest, then the setting is deemed an honest majority setting. In the set of 3 parties, at most one party can be corrupt, and this framework tackles semi-honest corruption, where the corrupt party follows the protocol steps but tries to glean more information from the computation. $\TSthis$ acts as a stepping stone towards achieving a stronger security guarantee against active corruption. Our protocol requires communication of $2$ ring elements per multiplication gate during the online phase, attaining a per-party cost of less than one element. This is achieved for the first time in the regime of 3PC.
	\item[--] $\Tthis$: Designed for 3 parties, this framework tackles one active corruption where the corrupt party can arbitrarily deviate from the computation. Building on $\TSthis$, $\Tthis$ provides a multiplication that improves the communication to $6$ ring elements from $21$ over the state-of-the-art, besides improving security from abort to robustness. In the regime of malicious 3PC, $\Tthis$ is the first robust and efficient PPML framework. It achieves a dot product protocol with communication independent of the vector size for the first time.
	\item[--] $\Fthis$: Designed for 4 parties in the honest majority, the fair multiplication protocol in $\Fthis$ requires communication of only $5$ ring elements instead of $6$ in the state-of-the-art. The fair framework is then extended to provide robustness without inflating the costs. A notable contribution is the design of the multiplication protocol that supports on-demand applications where the function to be computed is not known in advance.
	\item[--] $\TWthis$: Moving on to the stronger corruption model where a majority of the parties can be corrupt, we explore the base case of 2 parties~(2PC). Since we aim to achieve robustness which is proven to be impossible in active corruption, we restrict ourselves to semi-honest corruption. The prime contribution of this framework is the scalar product for which the online communication is two ring elements irrespective of the vector dimension. This is a feature achieved for the first time in the 2PC literature. 
\end{itemize}

Our frameworks provide the following contributions in addition to the ones mentioned above. First, we support multi-input multiplication for arithmetic and boolean worlds, improving the online phase in rounds and communication. Second, all our frameworks except $\Tthis$, incorporate truncation without incurring any overhead. Finally, we introduce efficient instantiation of garbled-world, tailor-made for the mixed-protocol framework for the first time. The mixed-protocol approach, combining arithmetic, boolean and garbled style computations, has demonstrated its potential in several practical use-cases like PPML. To facilitate the computation, we also provide the conversion mechanisms to switch between the computation styles.

The practicality of our framework is argued through improvements in the benchmarking of widely used ML algorithms -- Linear Regression, Logistic Regression, Neural Networks, and Support Vector Machines. We propose two variants for each of our frameworks, with one variant aiming to minimise the execution time while the other focuses on the monetary cost.

The concrete efficiency gains of our frameworks coupled with the stronger security guarantee of robustness make our platform an ideal choice for a real-time deployment of privacy-preserving machine learning techniques.

\prefacesection{Publications based on this Thesis}
The work in this dissertation is primarily related to the following articles. 
Publications in cryptography usually order authors alphabetically (using surnames) and conferences are more common than journals\footnote{Workshops without proceedings are marked in \textcolor{UniRot}{this} colour.}. 

\paragraph{Accepted Papers}
\begin{small}
\begin{enumerate}
	\item Nishat Koti, Arpita Patra, Rahul Rachuri and {\bf Ajith Suresh}.
	\ptitle{Tetrad: Actively Secure 4PC for Secure Training and Inference} \cite{EPRINT:KPRS21}. 
	\published{https://www.ndss-symposium.org/ndss2022/}{NDSS'22} \CORE{A*},
	\publishedW{https://ppml-workshop.github.io}{PPML'21 (CCS)}
	\item Nishat Koti, Mahak Pancholi, Arpita Patra and {\bf Ajith Suresh}.
	\ptitle{SWIFT: Super-fast and Robust Privacy-Preserving Machine Learning} \cite{USENIX:KPPS21}. 
	\published{https://www.usenix.org/conference/usenixsecurity21}{USENIX Security'21} \CORE{A*}, \publishedW{https://ppml-workshop.github.io}{PRIML/PPML'20 (NeurIPS)},
	\publishedW{https://dp-ml.github.io/2021-workshop-ICLR}{DPML'21 (ICLR)}
	\item Arpita Patra, Thomas Schneider, {\bf Ajith Suresh} and Hossein Yalame.
	\ptitle{ABY2.0: Improved Mixed-Protocol Secure Two-Party Computation} \cite{USENIX:PSSY21}.
	\published{https://www.usenix.org/conference/usenixsecurity21}{USENIX Security '21} \CORE{A*},
	\publishedW{https://ppml-workshop.github.io}{PPML'21 (CCS)},
	\publishedW{https://priml2021.github.io}{PriML'21 (NeurIPS)},
	\publishedW{https://crypto-ppml.github.io/2021/}{PPML'21 (CRYPTO)}
	\item Nishat Koti, Arpita Patra and {\bf Ajith Suresh}.
	\ptitle{MPCLeague: Robust and Efficient Mixed-protocol Framework for 4-party Computation} \cite{DPML:KPS21}. 
	\published{https://www.ieee-security.org/TC/SP2021/downloads/poster/poster25.pdf}{IEEE S\&P 2021 (Poster)},
	\publishedW{https://dp-ml.github.io/2021-workshop-ICLR}{DPML'21 (ICLR)}
	\item Harsh Chaudhari, Rahul Rachuri and {\bf Ajith Suresh}.
	\ptitle{Trident: Efficient 4PC Framework for Privacy Preserving Machine Learning} \cite{NDSS:ChaRacSur20}. \published{https://www.ndss-symposium.org/ndss2020/}{NDSS'20} \CORE{A*}
	\item Arpita Patra and {\bf Ajith Suresh}.
	\ptitle{BLAZE: Blazing Fast Privacy-Preserving Machine Learning} \cite{NDSS:PatSur20}.
	\published{https://www.ndss-symposium.org/ndss2020/}{NDSS'20} \CORE{A*}
	\item Harsh Chaudhari, Ashish Choudhury, Arpita Patra and {\bf Ajith Suresh}.
	\ptitle{ASTRA: High Throughput 3PC over Rings with Application to Secure Prediction} \cite{CCSW:CCPS19}.
	\published{https://ccsw.io/}{ACM CCSW'19}, \publishedW{https://ppml-workshop.github.io/ppml19/}{PPML'19 (CCS)}
\end{enumerate}
\end{small}

\prefacesection{Publications outside this Thesis}
\paragraph{Accepted Papers}
\begin{small}
\begin{enumerate}
	\item Arpita Patra, Thomas Scheider, Ajith Suresh and Hossein Yalame.
	\ptitle{SynCirc: Efficient Synthesis of Depth‐Optimized Circuits for Secure Computation} \cite{HOST:PatraSSY21}.
	\published{http://www.hostsymposium.org}{IEEE HOST'21}
	\item Megha Byali, Harsh Chaudhari, Arpita Patra and Ajith Suresh.
	\ptitle{FLASH: Fast and Robust Framework for Privacy‐preserving Machine Learning} \cite{PoPETS:BCPS20}.
	\published{https://petsymposium.org/2020/index.php}{PoPETS'20}
	\CORE{A}
\end{enumerate}
\end{small}

\paragraph{Preprints / Manuscripts}
\begin{small}
\begin{enumerate}
	\item Nishat Koti, Shravani Patil, Arpita Patra and Ajith Suresh.
	\ptitle{MPClan: Protocol Suite for Privacy-Conscious Computations}.
	\undersub
\end{enumerate}
\end{small}

\tableofcontents
\listoffigures
\listoftables

\mainmatter 
\setcounter{page}{1}

\chapter{Introduction}
\label{chap:introduction}
With the advent of the contemporary era of computing, machine learning techniques have proven their mettle in diverse sectors, such as finance and healthcare, that involve multi-party computation (MPC) to derive genuine inferences. 
Increased concerns about privacy coupled with policies such as European Union General Data Protection Regulation~(GDPR) make it harder for multiple parties to collaborate on machine learning~(ML) computations. The emerging field of privacy-preserving machine learning~(PPML) addresses this issue by offering tools to let parties perform computations without sacrificing the privacy of the underlying data. PPML can be deployed across various domains such as healthcare, recommendation systems, etc., with works like \cite{CORR:AVBCG20} demonstrating practicality. 

The primary challenge that inhibits widespread adoption of PPML is that the additional demand on privacy makes the already compute-intensive ML algorithms all the more demanding in terms of high computing power and other complexity measures such as communication complexity that the privacy-preserving techniques entail. Many everyday end-users are not equipped with computing infrastructure capable of efficiently executing these algorithms. It is economical and convenient for end-users to outsource an ML task to more powerful and specialized systems. However, even while outsourcing to servers, the privacy of data must be ensured. This is addressed by the Secure Outsourced Computation~(SOC) paradigm and thus is an apt fit for the moment's need. SOC allows end-users to securely outsource computation to a set of specialized and powerful cloud servers and avail of its services on a pay-per-use basis. SOC guarantees that individual data of the end-users remain private, tolerating reasonable collusion amongst the servers.
Both the training and prediction phases of PPML can be realized in the SOC setting. The common approach of outsourcing followed in the PPML literature, as well as by our work,  requires the users to secret-share\footnote{The threshold of the secret-sharing is decided based on the number of corrupt servers so that privacy is preserved.} their inputs between the set of hired (untrusted) servers, who jointly interact and compute the secret-shared output, and reconstruct it towards the users. Of late, MPC based techniques~\cite{SP:MohZha17,CCS:MohRin18,ASIACCS:RWTSSK18,PoPETS:WagGupCha19,RSA:MRSV19, CCSW:CCPS19,PoPETS:BCPS20,NDSS:ChaRacSur20,NDSS:PatSur20} have been gaining interest, where a server enacts the role of a party in the MPC protocol. 

MPC~\cite{FOCS:Yao82b,STOC:GolMicWig87,STOC:BenGolWig88},  the holy-grail problem of secure distributed computing,  enables a set of $n$ mutually distrusting parties to perform joint computation on their private inputs in a way that no coalition of $t$ parties can learn more information than the output (privacy) or affect the true output of the computation (correctness). The distrust among the parties is formalized by having an {\em adversary} that may corrupt some of the parties.  We usually consider a {\em monolithic or centralized} adversary, i.e., if two or more parties are corrupted, we assume that they collude with each other. We denote the corruption threshold of the adversary by $t$. Under the adversary's control, the parties are called "corrupt", and the remaining parties are called "honest". This thesis focuses on designing efficient MPC frameworks for 2, 3 and 4 parties,  with at most one corruption.

\vspace{-2mm}
\section{System Model}
\label{sec:systemmodel}

\paragraph{Adversarial Model}
\label{sec:adv}
The various traits of the adversary introduce several unique settings where MPC is explored in the literature. This thesis considers a static adversary that decides on the set of $t$ parties it would corrupt before the protocol begins. Moreover, the adversary is computationally bounded, meaning that it is restricted to run within probabilistic polynomial time. Based on the type of corruption, an adversary can be primarily categorized into two: i) {\em passive / semi-honest} - where the corrupt parties follow the protocol specifications but try to learn more information than what is allowed as per the security guarantees of the protocol, and ii) {\em active/malicious} - where the adversary exercises total control over the corrupt parties who may deviate from the protocol steps in any arbitrary manner. 

\vspace{-2mm}
\paragraph{High-throughput vs Low-latency MPC}
MPC protocols can be categorized as high-throughput~\cite{CCS:AFLNO16,EC:FLNW17,SP:ABFLLN17,CCS:MohRin18,CCSW:CCPS19,EPRINT:ADEN19,NDSS:ChaRacSur20,NDSS:PatSur20,USENIX:KPPS21,USENIX:PSSY21} and low-latency~\cite{CCS:MohRosZha15,C:PatRav18,CCS:BJPR18,CCS:BHPS19} protocols. The low-latency protocols are built using garbled circuits (GC)~\cite{FOCS:Yao86,STOC:BeaMicRog90,ICALP:KolSch08,EC:ZahRosEva15} and result in constant-round solutions. Secret-sharing (SS) based solutions have been used for high-throughput protocols, but require a number of communication rounds linear in the multiplicative depth of the circuit. However, less communication than GC-based protocols facilitates several instances of SS-based protocols to be executed in parallel, leading to high throughput. 
While high-throughput protocols enable efficient computation of functions such as addition, multiplication and dot-product, other functions such as division are best performed using garbled circuits. Activation functions such as ReLU used in neural networks~(NN) alternate between multiplication and comparison, wherein multiplication is better suited to the arithmetic world and comparison to the boolean world. Hence, MPC protocols working over different representations (arithmetic/boolean/garbled circuit based) can be mixed to achieve better efficiency. 
The characteristics of the categories mentioned above put forth the need for a mixed-protocol framework~\cite{NDSS:DemSchZoh15,SP:MohZha17,CCS:MohRin18,ASIACCS:RWTSSK18,INDOCRYPT:RotWoo19,NDSS:ChaRacSur20,C:EGKRS20,USENIX:PSSY21}, where the protocol is split into blocks. Each block is executed in one of the following three worlds: i) Arithmetic, ii) Boolean,  and iii) Garbled. While the arithmetic world performs operations on $\ell$-bit rings (or fields), both boolean and garbled world perform operations on bits.  Also, arithmetic and boolean worlds operate using an SS-based approach, while the garbled world uses a GC-based approach. 

Almost all high-throughput protocols evaluate a circuit that represents the function $f$ to be computed in a  secret-shared fashion. Informally,  the parties jointly maintain the invariant that for each wire in the circuit, the exact value over that wire is available in a secret-shared fashion among the parties so that the adversary learns no information about the exact value from the shares of the corrupt parties. Upon completion of the circuit evaluation, the parties jointly reconstruct the secret-shared function output. Intuitively, the security holds as no intermediate value is revealed during the computation. The deployed secret-sharing schemes are typically linear, ensuring non-interactive evaluation of the linear gates. The communication is required {\em only} for the non-linear  (i.e.multiplication) gates in the circuit. The focus then turns on improving the communication overhead per multiplication gate.  Recent literature has seen a  range of customized linear secret-sharing schemes over a small number of parties, boosting the performance for multiplication gate spectacularly.

\vspace{-2mm}
\paragraph{Pre-processing Paradigm}
To enhance practical efficiency,  MPC protocols resort to the pre-processing paradigm, which splits the computation into two phases; a pre-processing phase where input-independent (but function-dependent), computationally heavy tasks can be computed, followed by a fast online phase utilizing the pre-processing computation~\cite{C:Beaver91b}. Since the same functions in ML are evaluated several times, this paradigm naturally fits the case of PPML, where the ML algorithm is known beforehand. The parties can batch together the pre-computations and generate a large volume of pre-processing data to support the execution of multiple online phases. There are constructions abound that show effectiveness of this paradigm both in the theoretical \cite{C:Beaver91b,TCC:BeeHir06,TCC:BeeHir08,C:BenFehOst12,CP17} and practical  \cite{C:DPSZ12,NDSS:DemSchZoh15,CCS:KelOrsSch16,C:DamOrlSim18,EC:KelPasRot18,NDSS:ChaRacSur20,NDSS:PatSur20} regime. 

\paragraph{Fields vs Rings}
In yet another direction to improve practical efficiency,  secure computation for arithmetic circuits over rings has gained momentum of late, while traditionally, fields have been the default choice.  Computation over rings models computation in real-life computer architectures such as computation over CPU words of 32 or 64 bits. Moreover, operating over rings eliminates the need for external libraries to operate over fields ($10\times$-$100\times$ slower) than real-world system architectures based on 32-bit and 64-bit rings. The benchmarking results of \cite{PPDL} and the works of \cite{EC:CFIK03,ESORICS:BogLauWil08,NDSS:DemSchZoh15,C:DamOrlSim18} have showcased the efficiency improvements of protocols compared to rings over their field counterparts. Further, recent works~\cite{CCS:KelOrsSch16,C:CDESX18,SP:DEFKSV19,C:EGKRS20,CCS:Keller20} propose MPC protocols over $32$ or $64$ bit rings to leverage CPU optimizations. 

\paragraph{Security Guarantees}
Works such as ~\cite{CCS:MohRin18,PoPETS:WagGupCha19,USENIX:MLRG20} typically go for active security with abort, where the adversary can act maliciously to obtain the output and make honest parties abort. The stronger notion of fairness guarantees that either all or none of the parties obtain the output. This provides an incentive to the adversary to behave honestly in resources-expensive tasks such as PPML, as creating an abort scenario to cause a rerun will waste its resources. In cases where the risk of failure for the system is too high, for instance, when deploying PPML for healthcare applications, participants might want to avoid the case when none of them receives the output. The way to tackle this issue is to modify protocols to guarantee that the correct output is always delivered to the participants irrespective of an adversary's misbehaviour. This is provided by guaranteed output delivery (GOD) or robustness. A robust protocol prevents the adversary from repeatedly causing the computations to rerun, thereby upholding the trust in the system. 

Robustness is crucial for real-world deployment and usage of PPML techniques. Consider the following scenario wherein an ML model owner wishes to provide inference service. The model owner shares the model parameters between the servers, while the end-users share their queries. A protocol that provides security with abort or fairness will not suffice. In both cases, a malicious adversary can lead to the protocol aborting, resulting in the user not obtaining the desired output. This leads to denial of service and heavy economic losses for the service provider. For data providers, as more training data leads to more accurate models, collaboratively building a model enables them to provide better ML services, and consequently, attract more clients. A robust framework encourages active involvement from multiple data providers. Hence, for the seamless adoption of PPML solutions in the real world, the protocol's robustness is of utmost importance. 

\paragraph{MPC for small number of parties}
\label{sec:smallpopulation}
While MPC, in general, has been a subject of extensive research, the area of MPC with a small number of parties~\cite{CCS:MohRosZha15,NDSS:DemSchZoh15,CCS:AFLNO16,SP:MohZha17,CCS:CGMV17,CCS:MohRin18,CCS:BJPR18} has drawn popularity of late mainly due to its efficiency and simplicity. Furthermore, most real-time applications involve up to 5 parties. Applications such as statistical and financial data analysis \cite{FC:BogTalWil12}, email-filtering \cite{LaunchburyADM14}, distributed credential encryption \cite{CCS:MohRosZha15}, Danish sugar beet auction \cite{FC:BCDGJK09} involve 3 parties. Well-known  MPC frameworks such as VIFF \cite{Gei07}, Sharemind \cite{ESORICS:BogLauWil08} have been explored with three parties. Recent advances in secure machine learning  (ML) based on  MPC  have shown  applications with small number of parties~\cite{SP:MohZha17,CCS:MohRin18,ASIACCS:RWTSSK18,PoPETS:WagGupCha19,RSA:MRSV19, CCSW:CCPS19,PoPETS:BCPS20,NDSS:ChaRacSur20,NDSS:PatSur20,USENIX:PSSY21}. MPC with small parties aids in solving MPC over a large population via server-aided computation, where a small number of servers jointly hold the input data of the large population and run an MPC protocol evaluating the desired function.

Our protocols designed for 2, 3 and 4 parties operating over rings are cast in the pre-processing paradigm and achieve robustness. Before moving on to the contributions of the thesis, we outline the relevant literature next.

\section{Related Work}
\label{sec:related}
In the regime of PPML using MPC, the initial works considered the widely-used ML algorithms such as Decision Trees~\cite{C:LinPin00}, K-Means Clustering~\cite{JagannathanW05, CCS:BunOst07}, Support Vector Machines~\cite{YuVJ06, VaidyaYJ08}, Linear Regression~\cite{DuA01,DuHC04,SanilKLR04} and Logistic Regression~\cite{SlavkovicNT07}. However, these solutions are far from practical reach due to the huge performance overheads that they incur. We next discuss the literature concerning the following three algorithms -- Linear Regression, Logistic Regression, and Neural Networks, which are the focus of this thesis. The initial set of practical solutions for these algorithms were proposed in the dishonest majority (two-party) setting and are discussed below.

{\em Linear Regression:} Privacy-preserving linear regression on the two server model was first proposed by Nikolaenko et al.~\cite{SP:NWIJBT13}. Their solution focused on horizontally partitioned data and used a combination of linearly homomorphic encryption (LHE) and garbled circuits. Later, Gascon et al.~\cite{EPRINT:GSBRDZ16} and Giacomelli et al.\cite{ACNS:GJJPY18} extended these results to vertically partitioned data. Both papers, however, confine the problem to solving a linear system using Yao's garbled circuit protocol, which has a substantial training time overhead and cannot be applied to non-linear models. SecureML~\cite{SP:MohZha17} then used stochastic gradient descent (SGD) for training, as well as a mix of arithmetic, binary, and Yao sharing (using the ABY~\cite{NDSS:DemSchZoh15} framework) over two parties, to increase the performance of linear regression over horizontally partitioned data. Furthermore, they present a unique design for approximation fixed-point multiplication that avoids boolean operations for truncating decimal numbers while providing state-of-the-art performance for training linear regression models.

{\em Logistic Regression:} Wu et al.~\cite{Wu2013FS} explored privacy-preserving logistic regression and proposed approximating the logistic function with polynomials and training the model with LHE, with the complexity being exponential in the degree of the approximation polynomial. Aono et al.~\cite{EPRINT:AHPW16} considered a different security model where an additional untrusted server collects and mixes encrypted data from several clients and delivers it to a trusted client who trains the model on the plaintext on clear. 

{\em Neural Networks:} Privacy-preserving solutions for neural networks have also been studied.  For the case of training, Shokri and Shmatikov~\cite{CCS:ShoShm15} proposed a scheme where the two servers locally train their model using the horizontally partitioned data. Instead of exchanging the training data, they only share the changes in a portion of the coefficients in the locally trained model. Although the system is very efficient (no cryptographic operations are required), the leakage resulting from sharing these coefficient changes remains unclear, and no formal security guarantees are provided. The privacy-preserving training of neural networks was also considered in the work of SecureML~\cite{SP:MohZha17}, where the ABY framework was customized to achieve a new approximate fixed-point multiplication protocol that avoids binary circuits. For the case of inference, the works of ~\cite{Gilad-BachrachD16,HTG17,EPRINT:CWMMP17,C:BMMP18} consider fully homomorphic or somewhat homomorphic encryption to evaluate the model on encrypted data, while~\cite{CCS:LJLA17,EPRINT:RouRiaKou17} uses a combination of LHE and garbled circuits. 

Departing from the dishonest majority setting, a performance breakthrough in the above-mentioned PPML algorithms was observed in ABY3~\cite{CCS:MohRin18}, which explored the honest majority setting for three parties. After that, a plethora of works followed, such as ~\cite{CCSW:CCPS19,PoPETS:WagGupCha19,NDSS:PatSur20,NDSS:ChaRacSur20,PoPETS:BCPS20,PoPETS:WTBKMR21,USENIX:KPPS21,EPRINT:DalEscKel20,EPRINT:KPRS21}, which explored the setting of small population with honest-majority and showcased real-time efficiency even for complex neural-network architectures such as LeNet~\cite{lenet} and VGG16~\cite{vgg16}. 

While the literature above tackles only the line of works in PPML via MPC, other dimensions such as differential privacy, model attacks and defense mechanisms, etc., are relevant. However, the literature elaborating on the line of development in these areas is quite vast to be briefly explained in this section, and we refer the reader to~\cite{SWRH20,MTVSRE20,LiuXWZXYV21,PETS:CP21} for a detailed overview of the same. Next, we provide an elaborate summary of the most relevant related work that focuses on MPC frameworks for PPML. 

\paragraph{Honest Majority}

ABY3~\cite{CCS:MohRin18} was the first framework for the case of 3 parties, supporting both training and inference. It had variants for both passive and active security, with the former being based on~\cite{CCS:AFLNO16} and the latter on~\cite{EC:FLNW17,SP:ABFLLN17}. ASTRA~\cite{CCSW:CCPS19} improved upon the 3PC of~\cite{CCS:AFLNO16,EC:FLNW17,SP:ABFLLN17} by proposing faster protocols for the online phase with active security. As a result, secure inference of ASTRA is faster than ABY3. Building on~\cite{C:BBCGI19}, BLAZE~\cite{NDSS:PatSur20} proposed an actively secure framework that supports the inference of neural networks. BLAZE pushes the expensive zero-knowledge part of the computation to the preprocessing phase, making its online phase faster than that of~\cite{C:BBCGI19}. SWIFT~(3PC) improved upon BLAZE by using the distributed zero-knowledge protocol of \cite{CCS:BGIN19}, thereby achieving GOD. In an orthogonal line of work,~\cite{PoPETS:WagGupCha19,PoPETS:WTBKMR21} focused on enhancing the efficiency of actively secure protocols for large convolutional neural networks, supporting training and inference. 

In the high-throughput setting for 4PC, ~\cite{AC:GorRanWan18} explores protocols for the security notions of abort. Inspired by the theoretical GOD construction in~\cite{AC:GorRanWan18}, \cite{PoPETS:BCPS20} proposed practical protocols with GOD for secure inference.  Trident~\cite{NDSS:ChaRacSur20} improved protocols (in terms of communication) compared to~\cite{AC:GorRanWan18} with a focus on security with fairness. In addition, it was the first work to propose a mixed-protocol framework for the case of 4 parties. More recently,~\cite{USENIX:MLRG20} improved over~\cite{AC:GorRanWan18} to provide support for fixed-point arithmetic with applications to graph parallel computation, albeit with abort security. Improving the security of Trident to GOD, SWIFT~\cite{USENIX:KPPS21} presented an efficient, robust PPML framework with protocols as fast as Trident. SWIFT only supports the secure inference of neural networks and lacks conversions similar to Trident and the garbled world. Fantastic Four~\cite{EPRINT:DalEscKel20} also provides robust 4PC protocols which are on par with SWIFT. While they claim to provide a better security model called {\em private robustness} compared to SWIFT, it has been shown in SWIFT that the two security models are theoretically equivalent. 

In the regime of constant-round protocols,~\cite{CCS:MohRosZha15} presents 3PC protocols in the honest majority setting satisfying security with abort, which require communicating one garbled circuit and three rounds of interaction. The work of~\cite{C:IKKP15} presents a robust 4-party computation protocol (4PC) with GOD in $2$-rounds (which is optimal) at the expense of 12 garbled circuits. Further,~\cite{CCS:BJPR18} presents efficient 3PC and 4PC constructions providing security notions of fairness and GOD. 

\paragraph{Dishonest Majority}
The works of \cite{C:DPSZ12,CCS:KelSchSma13} proposed efficient SS-based solutions for the dishonest majority setting over fields, which was then extended to the ring setting in~\cite{C:CDESX18}. The solution involves the generation of Beaver multiplication triples~\cite{C:Beaver91b} in the setup phase and evaluation of the circuit (multiplication gates) in the online phase using the generated triples. For the 2PC case, the approach mentioned above requires two public reconstructions among the parties per multiplication gate in the online phase. Later, works like~\cite{CCS:KelOrsSch16,EC:KelPasRot18,RSA:OrsSmaVer20} focused on improving the setup cost using techniques like Oblivious Transfer (OT) and Homomorphic Encryption (HE). \cite{ACNS:BenNieOmr19} improved the number of public reconstructions required in the online phase from two to one using a function-dependent preprocessing but requires additional communication of four ring elements in the preprocessing phase.

In this line of work, the GMW protocol~\cite{STOC:GolMicWig87} takes a function represented as a Boolean circuit (i.e., $\ell=1$), and the values are secret-shared using  XOR-based secret sharing. To pre-compute, a multiplication triple, the solution of \cite{CCS:ALSZ13} proposed a solution which uses 1-out-of-2 Oblivious Transfer~(OT), which was later improved by factor $1.2\times$ by~\cite{NDSS:DKSSZZ17} using the 1-out-of-N OT extension of~\cite{C:KolKum13}.

\paragraph{Mixed-protocols}
A mixed-protocol framework for MPC was first shown to be practical, in the 2-party dishonest majority setting, by TASTY~\cite{KSS09,CCS:HKSSW10}. TASTY was a passively secure compiler supporting generation of protocols based on homomorphic encryption and garbled circuits. This was followed by ABY~\cite{NDSS:DemSchZoh15}, which proposed a mixed protocol framework, also with passive security, combining the arithmetic, boolean and garbled worlds. The recent work of ABY2~\cite{USENIX:PSSY21} improves upon the ABY framework, providing a faster online phase with applications to PPML.
The work of~\cite{INDOCRYPT:RotWoo19,C:EGKRS20} proposed efficient mixed world conversions for the case of $n$ parties with a dishonest majority. Both works have active security, with \cite{INDOCRYPT:RotWoo19} supporting the inference of SVMs, and~\cite{C:EGKRS20} supporting neural network inference.

In the honest majority setting, ABY3~\cite{CCS:MohRin18} extended the idea to 3 parties and provided specialized protocols for the case of PPML. ABY3 was the first work to support secure training in the case of 3 parties, while Trident~\cite{NDSS:ChaRacSur20, EPRINT:KPRS21} extended it to the 4-party setting.

HyCC~\cite{CCS:BDKKS18} provides a compiler to automatically partition a function~(specified in ANSI C) into sub-functions such that each sub-function is evaluated with either Arithmetic sharing, Boolean sharing or GCs. The partitioning takes into account the real-world setup, such as the network between the parties. The work of~\cite{CCS:IshMilZik19} has shown a method to find an optimal partitioning in polynomial time.

\paragraph{Multi-Input Multiplication}
In the boolean setting,~\cite{NDSS:DKSSZZ17} extended two-input AND gates to the general N-input case using lookup tables. \cite{FC:OhaNui20} extended the multiplication from two-input to arbitrary input using Beaver triple extension with a focus on minimizing the online rounds. However, the online communication of~\cite{FC:OhaNui20} scale with the fan-in of the multiplication gates. \cite{USENIX:PSSY21} improved~\cite{FC:OhaNui20} and achieved an online communication of $2$ ring elements. Recently, \cite{EPRINT:KPRS21} extended the technique of \cite{USENIX:PSSY21} to the four-party honest majority setting.

\section{The Contribution of this Thesis}
\label{sec:contributions}
In the dominion of PPML consisting of a small number of parties which is of practical interest to the community, we propose $\thisP$, an efficient and robust PPML platform for 2,3 and 4 parties with different corruption thresholds. In the honest majority setting, we explore protocols with three and four parties, amongst which at most one can be maliciously corrupt. In the dishonest majority setting, we consider the two-party setting with only semi-honest corruption as achieving robustness with malicious corruption is proven to be impossible in the dishonest-majority setting~\cite{STOC:Cleve86}. While some of our protocols are the first of a kind in their setting (robust 3PC and 4PC), the rest of the protocols improve upon their counterparts in the literature by several orders of magnitude. 

A major contribution of the thesis lies in unifying the protocol design of all four settings. This results in much simpler protocols and brings in efficiency improvements over the prior versions~\cite{CCSW:CCPS19,NDSS:PatSur20,NDSS:ChaRacSur20,USENIX:KPPS21,USENIX:PSSY21}.  
All our protocols fall back to a generalized architecture of 3 layers as shown in \figref{architecture}. The first layer forms the foundation of our constructions designed using MPC protocols, which is then built upon by the second layer to obtain the building blocks. Finally, layer 3 utilizes layers 1 and 2 to give rise to the realization of privacy-preserving ML algorithms, thus forming the end goal of our architecture. We elaborate on this next, starting with the base layer.

\begin{figure}[htb!] 
		\centering{
			\includegraphics[scale=0.8]{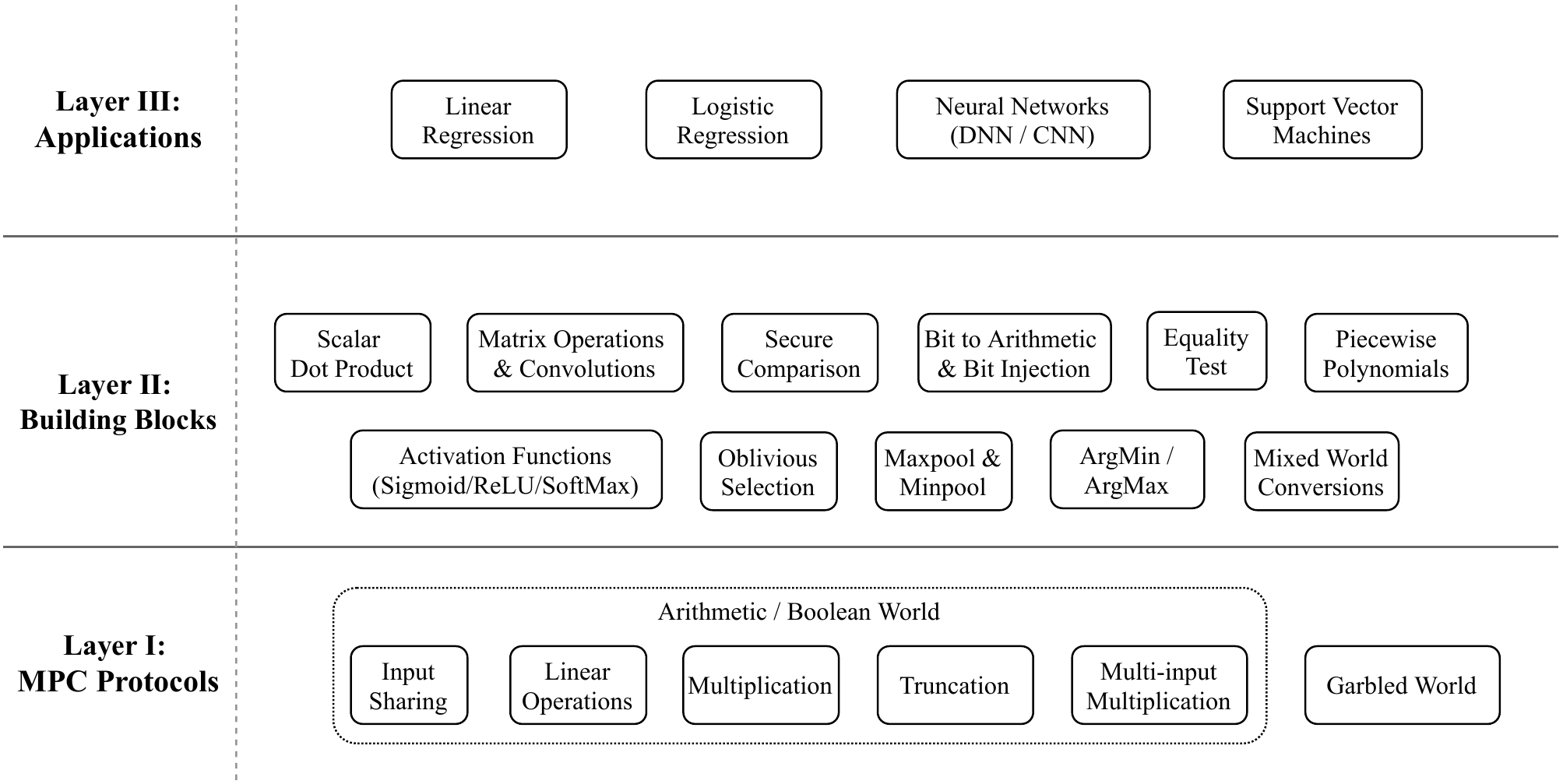}
		}
	\caption{Three-layer Architecture of $\thisP$}
	\smallskip
	\label{fig:architecture}
\end{figure}

\subsection{Layer I}
Layer I consisting of MPC protocols~($\TSthis, \Tthis, \Fthis, \TWthis$) form the basis of our architecture. We aim to realize efficient primitive operations such as input sharing, multiplication, and output reconstruction for all the considered frameworks. Although inspired by Beaver's multiplication-triple method~\cite{C:Beaver91b}, our multiplication protocol, which lies at the heart of this layer, adopts a new perspective that aids in realizing several efficient primitives discussed in \S\ref{sec:overview}. We believe that our new perspective can bring several further optimizations where Beaver's randomization technique is currently being used. 

To preserve privacy, we rely on computation and evaluation using our customized secret-sharing technique. This technique has two main advantages: It allows our protocols to be cast in the preprocessing paradigm leading to a blazing fast online phase. Further, it helps in minimizing the number of parties that need to be active for the majority of the computation in the online phase~(cf. \tabref{mpcCost}). We use the sharing over
both $\Z{\ell}$ and its special instantiation $\Z{1}$ and refer to them as {\em arithmetic} and {\em boolean} sharing respectively.

In most MPC-based PPML frameworks, we observe that a large part of the computation is done over the arithmetic and boolean worlds. The garbled world is used only to perform the non-linear operations~(e.g. softmax) that are expensive in the arithmetic/boolean world and switched back immediately after. Leveraging this observation, we propose tailor-made garbled world protocols with {\em end-to-end} conversion techniques. These protocols have the following advantages over the standalone variants -- i) no use of commitments for the inputs, and ii) no requirement of an explicit input sharing and output reconstruction phase, as explained later in the thesis.

Inspired by~\cite{USENIX:PSSY21, FC:OhaNui20}, we extend our multiplication protocol to the multi-input case, allowing multiplication of 3 and 4 inputs in one online round. Naively, performing a 4-input multiplication follows a tree-based approach, and the required communication is that of three 2-input multiplications and two online rounds. Our contribution lies in keeping the communication and the round of the online phase the same as that of 2-input multiplication (i.e. invariant of the number of inputs) by trading off the preprocessing cost. Looking ahead, multi-input multiplication, when coupled with the optimized parallel prefix adder circuit from~\cite{USENIX:PSSY21}, brings in a $2 \times$ improvement in online rounds. It also cuts down the online communication of secure comparison, impacting PPML applications.

\subsection{Layer II}
Layer II defines the building blocks that form the core of our architecture.
The primary building blocks constitute scalar dot product,  secure comparison, piece-wise polynomials and mixed world conversions. Although our building blocks improve over the state-of-the-art, our main contributions lie in the efficient realization of scalar dot product and mixed world conversions highlighted below. 

A naive approach to perform the dot product operation on two $\vl{d}$-length vectors is to perform $\vl{d}$ multiplications followed by adding the results. However, this leads to communication proportional to the length of the vectors. Our constructions remove the dependency of the communication on the length of the vectors in the setting of 3 and 4 parties. This is achieved for the first time in the setting of 3 parties with one active corruption. Moreover, in the 2PC literature, our construction achieves an online communication independent of the vector length for the first time. 

\begin{table}[htb!]
	\centering
	\resizebox{0.98\textwidth}{!}{
		\begin{NiceTabular}{c r c r|r r|r r|ccc}[notes/para]
			\toprule
			\Block{2-1}{\# Parties} & \Block{2-1}{Reference\tabularnote{Amortized costs are reported for 1 million operations}} 
			& \Block[c]{2-1}{\#Active\\Parties\tabularnote{parties that carry out most of the computation during online phase}}
			& \Block[c]{2-1}{Security} & \Block[c]{1-2}{Dot Product\tabularnote{$\ell$ - size of ring in bits, $\kappa$ - security parameter, $\vl{d}$ - length of the vectors.}} & 
			& \Block[c]{1-2}{Dot Product \\with Truncation} &  
			& \Block[c]{1-3}{Conversions\tabularnote{A, B, G indicate support for arithmetic, boolean, and garbled worlds respectively}} & & \\ \cmidrule{5-11}
		 &  &  &  
		 & Comm\textsubscript{pre}\tabularnote{`Comm' - communication, `pre' - preprocessing, `on' - online} 
		 & Comm\textsubscript{on} & Comm\textsubscript{pre} & Comm\textsubscript{on} & A & B & G\\ 
			\midrule
			\Block{4-1}{3}
			& ABY3~\cite{CCS:MohRin18} & 3 & semi-honest & $-$ & $3\ell$ & $\approx 6\ell$ & $4\ell$ & \cmark & \cmark & \cmark \\	
			& \textbf{$\TSthis$}~\cite{CCSW:CCPS19} & 2 & semi-honest & $1\ell$ & $2\ell$ & $1\ell$ & $2\ell$ & \cmark & \cmark & \cmark \\
			\cmidrule{2-11}
			& ABY3~\cite{CCS:MohRin18} & 3 & Abort & $12\vl{d}\ell$ & $9\vl{d}\ell$ & $12\vl{d}\ell + 84\ell$ & $9\vl{d}\ell + 3\ell$ & \cmark & \cmark & \cmark \\
			& \textbf{$\Tthis$}~\cite{NDSS:PatSur20,USENIX:KPPS21} & 2 & Robust & $3\ell$ & $3\ell$ & $9\ell$ & $3\ell$ & \cmark & \cmark & \cmark \\
			\midrule
			\Block{7-1}{4}
			& Mazloom et al.~\cite{USENIX:MLRG20} & 4 & Abort & $2\ell$ & $4\ell$ & $2\ell$ & $4\ell$ & \cmark & \cmark & \xmark \\
			& Trident~\cite{NDSS:ChaRacSur20} & 3 & Fair & $3\ell$ & $3\ell$ & $6\ell$ & $3\ell$ & \cmark & \cmark & \cmark \\		  
			& \textbf{$\Fthis$}~\cite{EPRINT:KPRS21} & 2 & Fair & $2\ell$ & $3\ell$ & $2\ell$ & $3\ell$ & \cmark & \cmark & \cmark \\
			\cmidrule{2-11}
			& SWIFT~(4PC)~\cite{USENIX:KPPS21} & 2 & Robust & $3\ell$ & $3\ell$ & $4\ell$ & $3\ell$ & \cmark & \cmark & \xmark \\
			& Fantastic Four~\cite{EPRINT:DalEscKel20} (Best)\tabularnote{cf. \S\ref{pa:fantasticfour} for details} 
			& 4 & Robust & $-$ & $6\ell$ &  $\ell$  & $9\ell$  & \cmark & \cmark & \xmark \\
			& Fantastic Four~\cite{EPRINT:DalEscKel20} (Worst) & 3 & Robust & $-$ & $6(\ell + \kappa$) &  $\approx80\ell+ 76\kappa$  & $9\ell + 6\kappa$  & \cmark & \cmark & \xmark \\
			& \textbf{$\Fthis$}~\cite{EPRINT:KPRS21} & 2 & Robust & $2\ell$ & $3\ell$ & $2\ell$ & $3\ell$ & \cmark & \cmark & \cmark \\
			\midrule
			\Block{2-1}{2}
			& SecureML~\cite{SP:MohZha17} & 2 & semi-honest & $2\vl{d}\ell (\kappa + \ell)$ & $4\vl{d}\ell$ & $2\vl{d}\ell (\kappa + \ell)$ & $4\vl{d}\ell$ & \cmark & \cmark & \cmark \\	
			& \textbf{$\TWthis$}~\cite{USENIX:PSSY21} & 2 & semi-honest & $2\vl{d}\ell (\kappa + \ell)$ & $2\ell$ & $2\vl{d}\ell (\kappa + \ell)$ & $2\ell$ & \cmark & \cmark & \cmark \\
			\bottomrule
		\end{NiceTabular}
	}
	\caption{\small Comparison of MPC frameworks~(small no. of parties) for PPML.}\label{tab:mpcCost}
	\vspace{-3mm}
\end{table}

For an operation that requires computing over the garbled domain in the mixed-world computation, the standard approach is to first switch from {\em Arithmetic to Garbled} and evaluate the garbled circuit to obtain a garbled-shared output. These shares are brought back to the arithmetic domain using a {\em Garbled to Arithmetic} conversion. Deviating from the standard approach, we propose new end-to-end conversion techniques that improve the round complexity by $2\times$. On a high level, our approach is to modify the garbled circuit such that the output is in the arithmetic domain. This eliminates the need for an explicit {\em Garbled to Arithmetic} conversion, saving in both communication and rounds in the online phase. More generally, end-to-end conversions are of the form "$\sf{x}$-Garbled-$\sf{x}$" where $\sf{x}$ can be either arithmetic or boolean and need a single round for the garbled world.

We summarize and compare the efficiency of layer II protocols with the state-of-the-art in \tabref{mpcCost}. We showcase the cost for a dot product operation in that table as it forms the fundamental building block of most PPML algorithms. As most computations in the PPML domain operate on decimal values, we provide the cost comparison for dot-product with truncation in the table. Finally, we highlight the conversions supported by our protocols and that of the stat-of-the-art.

\subsection{Layer III}
Layer III constitutes the realizations of the PPML algorithms that are widely used. We are the first to propose a robust PPML framework in the literature of three and four parties. We demonstrate the practicality of the framework, which combines the arithmetic, boolean, garbled worlds via benchmarking over a Wide Area Network (WAN), instantiated using n1-standard-64 instances of Google Cloud. We consider the training and inference phases of linear regression, logistic regression and deep neural networks such as LeNet~\cite{lenet} and VGG16~\cite{vgg16} along with the inference phase of Support Vector Machines.

The implementation section is presented through the lens of deployment scenarios with two different goals. Participants in the first scenario are interested in the shortest online runtime for the computation, whereas participants in the second one want to minimize the deployment cost. Correspondingly, there are variants of our framework that cater to both scenarios. The time-optimized~(${\sf T}$) variant has the fastest online phase considering online runtime as the metric.  On the other hand, the cost-optimized~(${\sf C}$) variant aims at minimizing deployment cost. This is measured via {\em monetary cost}~\cite{C:PRTY19}, which helps to capture the effect of the total runtime of the parties, and communication together.

\section{Organization of the Thesis}
\label{sec:organization}
The thesis is categorized into three parts. Each part represents a layer of the architecture~(\figref{architecture}) consisting of chapters devoted to $\TSthis, \Tthis, \Fthis, \TWthis$ frameworks. Moreover, chapters in each part are preceded by an overview.  \tabref{organization} summarizes the organization of these chapters. 

\begin{table}[htb!]
	\centering
	\begin{NiceTabular}{rrr rrr}
		\toprule
		\Block{2-1}{Framework} & \Block{2-1}{Setting} & \Block{2-1}{Security} 
		& \Block{1-3}{3-Layer Architecture~(\figref{architecture})} & & \\ \cmidrule{4-6}
		& & & Layer I & Layer II & Layer III \\ 
		\midrule
		$\TSthis$    & 3PC  & semi-honest  
		& Chapter~\ref{chap:layer1_3pcsemi} & Chapter~\ref{chap:layer2_3pcsemi} & Chapter~\ref{chap:layer3_3pcsemi} \\
		$\Tthis$      & 3PC  & robust  
		& Chapter~\ref{chap:layer1_3pcmal} & Chapter~\ref{chap:layer2_3pcmal} & Chapter~\ref{chap:layer3_3pcmal} \\
		$\Fthis$      & 4PC  & robust  
		& Chapter~\ref{chap:layer1_4pc} & Chapter~\ref{chap:layer2_4pc} & Chapter~\ref{chap:layer3_4pc} \\
		$\TWthis$   & 2PC  & semi-honest  
		& Chapter~\ref{chap:layer1_2pc} & Chapter~\ref{chap:layer2_2pc} & Chapter~\ref{chap:layer3_2pc} \\
		\bottomrule
	\end{NiceTabular}
	\caption{Organization of the thesis\label{tab:organization}}
\end{table}

The preliminaries and conclusion of the thesis appear in Chapter~\ref{chap:prelims} and~\ref{chap:conclusion} respectively. 


\chapter{Preliminaries}
\label{chap:prelims}
\begin{quote} \small
  This chapter presents the relevant background, including the notation, definitions, security model and an overview of some of the standard primitives used in our constructions. 
\end{quote}

\section{High Level Overview of Our Approach}
\label{sec:overview}
The MPC protocols in our framework rely on the well-known Beaver's circuit randomization technique~\cite{C:Beaver91b} but use a different perspective of the technique. This section presents a high-level overview of our scheme and a side-by-side comparison with Beaver's technique. The highlight of our scheme is its effectiveness towards efficient realizations for multiple input multiplication gates and dot product operations, as will be explained later in this thesis.
For simplicity, consider two parties $P_1, P_2$ with values $\vl{a}, \vl{b}$ secret-shared among them who want to compute a multiplication gate with output $\vl{z} = \vl{ab}$.  

\paragraph{Beaver's technique~\cite{C:Beaver91b} on gate inputs~(cf. left of \figref{OverviewBeaver})}

In Beaver's\cite{C:Beaver91b} circuit randomization technique~(cf. left side of \figref{OverviewBeaver}), the inputs of the multiplication gate are randomized first and the corresponding correlated randomness is generated independently (preferably in a setup phase). In detail, parties interactively generate an additive sharing of the multiplication triple $(\delta_{\vl{a}}, \delta_{\vl{b}}, \delta_{\vl{ab}})$ with $\delta_{\vl{ab}} = \delta_{\vl{a}} \delta_{\vl{b}}$ during the setup phase before the actual inputs are known. Now, we can write
\begin{align*}
	\vl{a}\cdot\vl{b}    &= ((\vl{a} + \delta_{\vl{a}}) - \delta_{\vl{a}}) ((\vl{b} +   \delta_{\vl{b}}) - \delta_{\vl{b}}) \\
	&= (\vl{a} + \delta_{\vl{a}})(\vl{b} + \delta_{\vl{b}}) - (\vl{a} + \delta_{\vl{a}})\delta_{\vl{b}} - (\vl{b} + \delta_{\vl{b}})\delta_{\vl{a}} + \delta_{\vl{a} \vl{b}}.
\end{align*} 

Let $\Delta_{\vl{a}} = (\vl{a} + \delta_{\vl{a}})$ and $\Delta_{\vl{b}} = (\vl{b} + \delta_{\vl{b}})$ be the randomized versions of the input values of a multiplication gate. Then, during the online phase, parties locally compute an additive sharing of $\Delta_{\vl{a}}$ using additive shares of $\vl{a}$ and $\delta_{\vl{a}}$. Similarly, an additive sharing of $\Delta_{\vl{b}}$ is computed. This is followed by the parties mutually exchanging the shares of $\Delta_{\vl{a}}$ and $\Delta_{\vl{b}}$ to enable public reconstruction of $\Delta_{\vl{a}}$ and $\Delta_{\vl{b}}$. Then using the above equation, parties can locally compute a sharing of $\vl{a}\cdot\vl{b}$. Note that this method requires reconstruction of two elements per multiplication gate. We observe that the communication is required for enabling parties to obtain the value of $\Delta_{\vl{a}}$ and $\Delta_{\vl{b}}$ in clear.

\tikzset{every picture/.style={line width=0.75pt}} 
\begin{figure}[htb!]
	\centering
	\resizebox{.8\textwidth}{!}{%
		\begin{tikzpicture}[x=0.75pt,y=0.75pt,yscale=-1,xscale=1]
		\draw   (194.33,78) .. controls (194.33,78) and (194.33,78) .. (194.33,78) -- (232.83,78) .. controls (232.83,78) and (232.83,78) .. (232.83,78) -- (232.83,99.08) .. controls (232.83,109.71) and (224.21,118.33) .. (213.58,118.33) -- (213.58,118.33) .. controls (202.95,118.33) and (194.33,109.71) .. (194.33,99.08) -- cycle ;
		\draw    (203.58,57.17) -- (203.58,76.17) ;
		\draw [shift={(203.58,78.17)}, rotate = 270] [color={rgb, 255:red, 0; green, 0; blue, 0 }  ][line width=0.75]    (10.93,-3.29) .. controls (6.95,-1.4) and (3.31,-0.3) .. (0,0) .. controls (3.31,0.3) and (6.95,1.4) .. (10.93,3.29)   ;
		
		\draw    (223.58,57.17) -- (223.58,76.17) ;
		\draw [shift={(223.58,78.17)}, rotate = 270] [color={rgb, 255:red, 0; green, 0; blue, 0 }  ][line width=0.75]    (10.93,-3.29) .. controls (6.95,-1.4) and (3.31,-0.3) .. (0,0) .. controls (3.31,0.3) and (6.95,1.4) .. (10.93,3.29)   ;
		
		\draw    (213.58,118.33) -- (213.58,137.33) ;
		\draw [shift={(213.58,139.33)}, rotate = 270] [color={rgb, 255:red, 0; green, 0; blue, 0 }  ][line width=0.75]    (10.93,-3.29) .. controls (6.95,-1.4) and (3.31,-0.3) .. (0,0) .. controls (3.31,0.3) and (6.95,1.4) .. (10.93,3.29)   ;

		\draw    (41.47,124.2) -- (99.87,124.59) ;
		\draw [shift={(101.87,124.6)}, rotate = 180.38] [color={rgb, 255:red, 0; green, 0; blue, 0 }  ][line width=0.75]    (10.93,-3.29) .. controls (6.95,-1.4) and (3.31,-0.3) .. (0,0) .. controls (3.31,0.3) and (6.95,1.4) .. (10.93,3.29)   ;
		
		\draw    (43.87,133.81) -- (102.27,134.2) ;
		
		\draw [shift={(41.87,133.8)}, rotate = 0.38] [color={rgb, 255:red, 0; green, 0; blue, 0 }  ][line width=0.75]    (10.93,-3.29) .. controls (6.95,-1.4) and (3.31,-0.3) .. (0,0) .. controls (3.31,0.3) and (6.95,1.4) .. (10.93,3.29)   ;
		
		\draw    (340.47,123.2) -- (398.87,123.59) ;
		\draw [shift={(400.87,123.6)}, rotate = 180.38] [color={rgb, 255:red, 0; green, 0; blue, 0 }  ][line width=0.75]    (10.93,-3.29) .. controls (6.95,-1.4) and (3.31,-0.3) .. (0,0) .. controls (3.31,0.3) and (6.95,1.4) .. (10.93,3.29)   ;
		
		\draw    (342.87,132.81) -- (401.27,133.2) ;
		
		\draw [shift={(340.87,132.8)}, rotate = 0.38] [color={rgb, 255:red, 0; green, 0; blue, 0 }  ][line width=0.75]    (10.93,-3.29) .. controls (6.95,-1.4) and (3.31,-0.3) .. (0,0) .. controls (3.31,0.3) and (6.95,1.4) .. (10.93,3.29)   ;
		
		\draw [color={rgb, 255:red, 254; green, 36; blue, 36 }  ,draw opacity=1 ] [dash pattern={on 0.84pt off 2.51pt}]  (214.13,1.3) -- (213.27,154.33) ;

		\draw    (9.85,159.93) -- (418.85,159.93) -- (418.85,182.93) -- (9.85,182.93) -- cycle  ;
		\draw (214.35,171.43) node [scale=0.9]  {$\vl{c}_{i} \ =\ (i-1)\cdot \Delta _{\vl{a}} \Delta _{\vl{b}} -\ \Delta _{\vl{a}}[ \delta _{\vl{b}}]_{i} \ -\ \Delta _{\vl{b}}[ \delta _{\vl{a}}]_{i} \ -\ [ \delta _{\vl{a}} \delta _{\vl{b}}]_{i} \ \ ;\ i\ \in \{1,2\}$};
		\draw (104.42,36.17) node [scale=0.9]  {$P_{i} \ :\ ( \vl{a}_{i} ,[ \delta _{\vl{a}}]_{i}) ,( \vl{b}_{i} ,[ \delta _{\vl{b}}]_{i}) ,[ \delta _{\vl{a}} \delta _{\vl{b}}]_{i}$};
		\draw (367.75,76.83) node [scale=0.9]  {$[ \Delta _{\vl{c}}]_{i} \ :\ \vl{c}_{i} +[ \delta _{\vl{c}}]_{i}$};
		\draw (120,163) node  [align=left] {};
		\draw (105,11) node  [align=left] {\underline{Beaver's\cite{C:Beaver91b}: On Gate Inputs}};
		\draw (324,11) node  [align=left] {\underline{\thisW : On Gate Output}};
		\draw (196.67,60) node [scale=0.9]  {$\vl{a}$};
		\draw (233,60.33) node [scale=0.9]  {$\vl{b}$};
		\draw (226,126.67) node [scale=0.9]  {$\vl{c}$};
		\draw (213.58,98.17) node [scale=0.8] [align=left] {$\Mult$};
		\draw (24,127.33) node [scale=0.9]  {$P_{1}$};
		\draw (122.67,126.67) node [scale=0.9]  {$P_{2}$};
		\draw (73.67,112) node [scale=0.9]  {$[ \Delta _{\vl{a}}]_{1} ,[ \Delta _{\vl{b}}]_{1}$};
		\draw (73.67,145) node [scale=0.9]  {$[ \Delta _{\vl{a}}]_{2} ,[ \Delta _{\vl{b}}]_{2}$};
		\draw (68.75,62.5) node [scale=0.9]  {$[ \Delta _{\vl{a}}]_{i} \ :\ \vl{a}_{i} +[ \delta _{\vl{a}}]_{i}$};
		\draw (68.75,83.83) node [scale=0.9]  {$[ \Delta _{\vl{b}}]_{i} \ :\ \vl{b}_{i} +[ \delta _{\vl{b}}]_{i}$};
		\draw (328.42,36.17) node [scale=0.9]  {$P_{i} \ :\ ( \Delta _{\vl{a}} ,[ \delta _{\vl{a}}]_{i}) ,( \Delta _{\vl{b}} ,[ \delta _{\vl{b}}]_{i}) ,[ \delta _{\vl{a}} \delta _{\vl{b}}]_{i}$};
		\draw (372.67,144) node [scale=0.9]  {$[ \Delta _{\vl{c}}]_{2}$};
		\draw (372.67,111) node [scale=0.9]  {$[ \Delta _{\vl{c}}]_{1}$};
		\draw (412.67,125.67) node [scale=0.9]  {$P_{2}$};
		\draw (329,126.33) node [scale=0.9]  {$P_{1}$};
		\end{tikzpicture}
	}%
    \small{
    	\caption{\small High level overview of Beaver's\cite{C:Beaver91b} and \thisW}
   	    \label{fig:OverviewBeaver} 
    }
\end{figure}
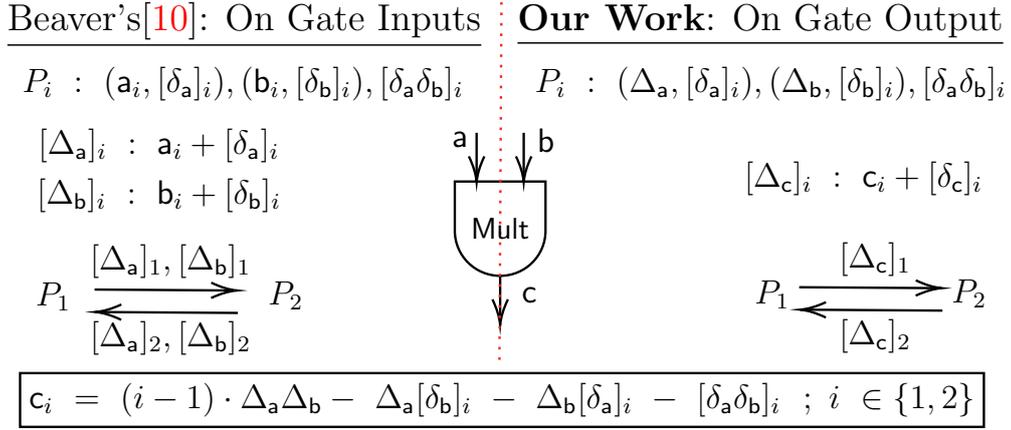
\vspace{-3mm}

\paragraph{Our technique on gate outputs~(cf. right of \figref{OverviewBeaver})}
With this insight,  we modify the sharing semantics so that the parties are ensured to have the $\Delta_{}$ value as a part of their share, corresponding to every wire value (including the inputs of a multiplication gate). As a result, the reconstructions  of $\Delta_{\vl{a}}$ and $\Delta_{\vl{b}}$ are no longer required. This may give the wrong impression that no communication is required for evaluating a multiplication gate.
It is true that now the parties can locally evaluate the additive sharing of $\vl{z} = \vl{ab}$. But to proceed further, a sharing for $\vl{z}$ according to the new sharing semantics needs to be generated. This requires both parties to obtain $\Delta_{\vl{z}}$ in the clear. Hence, the parties locally compute an additive sharing of $\Delta_{\vl{z}}$ using the shares of $\vl{z}$ computed earlier and mutually exchange their shares to reconstruct $\Delta_{\vl{z}}$.  

Our technique, in summary, shifts the need for reconstruction (which alone causes communication for a multiplication gate) from per input wire to the \emph{output} wire alone for a multiplication gate. For a traditional 2-input multiplication gate, we reduce the number of reconstructions (each involves sending two elements) from 2 to 1. As a result, we improve communication by a factor of $2\times$.
The impact is much higher for an $N$-input multiplication gate and a scalar product of two $N$-dimensional vectors. For scalar product, Beaver's circuit re-randomization required $2N$ reconstructions,  whereas our techniques need a {\em single} one, offering a gain of $2N\times$. Our constructions can be generalized to the $n$-party scenario (which is out of scope for this work) and bring a significant pay-off, as the cost per reconstruction depends linearly on the number of parties.

\section{Parameters and Notation}
\label{sec:notations}
In our framework, we have $n \in \{2,3,4\}$ parties, denoted by $\Partyset$ that are connected by pair-wise private and authentic channels in a synchronous network, and an adversary that can corrupt at most one party. 
Our protocols are designed to work over an $\ell$-bit ring denoted by $\Z{\ell}$. $\csec$ denotes the computational security parameter. In our implementation, we use $\ell = 64$ and $\csec = 128$.
Our protocols are cast into an {\em input-independent} preprocessing phase and an {\em input-dependent} online phase. 
Our protocols work over the arithmetic ring $\Z{\ell}$ or boolean ring $\Z{1}$. 

\paragraph{Secure Outsourced Computation~(SOC)}
In the secure outsourced computation (SOC) setting, the servers hired to carry out the computation enact the role of the parties mentioned above. For ML training, data owners who want to train a model collaboratively secret-share their data among the servers. For ML inference, a data owner shares its model while the client shares its query among the servers. The servers carry out the computation on secret-shared data and obtain the output in a secret-shared fashion. In the case of training, the output is reconstructed towards the data owners, whereas in the case of inference, the output is reconstructed towards the client. We assume that the corrupt server can collude with an arbitrary number of data-owners in the case of training. In contrast, we assume that the corrupt server can collude with the model owner or the client for inference.
In the case of inference, since the query response is available in the clear to the client, we do not guarantee the privacy of the training data against attacks such as attribute inference, membership inference, or model inversion~\cite{CCS:FreJhaRis15,USENIX:TZJRR16,SP:SSSS17}. This is an orthogonal problem, and we consider it as an out-of-scope of this thesis.

\paragraph{Dealing with decimal values}
For applications such as machine learning where the inputs are decimal numbers, we use the Fixed-Point Arithmetic~(FPA)~\cite{CCS:MohRin18,CCSW:CCPS19,NDSS:PatSur20,NDSS:ChaRacSur20,PoPETS:BCPS20} representation to embed the value in the underlying ring $\Z{\ell}$. Decimal value is treated as an $\ell$-bit integer in signed 2's complement representation. The most significant bit~($\msb$) represents the sign bit, and $x$ least significant bits are reserved for the fractional part. The $\ell$-bit integer is then treated as an element of $\Z{\ell}$, and operations are performed modulo $2^{\ell}$. For our implementation, we use $\ell = 64$, and $x = 13$, with $\ell - x - 1$ bits for the integral part.

\paragraph{Vectors and Matrices}
For a vector $\vct{a}$, ${\vl{a}}_i$ denotes the $i^{th}$ element in the vector. For two vectors $\vct{a}$ and $\vct{b}$ of length $\vl{d}$, the dot product is given by, $\vct{a} \band \vct{b} = \sum_{i = 1}^{\vl{d}} {\vl{a}}_i {\vl{b}}_i$. 
Given two matrices $\Mat{A}, \Mat{B}$, the operation $\Mat{A} \circ \Mat{B}$ denotes the matrix multiplication.

\begin{notation}\label{arval}
	For a bit $\bitb \in \bitset$, $\arval{\bitb}$ denotes the representation of the bit value $\bitb$ over the arithmetic ring $\Z{\ell}$. In detail, all the bits of $\arval{\bitb}$ will be zero except for the least significant bit, which is set to $\bitb$.
\end{notation}

\noindent \tabref{notations} depicts notation that we use throughout the thesis.

\begin{table}[htb!]
	\centering
	\resizebox{0.8\textwidth}{!}{
		\begin{NiceTabular}{r | l}
			\toprule
			$n$PC			          & $n$-party computation; $n \in$ \{2,3,4\} in this thesis\\[3pt]
			$\Partyset$			      & \makecell[l]{Set of all parties performing secure computation; \\ 2PC: $\Partyset = \{P_1, P_2\}$, 3PC: $\Partyset = \{P_1, P_2, P_3/P_0\}$, 4PC: $\Partyset = \{P_0, P_1, P_2, P_3\}$} \\[5pt]
			$\Z{\ell}$			         & Ring of size $\ell$ bits; $\ell = 64$ in this thesis\\[3pt]
			$\kappa$			      & Symmetric security parameter; $\kappa = 128$ in this thesis\\[3pt]
			${\vl{a}}_i$			   & $i^{th}$ element of vector $\vct{a}$\\[3pt]
			$\vct{a} \band \vct{b}$  & Scalar dot product between vectors $\vct{a}$ and $\vct{b}$ of length $\vl{d}$\\[3pt]
			$\Mat{X} \circ \Mat{Y}$  & Multiplication of two matrices $\Mat{X}$ and $\Mat{Y}$\\[3pt]
			$s \in \{{\bf A, B, G}\}$		 & Type of sharing: \textbf{A}rithmetic, \textbf{B}oolean, or \textbf{G}arbled \\[3pt]
			$\arval{\bitb}$  & Representation of the bit value $\bitb \in \bitset$ over the arithmetic ring $\Z{\ell}$\\[3pt]
			
			$\overline{\bitb}$ & Complement value $1 \xor \bitb$ for bit $\bitb \in \bitset$\\[3pt]
			$\Hash(\cdot)$  & A {\em collision-resistant} hash function\\[3pt]
			PRF           	         & Pseudo-random Function \\[3pt]
			FPA           	         & Fixed-point Arithmetic; $x$ denotes the precision and $x = 13$ in this thesis\\[3pt]
			$\msb~/~\lsb$           	     & Most / Least Significant Bit \\[3pt]
			OT           	         & Oblivious Transfer \\[3pt]
			$\COT{n}{\ell}$          & $n$ instances of Correlated OT on $\ell$-bit strings \\[3pt]
			HE           	         & Homomorphic Encryption \\[3pt]
			$\ppt$           	   & Probabilistic-polynomial Time \\[3pt]
			PPA           	   & Parallel-prefix Adder \\[3pt]
			\bottomrule
		\end{NiceTabular}
	}
	\caption{Notations used throughout this thesis.}\label{tab:notations}
\end{table}


\section{Definitions}
\label{sec:definitions}

\begin{definition}
	\label{def:negl}
	(Negligible functions) A function $\negl$ is negligible iff $\forall c \in \NN \; \exists n_0 \in \NN$ such that $\forall n > n_0, \negl(n) < n^{-c}$.
\end{definition}

\section{Security Model}
\label{sec:secmodel}
We prove the security of our protocols using the real-world/ ideal-word simulation paradigm~\cite{Goldreich04, EPRINT:Lindell16}. The security of protocols is analyzed by comparing what an adversary can do in the real world execution of the protocol with what it can do in an ideal world execution that is considered secure by definition (where there exists a trusted third party, denoted as~$\ttp$). In the ideal world, the parties send their inputs to the trusted third party over perfectly secure channels that carries out the computation and send the output to the parties. Informally, a protocol is said to be secure if whatever an adversary can do in the real world can also be done in the ideal world. We refer the readers to \cite{JC:Canetti00,Goldreich01,AC:CohLin14,EPRINT:Lindell16} for further details regarding the security model. 

Let $\Adv$ denote the probabilistic polynomial time ($\ppt$) real-world adversary corrupting at most one party in $\Partyset$, $\Sim$ denote the corresponding ideal world adversary, and $\Func{ }$ denote the ideal functionality. Let $\Ideal_{\Func{ }, \Sim}(\onesec, z)$ denote the joint output of the honest parties and $\Sim$ from the ideal execution with respect to the security parameter $\csec$ and auxiliary input $z$. Similarly, let $\Real_{\Pi, \Adv}(\onesec, z)$ denote the joint output of the honest parties and $\Adv$ from the real world execution. We say that the protocol $\Pi$ securely realizes $\Func{ }$ if for every $\ppt$ adversary $\Adv$ there exists an ideal world adversary $\Sim$ corrupting the same parties such that $\Ideal_{\Func{ }, \Sim}(\onesec, z)$ and $\Real_{\Pi, \Adv}(\onesec, z)$ are computationally indistinguishable.

\begin{definition}
	\label{def:realideal}
	For $n\in \NN$, let $\Func$ be a functionality and let $\Pi$ be a $n$-party protocol. We say that  $\Pi$ {\it securely realizes} $\Func$~if for every  $\ppt$ real world adversary $\Adv$, there exists a $\ppt$ ideal world adversary $\Sim$, corrupting the same parties, such that the following two distributions  are computationally indistinguishable:
	\[\Ideal_{\Func, \Sim}  \stackrel{c}{\approx}  \Real_{\Pi, \Adv}. \] 
\end{definition}

We analyze the security guarantees of
correctness and privacy separately in all our security proofs since we consider deterministic functionalities alone in this thesis~\cite{EPRINT:Lindell16}.

\paragraph{Ideal Functionalities.~\cite{AC:CohLin14,C:GorLiuShi15}}
\label{para:idealfunc}
For the secure computation of a function $f$ using MPC, we define the ideal functionalities $\Func[Semi]$, $\Func[Abort]$, $\Func[Fair]$, and $\Func[GOD]$ in \boxref{ideal:semi}, \boxref{ideal:abort}, \boxref{ideal:fair}, and \boxref{ideal:god} respectively.

\begin{systembox}{$\Func[Semi]$}{Semi-honest functionality for computing function $f$}{ideal:semi}
	Every party $P_i \in \Partyset$ $(i \in [n])$ sends its input $x_i$ to the functionality.
	
	\noindent \textbf{Input:} On message $(\INPUT, x_i)$ from $P_i$ $(i \in [n])$, do the following: if $(\INPUT, \ast)$ already received from $P_i$, then ignore the current message. Otherwise, record $x_i^{\prime} = x_i$ internally.
	
	\noindent \textbf{Output:} Compute $y = f(x_1^{\prime}, \dots, x_n^{\prime})$ and send $(\OUTPUT, y)$ to all parties.
\end{systembox}

\begin{systembox}{$\Func[Abort]$}{Abort functionality for computing function $f$}{ideal:abort}
	Every honest party $P_i \in \Partyset$ $(i \in [n])$ sends its input $x_i$ to the functionality. Corrupted parties may send arbitrary inputs as instructed by the adversary. While sending the inputs, the adversary is also allowed to send a special $\abort$ command.
	
	\noindent \textbf{Input:} On message $(\INPUT, x_i)$ from $P_i$ $(i \in [n])$, do the following: if $(\INPUT, \ast)$ already received from $P_i$, then ignore the current message. Otherwise, record $x_i^{\prime} = x_i$ internally. If $x_i$ is outside $P_i$'s domain, consider $x_i^{\prime} = \abort$.
	
	\noindent \textbf{Output to adversary:} If there exists an $i \in [n]$ such that $x_i^{\prime} = \abort$, send $(\OUTPUT, \bot)$ to all the parties. Else, compute $y = f(x_1^{\prime}, \dots, x_n^{\prime})$ and send $(\OUTPUT, y)$ to the adversary.
	
	\noindent \textbf{Output to selected honest parties:} Receive $(\SELECT, I)$  from adversary, where $I$ denotes a subset of the honest parties. If an honest party belongs to $I$, send $(\OUTPUT, y)$, else send $(\OUTPUT, \bot)$, where $y = f(x_1^{\prime}, \dots, x_n^{\prime})$. We require that $I$ includes all honest parties in case the adversary corrupts no party actively.
\end{systembox}

\begin{systembox}{$\Func[Fair]$}{Fair functionality for computing function $f$}{ideal:fair}
	Every honest party $P_i \in \Partyset$ $(i \in [n])$ sends its input $x_i$ to the functionality. Corrupted parties may send arbitrary inputs as instructed by the adversary. While sending the inputs, the adversary is also allowed to send a special $\abort$ command.
	
	\noindent \textbf{Input:} On message $(\INPUT, x_i)$ from $P_i$ $(i \in [n])$, do the following: if $(\INPUT, \ast)$ already received from $P_i$, then ignore the current message. Otherwise, record $x_i^{\prime} = x_i$ internally. If $x_i$ is outside $P_i$'s domain, consider $x_i^{\prime} = \abort$.
	
	\noindent \textbf{Output:} If there exists an $i \in [n]$ such that $x_i^{\prime} = \abort$, send $(\OUTPUT, \bot)$ to all the parties. Else, compute $y = f(x_1^{\prime}, \dots, x_n^{\prime})$ and send $(\OUTPUT, y)$ to all parties.
\end{systembox}

\begin{systembox}{$\Func[GOD]$}{GOD functionality for computing function $f$}{ideal:god}
	Every honest party $P_i \in \Partyset$ $(i \in [n])$ sends its input $x_i$ to the functionality. Corrupted parties may send arbitrary inputs as instructed by the adversary.
	
	\noindent \textbf{Input:} On message $(\INPUT, x_i)$ from $P_i$ $(i \in [n])$, do the following: if $(\INPUT, \ast)$ already received from $P_i$, then ignore the current message. Otherwise, record $x_i^{\prime} = x_i$ internally. If $x_i$ is outside $P_i$'s domain, consider $x_i^{\prime}$ to be some predetermined default value.
	
	\noindent \textbf{Output:} Compute $y = f(x_1^{\prime}, \dots, x_n^{\prime})$ and send $(\OUTPUT, y)$ to all parties.
\end{systembox}


\section{Primitives}
\label{sec:primitives}

\subsection{Shared-Key Setup}
\label{sec:KeySetupprelims}
To enable parties to non-interactively sample a random value, parties rely on a one-time shared key-setup~\cite{CCS:MohRin18,CCSW:CCPS19,NDSS:PatSur20,NDSS:ChaRacSur20,PoPETS:BCPS20,USENIX:KPPS21,USENIX:PSSY21}, denoted by $\Func[Key]$. The key-setup can be instantiated using any standard MPC protocol in the respective setting. The key-setup establishes  random keys  among the parties for a pseudo-random function~(PRF) which can be instantiated, for instance, using AES in counter mode. 

Let $F : \{0, 1\}^{\csec} \times \{0, 1\}^{\csec} \rightarrow X$ be a secure  pseudo-random function (PRF), with co-domain $X$ being $\Z{\ell}$. In $\Func[Key]$, the key $\Key{\Partyset}$ is established among all the parties in $\Partyset$. In addition, the following set of keys are established depending on the underlying framework.
\begin{enumerate}
	\item Three-party frameworks~($\TSthis$ \& $\Tthis$):
			\begin{enumerate}
				\item[--] One key between every pair -- $\Key{ij}$ for $P_i, P_j$.
			\end{enumerate}
	\item Four-party framework~($\Fthis$):
			\begin{enumerate}
				\item[--] One key between every pair -- $\Key{ij}$ for $P_i, P_j$.
				\item[--] One key between every set of three parties -- $\Key{ijk}$ for $P_i, P_j, P_k$.
			\end{enumerate}
\end{enumerate}

A simple instantiation for the case of $\TSthis$ with $\Partyset = \{P_0, P_1, P_2\}$ is as follows. $P_0$ samples key $\Key{0i}, \Key{\Partyset}$ and sends to $P_i$ for $i \in \{1,2\}$. $P_1$ samples $\Key{12}$ and sends to $P_2$. The instantiations for other frameworks can be derived similarly.

\subsection{Collision Resistant Hash Function}
\label{sec:hashprelims}
Consider a hash function family $\Hash = \mathcal{K}\times \mathcal{L} \rightarrow \mathcal{Y}$. The hash function $\Hash$ is said to be collision resistant if, for all probabilistic polynomial-time adversaries $\Adv$, given the description of $\Hash_k$ where {$k \in_R \mathcal{K}$}, there exists a negligible function $\negl()$ such that $\Pr[ (x_1,x_2) \leftarrow \Adv(k):(x_1 \ne x_2) \wedge \Hash_k(x_1)=\Hash_k(x_2)] \leq \negl(\csec)$, where $m = \mathsf{poly}(\csec)$ and $x_1,x_2 \in_R \{0,1\}^m$.

\subsection{Commitment Scheme}
\label{sec:commitprelims}
Let $\commit(x)$ denote the commitment of a value $x$. The commitment scheme $\commit(x)$ possesses two properties; {\em hiding} and {\em binding}. The former ensures privacy of the value $\vl{v}$ given just its commitment $\commit(\vl{v})$, while the latter prevents a corrupt server from opening the commitment to a different value $x' \ne x$. The practical realization of  a commitment scheme is via a hash function $\mathcal{H}()$ given below, whose  security can be proved in the random-oracle model (ROM)--  for  $(c, o) =  (\mathcal{H}(x||r), \allowbreak x||r) = \commit (x; r)$.

\subsection{Replicated Secret Sharing~\cite{TCC:CraDamIsh05}}
\label{sec:rsssharing}
Informally, a $t$-out-of-$n$ replicated secret sharing scheme distributes a secret among $n$ parties in such a way that any group of $t + 1$  or more parties can together reconstruct the secret but no group of fewer than $t +  1$ parties can. We present the formal definition below. 
\begin{definition}
	A $t$-out-of-$n$ replicated secret sharing scheme, defined for a finite set of secrets $K$ and a set of $\Partyset$ parties, comprises of two protocols-- Sharing~$(\Sh)$ and Reconstruction~$(\Rec)$, with the following requirements:
	\begin{description}
		\item[-] \textit{Correctness.} The secret can be reconstructed by any set of $(t+1)$ parties via $\Rec$. That is,  $\forall s \in K$ and  $\forall S = \{i_1, \dots i_{t + 1}\} \subseteq \{1, \dots n\}$ of size $(t + 1)$, $\Pr[\Rec(s_{i_1} \dots s_{i_{t + 1}}) = s] = 1$. 
		\item[-]\textit{Privacy.} Any set of $t$ parties cannot learn anything about the secret from their shares. That is: $\forall s^1,s^2 \in K$,  $\forall S = \{i_1, \dots i_{t}\} \subseteq \{1, \dots n\}$ of size $t$, and for every possible vector of shares $\{s_j\}_{j \in S}$,  $\Pr[\{\{\Sh(s^{1})\}_S = \{s_j\}_{i_j \in S}] =  \Pr[\{\{\Sh(s^{2})\}_S = \{s_j\}_{i_j \in S}]$, where $\{\Sh(s^{i})\}_S$ denotes the set of shares assigned to the set $S$ as per $\Sh$ when $s^i$ is the secret for $i \in \{1,2\}$. 
	\end{description} 
\end{definition}

\subsection{Garbling scheme and properties}
\label{sec:garbledprelims}
Here, we provide the pre-requisites for the two-party garbled circuit based computation of Yao~\cite{FOCS:Yao82b}. All the garbled circuit computations in this thesis can be viewed as an instance of a two-party case, and hence we omit the details for the multi-party case~\cite{STOC:BeaMicRog90,CCS:BelHoaRog12}.
As per Yao's garbling circuit paradigm~\cite{FOCS:Yao82b}, every wire in the circuit is assigned two $\csec$-bit strings, called ``keys'', one each for bit value $0$ and $1$ on that wire. Let $(\key{\vl{x}}{0},\key{\vl{x}}{1})$ denote the zero-key and one-key, respectively, on wire $\vl{x}$ in the circuit.
For simplicity, the same notation is used for wire identity as well as the value on the wire. For instance,  the key-pair for wire $\vl{x}$ is denoted as $(\key{\vl{x}}{0},\key{\vl{x}}{1})$, while the key corresponding to bit  $\vl{x}$ on the wire is denoted as $\key{\vl{x}}{\vl{x}}$.
Then, each gate is constructed by encrypting the output-wire key with the appropriate input-wire keys. For example, for an AND gate with input wires $\vl{x}, \vl{y}$ and output wire $\vl{z}$, $\key{\vl{z}}{0}$ is double encrypted with keys $\key{\vl{x}}{0}, \key{\vl{y}}{0}$, with $\key{\vl{x}}{0}, \key{\vl{y}}{1}$, and with $\key{\vl{x}}{1}, \key{\vl{y}}{0}$, while $\key{\vl{z}}{1}$ is double encrypted with $\key{\vl{x}}{1}, \key{\vl{y}}{1}$. Give one key on each input wire, the output wire key can be obtained by decrypting the ciphertext which was encrypted using the corresponding input wire keys. These ciphertexts are provided in a permuted order so that the evaluating party does not learn which key, $\key{\vl{z}}{0}$ or $\key{\vl{z}}{1}$, it obtains after decryption.

A garbling scheme $\GS$, consists of four algorithms $(\Gb, \En, \allowbreak \Ev, \De)$ defined as follows:
\begin{enumerate}
	\item $\Gb(\onesec,\Ckt) \rightarrow (\GC,e,d)$: $\Gb$ takes as input the security parameter $\kappa$ and the circuit $\Ckt$ to be garbled, and outputs a garbled circuit $\GC$, encoding information $e$ and decoding information $d$.
	\item $\En(x,e) \rightarrow \X$: $\En$ encodes input $x$ using $e$ to output encoded input $\X$. $\X$ is referred to as encoded input or encoded keys interchangeably.
	\item $\Ev(\GC,\X) \rightarrow \Y$:  $\Ev$ evaluates the garbled circuit $\GC$ on the encoded input $\X$ and produces the encoded output $\Y$.
	\item $\De(\Y,d) \rightarrow y$: The encoded output $\Y$ is decoded into the clear output $y$ by running the $\De$ algorithm on $\Y$ and $d$.
\end{enumerate}

We rely on the following properties of garbling scheme~\cite{CCS:BelHoaRog12} in our constructions.
\begin{enumerate}
	\item A garbling scheme $\GS = (\Gb, \En, \Ev, \De)$ is {\em correct} if for all input lengths $n \leq \poly(\csec)$, circuits $C: \{0,1\}^{n} \rightarrow \{0,1\}^{m}$ and inputs $x \in \{0,1\}^{n}$, the following holds.

	\begin{equation*}
		\hspace{7mm} \Prob[\De(\Ev(\GC, \En(x, e)), d) \neq C(x) : (\GC, e, d) \leftarrow \Gb(\onesec, C)] < \negl(\csec)
	\end{equation*}

	\item A garbling scheme $\GS$ is said to be \textit{private} if for all $n \leq \poly(\csec)$, circuit $C: \{0,1\}^{n} \rightarrow \{0,1\}^{m}$, there exists a $\ppt$ simulator $\Sim_{\priv}$ such that for all $x \in \{0,1\}^{n}$, for all $\ppt$ adversary $\Adv$ the following distributions are computationally indistinguishable.
	\begin{itemize}
		\item[-] $\Real(C, x)$: run $(\GC, e, d) \leftarrow \Gb(\onesec, C)$ and output $(\GC, \En(x, e), d)$.
		\item[-] $\Ideal(C, C(x))$: run $(\GC^{\prime}, \textbf{X}, d^{\prime}) \leftarrow \Sim_{\priv}(\onesec, C, C(x))$ and output $(\GC^{\prime}, \textbf{X}, d^{\prime})$.
	\end{itemize}
	
	\item A garbling scheme $\GS$ is \textit{authentic} if for all $n \leq \poly(\csec)$, circuit $C: \{0,1\}^{n} \rightarrow \{0,1\}^{m}$, input $x \in \{0,1\}^{n}$ and for all $\ppt$ adversary $\Adv$, the following probability is $\negl(\csec)$.
	
	\begin{equation*}
		\hspace{7mm} \Prob \Bigg(
		\begin{aligned}
			&\hat{\textbf{Y}} \neq \Ev(\GC, \textbf{X}) \\
			&\wedge \De(\hat{\textbf{Y}}, d) \neq \bot
		\end{aligned}
		:
		\quad
		\!
		\begin{aligned}
			\textbf{X} = \En(x, e), &(\GC, e, d) \leftarrow \Gb(\csec, \Ckt), \\
			&\hat{\textbf{Y}} \leftarrow \Adv(\GC, \textbf{X})
		\end{aligned}
		\Bigg)
	\end{equation*}
	
\end{enumerate}

\part{Layer I: MPC Protocols}
\label{part:layer1}

\chapter*{Introduction to Layer I}
\label{chap:layer1_intro}
\begin{quote} \small
	In this part, we provide the details of the Layer I blocks of our three-layer architecture~(\boxref{fig:architecture}). Before going into the details of each of our frameworks, we provide an abstraction of the underlying secret sharing semantics. This is followed by an overview of the basic blocks of our MPC frameworks.
\end{quote}

\subsection*{An Abstraction of Our Sharing Semantics}
\label{sec:semanticsoverview}

To enforce security, we perform computation on secret-shared data.  For the arithmetic and boolean sharing, we follow replicated secret sharing~(RSS), where a value $\vl{v} \in \Z{\ell}$ is split into shares and is denoted by $\shr{\cdot}$. To leverage the benefits of the preprocessing paradigm, we associate meaning to the shares and demarcate the parties in terms of their roles.  The parties are categorized into two sets -- i) $\PartysetO$ - online parties that perform the computation in the online phase, and ii) $\PartysetV$ - verifiers that help in generating preprocessing data and has almost no role in the online phase\footnote{Except operations like input sharing, output reconstruction, final stages of verification etc.}. 

\begin{table}[htb!]
	\centering
	\resizebox{0.9\textwidth}{!}{
		\begin{NiceTabular}{r ||c|c|c || c|c|c|c}[notes/para]
			\toprule
			\Block{2-1}{Framework} & \multicolumn{3}{c}{Parties\tabularnote{$\PartysetO$ - Online parties, $\PartysetV$ - Verifiers,}} 
			& \multicolumn{4}{c}{$\shr{\cdot}$-shares of value $\vl{v}$\tabularnote{$\mk{\vl{v}} = \vl{v} + \pad{\vl{v}}{}$, $\pad{\vl{v}}{} = \pad{\vl{v}}{1} + \pad{\vl{v}}{2}$  or $\pad{\vl{v}}{1} + \pad{\vl{v}}{2}  + \pad{\vl{v}}{3}$}} \\ \cmidrule{2-4} \cmidrule{5-8}
			& $\Partyset$  & $\PartysetO$ & $\PartysetV$ & $P_0$ & $P_1$ & $P_2$ & $P_3$\\
			\midrule
			$\TSthis$ & $P_0, P_1, P_2\hspace*{18pt}$  & $P_1, P_2\hspace*{18pt}$ & $P_0$
			& $\pad{\vl{v}}{1}, \pad{\vl{v}}{2}\hspace*{18pt}$ 
			& $\mk{\vl{v}}, \pad{\vl{v}}{1}\hspace*{18pt}$  
			& $\mk{\vl{v}}, \pad{\vl{v}}{2}\hspace*{18pt}$ 
			& $-$\\[5pt] 
			$\Tthis$ & $\hspace*{18pt}P_1, P_2, P_3$  & $P_1, P_2, P_3$ & $-$  
			& $-$ 
			& $\mk{\vl{v}}, \pad{\vl{v}}{1}, \pad{\vl{v}}{3}$  
			& $\mk{\vl{v}}, \pad{\vl{v}}{2}, \pad{\vl{v}}{3}$ 
			& $\mk{\vl{v}}, \pad{\vl{v}}{1}, \pad{\vl{v}}{2}$\\[5pt] 
			$\Fthis$ & $P_0, P_1, P_2, P_3$  & $P_1, P_2, P_3$ & $P_0$  
			& $\pad{\vl{v}}{1}, \pad{\vl{v}}{2}, \pad{\vl{v}}{3}$ 
			& $\mk{\vl{v}}, \pad{\vl{v}}{1}, \pad{\vl{v}}{3}$  
			& $\mk{\vl{v}}, \pad{\vl{v}}{2}, \pad{\vl{v}}{3}$ 
			& $\mk{\vl{v}}, \pad{\vl{v}}{1}, \pad{\vl{v}}{2}$\\[5pt]  
			$\TWthis$ & $\hspace*{18pt}P_1, P_2\hspace*{18pt}$  & $P_1, P_2\hspace*{18pt}$ & $-$  
			& $-$ 
			& $\mk{\vl{v}}, \pad{\vl{v}}{1}\hspace*{18pt}$  
			& $\mk{\vl{v}}, \pad{\vl{v}}{2}\hspace*{18pt}$ 
			& $-$\\[5pt] 
			\bottomrule
		\end{NiceTabular}
	}
	\caption{Sharing semantics~($\shr{\cdot}$) for value $\vl{v} \in \Z{\ell}$ across various frameworks.}\label{tab:shareabstraction}
\end{table}

For every value $\vl{v} \in \Z{\ell}$, we associate a mask denoted by $\pad{\vl{v}}{}$ and their sum is denoted by the masked value $\mk{\vl{v}}= \vl{v} +\pad{\vl{v}}{}$. The share distribution is done in a specific manner to achieve practical efficiency. The masked value $\mk{\vl{v}}$ is given in clear to all the parties in $\PartysetO$ and the mask $\pad{\vl{v}}{}$ is made available to them in a replicated fashion. For the case when there are $p$ parties in $\PartysetO$, the mask $\pad{\vl{v}}{}$ is split into $p$ shares, denoted by $\pad{\vl{v}}{1},\ldots, \pad{\vl{v}}{p}$, such that $\pad{\vl{v}}{} = \sum_{j=1}^{p}\pad{\vl{v}}{j}$. Each party $P_j \in \PartysetO$ gets all but one share of $\pad{\vl{v}}{}$ guaranteeing privacy.

On the other hand, parties in $\PartysetV$ obtain all the shares of the mask $\pad{\vl{v}}{}$, enabling them to compute $\pad{\vl{v}}{}$ in clear. The parties in $\PartysetV$ are refrained from obtaining the mask $\mk{\vl{v}}$ to ensure privacy. The sharing semantics for our frameworks are summarized in \tabref{shareabstraction}.

The idea of using a {\em masked} evaluation goes back to the work of Lindell et al.~\cite{C:LPSY15} in the context of multi-party garbling over boolean circuits. Here, a {\em masking bit} is assigned to every wire in the circuit to prevent the parties from knowing the actual value on the wire.  Wang et al.~\cite{CCS:WanRanKat17a} adopted this idea to achieve efficient authenticated two-party garbling schemes. Inspired from~\cite{CCS:WanRanKat17a}, Katz et al.~\cite{CCS:KatKolWan18} proposed an $n$-party semi-honest protocol in the dishonest majority setting using the idea of masked evaluation. Concretely, every party holds an $n$-out-of-$n$ secret sharing of a random boolean mask along with the (public) masked value. The resultant protocol is then used to construct an efficient MPC-in-the-head style zero-knowledge protocol. In an orthogonal line of work,  Ben-Efraim et al.~\cite{ACNS:BenNieOmr19} adopted this strategy and improved the online communication of SPDZ-style protocols (dishonest majority) by using function-dependent pre-processing.

\subsection*{The Complete MPC} 
\label{sec:MPCmain} 
In order to compute an arithmetic circuit $\ckt$ over $\Z{\ell}$, parties first invoke the key-setup functionality $\Func[Key]$~(\S\ref{sec:KeySetupprelims}) for the key distribution. The computation is divided mainly into three stages -- i) Input sharing, ii) Evaluation, and iii) Output Reconstruction. Using the description of the $\ckt$, parties prepare the necessary preprocessing data by invoking the preprocessing phase of the respective stages. Concretely, all the mask values~($\pad{}{}$) for every wire in the $\ckt$ along with other input-independent data will be ready after the preprocessing.

During the online phase,  $P_i \in \Partyset$ shares its input $\vl{v}_i$ by executing the input sharing protocol $\prot{\Sh}$. That is, using the mask $\pad{\vl{v}_i}{}$, $P_i$ computes the masked value $\mk{\vl{v}_i}$ and communicates it to the parties in $\PartysetO$. This is followed by the circuit evaluation phase, where parties evaluate the gates in the circuit in the topological order, with addition gates (and multiplication-by-a-constant gates)  being computed locally and multiplication gates being computed via the multiplication protocol $\prot{\Mult}$. At every gate output wire $\vl{z}$, the goal is to compute the masked value~($\mk{\vl{z}}$) using the shares of the input wires. Finally, parties execute the reconstruction protocol $\prot{\Rec}$ on the output wires to reconstruct the function output.

\subsection*{Other blocks in Layer I}
\label{sec:otherblocks_layer1}

\paragraph{Truncation}
Repeated multiplications in Fixed-Point Arithmetic~(FPA) result in an overflow with the fractional part doubling up in size after each multiplication. This can result in the loss of significant bits of information eventually. The naive solution of choosing a large enough ring to avoid the overflow is impractical for ML algorithms where the number of sequential multiplications is large. To tackle this, truncation~\cite{SP:MohZha17,CCS:MohRin18,NDSS:PatSur20,NDSS:ChaRacSur20,PoPETS:BCPS20,USENIX:KPPS21} is used where the result of the multiplication is brought back to the FPA representation by chopping off the last $x$ bits. 

For a value $\vl{v} = \vl{v}_1 + \vl{v}_2$, SecureML~\cite{SP:MohZha17} showed that the truncated value $\vl{v}/2^x$, denoted by $\vl{v}^\vl{t}$, can be computed as $\vl{v}_1^{\vl{t}} + \vl{v}_2^{\vl{t}}$. With high probability, a truncated value having at most one bit error in the least significant position is generated. It was shown in SecureML that accuracy drop for ML algorithms due to the one bit error is minimal. However, the method cannot be generalized to more than two parties. 
ABY3~\cite{CCS:MohRin18} demonstrated the extension to 3-party setting with a generic design that uses a truncation pair of the form $(\vl{r}, {\vl{r}}^{\vl{t}})$. Here, $\vl{r}$ is a random value and $\vl{r}^{\vl{t}}$ denotes its truncated version. Given this pair, $\vl{z}$ can be truncated by opening $\vl{z} - \vl{r}$ towards all, and computing $\vl{z}^{\vl{t}}$ as $\vl{z}^{\vl{t}} = (\vl{z-r})^{\vl{t}} + \vl{r}^{\vl{t}}$. Note that all operations are carried out on shares. The design of our multiplication protocol allows for truncation to be carried out this way without any additional overhead in communication. 

\paragraph{Multi-input Multiplication}
Given the $\shr{\cdot}$-shares of values, $\vl{a}, \vl{b}, \vl{c}, \vl{d} \in \Z{\ell}$, we design 3-input and 4-input multiplication protocols in our frameworks. For the three-input case, the goal is to compute $\vl{z} = \vl{abc}$, without the need for performing two sequential multiplications~(i.e. first $\vl{y} = \vl{ab}$ then $\vl{yc}$).  Similarly, $\vl{z} = \vl{abcd}$ for the four-input case. We remark that our multi-input multiplication, when coupled with the optimized parallel prefix adder circuit from~\cite{USENIX:PSSY21}, brings in a $2\times$ improvement in online rounds, as well as an improvement in online communication of secure comparison, as will be shown later in the thesis.

\paragraph{NOT operation in Boolean world}
Given the boolean shares of a bit $\bitb \in \{0,1\}$, denoted by $\shrB{\bitb}$, parties can locally compute the boolean shares corresponding to its complement $\overline{\bitb}$. For this, parties locally set $\mk{\overline{\bitb}} = 1 \xor \mk{\bitb}$ and the $\pad{\overline{\bitb}}{}$ shares are set to be the same as $\pad{\bitb}{}$. It is easy to verify that $\overline{\bitb} = \mk{\overline{\bitb}} \xor \pad{\overline{\bitb}}{} = (1 \xor \mk{\bitb}) \xor \pad{\bitb}{} = 1 ( (\mk{\bitb} \xor \pad{\bitb}{}) = 1 \xor \bitb$. We use $\NOTBool$ to denote this operation.

\paragraph{Garbled World}
In our frameworks, we build GC-based protocol, tailor-made for PPML applications where only a small portion of the computation is done over the garbled world. We propose 2 GC protocols -- one requiring communication of 2 GC evaluations and one online round, and the other one requiring 1 GC and two rounds. 

Garbled evaluation proceeds in three phases-- i) Input phase,  ii) Evaluation, and iii) Output phase. The input phase involves transferring the keys to the evaluators for every input to the GC. The evaluation consists of GC transfer followed by GC evaluation.  Lastly, in the output phase, evaluators obtain the encoded output. Moreover, the state-of-the-art GC optimizations of free-XOR~\cite{ICALP:KolSch08,C:KolMohRos14}, half gates~\cite{EC:ZahRosEva15,JC:GLNP18}, and fixed AES-key~\cite{SP:BHKR13} are deployed in our protocols.  

Preliminary details about the garbling scheme and properties are described in \S\ref{sec:garbledprelims}. In the thesis, to simplify the presentation, we assume single bit values; for $\ell$-bit values, each operation is performed $\ell$ times in parallel.

\chapter{$\TSthis$: 3PC Semi-honest Protocols}
\label{chap:layer1_3pcsemi}
This chapter provides details for the Layer I blocks of our 3PC framework $\TSthis$. Some of the results in this chapter resulted in a publication at ACM CCSW'19~\cite{CCSW:CCPS19}. 
Comparison of $\TSthis$ with passively secure 3PC PPML framework of ABY3~\cite{CCS:MohRin18}, in terms of the communication for multiplication,  is presented in \tabref{3pcSCost}.

\begin{table}[htb!]
	\centering
	\resizebox{0.98\textwidth}{!}{
		\begin{NiceTabular}{r c r|r r|r r|c}
			\toprule
			\Block{2-1}{Work}
			& \Block[c]{2-1}{\#Active\\Parties}
			& \Block{2-1}{Security}
			& \multicolumn{2}{c}{Multiplication} 
			& \multicolumn{2}{c}{Multiplication with Truncation\tabularnote{$\ell$ - size of ring in bits, $x$ - number of bits for the fractional part in FPA semantics.}} 
			& \Block{2-1}{Conversions\tabularnote{A, B, G indicate support for arithmetic, boolean, and garbled worlds respectively.}} \\ \cmidrule{4-7}
			&  & 
			& Comm\textsubscript{pre} 
			& Comm\textsubscript{on}\tabularnote{`Comm' - communication, `pre' - preprocessing, `on' - online} 
			& Comm\textsubscript{pre} 
			& Comm\textsubscript{on} &  \\ 
			\midrule
			ABY3~\cite{CCS:MohRin18} & 3 & Semi-honest & $-$ & $3\ell$ & $14\ell -6x -6$ & $4\ell$ & A-B-G\\		
			\textbf{$\TSthis$} & 2 & Semi-honest & $\ell$ & $2\ell$ & $\ell$ & $2\ell$ & A-B-G\\	
			\bottomrule
		\end{NiceTabular}
    }
	\caption{Comparison of semi-honest 3PC frameworks for PPML}\label{tab:3pcSCost}
\end{table}

\section{Preliminaries and Definitions}
\label{sec:3pcSPrelim}
We consider $3$ parties denoted by $\Partyset = \{ P_0, P_1, P_2 \}$ that are connected by pair-wise private and authentic channels in a synchronous network, and a static, semi-honest adversary that can corrupt at most one party.

\subsection{Sharing Semantics}
\label{sec:3pcSsematics}
For the arithmetic and boolean sharing, we follow  a $(3, 1)$ replicated secret sharing~(RSS), where a value $\vl{v} \in \Z{\ell}$ is split into three shares. Two of the shares~($\pad{\vl{v}}{1}, \pad{\vl{v}}{2})$ can be generated in the preprocessing phase independent of the value to be shared, and their sum can be interpreted as a mask~($\pad{\vl{v}}{}$). The third share, dependent on $\vl{v}$,  can be computed in the online phase and can be treated as the masked value $\mk{\vl{v}}= \vl{v} +\pad{\vl{v}}{}$. 

\begin{table}[htb!]
        \centering
		\begin{NiceTabular}{r r r r}[notes/para]
			\toprule
			Sharing Type  & $P_0$ & $P_1$ & $P_2$\\
			\midrule
			$\sqr{\cdot}$-sharing\tabularnote{$\vl{v} = \vl{v}^1 + \vl{v}^2$} 
			& $-$      & ${\vl{v}}^1$     & ${\vl{v}}^2$\\
			$\shr{\cdot}$-sharing\tabularnote{$\pad{\vl{v}}{} = \pad{\vl{v}}{1} + \pad{\vl{v}}{2}$, $\mk{\vl{v}} = \vl{v} + \pad{\vl{v}}{}$}   
			& $(\pad{\vl{v}}{1}, \pad{\vl{v}}{2})$
			& $(\mk{\vl{v}}, \pad{\vl{v}}{1})$  
			& $(\mk{\vl{v}}, \pad{\vl{v}}{2})$ \\
			\bottomrule
		\end{NiceTabular}
	\caption{Semantics for $\vl{v} \in \Z{\ell}$ in \TSthis.}\label{tab:3pcSsharing}
\end{table}

Next, we distinguish the three parties into two sets; the {\em eval} set $\SetE = \{P_1,P_2\}$ which is assigned the task of carrying out the computation, and is active throughout the online phase. The {\em helper} set $\SetD = \{P_0\}$, is used to assist $\SetE$ in preparing the preprocessing material, and so it is only active in the preprocessing phase. Complying with the roles and RSS format, the distribution is done as follows: $P_0: \{\pad{\vl{v}}{1}, \pad{\vl{v}}{2}\}, P_1: \{\pad{\vl{v}}{1}, \mk{\vl{v}}\}$, and $P_2: \{\pad{\vl{v}}{2}, \mk{\vl{v}}\}$. 

The RSS sharing semantics is presented in \tabref{3pcSsharing}, denoted by $\shr{\cdot}$, along with the semantics for $\sqr{\cdot}$-sharing. Both the sharings used are linear i.e. given sharings of $\vl{v}_1,\ldots, \vl{v}_m$ and public constants $c_1,\ldots,c_m$, sharing of $\sum_{i=1}^m c_i \vl{v}_i$ can be computed non-interactively for an integer $m$.

\begin{notation} \label{notation:3pcSconcise}
	(a) For the $\shr{\cdot}$-shares of $n$ values $\vl{a}_1,\ldots,\vl{a}_n$, $\gm{\vl{a}_1 \ldots \vl{a}_n}{} = \prod\limits_{i=1}^{n} \pad{\vl{a}_i}{}$ and $\mk{\vl{a}_1 \ldots \vl{a}_n}{} = \prod\limits_{i=1}^{n} \mk{\vl{a}_i}{}$ (b) We use superscripts ${\bf B}$, and ${\bf G}$ to denote sharing semantics in boolean, and garbled world, respectively-- $\shrB{\cdot}$,  $\shrG{\cdot}$. We omit the superscript for arithmetic world. 
\end{notation}

Sharing semantics for boolean sharing over $\Z{}$ is similar to arithmetic sharing except that addition is replaced with XOR. The semantics for garbled sharing are described in \S\ref{sec:3pcSGCWorld} with the relevant context.

\section{Arithmetic / Boolean 3PC}
\label{sec:3pcSFourPC}

This section covers the details of our 3PC semi-honest protocol $\TSthis$ over an arithmetic ring $\Z{\ell}$. The protocol primarily consists of the following primitives -- i) Sharing~\secref{share3pcS}, ii) Multiplication~\secref{mult3pcS}, and iii) Reconstruction~\secref{rec3pcS}. 

\subsection{Sharing}
\label{sec:share3pcS}
Protocol $\prot{\Sh}$~(\boxref{fig:piSh3pcS}) enables $P_i$ to generate $\shr{\cdot}$-share of a value $\vl{v}$. During the preprocessing phase, $\pd{}$-shares are sampled non-interactively using the pre-shared keys~(cf. \S\ref{sec:KeySetupprelims}) in a way that $P_i$ will get the entire mask $\pd{\vl{v}}$. During the online phase, $P_i$ computes $\mk{\vl{v}} = \vl{v} + \pd{\vl{v}}$ and sends to $P_1, P_2$. 
For the special case when $P_0$ wants to perform a $\shr{\cdot}$-sharing of $\vl{v}$ in the preprocessing, the communication can be optimized further. For this, parties set $\mk{\vl{v}} = 0$. $P_0, P_1$ sample $\pad{\vl{v}}{1}$ non-interactively. $P_0$ computes and sends $\pad{\vl{v}}{2} = - (\vl{v} + \pad{\vl{v}}{1})$ to $P_2$.

\begin{protocolbox}{$\prot{\Sh}(P_i, \vl{v})$}{$\shr{\cdot}$-sharing of a value $\vl{v}$ by party $P_i$ in $\TSthis$.}{fig:piSh3pcS}
	\detail{
		{\bf Input(s):} $P_i : \vl{v}$,~~{\bf Output:} $\shr{\vl{v}}$.
	}
	\justify
	\algoHead{Preprocessing:} 
	Sample as follows: $P_i, P_0, P_1: \pad{\vl{v}}{1}$,~~$P_i, P_0, P_2: \pad{\vl{v}}{2}$.
	\justify
	\vspace{-2mm}
	\algoHead{Online:} $P_i$ computes $\mk{\vl{v}} = \vl{v} + \pd{\vl{v}}$ and sends to $P_1, P_2$.  
\end{protocolbox}

\begin{lemma}[Communication]
	\label{lemma:pish3pcS}
	Protocol $\prot{\Sh}$~(\boxref{fig:piSh3pcS}) requires a communication of at most $2\ell$ bits and $1$ round in the online phase.
\end{lemma}
\begin{proof}
	The preprocessing of $\prot{\Sh}$ is non-interactive as the parties sample non interactively using key setup $\Func[Key]$~(\S\ref{sec:KeySetupprelims}). In the online phase, $P_i$ sends $\mk{\vl{v}}$ to $P_1, P_2$ resulting in 1 round and communication of at most $2\ell$ bits~($P_i = P_0$).
\end{proof}

\subsubsection{Joint Sharing}
\label{sec:jsh3pcS}
Protocol $\prot{\JSh}$ enables parties $P_i, P_j$ to generate $\shr{\cdot}$-share of a value $\vl{v}$. In $\TSthis$, protocol $\prot{\JSh}$ is used to enable $P_1, P_2$ generate $\shr{\vl{v}}$ non-interactively. For this, parties set $\pad{\vl{v}}{1} = \pad{\vl{v}}{2} = 0$ and $\mk{\vl{v}} = \vl{v}$.

\subsection{Multiplication}
\label{sec:mult3pcS}
Given the shares of $\vl{a}, \vl{b}$, the goal of the multiplication protocol is to generate shares of $\vl{z} = \vl{ab}$. The protocol is designed such that parties $P_1, P_2$ obtain a masked version of the output $\vl{z}$, say $\vl{z} - \vl{r}$ in the online phase, and $P_0$ obtain the mask $\vl{r}$ in the preprocessing phase. Parties then generate $\shr{\cdot}$-sharing of these values, and locally compute $\shr{\vl{z} - \vl{r}} + \shr{\vl{r}}$ to obtain the final output. 

\paragraph{Online} 
Note that,

\begin{align}\label{eq:3pcSmult}
	\vl{z} - \vl{r} &= \vl{ab} - \vl{r} = (\mk{\vl{a}} - \pd{\vl{a}})(\mk{\vl{b}} - \pd{\vl{b}}) - \vl{r} \nonumber\\ 
	&= \mk{\vl{ab}} - \mk{\vl{a}}\pd{\vl{b}} - \mk{\vl{b}}\pd{\vl{a}} + \gm{\vl{ab}}{} - \vl{r}
	~~\text{\footnotesize{(cf. notation~\ref{notation:3pcSconcise})}}
\end{align}

In Eq~\ref{eq:3pcSmult}, $P_1, P_2$ can compute $\mk{\vl{ab}}$ locally, and hence we are interested in computing $\vl{y} = (\vl{z - r}) - \mk{\vl{ab}}$. Let $\vl{y} = \vl{y}_1 + \vl{y}_2$, where $\vl{y}_1$ and $\vl{y}_2$ can be computed respectively by $P_1$ and $P_2$.

\begin{align}\label{eq:3pcSmulty}
	P_1: \vl{y}_1 &= - \pad{\vl{a}}{1} \mk{\vl{b}} - \pad{\vl{b}}{1} \mk{\vl{a}} + \sqr{\gm{\vl{ab}}{} - \vl{r}}_1 \nonumber\\
	P_2: \vl{y}_2 &= - \pad{\vl{a}}{2} \mk{\vl{b}} - \pad{\vl{b}}{2} \mk{\vl{a}} + \sqr{\gm{\vl{ab}}{} - \vl{r}}_2 
\end{align}

The preprocessing is set up such that $P_1, P_2$ receive an additive sharing~($\sqr{\cdot}$) of $\gm{\vl{ab}}{} - \vl{r}$. Parties $P_1, P_2$ mutually exchange the missing share to reconstruct $\vl{y}$ and subsequently $\vl{z - r}$. 

\smallskip
\begin{protocolbox}{$\prot{\Mult}(\vl{a}, \vl{b}, \isTr)$}{Multiplication with / without truncation in $\TSthis$.}{fig:piMult3pcS}
	$\isTr$ is a bit denoting whether truncation is required ($\isTr =1$) or not ($\isTr=0$). \\
	\detail{
		{\bf Input(s):} $\shr{\vl{a}}, \shr{\vl{b}}$.\\
		{\bf Output:} $\shr{\vl{o}}$ where $\vl{o} = \vl{z}^{\vl{t}}$ if $\isTr = 1$ and $\vl{o} = \vl{z}$ if $\isTr = 0$ and $\vl{z} = \vl{ab}$.
	}
	\justify 
	\vspace{-2mm}
	\algoHead{Preprocessing:} 
	\begin{enumerate}
		\item $P_0, P_j$ sample ${\vl{u}}^j \in_R \Z{\ell}$ for $j \in \{1,2\}$. Let ${\vl{u}^1} + \vl{u}^2 = \gm{\vl{ab}}{} - \vl{r}$ for $\vl{r} \in_R \Z{\ell}$.  
		\item Party $P_0$: Computes $\vl{r} = \gm{\vl{ab}}{} - {\vl{u}^1} - \vl{u}^2$. If $\isTr = 1$, sets $\vl{q} = \vl{r}^{\vl{t}}$, else $\vl{q} = \vl{r}$.\newline Executes $\prot{\Sh}(P_0, \vl{q})$ to generate $\shr{\vl{q}}$.
	\end{enumerate}
	\justify
	\vspace{-2mm}
	\algoHead{Online:} Let $\vl{y} = (\vl{z} - \vl{r}) - \mk{\vl{ab}}$.
	\begin{enumerate}
		\item Compute: $P_1: \vl{y}_1 = - \pad{\vl{a}}{1} \mk{\vl{b}} - \pad{\vl{b}}{1} \mk{\vl{a}} + {\vl{u}}^1,~~
		P_2: \vl{y}_2 = - \pad{\vl{a}}{2} \mk{\vl{b}} - \pad{\vl{b}}{2} \mk{\vl{a}} + {\vl{u}}^2$
		\item $P_1$ sends $\vl{y}_1$ to $P_2$, while $P_2$ sends $\vl{y}_2$ to $P_1$, and they locally compute $\vl{z} - \vl{r} = \vl{y}_1 + \vl{y}_2 + \mk{\vl{ab}}$.
		\item $P_1, P_2$: If $\isTr = 1$, set $\vl{p} = (\vl{z} - \vl{r})^{\vl{t}}$, else $\vl{p} = \vl{z} - \vl{r}$. Execute $\prot{\JSh}(P_1, P_2, \vl{p})$ to generate $\shr{\vl{p}}$. 
		\item Compute $\shr{\vl{o}} = \shr{\vl{p}} + \shr{\vl{q}}$. Here $\vl{o} = \vl{z}^{\vl{t}}$ if $\isTr = 1$ and $\vl{z}$ otherwise.
	\end{enumerate}     
\end{protocolbox}

\paragraph{Preprocessing} 
Parties $P_1, P_2$ should obtain $\sqr{\gm{\vl{ab}}{} - \vl{r}}$ while $P_0$ should obtain $\vl{r}$. For this, $P_0, P_i$ for $i \in \{1,2\}$ non-interactively sample $\sqr{\gm{\vl{ab}}{} - \vl{r}}_i$. This enables $P_0$ to obtain $\vl{r}$ in clear as it can compute $\gm{\vl{ab}}{}$ locally.

\begin{lemma}[Communication]
	\label{lemma:piMult3pcS}
	Protocol $\prot{\Mult}$~(\boxref{fig:piMult3pcS})~(in $\TSthis$) requires $\ell$ bits of communication in the preprocessing, and $1$ round and $2 \ell$ bits of communication in the online phase.
\end{lemma}
\begin{proof}
	During preprocessing, sampling of ${\vl{u}}^1, {\vl{u}}^2$ are performed non-interactively using $\Func[Key]$. A communication of $\ell$ bits is required for the sharing of $\vl{q}$ by $P_0$.
	During online, $P_1, P_2$ exchange $\vl{y}_1, \vl{y}_2$ values  in parallel resulting in a communication of $2\ell$ bits and 1 round. 
\end{proof}

\subsubsection{Truncation}
To accommodate truncation, the multiplication protocol is modified as follows. $P_1, P_2$ locally truncate $(\vl{z - r})$ and generate $\shr{\cdot}$-shares of it in the online phase. Similarly, $P_0$ truncates $\vl{r}$ in the preprocessing  and generates its $\shr{\cdot}$-shares. Parties locally compute $\shr{\vl{z}^{\vl{t}}} = \shr{(\vl{z-r})^{\vl{t}}} + \shr{\vl{r}^{\vl{t}}}$.

\subsubsection{Multiplication with constant}
Multiplication by a constant in MPC is typically local. Given constant $\alpha$ and $\shr{\vl{v}}$, the $\shr{\cdot}$-shares of the product $\vl{y} = \alpha\vl{v}$ can be locally computed as per \eqref{eq:mutconst3pcS}. 
\begin{equation}\label{eq:mutconst3pcS}
	\mk{\vl{y}} = \alpha \mk{\vl{u}},~~~\pad{\vl{y}}{1} = \alpha \pad{\vl{v}}{1},~~~\pad{\vl{y}}{2} = \alpha \pad{\vl{v}}{2}
\end{equation}

However, in FPA, we need to perform a truncation on the output. Let $\alpha\vl{v} = \beta^1 + \beta^2$ where $\beta^1 = \alpha.\mk{\vl{v}}$ and $\beta^2 =  \alpha.(-\pad{\vl{v}}{1} - \pad{\vl{v}}{2})$. $P_1, P_2$ truncate $\beta^{1}$ and generate its arithmetic sharing using $\prot{\JSh}$, while $P_0$ does the same with $\beta^{2}$.

\subsection{Reconstruction}
\label{sec:rec3pcS}
Protocol $\prot{\Rec}(\Partyset, \vl{v})$~(\boxref{fig:piRec3pcS}) enables parties in $\Partyset$ to compute $\vl{v}$, given its $\shr{\cdot}$-share. Note that each party misses one share to reconstruct the output, and the other two parties hold this share. One out of the two parties will send the missing share to the party that lacks it. Reconstruction towards a single party can be viewed as a special case.

\begin{protocolbox}{$\prot{\Rec}(\Partyset, \shr{\vl{v}})$}{Reconstruction of value $\vl{v}$ among $\Partyset$ in $\TSthis$.}{fig:piRec3pcS}
	\justify
	\detail{
		{\bf Input(s):} $\shr{\vl{v}}$,~~{\bf Output:} $\vl{v}$.
	}
	\begin{enumerate}[itemsep=0mm]
		\item $P_0$ sends $\pad{\vl{v}}{1}$ to $P_2$;~~~$P_0$ sends $\pad{\vl{v}}{2}$ to $P_1$;~~~$P_1$ sends $\mk{\vl{v}}$ to $P_0$.
		\item Compute $\vl{v} = \mk{\vl{v}} - \pad{\vl{v}}{1} - \pad{\vl{v}}{2}$. 
	\end{enumerate}
\end{protocolbox}

\begin{lemma}[Communication]
	\label{lemma:pirec3pcS}
	Protocol $\prot{\Rec}$~(\boxref{fig:piRec3pcS}) requires a communication of $3\ell$ bits and $1$ round in the online phase.
\end{lemma}

\subsection{Multi-input Multiplication}
\label{sec:multT3pcS}

\paragraph{3-input multiplication}
To compute $\shr{\cdot}$-shares of $\vl{z} = \vl{abc}$, note that  

\begin{align}\label{eq:3pcSmultT}
	\vl{z} - \vl{r} &= \vl{abc} - \vl{r} = (\mk{\vl{a}} - \pd{\vl{a}})(\mk{\vl{b}} - \pd{\vl{b}})(\mk{\vl{c}} - \pd{\vl{c}}) - \vl{r} \nonumber\\ 
	&= \mk{\vl{abc}} - \mk{\vl{ac}}\pd{\vl{b}} - \mk{\vl{bc}}\pd{\vl{a}} - \mk{\vl{ab}}\pd{\vl{c}} + \mk{\vl{a}}\gm{\vl{bc}}{} + \mk{\vl{b}}\gm{\vl{ac}}{} + \mk{\vl{c}}\gm{\vl{ab}}{} - \gm{\vl{abc}}{} - \vl{r}
	~~\text{\footnotesize{(cf. notation~\ref{notation:3pcSconcise})}}
\end{align}

Similar to $\prot{\Mult}$, for $\vl{y} = (\vl{z - r}) - \mk{\vl{abc}}$, let $\vl{y} = \vl{y}_1 + \vl{y}_2$, where $\vl{y}_1$ and $\vl{y}_2$ can be computed respectively by $P_1$ and $P_2$.

\begin{align}\label{eq:3pcSmultTy}
	P_1: \vl{y}_1 &= - \pad{\vl{a}}{1} \mk{\vl{bc}} - \pad{\vl{b}}{1} \mk{\vl{ac}} - \pad{\vl{c}}{1} \mk{\vl{ab}} + \sqr{\gm{\vl{ab}}{}}_1 \mk{\vl{c}} + \sqr{\gm{\vl{ac}}{}}_1 \mk{\vl{b}} + \sqr{\gm{\vl{bc}}{}}_1 \mk{\vl{a}} - \sqr{\gm{\vl{abc}}{} + \vl{r}}_1 \nonumber\\
	P_2: \vl{y}_2 &= - \pad{\vl{a}}{2} \mk{\vl{bc}} - \pad{\vl{b}}{2} \mk{\vl{ac}} - \pad{\vl{c}}{2} \mk{\vl{ab}} + \sqr{\gm{\vl{ab}}{}}_1 \mk{\vl{c}} + \sqr{\gm{\vl{ac}}{}}_2 \mk{\vl{b}} + \sqr{\gm{\vl{bc}}{}}_2 \mk{\vl{a}} - \sqr{\gm{\vl{abc}}{} + \vl{r}}_2 
\end{align}

To generate $\sqr{\vl{x}}$ for $\vl{x} \in \{\gm{\vl{ab}}{}, \gm{\vl{bc}}{}, \gm{\vl{ac}}{}\}$, $P_0, P_1$ non-interactively sample $P_1$'s share. $P_0$ computes the share of $P_2$ and communicates to it. The generation of $\sqr{\gm{\vl{abc}}{} + \vl{r}}$ and the rest of the steps follow similar to that of 2-input multiplication protocol $\prot{\Mult}$ in \S\ref{sec:mult3pcS}. The formal protocol appears in \boxref{fig:piMultT3pcS}.

\smallskip
\begin{protocolsplitbox}{$\prot{\MultT}(\vl{a}, \vl{b}, \vl{c}, \isTr)$}{Three-input Multiplication with / without truncation in $\TSthis$.}{fig:piMultT3pcS}
	$\isTr$ is a bit denoting whether truncation is required ($\isTr =1$) or not ($\isTr=0$). \\
	\detail{
		{\bf Input(s):} $\shr{\vl{a}}, \shr{\vl{b}}, \shr{\vl{c}}$.\\
		{\bf Output:} $\shr{\vl{o}}$ where $\vl{o} = \vl{z}^{\vl{t}}$ if $\isTr = 1$ and $\vl{o} = \vl{z}$ if $\isTr = 0$ and $\vl{z} = \vl{abc}$.
	}
	\justify 
	\vspace{-2mm}
	\algoHead{Preprocessing:} 
	\begin{enumerate}
		\item For each $\vl{x} \in \{\gm{\vl{ab}}{}, \gm{\vl{bc}}{}, \gm{\vl{ac}}{}\}$, $P_0, P_1$ sample ${\vl{x}}^1 \in_R \Z{\ell}$. $P_0$ computes and sends ${\vl{x}}^2 = {\vl{x}} - {\vl{x}}^1$ to $P_2$.
		\item $P_0, P_j$ sample ${\vl{u}}^j \in_R \Z{\ell}$ for $j \in \{1,2\}$. Let ${\vl{u}^1} + \vl{u}^2 = \gm{\vl{abc}}{} + \vl{r}$ for $\vl{r} \in_R \Z{\ell}$.  
		\item Party $P_0$: Computes $\vl{r} = {\vl{u}^1} + \vl{u}^2 - \gm{\vl{abc}}{}$. If $\isTr = 1$, sets $\vl{q} = \vl{r}^{\vl{t}}$, else $\vl{q} = \vl{r}$.\newline Executes $\prot{\Sh}(P_0, \vl{q})$ to generate $\shr{\vl{q}}$.
	\end{enumerate}
	\justify
	\vspace{-2mm}
	\algoHead{Online:} Let $\vl{y} = (\vl{z} - \vl{r}) - \mk{\vl{ab}}$.
	\begin{enumerate}
		\item Locally compute: 
		\begin{align*}
		      P_1: \vl{y}_1 &= - \pad{\vl{a}}{1} \mk{\vl{bc}} - \pad{\vl{b}}{1} \mk{\vl{ac}} - \pad{\vl{c}}{1} \mk{\vl{ab}} + \sqr{\gm{\vl{ab}}{}}_1 \mk{\vl{c}} + \sqr{\gm{\vl{ac}}{}}_1 \mk{\vl{b}} + \sqr{\gm{\vl{bc}}{}}_1 \mk{\vl{a}} - {\vl{u}}^1,\\
		      P_2: \vl{y}_2 &= - \pad{\vl{a}}{2} \mk{\vl{bc}} - \pad{\vl{b}}{2} \mk{\vl{ac}} - \pad{\vl{c}}{2} \mk{\vl{ab}} + \sqr{\gm{\vl{ab}}{}}_1 \mk{\vl{c}} + \sqr{\gm{\vl{ac}}{}}_2 \mk{\vl{b}} + \sqr{\gm{\vl{bc}}{}}_2 \mk{\vl{a}} - {\vl{u}}^2
		\end{align*}
		\item $P_1$ sends $\vl{y}_1$ to $P_2$, while $P_2$ sends $\vl{y}_2$ to $P_1$, and they locally compute $\vl{z} - \vl{r} = \vl{y}_1 + \vl{y}_2 + \mk{\vl{ab}}$.
		\item $P_1, P_2$: If $\isTr = 1$, set $\vl{p} = (\vl{z} - \vl{r})^{\vl{t}}$, else $\vl{p} = \vl{z} - \vl{r}$. Execute $\prot{\JSh}(P_1, P_2, \vl{p})$ to generate $\shr{\vl{p}}$. 
		\item Compute $\shr{\vl{o}} = \shr{\vl{p}} + \shr{\vl{q}}$. Here $\vl{o} = \vl{z}^{\vl{t}}$ if $\isTr = 1$ and $\vl{z}$ otherwise.
	\end{enumerate}     
\end{protocolsplitbox}

\begin{lemma}[Communication]
	\label{lemma:piMultT3pcS}
	Protocol $\prot{\MultT}$~(\boxref{fig:piMultT3pcS})~(in $\TSthis$) requires $4\ell$ bits of communication in the preprocessing, and $1$ round and $2 \ell$ bits of communication in the online phase.
\end{lemma}
\begin{proof}
	During preprocessing, $\ell$ bits of communication from $P_0$ to $P_2$ is required to generate $\sqr{\cdot}$-shares of each of $\gm{\vl{ab}}{}, \gm{\vl{bc}}{}$, and $\gm{\vl{ac}}{}$. The sampling of ${\vl{u}}^1, {\vl{u}}^2$ are performed non-interactively using $\Func[Key]$. Another $\ell$ bits are required for the sharing of $\vl{q}$ by $P_0$.
	During online, $P_1, P_2$ exchange $\vl{y}_1, \vl{y}_2$ values  in parallel resulting in a communication of $2\ell$ bits and 1 round. 
\end{proof}

\paragraph{4-input multiplication}
For the case of 4-input multiplication with $\vl{z} = \vl{abcd}$, note that  

\begin{align}\label{eq:3pcSmultF}
	\vl{z} - \vl{r} &= \vl{abcd} - \vl{r} = (\mk{\vl{a}} - \pd{\vl{a}})(\mk{\vl{b}} - \pd{\vl{b}})(\mk{\vl{c}} - \pd{\vl{c}})(\mk{\vl{d}} - \pd{\vl{d}}) - \vl{r} \nonumber\\ 
	&= \mk{\vl{abcd}} - \mk{\vl{abc}}\pd{\vl{d}} - \mk{\vl{abd}}\pd{\vl{c}} - \mk{\vl{acd}}\pd{\vl{b}} - \mk{\vl{bcd}}\pd{\vl{a}} 
	     + \mk{\vl{ab}}\gm{\vl{cd}}{} + \mk{\vl{ac}}\gm{\vl{bd}}{} + \mk{\vl{ad}}\gm{\vl{bc}}{} + \mk{\vl{bc}}\gm{\vl{ad}}{} \nonumber\\
	     &~~~+ \mk{\vl{bd}}\gm{\vl{ac}}{} + \mk{\vl{cd}}\gm{\vl{ab}}{} 
	     - \mk{\vl{a}}\gm{\vl{bcd}}{} - \mk{\vl{b}}\gm{\vl{acd}}{} - \mk{\vl{c}}\gm{\vl{abd}}{} - \mk{\vl{d}}\gm{\vl{abc}}{} + \gm{\vl{abcd}}{} - \vl{r}
	~~\text{\footnotesize{(cf. notation~\ref{notation:3pcSconcise})}}
\end{align}

Here the parties need to generate $\sqr{\cdot}$-shares of $\gm{\vl{ab}}{},\gm{\vl{ac}}{},\gm{\vl{ad}}{},\gm{\vl{bc}}{},\gm{\vl{bd}}{},\gm{\vl{cd}}{},\gm{\vl{abc}}{},\gm{\vl{abd}}{},\gm{\vl{acd}}{},\gm{\vl{bcd}}{}$ and $\gm{\vl{abcd}}{} - \vl{r}$. This is computed similarly as in 3-input multiplication and the protocol is denoted as $\prot{\MultF}$.

\begin{lemma}[Communication]
	\label{lemma:piMultF3pcS}
	Protocol $\prot{\MultF}$~(in $\TSthis$) requires $11\ell$ bits of communication in the preprocessing, and $1$ round and $2 \ell$ bits of communication in the online phase.
\end{lemma}
\begin{proof}
	During preprocessing, $\ell$ bits of communication from $P_0$ to $P_2$ is required to generate $\sqr{\cdot}$-shares of each of the ten values $\gm{\vl{ab}}{},\gm{\vl{ac}}{},\gm{\vl{ad}}{},\gm{\vl{bc}}{},\gm{\vl{bd}}{},\gm{\vl{cd}}{},\gm{\vl{abc}}{},\gm{\vl{abd}}{},\gm{\vl{acd}}{},\gm{\vl{bcd}}{}$. The sampling of ${\vl{u}}^1, {\vl{u}}^2$ are performed non-interactively using $\Func[Key]$. A communication of $\ell$ bits is required for the sharing of $\vl{q}$ by $P_0$.
	During online, $P_1, P_2$ exchange $\vl{y}_1, \vl{y}_2$ values  in parallel resulting in a communication of $2\ell$ bits and 1 round. 
\end{proof}

\paragraph{$N$-input multiplication}
Consider an $N$-input multiplication gate with inputs $\vl{a}_1,\dots,\vl{a}_N$ and output $\vl{z}$. Then, we can write
\begin{equation}
	\vl{z - r} = \prod_{j=1}^{N} (\mk{\vl{a}_j} - \pad{\vl{a}_j}{}) - \vl{r} = \left(  \sum_{I \subseteq \{1,\dots,N\} }(-1)^{|I|}  \prod_{j \in I} \pad{\vl{a}_j}{} \prod_{k\notin I} \mk{\vl{a}_j} \right) - \vl{r}
\end{equation}

Here $I \subseteq \{1,\dots,N\}$ denotes a subset of indices from 1 to $N$, while $|I|$ denotes the cardinality of the set. 

We note that for an $N$-Input multiplication gate, we would require a total of $2^N - N - 1$ terms to be processed in the preprocessing, while the online phase still requires a communication of just $2$ ring elements. Hence, to maintain a balance between the online communication and the overhead in the preprocessing, we consider $N = 3$ and $N = 4$ in our platform.

\subsection{Supporting on-demand computations}
\label{sec:nopre3pcS}
For on-demand applications where the underlying function to be computed is not known in advance, the preprocessing model is not desirable. We observe that the $\TSthis$ protocol can be modified by executing the preprocessing steps in the online phase itself, keeping the same overall communication cost and online rounds. The formal protocol appears in \boxref{fig:piMultOn3pcS}.

\smallskip
\begin{protocolsplitbox}{$\piMultO(\vl{a}, \vl{b}, \isTr)$}{Multiplication for on-demand applications  in $\TSthis$.}{fig:piMultOn3pcS}
	$\isTr$ is a bit denoting whether truncation is required ($\isTr =1$) or not ($\isTr=0$). \\
	\detail{
		{\bf Input(s):} $\shr{\vl{a}}, \shr{\vl{b}}$.\\
		{\bf Output:} $\shr{\vl{o}}$ where $\vl{o} = \vl{z}^{\vl{t}}$ if $\isTr = 1$ and $\vl{o} = \vl{z}$ if $\isTr = 0$ and $\vl{z} = \vl{ab}$.
	}
	\justify 
	\vspace{-2mm}
	\algoHead{Online:} 
	\begin{enumerate}
		\item $P_0, P_j$ sample ${\vl{u}}^j \in_R \Z{\ell}$ for $j \in \{1,2\}$. Let ${\vl{u}^1} + \vl{u}^2 = \gm{\vl{ab}}{} - \vl{r}$ for $\vl{r} \in_R \Z{\ell}$.   
		\item Let $\vl{y} = (\vl{z} - \vl{r}) - \mk{\vl{ab}}$. Compute: $P_1: \vl{y}_1 = - \pad{\vl{a}}{1} \mk{\vl{b}} - \pad{\vl{b}}{1} \mk{\vl{a}} + {\vl{u}}^1,~~
		P_2: \vl{y}_2 = - \pad{\vl{a}}{2} \mk{\vl{b}} - \pad{\vl{b}}{2} \mk{\vl{a}} + {\vl{u}}^2$.
		\item $P_1$ sends $\vl{y}_1$ to $P_2$, while $P_2$ sends $\vl{y}_2$ to $P_1$.
		\item Parties proceed as follows:
		\begin{enumerate}
			\item $P_0$: $\vl{r} = \gm{\vl{ab}}{} - \vl{u}^1 - \vl{u}^2$; $\vl{q} = \vl{r}^{\vl{t}}$ if $\isTr = 1$, else $\vl{q} = \vl{r}$. Executes $\prot{\Sh}(P_0, \vl{q})$.
			\item $P_1, P_2$: $\vl{z} - \vl{r} = (\vl{y}_1 + \vl{y}_2) + \mk{\vl{ab}}$; $\vl{p} = (\vl{z} - \vl{r})^{\vl{t}}$ if $\isTr = 1$, else $\vl{p} = \vl{z} - \vl{r}$. Execute $\prot{\JSh}(P_1, P_2, \vl{p})$.
		\end{enumerate}
		\item Locally compute $\shr{\vl{o}} = \shr{\vl{p}} + \shr{\vl{q}}$. Here $\vl{o} = \vl{z}^{\vl{t}}$ if $\isTr = 1$ and $\vl{z}$ otherwise.
	\end{enumerate}   
\end{protocolsplitbox}

\begin{lemma}[Communication]
	\label{lemma:piMultOn3pcS}
	Protocol $\piMultO$~(\boxref{fig:piMult3pcS})~(in $\TSthis$) requires $1$ round and $3 \ell$ bits of communication in the online phase.
\end{lemma}
\begin{proof}
	Steps 3 and 4 (a) of $\piMultO$ can be executed in parallel resulting in $1$ round and $3\ell$ bits of communication.
\end{proof}

\section{Garbled World}
\label{sec:3pcSGCWorld}
We propose 2 GC protocols -- $\TSthisT$ requiring communication of 2 GCs and 1 online round, and $\TSthisC$ requiring 1 GC and 2 rounds. 
The 2 GC variant has two parallel executions, each comprising of 2 garblers and 1 evaluator. $P_1, P_2$ act as evaluators in two independent executions and the parties in $\PlSet{1} = \{P_0, P_2\}$, $\PlSet{2} = \{P_0, P_1\}$ act as garblers, respectively. The 1 GC variant comprises of a single execution with $\PlSet{1}$ acting as garblers and $P_1$ as the evaluator.

\subsection{2 GC Variant}
\label{sec:GCT3pcS}

\paragraph{Input Phase}
\label{p:GCIpT3pcS}
Given that the function  input $\vl{x}$ is already available as $\shrB{\vl{x}}$, the boolean values $\mk{\vl{x}}, \pad{\vl{x}}{}$ act as the {\em new} inputs for the garbled computation, and garbled sharing ($\shrG{\cdot}$) is generated for each of these values. The semantics of $\shrB{\cdot}$-sharing ensures that each of these shares ($\mk{\vl{x}}, \pad{\vl{x}}{}$) is available with at least one garbler in each garbling instance. Thus, the goal of our input phase is to create the compound sharing, $\shrC{\vl{x}} = (\shrG{\mk{\vl{x}}}, \shrG{\pad{\vl{x}}{}})$ for every input $\vl{x}$ to the function to be evaluated via the GC. We first discuss the semantics for $\shrG{\cdot}$-sharing followed by steps for generating $\shrC{\cdot}$-sharing.

\paragraph{Garbled sharing semantics}
\label{p:GCSemT3pcS}
A value $\vl{v} \in \Z{}$  is $\shrG{\cdot}$-shared (garbled shared) amongst  $\Partyset$ if $P_0$ holds $\shrG{\vl{v}}_{0}= (\key{{\vl{v}}}{0,1}, \key{{\vl{v}}}{0,2})$, $P_1$ holds $\shrG{\vl{v}}_{1} = (\key{{\vl{v}}}{\vl{v},1}, \key{{\vl{v}}}{0,2})$ and $P_2$ holds $\shrG{\vl{v}}_{2} = (\key{{\vl{v}}}{0,1}, \key{{\vl{v}}}{\vl{v},2})$. Here, $\key{{\vl{v}}}{\vl{v}, j} = \key{{\vl{v}}}{0, j} \xor \vl{v} \Delta^{j}$ for $j \in \{1, 2\}$, and $\Delta^{j}$, which is known only to the garblers in $\PlSet{j}$, denotes the global offset with its least significant bit set to $1$ and is same for every wire in the circuit. 
A value $\vl{x} \in \Z{}$ is said to be $\shrC{\cdot}$-shared (compound shared) if each value  from $(\mk{\vl{x}}, \pad{\vl{x}}{})$ is $\shrG{\cdot}$-shared. We write $\shrC{\vl{x}} = (\shrG{\mk{\vl{x}}},\shrG{\pad{\vl{x}}{}})$. 

\paragraph{Generation of $\shrG{\vl{v}}$ and $\shrC{\vl{x}}$} 
Protocol $\pigsh(\Partyset, \vl{v})$~(\boxref{fig:pigsh3pcS}) enables generation of $\shrG{\vl{v}}$ where two garblers in each garbling instance hold $\vl{v}$, and proceeds as follows. Consider the first garbling instance with evaluator $P_1$. Garblers in $\PlSet{1}$ generate $\{\key{{{\vl{v}}}}{\bitb, 1}\}_{\bitb \in \{0, 1\}}$ which denotes the key for value $\bitb$ on wire $\vl{v}$, following the free-XOR technique~\cite{ICALP:KolSch08,C:KolMohRos14}. $P_s \in \PlSet{1}$ sends $\key{{{\vl{v}}}}{\vl{v}, 1}$ to evaluator $P_1$ where $P_s \in \PlSet{1}$ denotes the garbler that knows $\vl{v}$ in clear. 
Similar steps carried out with respect to the second garbling instance, at the end of which, garblers in $\PlSet{2}$ possess $\{\key{\vl{v}}{\bitb, 2}\}_{\bitb \in \{0,1\}}$ while the evaluator $P_2$ holds $\key{\vl{v}}{\vl{v}, 2}$. Following this, the shares $\shrG{\vl{v}}_s$ held by $P_s \in \Partyset$ are defined as $\shrG{\vl{v}}_0 = (\key{\vl{v}}{0, 1}, \key{\vl{v}}{0, 2})$, $\shrG{\vl{v}}_1 = (\key{\vl{v}}{\vl{v}, 1}, \key{\vl{v}}{0, 2})$, $\shrG{\vl{v}}_2 = (\key{\vl{v}}{0, 1}, \key{\vl{v}}{\vl{v}, 2})$. 
To generate $\shrC{\vl{x}}$, $\pigsh$ is invoked for each of $\mk{\vl{x}}$ and $\pad{\vl{x}}{}$.  

\begin{protocolbox}{$\pigsh(\Partyset, \vl{v})$}{Generation of $\shrG{\vl{v}}$ in $\TSthis$.}{fig:pigsh3pcS}
	\justify
	\detail{
		{\bf Input(s):} $\vl{v}$,~~{\bf Output:} $\shrG{\vl{v}}$.
	} \\
    Let $P_s \in \PlSet{j}$ be the garbler that knows $\vl{v}$ in clear where $j \in \{1, 2\}$.
	\begin{enumerate}
		\item Garblers in $\PlSet{j}$ generate keys $\key{{\vl{v}}}{0, j}, \key{{\vl{v}}}{1, j}$ for wire $\vl{v}$, using free-XOR technique.
		\item $P_s \in \PlSet{j}$ sends $\key{\vl{v}}{\vl{v}, j}$ to evaluator $P_j$ for the $j^{\text{th}}$ garbling instance.
		\item $P_0$ sets $\shrG{\vl{v}}_0 = (\key{{\vl{v}}}{0,1}, \key{{\vl{v}}}{0,2})$, $P_1$ sets $\shrG{\vl{v}}_{1} = (\key{{\vl{v}}}{\vl{v},1}, \key{{\vl{v}}}{0,2})$ and $P_2$ sets $\shrG{\vl{v}}_{2} = (\key{{\vl{v}}}{0,1}, \key{{\vl{v}}}{\vl{v},2})$.
	\end{enumerate}
\end{protocolbox}

\paragraph{Evaluation} 
\label{p:GCEvT3pcS}
Let $f(\vl{x})$ be the function to be evaluated. At this point, the function input is $\shrC{\cdot}$-shared. This renders $\shrG{\cdot}$-sharing for the input of the GC that corresponds to the function $f'\big({\mk{\vl{x}}}, {\av{\vl{x}}} \big)$ which first combines the given boolean-shares to compute the actual input and then applies $f$ on it. Let $\GC_j$ denote the garbled circuit to be sent to $P_j \in \{P_1, P_2\}$ by garblers in $\PlSet{j}$. Sending of $\GC_j$ is overlapped  with the key transfer (during generation of $\shrC{\vl{x}}$), to save rounds, where garbler $P_0$ sends $\GC_j$ to $P_j$. On receiving the $\GC$, evaluators evaluate their respective GCs and obtain the key corresponding to the output, say $\vl{z}$. This generates $\shrG{\vl{z}}$. 

\paragraph{Output phase} 
\label{p:GCOpT3pcS}
The goal of output computation is to compute the output $\vl{z}$ from $\shrG{\vl{z}}$.
To reconstruct $\vl{z}$ towards $P_j \in \{P_1, P_2\}$, $P_0$ sends the least significant bit $\vl{p}^j$ of $\key{\vl{z}}{0, j}$, referred to as the decoding information, to $P_j$. $P_j$ uses the received $\vl{p}^j$ to reconstruct $\vl{z}$ as $\vl{z} = \vl{p}^j \xor \vl{q}^j$, where $\vl{q}^j$ denotes the least significant bit of $\key{\vl{z}}{\vl{z}, j}$.
To reconstruct $\vl{z}$ towards $P_0$, one evaluator, say $P_1$ sends the least significant bit, $\vl{q}^1$, of $\key{{\vl{z}}}{\vl{z}, 1}$ to $P_0$. Reconstruction is lightweight and requires a single round for garblers while reconstruction towards evaluators can be overlapped with key transfer and does not incur extra rounds.
The protocol appears in \boxref{fig:pirec3pcS}.

\begin{protocolbox}{$\pigrec(\Partyset, \shrG{\vl{z}})$}{Output computation: reconstruction of $\vl{z}$ in $\TSthis$.}{fig:pirec3pcS}
	\justify
	\detail{
		{\bf Input(s):} $\shrG{\vl{z}}$,~~{\bf Output:} $\vl{z}$.
	}
	\begin{enumerate}
		\item For an output wire $\vl{z}$, let $\vl{p}^j$ denote the least significant bit of $\key{{\vl{z}}}{0,j}$ and $\vl{q}^j$ denote the least significant bit of $\key{{\vl{z}}}{\vl{z},j}$for $j \in \{1, 2\}$.
		\item {\em Reconstruction towards $P_j \in \{P_1, P_2\}$}: $P_0$ sends $\vl{p}^j$ to $P_j$ who reconstructs $\vl{z} = \vl{p}^j \xor \vl{q}^j$.
		\item {\em Reconstruction towards $P_0$}: $P_1$~(or $P_2$) sends $\vl{q}^1$ to $P_0$ who reconstructs $\vl{z} = \vl{p}^1 \xor \vl{q}^1$. 
	\end{enumerate}
\end{protocolbox}

\paragraph{Optimizations when deployed in mixed framework}
\label{p:GCMixT3pcS}
Working in the preprocessing model enables transfer of the (communication-intensive) GC and generating $\shrG{\cdot}$-shares of the input-independent shares of $\vl{x}$ (i.e. $ \pad{\vl{x}}{}$) in the preprocessing. Thus, the online phase is very light and only requires one round to generate $\shrG{\cdot}$-shares  for the input-dependent data (i.e. ${\mk{\vl{x}}}$). Since evaluation is local, evaluators obtain $\shrG{\cdot}$-sharing of the GC output at the end of $1$ round. Moreover, we require the garbled output to be reconstructed towards both $P_1$ and $P_2$ in clear. Thus, the steps for reconstruction towards $P_0$ can be avoided in $\pigrec$ protocol~(\boxref{fig:pirec3pcS}).

\subsection{1 GC Variant}
\label{sec:GCO3pcS}
The garbling scheme here is similar to the 2GC variant except that now there exists only a single garbling instance. Parties in $\PlSet{1} = \{P_0, P_2\}$ act as the garblers while $P_1$ act as the evaluator. Looking ahead, in the mixed protocol framework, the output has to be reconstructed towards $P_1, P_2$. Reconstruction towards $P_1$ does not incur additional rounds since sending of decoding information can be overlapped with the key transfer. However, unlike in the 2GC variant, an additional round is required for $P_1$ to send the output to $P_2$. This incurs one extra round as opposed to the 2GC variant.

\section{Security proofs}
\label{sec:GCSec3pcS}
The simulation for the semi-honest 3PC case is straightforward in the $\FSETUP$-hybrid model, where $\FSETUP$~(\S\ref{sec:KeySetupprelims}) denotes the ideal functionality for the shared-key setup. The strategy for simulating the computation of function $f$ (represented by a circuit $\Ckt$) is as follows. The simulation begins with the simulator emulating the shared-key setup~($\FSETUP$) functionality and giving the respective keys to the adversary $\Adv$. Since $\Sim$ is given the input and output of the $\Adv$, it can compute all the intermediate values of the circuit $\Ckt$ in clear.

For the input sharing of value $\vl{v}$, $\Sim$ receives the  $\mk{\vl{v}}$ from $\Adv$ on behalf of the honest parties. Similarly, for the inputs of honest parties, $\Sim$ interacts with the $\Adv$ with the inputs set to $0$. The simulated view is indistinguishable from the ideal view due to the privacy of the underlying sharing scheme. The linear gates involve no communication, while simulation of the multiplication protocol is straightforward. Moreover, simulation for the joint sharing~($\prot{\JSh}$) instances is similar to that of the sharing protocol. The protocol's design is such that $\Sim$ will always know the value to be sent as part of the joint sharing protocol. Finally, for the reconstruction towards $\Adv$, $\Sim$ calculates the missing share of $\Adv$ using $\vl{y}$ and the other shares. The missing share is then communicated to $\Adv$ as per the reconstruction protocol. 

\chapter{$\Tthis$: 3PC Fair and Robust Protocols}
\label{chap:layer1_3pcmal}
This chapter provides details for the Layer I blocks of our 3PC framework $\Tthis$. Some of the results in this chapter resulted in publications at NDSS'20~\cite{NDSS:PatSur20} and USENIX Security'21~\cite{USENIX:KPPS21}. Comparison of $\Tthis$ with actively secure 3PC PPML framework of ABY3~\cite{CCS:MohRin18}, in terms of the communication for multiplication, is presented in \tabref{3pcMCost}.

\begin{table}[htb!]
	\centering
	\resizebox{0.98\textwidth}{!}{
		\begin{NiceTabular}{r c r|r r|r r|c}
			\toprule
			\Block{2-1}{Work}
			& \Block[c]{2-1}{\#Active\\Parties}
			& \Block{2-1}{Security}
			& \multicolumn{2}{c}{Multiplication} 
			& \multicolumn{2}{c}{Multiplication with Truncation\tabularnote{$\ell$ - size of ring in bits, $x$ - number of bits for the fractional part in FPA semantics.}} 
			& \Block{2-1}{Conversions\tabularnote{A, B, G indicate support for arithmetic, boolean, and garbled worlds respectively.}} \\ \cmidrule{4-7}
			&  & 
			& Comm\textsubscript{pre} 
			& Comm\textsubscript{on}\tabularnote{`Comm' - communication, `pre' - preprocessing, `on' - online} 
			& Comm\textsubscript{pre} 
			& Comm\textsubscript{on} &  \\ 
			\midrule
			ABY3~\cite{CCS:MohRin18} & 3 & Abort & $12\ell$ & $9\ell$ & $100\ell -44x -84$ & $12\ell$ & A-B-G\\		
			\textbf{$\Tthis$} & 2 & GOD & $3\ell$ & $3\ell$ & $15\ell$ & $3\ell$ & A-B-G\\	
			\bottomrule
		\end{NiceTabular}
	}
	\caption{Comparison of malicious 3PC frameworks for PPML}\label{tab:3pcMCost}
\end{table}

\section{Preliminaries and Definitions}
\label{sec:3pcMPrelim}
We consider $3$ parties denoted by $\Partyset = \{ P_1, P_2, P_3 \}$ that are connected by pair-wise private and authentic channels in a synchronous network, and a static, malicious adversary that can corrupt at most 1 party.

\subsection{Sharing Semantics}
\label{sec:3pcMsematics}
For the arithmetic and boolean sharing, we follow  a $(3, 1)$ RSS scheme similar to $\TSthis$, except that a value $\vl{v} \in \Z{\ell}$ is split into four shares. Three of the shares~($\pad{\vl{v}}{1}, \pad{\vl{v}}{2}, \pad{\vl{v}}{3})$ can be generated in the preprocessing phase independent of the value to be shared, and their sum can be interpreted as a mask~($\pad{\vl{v}}{}$). The fourth share, dependent on $\vl{v}$,  can be computed in the online phase and can be treated as the masked value $\mk{\vl{v}}= \vl{v} +\pad{\vl{v}}{}$.

Similar to $\TSthis$, we distinguish the three parties into two sets; the {\em eval} set $\SetE = \{P_1,P_2\}$ which is assigned the task of carrying out the computation, and is active throughout the online phase. The {\em helper} set $\SetD = \{P_3\}$, is used to assist $\SetE$ in verification, and so it is only active towards the end of the computation. Moreover, the share distribution is done as follows: $P_1: \{\pad{\vl{v}}{1},  \pad{\vl{v}}{3},  \mk{\vl{v}}\}, P_2: \{\pad{\vl{v}}{2},  \pad{\vl{v}}{3},  \mk{\vl{v}}\}$, and $P_3: \{\pad{\vl{v}}{1}, \pad{\vl{v}}{2}, \mk{\vl{v}}\}$. 

\begin{table}[htb!]
	\centering
	\begin{NiceTabular}{r r r r}[notes/para]
		\toprule
		Sharing Type  & $P_1$ & $P_2$ & $P_3$\\
		\midrule
		$\sqr{\cdot}$-sharing 
		& ${\vl{v}}^1$     & ${\vl{v}}^2$ & $-$      \\
		$\sgr{\cdot}$-sharing\tabularnote{$\vl{v} = \vl{v}^1 + \vl{v}^2 + \vl{v}^3$}  
		& $({\vl{v}}^1, {\vl{v}}^3)$ 
		& $({\vl{v}}^2, {\vl{v}}^3)$     & $({\vl{v}}^1, {\vl{v}}^2)$\\
		$\shr{\cdot}$-sharing\tabularnote{$\pad{\vl{v}}{} = \pad{\vl{v}}{1} + \pad{\vl{v}}{2}  + \pad{\vl{v}}{3}$, $\mk{\vl{v}} = \vl{v} + \pad{\vl{v}}{}$}  
		& $(\pad{\vl{v}}{1},  \pad{\vl{v}}{3},  \mk{\vl{v}})$ 
		& $(\pad{\vl{v}}{2},  \pad{\vl{v}}{3},  \mk{\vl{v}})$  
		& $(\pad{\vl{v}}{1}, \pad{\vl{v}}{2}, \mk{\vl{v}})$ \\
		\bottomrule
	\end{NiceTabular}
	\caption{Semantics for $\vl{v} \in \Z{\ell}$ in \Tthis.}\label{tab:3pcMsharing}
\end{table}

The RSS sharing semantics is presented in \tabref{3pcMsharing}, denoted by $\shr{\cdot}$, in a modular way with the help of two intermediate sharing semantics $\sqr{\cdot}$, and $\sgr{\cdot}$. All the sharings used are linear i.e. given sharings of values $\vl{v}_1,\ldots, \vl{v}_m$ and public constants $c_1,\ldots,c_m$, sharing of $\sum_{i=1}^m c_i \vl{v}_i$ can be computed non-interactively for an integer $m$.

\begin{notation} \label{notation:3pcMconcise}
	(a) For the $\shr{\cdot}$-shares of $n$ values $\vl{a}_1,\ldots,\vl{a}_n$, $\gm{\vl{a}_1 \ldots \vl{a}_n}{} = \prod\limits_{i=1}^{n} \pad{\vl{a}_i}{}$ and $\mk{\vl{a}_1 \ldots \vl{a}_n}{} = \prod\limits_{i=1}^{n} \mk{\vl{a}_i}{}$ (b) We use superscripts ${\bf B}$, and ${\bf G}$ to denote sharing semantics in boolean, and garbled world, respectively-- $\shrB{\cdot}$,  $\shrG{\cdot}$. We omit the superscript for arithmetic world. 
\end{notation}

Sharing semantics for boolean sharing over $\Z{}$ is similar to arithmetic sharing except that addition is replaced with XOR. The semantics for garbled sharing are described in \S\ref{sec:3pcMGCWorld} with the relevant context.

\subsubsection{$\FZero$ - Generating additive shares of zero}
\label{sec:3pcMFZero}
In $\Tthis$, we make use of a functionality $\FZero$ to enable $P_i$ obtain $Z_i$ for $i \in \{1,2,3\}$ such that $Z_1 + Z_2 + Z_3 = 0$. We observe that the functionality can be instantiated non-interactively using the pre-shared keys~(cf. \S\ref{sec:KeySetupprelims}). For this, parties in $\Partyset \setminus \{P_j\}$ sample random value $\vl{r}_j$ for $j \in \{1,2,3\}$. The shares are then defined as $Z_1 = \vl{r}_3 - \vl{r}_2, Z_2 = \vl{r}_1 - \vl{r}_3$ and $Z_3 = \vl{r}_2 - \vl{r}_1$.

\subsection{Joint-Send~($\jsend$) Primitive}
\label{sec:3pcjsend}
The Joint-Send~($\jsend$) primitive, for the case of security with fairness, allows parties $P_i, P_j$ to relay a message $\vl{v}$ to a third party $P_k$ ensuring either the delivery of the message or $\abort$ in case of inconsistency. Towards this, $P_i$ sends $\vl{v}$ to $P_k$, while $P_j$ sends a hash of the same~($\Hash(\vl{v})$) to $P_k$. Party $P_k$ accepts the message if the hash values are consistent and $\abort$ otherwise. Note that the communication of the hash can be clubbed together for several instances and be deferred to the end of the protocol, amortizing the cost.

\paragraph{Joint-Send~($\jsend$) for robust protocols}
The $\jsend$ primitive, for the case of robustness, allows $P_i, P_j$ to relay a common message, $\vl{v} \in \Z{\ell}$, to recipient $P_k$, either by ensuring successful delivery of $\vl{v}$, or by establishing a Trusted Third Party ($\TTP$). The striking feature of $\jsend$  is that it offers a rate-$1$ communication, i.e. for a message of $\ell$ elements, it only incurs a communication of $\ell$ elements (in an amortized sense). The task of $\jsend$ is  captured in an ideal functionality (\boxref{fig:Funcjsend3pcM}) and the protocol for the same appears in \boxref{fig:3pcMjsendR}. Next, we give an overview.

\vspace{-2mm}
\begin{systembox}{$\Func[\jsend]$}{Ideal functionality for robust $\jsend$ primitive in $\Tthis$}{fig:Funcjsend3pcM}
	\justify
	$\Func[\jsend]$ interacts with the parties in $\Partyset$ and the adversary $\Sim$. 
	\begin{myitemize}
		\item[\bf Step 1:] $\Func[\jsend]$ receives $(\INPUT,\vl{v}_s)$ from $P_s$ for $s \in \{i,j\}$, while it receives $(\SELECT,\ttp)$ from $\Sim$. $\ttp$ denotes the party that $\Sim$ wants to choose as the $\TTP$ and $P^{\star} \in \Partyset$ denotes the corrupt party.
		\item[\bf Step 2:] If $\vl{v}_i = \vl{v}_j$ and $\ttp =\bot$, then set $\msg_i = \msg_j = \bot, \msg_k = \vl{v}_i$ and go to {\bf Step 5}.
		\item[\bf Step 3:] If $\ttp \in \Partyset\setminus\{P^{\star}\}$, then set $\msg_i = \msg_j = \msg_k = \ttp$ and go to {\bf Step 5}.
		\item[\bf Step 4:] $\TTP$ is the honest party with smallest index. Set $\msg_i = \msg_j = \msg_k = \TTP$
		\item[\bf Step 5:] Send $(\OUTPUT, \msg_s)$ to $P_s$ for $s \in \{1,2,3\}$.
	\end{myitemize}
\end{systembox}

Given two parties $P_i, P_j$ possessing a common value $\vl{v} \in \Z{\ell}$, protocol $\prot{\jsend}$ proceeds as follows. First, $P_i$ sends $\vl{v}$ to $P_k$ while $P_j$ sends a hash of $\vl{v}$ to $P_k$. The communication of the hash is done once and for all from $P_j$ to $P_k$. In the simplest case,  $P_k$ receives a consistent (value, hash) pair, and the protocol terminates. In all other cases,  a  $\TTP$ is identified as follows without having to communicate $\vl{v}$ again.  Importantly, the following part can be run once and for all instances of $\prot{\jsend}$  with $P_i,P_j,P_k$ in the same roles, invoked in the final 3PC protocol. Consequently, the  cost relevant to this part vanishes in an amortized sense, making the construction rate-1.

\begin{protocolbox}{$\prot{\jsend}(P_i, P_j,\vl{v},P_k)$}{Joint-Send for robust protocols in $\Tthis$}{fig:3pcMjsendR}
	\detail{
		{\bf Input(s):} $P_i, P_j : \vl{v}$, $P_k : \bot$,~~{\bf Output:} $P_i, P_j : \bot / \TTP$, $P_k : \vl{v} / \TTP$.
	}\\
	Each party $P_s$ for $s \in \{i,j,k\}$ initializes bit $\bitb_s = 0$.
	\justify
	{\em Send:} $P_i$ sends $\vl{v}$ to $P_k$.
	
	\noindent {\em Verify:} $P_j$ sends $\Hash(\vl{v})$ to $P_k$. 
	\begin{enumerate}
		\item[--]   $P_k$ broadcasts "\texttt{(accuse,$\mathtt{P_i}$)}", if $P_i$ is silent and $\TTP$ = $P_j$. Analogously for $P_j$. If $P_k$ accuses both $P_i,P_j$, then $\TTP$ = $P_i$. Otherwise,  $P_k$ receives some $\tilde{\vl{v}}$ and  either sets $\bitb_k = 0$ when the value and the hash are consistent or  sets $\bitb_k = 1$. $P_k$ then sends $\bitb_k$ to $P_i,P_j$ and terminates if $\bitb_k = 0$.
		
		\item[--] If $P_i$ does not receive a bit from $P_k$, it broadcasts "\texttt{(accuse,$\mathtt{P_k}$)}" and $\TTP$ = $P_j$. Analogously for $P_j$. If both $P_i,P_j$ accuse $P_k$, then $\TTP$ = $P_i$. Otherwise,  $P_s$ for $s \in \{i,j\}$  sets $\bitb_s = \bitb_k$.
		\item[--] $P_i,P_j$ exchange their bits to each other.  If $P_i$ does  not receive $\bitb_j$ from $P_j$, it broadcasts "\texttt{(accuse,$\mathtt{P_j}$)}" and  $\TTP$ = $P_k$. Analogously for $P_j$. Otherwise, $P_i$ resets its bit to $\bitb_i \vee \bitb_j$ and likewise $P_j$ resets its bit to $\bitb_j \vee \bitb_i$.
		\item[--] $P_s$ for $s \in \{i,j,k\}$ broadcasts $\Hash_s = \Hash(\vl{v}^*)$ if $b_s = 1$, where $\vl{v}^* = \vl{v}$ for $s \in \{i,j\}$ and $\vl{v}^* = \tilde{\vl{v}}$ otherwise.  If $P_k$ does not broadcast, terminate.  If either  $P_i$ or $P_j$ does not broadcast, then  $\TTP$ = $P_k$. Otherwise,
		\begin{myitemize}
			\item If $\Hash_i \neq \Hash_j$: $\TTP$ = $P_k$.
			\item Else if $\Hash_i \neq \Hash_k$: $\TTP$ = $P_j$.
			\item Else if $\Hash_i = \Hash_j = \Hash_k$: $\TTP$ = $P_i$.
		\end{myitemize}	
	\end{enumerate}
\end{protocolbox}

Each $P_s$ for $s \in \{i,j,k\}$ maintains a bit $\bitb_s$ initialized to $0$, as an indicator for inconsistency. When $P_k$  receives an inconsistent (value, hash) pair, it sets $\bitb_k = 1$ and sends the bit to both $P_i,P_j$. Parties $P_i,P_j$ cross-check with each other by exchanging the bit and turning on their inconsistency bit if the bit received from either $P_k$ or its fellow sender is turned on.  A party broadcasts a hash of its value when its inconsistency bit is on;\footnote{hash can be computed on a combined message across many calls of $\jsend$.} $P_k$'s value is the one it receives from $P_i$. There are a bunch of possible cases at this stage, and a detailed analysis determines an eligible $\TTP$ in each case.  

When $P_k$ is silent, the protocol is understood to be complete. This is fine irrespective of the status of $P_k$-- an honest $P_k$ never skips this broadcast with inconsistency bit on, and a corrupt $P_k$ implies honest senders. If either $P_i$ or  $P_j$ is silent, then $P_k$ is picked as $\TTP$ which is surely honest. A corrupt $P_k$ could not make one of $\{P_i, P_j\}$ speak, as the senders (honest in this case) agree on their inconsistency bit (due to their mutual exchange of inconsistency bit). When all of them speak and (i) the senders' hashes do not match, $P_k$  is picked as $\TTP$;  (ii) one of the senders conflicts with $P_k$, the other sender is picked as $\TTP$;  and lastly  (iii) if there is no conflict, $P_i$ is picked as $\TTP$. The first two cases are self-explanatory.  In the last case, either $P_j$ or $P_k$ is corrupt.  If not, a corrupt $P_i$ can have honest $P_k$ speak (and hence turn on its inconsistency bit) by sending a $\vl{v}'$ whose hash is not the same as that of $\vl{v}$ and so inevitably, the hashes of honest $P_j$ and $P_k$ will conflict, contradicting (iii).  
As a final touch, we ensure that,  in each step,  a party raises a public alarm (via broadcast) accusing a silent party when it is not supposed to be. Then the protocol terminates immediately by labelling the party as $\TTP$ who is neither the complainer nor the accused.

{\em Using $\jsend$ in protocols.} As mentioned earlier, the $\jsend$ protocol needs to be viewed as consisting of two phases ({\em send, verify}), where {\em send} phase consists of $P_i$ sending $\vl{v}$ to $P_k$ and the rest  goes to  {\em verify} phase. Looking ahead,  most of our protocols use $\jsend$, and consequently, our final construction, either of general MPC or any PPML task, will have several calls to $\jsend$. To leverage amortization,  the {\em send} phase will be executed in all protocols invoking $\jsend$ on the flow, while the {\em verify} for a fixed ordered pair of senders will be executed once and for all in the end. The {\em verify} phase will determine if all the sends were correct. If not, a $\TTP$ is identified, as explained, and the computation completes with the help of $\TTP$, just as in the ideal world.  

\begin{lemma}[Communication]
	\label{lemma:pijsend3pcM}
	Protocol $\prot{\jsend}$ (\boxref{fig:3pcMjsendR}) requires $1$ round and an amortized communication of $\ell$ bits overall.
\end{lemma}
\begin{proof}
	Party $P_i$ sends value $\vl{v}$ to $P_k$ while $P_j$ sends hash of the same to $P_k$. This accounts for one round and communication of $\ell$ bits. $P_k$ then sends back its inconsistency bit to $P_i,P_j$, who then exchange it; this takes another two rounds. This is followed by parties broadcasting hashes on their values and selecting a $\TTP$ based on it, which takes one more round. All except the first round can be combined for several instances of $\prot{\jsend}$ protocol and hence the cost gets amortized.
\end{proof}

Note that the appropriate instantiation of $\jsend$ is used depending on the security guarantee. For simplicity, protocols where the fair and robust variants only differ in the instantiation of $\jsend$ used, we give a common construction for both.

\begin{notation}\label{notation_jsend}
	Protocol $\prot{\jsend}$ denotes the instantiation of Joint-Send~($\jsend$) primitive. We say that $P_i, P_j$ $\jsend$ $\vl{v}$ to $P_k$ when they invoke $\prot{\jsend}(P_i, P_j, \vl{v}, P_k)$.
\end{notation}

\section{Arithmetic / Boolean 3PC}
\label{sec:3pcMThreePC}
This section covers the details of our 3PC protocol $\Tthis$ over an arithmetic ring $\Z{\ell}$. We begin by explaining the sharing protocol in \secref{share3pcM}, multiplication with abort in \secref{mult3pcM}, and the reconstruction in~\secref{rec3pcM}. Lastly, the details on elevating the security to fairness are presented in \secref{recfair3pcM} and to robustness in \secref{recR3pcM}. 

\subsection{Sharing}
\label{sec:share3pcM}
Protocol $\prot{\Sh}$~(\boxref{fig:piSh3pcM}) enables $P_i$ to generate $\shr{\cdot}$-share of a value $\vl{v}$. During the preprocessing phase, $\pd{}$-shares are sampled non-interactively using the pre-shared keys~(cf. \S\ref{sec:KeySetupprelims}) in a way that $P_i$ will get the entire mask $\pd{\vl{v}}$. During the online phase, $P_i$ computes $\mk{\vl{v}} = \vl{v} + \pd{\vl{v}}$ and sends to $P_1$. Parties $P_i, P_1$ then communicates $\mk{\vl{v}}$ to $P_2$ and $P_3$ using $\jsend$ primitive.

\begin{protocolbox}{$\prot{\Sh}(P_i, \vl{v})$}{$\shr{\cdot}$-sharing of a value $\vl{v}$ by party $P_i$ in $\Tthis$.}{fig:piSh3pcM}
	\detail{
		{\bf Input(s):} $P_i : \vl{v}$,~~{\bf Output:} $\shr{\vl{v}}$.
	}
	\justify
	\algoHead{Preprocessing:} 
	Sample as follows: $P_i, P_1, P_3: \pad{\vl{v}}{1}$,~~$P_i, P_2, P_3: \pad{\vl{v}}{2}$,~~$P_i, P_1, P_2: \pad{\vl{v}}{3}$.
	\justify
	\vspace{-2mm}
	\algoHead{Online:}
	\begin{enumerate}
		\item $P_i$ computes $\mk{\vl{v}} = \vl{v} + \pd{\vl{v}}$ and sends to $P_j$. Here $P_j = P_1$ if $P_i \neq P_1$, else $P_j = P_2$.
		\item $P_i, P_j$ $\jsend$ $\mk{\vl{v}}$ to $P_2$ and $P_3$. 
	\end{enumerate}       
\end{protocolbox}

For the case when sharing happens in the preprocessing, the communication can be optimized to $\ell$ bits. For this, parties set $\mk{\vl{v}} = 0$ and the $\pad{\vl{v}}{}$-shares are computed as follows:
\begin{enumerate}
	\item[--] $P_i = P_1$: $\Partyset \setminus \{P_2\} \leftarrow_R \pad{\vl{v}}{1}$;~~$\Partyset \leftarrow_R \pad{\vl{v}}{2}$;~~$P_1$ sends $\pad{\vl{v}}{3} = -(\vl{v} + \pad{\vl{v}}{1} + \pad{\vl{v}}{2})$ to $P_2$.
	\item[--] $P_i = P_2$: $\Partyset \setminus \{P_1\} \leftarrow_R \pad{\vl{v}}{2}$;~~$\Partyset \leftarrow_R \pad{\vl{v}}{1}$;~~$P_2$ sends $\pad{\vl{v}}{3} = -(\vl{v} + \pad{\vl{v}}{1} + \pad{\vl{v}}{2})$ to $P_1$.
	\item[--] $P_i = P_3$: $\Partyset \setminus \{P_1\} \leftarrow_R \pad{\vl{v}}{2}$;~~$\Partyset \leftarrow_R \pad{\vl{v}}{3}$;~~$P_3$ sends $\pad{\vl{v}}{1} = -(\vl{v} + \pad{\vl{v}}{2} + \pad{\vl{v}}{3})$ to $P_1$.
\end{enumerate} 

\begin{lemma}[Communication]
	\label{lemma:pish3pcM}
	Protocol $\prot{\Sh}$~(\boxref{fig:piSh3pcM}) requires an amortized communication of at most $2\ell$ bits and $2$ rounds in the online phase.
\end{lemma}
\begin{proof}
	The preprocessing of $\prot{\Sh}$ is non-interactive as the parties sample non interactively using key setup $\Func[Setup]$~(\S\ref{sec:KeySetupprelims}). In the online phase, $P_i$ sends $\mk{\vl{v}}$ to $P_1$ resulting in 1 round and communication of $\ell$ bits. The next round consists of one instance of $\prot{\jsend}$ protocol and the cost follows from Lemma~\ref{lemma:pijsend3pcM}.
\end{proof}

\subsubsection{Joint Sharing}
\label{sec:jsh3pcM}
Protocol $\prot{\JSh}$ enables parties $P_i, P_j$ to generate $\shr{\cdot}$-share of a value $\vl{v}$. The protocol is similar to $\prot{\Sh}$ except that $P_j$ ensures the correctness of the sharing performed by $P_i$. During the preprocessing, $\pd{}$-shares are sampled such that both $P_i, P_j$ will get the entire mask $\pd{\vl{v}}$. During the online phase, $P_i, P_j$ compute and $\jsend$ $\mk{\vl{v}} = \vl{v} + \pd{\vl{v}}$ to parties $P_1, P_2, P_3$.

When the value $\vl{v}$ is available to both $P_i, P_j$ in the preprocessing, protocol $\prot{\JSh}$ can be made non-interactive by setting the shares as given in \tabref{jsh3pcM}.

\begin{table}[htb!]
	\centering
	\begin{NiceTabular}{c c c c c}
		\toprule
		($P_i, P_j$)  & $\pad{\vl{v}}{1}$ & $\pad{\vl{v}}{2}$ & $\pad{\vl{v}}{3}$ & $\mk{\vl{v}}$\\
		\midrule
		$(P_1, P_2)$   & $0$         & $0$           & $- \vl{v}$  & $0$       \\   
		$(P_1, P_3)$   & $- \vl{v}$ & $0$           & $0$          & $0$       \\ 
		$(P_2, P_3)$  & $0$          & $- \vl{v}$  & $0$           & $0$       \\ 
		\bottomrule
	\end{NiceTabular}
	\caption{Shares for $\prot{\JSh}$ in the preprocessing in \Tthis.}\label{tab:jsh3pcM}
\end{table}

\begin{lemma}[Communication]
	\label{appl:pijsh3pcM}
	Protocol $\prot{\JSh}$ is non-interactive in the preprocessing and requires an amortized communication of $\ell$ bits and $1$ round in the online phase.
\end{lemma}
\begin{proof}
	The protocol involves one invocation of $\prot{\jsend}$ protocol in the online and the cost follows from Lemma~\ref{lemma:pijsend3pcM}.
\end{proof}

\subsection{Multiplication}
\label{sec:mult3pcM}
Given the shares of $\vl{a}, \vl{b}$, the goal of the multiplication protocol is to generate shares of $\vl{z} = \vl{ab}$. The protocol is designed such that parties $P_1, P_2$ obtain a masked version of the output $\vl{z}$, say $\vl{z} - \vl{r}$ in the online phase. Moreover, parties obtain the $\shr{\cdot}$-sharing of the mask $\vl{r}$ in the preprocessing. $P_1, P_2$ then generate $\shr{\cdot}$-sharing of $(\vl{z} - \vl{r})$ by executing $\prot{\JSh}$. Parties locally compute the final output as $\shr{\vl{z} - \vl{r}} + \shr{\vl{r}}$. 

\paragraph{Online} 
Similar to $\TSthis$, we have,

\begin{align}\label{eq:3pcMmult}
	\vl{z} - \vl{r} &= \vl{a}\vl{b} - \vl{r} = (\mk{\vl{a}} - \pd{\vl{a}})(\mk{\vl{b}} - \pd{\vl{b}}) - \vl{r} \nonumber\\ 
	&= \mk{\vl{ab}} - \mk{\vl{a}}\pd{\vl{b}} - \mk{\vl{b}}\pd{\vl{a}} + \gm{\vl{a}\vl{b}}{} - \vl{r}
	~~\text{\footnotesize{(cf. notation~\ref{notation:3pcMconcise})}}
\end{align}

In Eq~\ref{eq:3pcMmult}, all the parties can compute $\mk{\vl{ab}}$ locally, and hence we are interested in computing $\vl{y} = (\vl{z - r}) - \mk{\vl{ab}}$. Let $\vl{y} = \vl{y}_1 + \vl{y}_2 + \vl{y}_3$, where $\vl{y}_1, \vl{y}_2, \vl{y}_3$ can be computed respectively by the pairs $(P_1, P_3), (P_2, P_3)$ and $(P_1, P_2)$. Given a preprocessing that enables parties to obtain a $\sgr{\cdot}$-sharing of $(\gm{\vl{ab}}{} - \vl{r})$, parties locally compute the additive shares of $\vl{y}$ according to \eqref{eq:3pcMmulty}.
\begin{align}\label{eq:3pcMmulty}
	P_1, P_3: \vl{y}_1  &= - \pad{\vl{a}}{1} \mk{\vl{b}} - \pad{\vl{b}}{1} \mk{\vl{a}} + {(\gm{\vl{ab}}{} - \vl{r})}^1 \nonumber\\
	P_2, P_3: \vl{y}_2 &= - \pad{\vl{a}}{2} \mk{\vl{b}} - \pad{\vl{b}}{2} \mk{\vl{a}} + {(\gm{\vl{ab}}{} - \vl{r})}^2 \nonumber\\
	P_1, P_2: \vl{y}_3  &= - \pad{\vl{a}}{3} \mk{\vl{b}} - \pad{\vl{b}}{3} \mk{\vl{a}} + {(\gm{\vl{ab}}{} - \vl{r})}^3
\end{align}

Once the shares are computed, $P_1, P_3$ $\jsend$ $\vl{y}_1$ to $P_2$ and $P_2, P_3$ $\jsend$ $\vl{y}_2$ to $P_1$. Parties $P_1, P_2$ reconstruct $\vl{y}$ using the shares received and subsequently $\vl{z - r}$. 

\smallskip
\begin{protocolsplitbox}{$\prot{\Mult}(\vl{a}, \vl{b}, \isTr)$}{Multiplication with / without truncation in $\Tthis$.}{fig:piMult3pcM}
	$\isTr$ is a bit denoting whether truncation is required ($\isTr =1$) or not ($\isTr=0$). \\
	\detail{
		{\bf Input(s):} $\shr{\vl{a}}, \shr{\vl{b}}$.\\
		{\bf Output:} $\shr{\vl{o}}$ where $\vl{o} = \vl{z}^{\vl{t}}$ if $\isTr = 1$ and $\vl{o} = \vl{z}$ if $\isTr = 0$ and $\vl{z} = \vl{ab}$.
	}
	\justify 
	\vspace{-2mm}
	\algoHead{Preprocessing:} 
	\begin{enumerate}
		\item Invoke $\Func[\MultPre]$ on $\sgr{\pad{\vl{a}}{}}$ and $\sgr{\pad{\vl{b}}{}}$ to obtain $\sgr{\gm{\vl{ab}}{}}$.
		\item If $\isTr = 0$:
		         \begin{enumerate}
		         	\item Local computation of $\sgr{\vl{r}}$: $\Partyset \setminus \{P_2\} \leftarrow_R {\vl{r}}^1$;~~~$\Partyset \setminus \{P_1\} \leftarrow_R {\vl{r}}^2$;~~~$\Partyset \setminus \{P_3\} \leftarrow_R {\vl{r}}^3$.
		         	\item Local computation of $\shr{\vl{r}}$: $\pad{\vl{r}}{1} = - {\vl{r}}^1$,~~$\pad{\vl{r}}{2} = - {\vl{r}}^2$,~~$\pad{\vl{r}}{3} = - {\vl{r}}^3$,~~$\mk{\vl{r}} = 0$. Set $\shr{\vl{q}} = \shr{\vl{r}}$. 
		         \end{enumerate}
		\item If $\isTr = 1$, invoke $\prot{\trgen}$~(\boxref{fig:trgen3pcM}) to generate $(\sgr{\vl{r}}, \shr{\vl{r}^{\vl{t}}})$. Set $\shr{\vl{q}} =  \shr{\vl{r}^{\vl{t}}}$.
		\item Locally compute $\sgr{(\gm{\vl{ab}}{} - \vl{r})} = \sgr{\gm{\vl{ab}}{}} - \sgr{\vl{r}}$.
	\end{enumerate}
	\justify
	\vspace{-2mm}
	\algoHead{Online:} Let $\vl{y} = (\vl{z} - \vl{r}) - \mk{\vl{ab}}$.
	\begin{enumerate}
		\item Parties locally compute the following:
		\begin{align*}
			P_1, P_3: \vl{y}_1  &= - \pad{\vl{a}}{1} \mk{\vl{b}} - \pad{\vl{b}}{1} \mk{\vl{a}} + {(\gm{\vl{ab}}{} - \vl{r})}^1 \\
			P_2, P_3: \vl{y}_2 &= - \pad{\vl{a}}{2} \mk{\vl{b}} - \pad{\vl{b}}{2} \mk{\vl{a}} + {(\gm{\vl{ab}}{} - \vl{r})}^2 \\
			P_1, P_2: \vl{y}_3  &= - \pad{\vl{a}}{3} \mk{\vl{b}} - \pad{\vl{b}}{3} \mk{\vl{a}} + {(\gm{\vl{ab}}{} - \vl{r})}^3
		\end{align*}
		\item $P_1, P_3$ $\jsend$ $\vl{y}_1$ to $P_2$, while $P_2, P_3$ $\jsend$ $\vl{y}_2$ to $P_1$. They locally compute $\vl{z} - \vl{r} = (\vl{y}_1 + \vl{y}_2 + \vl{y}_3) + \mk{\vl{ab}}$.
		\item $P_1, P_2$: If $\isTr = 1$, set $\vl{p} = (\vl{z} - \vl{r})^{\vl{t}}$, else $\vl{p} = \vl{z} - \vl{r}$. Execute $\prot{\JSh}(P_1, P_2, \vl{p})$ to generate $\shr{\vl{p}}$. 
		\item Compute $\shr{\vl{o}} = \shr{\vl{p}} + \shr{\vl{q}}$. Here $\vl{o} = \vl{z}^{\vl{t}}$ if $\isTr = 1$ and $\vl{z}$ otherwise.
	\end{enumerate}     
\end{protocolsplitbox}

\paragraph{Verification}
To leverage amortization, the {\em send} phase of $\jsend$ alone is executed on the fly and {\em verify} is performed once for multiple instances of $\jsend$. Further, observe that $P_1,P_2$ possess the required shares in the online phase to compute the entire circuit.  Hence, $P_3$ can come in only during {\em verify} of $\jsend$ towards $P_1, P_2$, which can be deferred towards the end. Hence, the $\jsend$ to $P_3$ (as part of $\prot{\JSh}$ by $P_1, P_2$ during the online) can be performed once, towards the end, thereby requiring a single round for multiple instances of $\prot{\JSh}$. Following this, the {\em verify} of $\jsend$ towards $P_3$ is performed first, followed by  performing the {\em verify} of $\jsend$ towards $P_1, P_2$ in parallel. 

\paragraph{Preprocessing} 
As mentioned above, parties should obtain a $\sgr{\cdot}$-sharing of $(\gm{\vl{ab}}{} - \vl{r})$ from the preprocessing. The $\sgr{\cdot}$-shares for a random $\vl{r} \in \Z{\ell}$ can be generated non-interactively using the key setup $\Func[Setup]$~(\S\ref{sec:KeySetupprelims}). To compute $\sgr{\gm{\vl{ab}}{}}$, we rely on a 3-party multiplication protocol, say $\prot{\MultPre}$, abstracted in a functionality $\Func[\MultPre]$~(\boxref{fig:FMultPre3pcM}). The security of $\prot{\MultPre}$ depends on the security required in our framework. For instance, instantiating $\Func[\MultPre]$ with the protocols of \cite{C:BBCGI19} and \cite{EPRINT:ADEN19} will result in abort or fairness guarantees whereas using the robust 3 party protocol of \cite{CCS:BGIN19} will result in a multiplication protocol with robustness. In $\Tthis$, we use the protocol of \cite{CCS:BGIN19}  in a black-box manner resulting in a communication of $3\ell$ bits (amortized) for $\prot{\MultPre}$. This leaves room for further improvements in the overall efficiency of our multiplication, which can be obtained by instantiating the black-box with efficient protocols.

\begin{systembox}{$\Func[\MultPre]$}{Ideal functionality for $\prot{\MultPre}$ in $\Tthis$.}{fig:FMultPre3pcM}
	\justify
	$\Func[\MultPre]$ interacts with the parties in $\Partyset$ and the adversary $\Sim$. $\Func[\MultPre]$ receives $\sgr{\cdot}$-shares of $\vl{d}, \vl{e}$ from the parties. Let $P^{\star}$ denotes the party corrupted by $\Sim$. $\Func[\MultPre]$ receives $(\vl{f}_i, \vl{f}_j)$ from $\Sim$ as its share for $\sgr{\vl{f}}$ where $\vl{f} = \vl{d} \vl{e}$. $\Func[\MultPre]$ proceeds as follows:
	\begin{enumerate}
		\item Reconstructs $\vl{d}, \vl{e}$ using the shares received from honest parties and compute $\vl{f} = \vl{d} \vl{e}$.
		\item Computes the third share $\vl{f}_k = \vl{f} - \vl{f}_i -\vl{f}_j$ and sets $\sgr{\vl{f}}_1 = (\vl{f}_1, \vl{f}_3), \sgr{\vl{f}}_2 = (\vl{f}_2, \vl{f}_3),  \sgr{\vl{f}}_3 = (\vl{f}_1, \vl{f}_2)$.
		\item Send $(\OUTPUT, \sgr{\vl{f}}_s)$ to $P_s \in \Partyset$.
	\end{enumerate}
\end{systembox}

\begin{lemma}[Communication]
	\label{lemma:piMult3pcM}
	Protocol $\prot{\Mult}$~(\boxref{fig:piMult3pcM}) without truncation~(in $\Tthis$) requires $3 \ell$ bits of communication in the preprocessing, and $1$ round and $3 \ell$ bits of communication in the online phase.
\end{lemma}
\begin{proof}
	During preprocessing, sampling of the shares for $\sgr{\vl{r}}$ is performed non-interactively using $\Func[Setup]$. The $\prot{\MultPre}$ protocol, instantiated using the protocol of \cite{CCS:BGIN19} requires a communication of $3\ell$ bits in the preprocessing.
	During online, two instances of $\prot{\jsend}$ are executed in parallel resulting in a communication of $2\ell$ bits and 1 round. This is followed by a joint sharing by $P_1,P_2$ for which an additional communication of $\ell$ bits are required. However, in joint sharing, the communication is from $P_1$ to $P_3$ and the same can be deferred till the verification stage. Thus the online round is retained as $1$ in an amortized sense. 
\end{proof}

\subsubsection{Truncation}
To incorporate truncation, the multiplication protocol is modified such that $P_1, P_2$ execute joint sharing on the truncated value of $(\vl{z - r})$ in the online phase. To complete the protocol, the $\shr{\cdot}$-shares of the truncated $\vl{r}$, denoted by $\vl{r}^t$, is needed. For this, we use $\prot{\trgen}$~(\boxref{fig:trgen3pcM}) protocol in the preprocessing that generates a pair of the form $(\sgr{\vl{r}}, \shr{\vl{r}^{\vl{t}}})$. More details on $\prot{\trgen}$ are provided in \S\ref{sec:trgen3pcM}. Parties locally compute $\shr{\vl{z}^{\vl{t}}} = \shr{(\vl{z-r})^{\vl{t}}} + \shr{\vl{r}^{\vl{t}}}$.

\subsubsection{Multiplication with constant}
Multiplication by a constant in MPC is typically local. Given constant $\alpha$ and $\shr{\vl{v}}$, the $\shr{\cdot}$-shares of the product $\vl{y} = \alpha\vl{v}$ can be locally computed as per \eqref{eq:mutconst3pcM}. 
\begin{equation}\label{eq:mutconst3pcM}
	\mk{\vl{y}} = \alpha \mk{\vl{u}},~~~\pad{\vl{y}}{1} = \alpha \pad{\vl{v}}{1},~~~\pad{\vl{y}}{2} = \alpha \pad{\vl{v}}{2},~~~\pad{\vl{y}}{3} = \alpha \pad{\vl{v}}{3}
\end{equation}

In FPA, parties should obtain truncated $\vl{y}$ as both $\alpha$ and $\vl{v}$ are decimal values. For this, parties invoke $\prot{\trgen}$~(\boxref{fig:trgen3pcM}) in the preprocessing to generate $(\sgr{\vl{r}}, \shr{\vl{r}^{\vl{t}}})$ for a random $\vl{r} \in \Z{\ell}$. The $\shr{\cdot}$-shares of $\vl{r}$ are locally computed from $\sgr{\vl{r}}$ locally similar to $\prot{\Mult}$~(\boxref{fig:piMult3pcM}). During online, parties locally compute $\shr{\vl{v-r}}$ and reconstructs $\vl{z} = \vl{v-r}$ using $\prot{\Rec}$~(\boxref{fig:piRec3pcM}). Parties locally compute  $\shr{\vl{z}^{\vl{t}}}$ by setting $\mk{{\vl{z}^{\vl{t}}}} = \vl{z}^{\vl{t}}$ and $\pad{{\vl{z}^{\vl{t}}}}{1} = \pad{{\vl{z}^{\vl{t}}}}{2} = \pad{{\vl{z}^{\vl{t}}}}{3} = 0$. Lastly, parties locally compute  $\shr{\vl{v}^{\vl{t}}} - \shr{\vl{z}^{\vl{t}}} + \shr{\vl{r}^{\vl{t}}}$.

\subsection{Reconstruction}
\label{sec:rec3pcM}
Protocol $\prot{\Rec}(\Partyset, \vl{v})$ (\boxref{fig:piRec3pcM}) enables parties to compute $\vl{v}$, given its $\shr{\cdot}$-share and achieves security with abort. Note that each party misses one share to reconstruct the output, and the other two parties hold this share. They will $\jsend$~(abort variant) the missing share to the party that lacks it. Reconstruction towards a single party can be viewed as a special case.

\begin{protocolbox}{$\prot{\Rec}(\Partyset, \shr{\vl{v}})$}{Reconstruction (with abort security) of value $\vl{v}$ among $\Partyset$ in $\Tthis$.}{fig:piRec3pcM}
	\justify
	\detail{
		{\bf Input(s):} $\shr{\vl{v}}$,~~{\bf Output:} $\vl{v}$.
	}
	\begin{enumerate}
		\item $P_1, P_3$ $\jsend$ $\pad{\vl{v}}{1}$ to $P_2$;~~~$P_2, P_3$ $\jsend$ $\pad{\vl{v}}{2}$ to $P_1$; $P_1, P_2$ $\jsend$ $\pad{\vl{v}}{3}$ to $P_3$.
		\item Parties compute $\vl{v} = \mk{\vl{v}} - \pad{\vl{v}}{1} - \pad{\vl{v}}{2} - \pad{\vl{v}}{3}$. 
	\end{enumerate}
\end{protocolbox}

\begin{lemma}[Communication]
	\label{lemma:pirec3pcM}
	Protocol $\prot{\Rec}$~(abort security, \boxref{fig:piRec3pcM}) requires an amortized communication of $3\ell$ bits and $1$ round.
\end{lemma}
\begin{proof}
	The protocol involves three invocations of $\prot{\jsend}$ protocol and the cost follows from Lemma~\ref{lemma:pijsend3pcM}.
\end{proof}

\begin{protocolbox}{$\prot{\Rec}(\Partyset, \shr{\vl{v}})$}{Fair Reconstruction of value $\vl{v}$ among $\Partyset$ in $\Tthis$.}{fig:piRecfair3pcM}
	\justify
	\detail{
		{\bf Input(s):} $\shr{\vl{v}}$,~~{\bf Output:} $\vl{v}$.
	}
	\justify 
	\vspace{-2mm}
	\algoHead{Preprocessing:} 
	\begin{enumerate}
		\item Parties locally compute the commitments on the $\pad{\vl{v}}{}$ shares as:
		\begin{equation*}
			P_1, P_3: \Commit{\pad{\vl{v}}{1}},~~P_2, P_3: \Commit{\pad{\vl{v}}{2}},~~P_1, P_2: \Commit{\pad{\vl{v}}{3}}
		\end{equation*}
		\item $P_1, P_3$ $\jsend$ $\Commit{\pad{\vl{v}}{1}}$ to $P_2$;~~~$P_2, P_3$ $\jsend$ $\Commit{\pad{\vl{v}}{2}}$ to $P_1$; $P_1, P_2$ $\jsend$ $\Commit{\pad{\vl{v}}{3}}$ to $P_3$
	\end{enumerate}
    \justify 
    \vspace{-2mm}
    \algoHead{Online:} Parties set their aliveness bit $\bitb = \continue$, if the verification phase is successful. Else $\bitb = \abort$.
	\begin{enumerate}
		\item Party $P_s \in \Partyset$ broadcasts $\bitb_s$ and parties accept the value that forms the majority. 
		\item If the accepted value is $\abort$, parties abort. Else $P_1, P_3$ open $\Commit{\pad{\vl{v}}{1}}$ towards $P_2$; $P_2, P_3$ open $\Commit{\pad{\vl{v}}{2}}$ towards $P_1$; $P_1, P_2$ open $\Commit{\pad{\vl{v}}{3}}$ towards $P_3$. Parties use the correct opening to obtain their missing share.
		\item Parties compute $\vl{v} = \mk{\vl{v}} - \pad{\vl{v}}{1} - \pad{\vl{v}}{2} - \pad{\vl{v}}{3}$. 
	\end{enumerate}
\end{protocolbox}

\subsubsection{Achieving Fairness}
\label{sec:recfair3pcM}
Here, we show how to extend the security of $\Tthis$ from abort to fairness by modifying the reconstruction protocol. During preprocessing, each pair of parties together prepare a commitment on the $\pad{\vl{v}}{}$ share missing at the third party. The commitments are then communicated via $\jsend$~(abort variant), and the privacy is guaranteed by the hiding property of the underlying commitment scheme~(cf. \S\ref{sec:commitprelims}). 
Before proceeding with the output reconstruction in the online phase, we need to ensure that all the honest parties are alive after the verification phase. For this, all the parties maintain an {\em aliveness} bit, say $\bitb$, which is initialized to $\continue$. If the verification phase is not successful for a party, it sets $\bitb = \abort$. In the first round of reconstruction, the parties broadcast their $\bitb$ bit and accept the value that forms the majority.  If $\bitb = \continue$, then a pair of parties open the commitment (communicated in the preprocessing) towards the third party. This method is fair because at least one honest party would have provided the correct opening to allow the third party to obtain its missing share. The formal protocol appears in \boxref{fig:piRecfair3pcM}.

\subsection{Achieving Robustness}
\label{sec:recR3pcM}
To elevate the security of $\Tthis$ to robustness, we use the robust variant of $\jsend$ in all the protocols. Moreover, for reconstruction, we use the fair reconstruction protocol in \boxref{fig:piRecfair3pcM} except that the aliveness check~(Online, Step 1) is no longer required. This is because the verification in robust $\jsend$ guarantees identification of a $\TTP$ in case of any inconsistency, and the parties wouldn't have executed the reconstruction protocol. 

\subsection{Multi-input Multiplication}
\label{sec:multT3pcM}

\paragraph{3-input multiplication}
To compute $\shr{\cdot}$-shares of $\vl{z} = \vl{abc}$, note that  

\begin{align}\label{eq:3pcMmultT}
	\vl{z} - \vl{r} &= \vl{abc} - \vl{r} = (\mk{\vl{a}} - \pd{\vl{a}})(\mk{\vl{b}} - \pd{\vl{b}})(\mk{\vl{c}} - \pd{\vl{c}}) - \vl{r} \nonumber\\ 
	&= \mk{\vl{abc}} - \mk{\vl{ac}}\pd{\vl{b}} - \mk{\vl{bc}}\pd{\vl{a}} - \mk{\vl{ab}}\pd{\vl{c}} + \mk{\vl{a}}\gm{\vl{bc}}{} + \mk{\vl{b}}\gm{\vl{ac}}{} + \mk{\vl{c}}\gm{\vl{ab}}{} - \gm{\vl{abc}}{} - \vl{r}
	~~\text{\footnotesize{(cf. notation~\ref{notation:3pcMconcise})}}
\end{align}

Similar to $\prot{\Mult}$, for $\vl{y} = (\vl{z - r}) - \mk{\vl{abc}}$, let $\vl{y} = \vl{y}_1 + \vl{y}_2 + \vl{y}_3$.

\begin{align}\label{eq:3pcMmultTy}
	P_1, P_3: \vl{y}_1 &= - \pad{\vl{a}}{1} \mk{\vl{bc}} - \pad{\vl{b}}{1} \mk{\vl{ac}} - \pad{\vl{c}}{1} \mk{\vl{ab}} + \gm{\vl{ab}}{1} \mk{\vl{c}} + \gm{\vl{ac}}{1} \mk{\vl{b}} + \gm{\vl{bc}}{1} \mk{\vl{a}} - {(\gm{\vl{abc}}{} + \vl{r})}^1 \nonumber\\
	P_2, P_3: \vl{y}_2 &= - \pad{\vl{a}}{2} \mk{\vl{bc}} - \pad{\vl{b}}{2} \mk{\vl{ac}} - \pad{\vl{c}}{2} \mk{\vl{ab}} + \gm{\vl{ab}}{2} \mk{\vl{c}} + \gm{\vl{ac}}{2} \mk{\vl{b}} + \gm{\vl{bc}}{2} \mk{\vl{a}} - {(\gm{\vl{abc}}{} + \vl{r})}^2 \nonumber\\
	P_1, P_2: \vl{y}_3 &= - \pad{\vl{a}}{3} \mk{\vl{bc}} - \pad{\vl{b}}{3} \mk{\vl{ac}} - \pad{\vl{c}}{3} \mk{\vl{ab}} + \gm{\vl{ab}}{3} \mk{\vl{c}} + \gm{\vl{ac}}{3} \mk{\vl{b}} + \gm{\vl{bc}}{3} \mk{\vl{a}} - {(\gm{\vl{abc}}{} + \vl{r})}^3 
\end{align}

To generate $\sgr{\vl{x}}$ for $\vl{x} \in \{\gm{\vl{ab}}{}, \gm{\vl{ac}}{}, \gm{\vl{bc}}{}\}$, parties rely on $\prot{\MultPre}$ protocol. Parties then use another instance of $\prot{\MultPre}$ on the inputs $\gm{\vl{ab}}{}$ and $\pad{\vl{c}}{}$ to generate $\sgr{\gm{\vl{abc}}{}}$. 
The generation of $\sgr{\gm{\vl{abc}}{} + \vl{r}}$ and the rest of the steps follow similar to that of 2-input multiplication protocol $\prot{\Mult}$ in \S\ref{sec:mult3pcM}. The formal protocol appears in \boxref{fig:piMultT3pcM}.

\begin{protocolbox}{$\prot{\Mult}(\vl{a}, \vl{b}, \isTr)$}{Three-input Multiplication with / without truncation in $\Tthis$.}{fig:piMultT3pcM}
	$\isTr$ is a bit denoting whether truncation is required ($\isTr =1$) or not ($\isTr=0$). \\
	\detail{
		{\bf Input(s):} $\shr{\vl{a}}, \shr{\vl{b}}, \shr{\vl{c}}$.\\
		{\bf Output:} $\shr{\vl{o}}$ where $\vl{o} = \vl{z}^{\vl{t}}$ if $\isTr = 1$ and $\vl{o} = \vl{z}$ if $\isTr = 0$ and $\vl{z} = \vl{abc}$.
	}
	\justify 
	\vspace{-2mm}
	\algoHead{Preprocessing:} 
	\begin{enumerate}
		\item Invoke $\Func[\MultPre]$ on the $\sgr{\cdot}$-shares of $(\pad{\vl{a}}{}, \pad{\vl{b}}{})$, $(\pad{\vl{a}}{}, \pad{\vl{c}}{})$, and $(\pad{\vl{b}}{}, \pad{\vl{c}}{})$ to obtain $\sgr{\gm{\vl{ab}}{}}, \sgr{\gm{\vl{ac}}{}}, \sgr{\gm{\vl{bc}}{}}$ respectively.
		\item Invoke $\Func[\MultPre]$ on the $\sgr{\cdot}$-shares of $\gm{\vl{ab}}{}$ and $\pad{\vl{c}}{}$ to obtain  $\sgr{\gm{\vl{abc}}{}}$.
		\item If $\isTr = 0$:
		\begin{enumerate}
			\item Local computation of $\sgr{\vl{r}}$: $\Partyset \setminus \{P_2\} \leftarrow_R {\vl{r}}^1$;~~~$\Partyset \setminus \{P_1\} \leftarrow_R {\vl{r}}^2$;~~~$\Partyset \setminus \{P_3\} \leftarrow_R {\vl{r}}^3$.
			\item Local computation of $\shr{\vl{r}}$: $\pad{\vl{r}}{1} = - {\vl{r}}^1$,~~$\pad{\vl{r}}{2} = - {\vl{r}}^2$,~~$\pad{\vl{r}}{3} = - {\vl{r}}^3$,~~$\mk{\vl{r}} = 0$. Set $\shr{\vl{q}} = \shr{\vl{r}}$. 
		\end{enumerate}
		\item If $\isTr = 1$, invoke $\prot{\trgen}$~(\boxref{fig:trgen3pcM}) to generate $(\sgr{\vl{r}}, \shr{\vl{r}^{\vl{t}}})$. Set $\shr{\vl{q}} =  \shr{\vl{r}^{\vl{t}}}$.
		\item Locally compute $\sgr{(\gm{\vl{abc}}{} + \vl{r})} = \sgr{\gm{\vl{abc}}{}} + \sgr{\vl{r}}$.
	\end{enumerate}
	\justify
	\vspace{-2mm}
	\algoHead{Online:} Let $\vl{y} = (\vl{z} - \vl{r}) - \mk{\vl{abc}}$.
	\begin{enumerate}
		\item Parties locally compute the following:
		\begin{align*}
			P_1, P_3: \vl{y}_1 &= - \pad{\vl{a}}{1} \mk{\vl{bc}} - \pad{\vl{b}}{1} \mk{\vl{ac}} - \pad{\vl{c}}{1} \mk{\vl{ab}} + \gm{\vl{ab}}{1} \mk{\vl{c}} + \gm{\vl{ac}}{1} \mk{\vl{b}} + \gm{\vl{bc}}{1} \mk{\vl{a}} - {(\gm{\vl{abc}}{} + \vl{r})}^1 \nonumber\\
			P_2, P_3: \vl{y}_2 &= - \pad{\vl{a}}{2} \mk{\vl{bc}} - \pad{\vl{b}}{2} \mk{\vl{ac}} - \pad{\vl{c}}{2} \mk{\vl{ab}} + \gm{\vl{ab}}{2} \mk{\vl{c}} + \gm{\vl{ac}}{2} \mk{\vl{b}} + \gm{\vl{bc}}{2} \mk{\vl{a}} - {(\gm{\vl{abc}}{} + \vl{r})}^2 \nonumber\\
			P_1, P_2: \vl{y}_3 &= - \pad{\vl{a}}{3} \mk{\vl{bc}} - \pad{\vl{b}}{3} \mk{\vl{ac}} - \pad{\vl{c}}{3} \mk{\vl{ab}} + \gm{\vl{ab}}{3} \mk{\vl{c}} + \gm{\vl{ac}}{3} \mk{\vl{b}} + \gm{\vl{bc}}{3} \mk{\vl{a}} - {(\gm{\vl{abc}}{} + \vl{r})}^3 
		\end{align*}
		\item $P_1, P_3$ $\jsend$ $\vl{y}_1$ to $P_2$, while $P_2, P_3$ $\jsend$ $\vl{y}_2$ to $P_1$. They locally compute $\vl{z} - \vl{r} = (\vl{y}_1 + \vl{y}_2 + \vl{y}_3) + \mk{\vl{abc}}$.
		\item $P_1, P_2$: If $\isTr = 1$, set $\vl{p} = (\vl{z} - \vl{r})^{\vl{t}}$, else $\vl{p} = \vl{z} - \vl{r}$. Execute $\prot{\JSh}(P_1, P_2, \vl{p})$ to generate $\shr{\vl{p}}$. 
		\item Compute $\shr{\vl{o}} = \shr{\vl{p}} + \shr{\vl{q}}$. Here $\vl{o} = \vl{z}^{\vl{t}}$ if $\isTr = 1$ and $\vl{z}$ otherwise.
	\end{enumerate}     
\end{protocolbox}

\begin{lemma}[Communication]
	\label{lemma:piMultT3pcM}
	Protocol $\prot{\MultT}$~(\boxref{fig:piMultT3pcM})~(in $\Tthis$) requires $12\ell$ bits of communication in the preprocessing, and $1$ round and $3 \ell$ bits of communication in the online phase.
\end{lemma}
\begin{proof}
	During preprocessing, sampling of the shares for $\sgr{\vl{r}}$ is performed non-interactively using $\Func[Setup]$. Also, four instances of $\prot{\MultPre}$ protocol are executed in the preprocessing. Instantiating $\prot{\MultPre}$ using~\cite{CCS:BGIN19} requires a communication of $3\ell$ bits for each of the instances.
	The online phase is similar to that of $\prot{\Mult}$ and the costs follow from Lemma~\ref{lemma:piMult3pcM}.
\end{proof}

\paragraph{4-input multiplication}
For the case of 4-input multiplication with $\vl{z} = \vl{abcd}$, note that  

\begin{align}\label{eq:3pcMmultF}
	\vl{z} - \vl{r} &= \vl{abcd} - \vl{r} = (\mk{\vl{a}} - \pd{\vl{a}})(\mk{\vl{b}} - \pd{\vl{b}})(\mk{\vl{c}} - \pd{\vl{c}})(\mk{\vl{d}} - \pd{\vl{d}}) - \vl{r} \nonumber\\ 
	&= \mk{\vl{abcd}} - \mk{\vl{abc}}\pd{\vl{d}} - \mk{\vl{abd}}\pd{\vl{c}} - \mk{\vl{acd}}\pd{\vl{b}} - \mk{\vl{bcd}}\pd{\vl{a}} 
	+ \mk{\vl{ab}}\gm{\vl{cd}}{} + \mk{\vl{ac}}\gm{\vl{bd}}{} + \mk{\vl{ad}}\gm{\vl{bc}}{} + \mk{\vl{bc}}\gm{\vl{ad}}{} \nonumber\\
	&~~~+ \mk{\vl{bd}}\gm{\vl{ac}}{} + \mk{\vl{cd}}\gm{\vl{ab}}{} 
	- \mk{\vl{a}}\gm{\vl{bcd}}{} - \mk{\vl{b}}\gm{\vl{acd}}{} - \mk{\vl{c}}\gm{\vl{abd}}{} - \mk{\vl{d}}\gm{\vl{abc}}{} + \gm{\vl{abcd}}{} - \vl{r}
	~~\text{\footnotesize{(cf. notation~\ref{notation:3pcMconcise})}}
\end{align}

Here the parties first generate $\sgr{\cdot}$-shares of $\gm{\vl{ab}}{},\gm{\vl{ac}}{},\gm{\vl{ad}}{},\gm{\vl{bc}}{},\gm{\vl{bd}}{}$, and $\gm{\vl{cd}}{}$ using $\prot{\MultPre}$ on the respective inputs. In the next round, parties make use of these shares and $\prot{\MultPre}$ to generate $\gm{\vl{abc}}{},\gm{\vl{abd}}{},\gm{\vl{acd}}{},\gm{\vl{bcd}}{}$ and $\gm{\vl{abcd}}{}$. Thus, the preprocessing involves a total of eleven instances of $\prot{\MultPre}$ protocol.

\begin{lemma}[Communication]
	\label{lemma:piMultF3pcM}
	Protocol $\prot{\MultF}$~(in $\Tthis$) requires $33\ell$ bits of communication in the preprocessing, and $1$ round and $3 \ell$ bits of communication in the online phase.
\end{lemma}
\begin{proof}
	During preprocessing, 11 instances of $\prot{\MultPre}$ protocol are executed. This results in communication of $33\ell$ bits when the $\prot{\MultPre}$ protocol is instantiated with~\cite{CCS:BGIN19}.
	The online phase is similar to that of $\prot{\Mult}$, and the costs follow from Lemma~\ref{lemma:piMult3pcM}.
\end{proof}

\section{Garbled World}
\label{sec:3pcMGCWorld}
Similar to $\TSthis$, we have two variants -- $\TthisT$ requiring communication of 2 GCs and one online round, and $\TthisC$ requiring 1 GC and two rounds. 
The 2 GC variant has two parallel executions, each comprising of 2 garblers and 1 evaluator. $P_1, P_2$ act as evaluators in two independent executions and the parties in $\PlSet{1} = \{P_2, P_3\}$, $\PlSet{2} = \{P_1, P_3\}$ act as garblers, respectively. The 1 GC variant comprises of a single execution with $\PlSet{1}$ acting as garblers and $P_1$ as the evaluator.
%

\subsection{2 GC Variant}
\label{sec:GCT3pcM}

\paragraph{Input Phase}
\label{p:GCIpT3pcM}
Here, the boolean values $(\mk{\vl{x}}, \pad{\vl{x}}{1}, \pad{\vl{x}}{2}, \pad{\vl{x}}{3})$ act as the {\em new} inputs for the garbled computation.  The semantics of $\shrB{\cdot}$-sharing ensures that each of these shares is available with at least two parties (including evalator) in each garbling instance. Thus, the goal of our input phase is to create the compound sharing, $\shrC{\vl{x}} = (\shrG{\mk{\vl{x}}}, \shrG{\pad{\vl{x}}{1}}, \shrG{\pad{\vl{x}}{2}}, \shrG{\pad{\vl{x}}{2}})$ for every input $\vl{x}$ to the function to be evaluated via the GC. 
Consider the garbling instance with $P_1$ as the evaluator. The number of input keys for this instance can be further reduced by treating $(\mk{\vl{x}} \xor \pad{\vl{x}}{2})$ as a single input. For the case of arithmetic values, the input changes to $(\mk{\vl{x}} - \pad{\vl{x}}{2})$. Similarly, the other instance with $P_2$ as evaluator uses $(\mk{\vl{x}} \xor \pad{\vl{x}}{1})$ as input.
We first discuss the semantics for $\shrG{\cdot}$-sharing followed by steps for generating $\shrC{\cdot}$-sharing.

\paragraph{Garbled sharing semantics}
\label{p:GCSemT3pcM}
A value $\vl{v} \in \Z{}$  is $\shrG{\cdot}$-shared (garbled shared) amongst  $\Partyset$ if $P_3$ holds $\shrG{\vl{v}}_{i}= (\key{{\vl{v}}}{0,1}, \key{{\vl{v}}}{0,2})$, $P_1$ holds $\shrG{\vl{v}}_{1} = (\key{{\vl{v}}}{\vl{v},1}, \key{{\vl{v}}}{0,2})$ and $P_2$ holds $\shrG{\vl{v}}_{2} = (\key{{\vl{v}}}{0,1}, \key{{\vl{v}}}{\vl{v},2})$. Here, $\key{{\vl{v}}}{\vl{v}, j} = \key{{\vl{v}}}{0, j} \xor \vl{v} \Delta^{j}$ for $j \in \{1, 2\}$, and $\Delta^{j}$, which is known only to the garblers in $\PlSet{j}$, denotes the global offset with its least significant bit set to $1$ and is same for every wire in the circuit. 
A value $\vl{x} \in \Z{}$ is said to be $\shrC{\cdot}$-shared (compound shared) if each value from $(\mk{\vl{x}}, \pad{\vl{x}}{1}, \pad{\vl{x}}{2}, \pad{\vl{x}}{3})$ is $\shrG{\cdot}$-shared. We write $\shrC{\vl{x}} = (\shrG{\mk{\vl{x}}}, \shrG{\pad{\vl{x}}{1}}, \shrG{\pad{\vl{x}}{2}}, \shrG{\pad{\vl{x}}{2}})$. 

\paragraph{Generation of $\shrG{\vl{v}}$ and $\shrC{\vl{x}}$} 
Protocol $\pigsh(\Partyset, \vl{v})$~(\boxref{fig:pigsh3pcM}) enables generation of $\shrG{\vl{v}}$ where two garblers in each garbling instance hold $\vl{v}$, and proceeds as follows. Consider the first garbling instance with evaluator $P_1$ and garblers $P_s, P_t$. Garblers in $\PlSet{1}$ generate $\{\key{{{\vl{v}}}}{\bitb, 1}\}_{\bitb \in \{0, 1\}}$ which denotes the key for value $\bitb$ on wire $\vl{v}$, following the free-XOR technique~\cite{ICALP:KolSch08,C:KolMohRos14}. If the value $\vl{v}$ is known to both $P_s, P_t$, they $\jsend$ the respective key to $P_1$. Else, w.l.o.g. let $P_s \in \PlSet{j}$ be the garbler that knows $\vl{v}$. To ensure the correct key delivery towards $P_1$, we make garblers $P_s, P_t$ commit to both the keys to $P_1$ via $\jsend$. $P_s$ then sends the opening for commitment of $\key{\vl{v}}{\vl{v}, j}$ to $P_1$. If the decommitment fails, $P_1$ abort for the case of secuirty with abort or fairness. For robustness it accuses $P_s$ and $P_t$ is chosen as the $\TTP$.  

Similar steps carried out with respect to the second garbling instance, at the end of which, garblers in $\PlSet{2}$ possess $\{\key{\vl{v}}{\bitb, 2}\}_{\bitb \in \{0,1\}}$ while the evaluator $P_2$ holds $\key{\vl{v}}{\vl{v}, 2}$. Following this, the shares $\shrG{\vl{v}}_s$ held by $P_s \in \Partyset$ are defined as $\shrG{\vl{v}}_0 = (\key{\vl{v}}{0, 1}, \key{\vl{v}}{0, 2})$, $\shrG{\vl{v}}_1 = (\key{\vl{v}}{\vl{v}, 1}, \key{\vl{v}}{0, 2})$, $\shrG{\vl{v}}_2 = (\key{\vl{v}}{0, 1}, \key{\vl{v}}{\vl{v}, 2})$. 
To generate $\shrC{\vl{x}}$, $\pigsh$ is invoked for each of $\mk{\vl{x}}, \pad{\vl{x}}{1}, \pad{\vl{x}}{2}$, and $\pad{\vl{x}}{3}$ .  

\begin{protocolbox}{$\pigsh(\Partyset, \vl{v})$}{Generation of $\shrG{\vl{v}}$ in $\Tthis$.}{fig:pigsh3pcM}
	\justify
	\detail{
		{\bf Input(s):} $\vl{v}$,~~{\bf Output:} $\shrG{\vl{v}}$.
	} \\
	Let $P_s \in \PlSet{j}$ be the garbler that knows $\vl{v}$ in clear where $j \in \{1, 2\}$ and $P_t$ be the co-garbler in $\PlSet{j}$.
	\begin{enumerate}
		\item Garblers in $\PlSet{j}$ generate keys $\key{{\vl{v}}}{0, j}, \key{{\vl{v}}}{1, j}$ for wire $\vl{v}$, using free-XOR technique.
		\item If both $P_s$ and $P_t$ know $\vl{v}$ in clear, $P_s, P_t$ $\jsend$ $\key{\vl{v}}{\vl{v}, j}$ to evaluator $P_j$.
		\item Else, parties proceed as follows:
		\begin{enumerate}
			\item $P_s, P_t$ prepare commitment on both the keys as $\Commit{\key{{\vl{v}}}{0, j}}, \Commit{\key{{\vl{v}}}{1, j}}$ and communicates to evaluator $P_j$ in a random permuted order using $\jsend$.
			\item $P_s$ sends the opening for commitment of $\key{\vl{v}}{\vl{v}, j}$ to $P_j$.
			\item If the decommitment using the opening fails, $P_j$ $\abort$ for the case of security with abort or fairness. For robustness, $P_j$ broadcasts "\texttt{(accuse,$P_s$)}" and $P_t$ is chosen as the $\TTP$.
		\end{enumerate}
		\item $P_0$ sets $\shrG{\vl{v}}_0 = (\key{{\vl{v}}}{0,1}, \key{{\vl{v}}}{0,2})$, $P_1$ sets $\shrG{\vl{v}}_{1} = (\key{{\vl{v}}}{\vl{v},1}, \key{{\vl{v}}}{0,2})$ and $P_2$ sets $\shrG{\vl{v}}_{2} = (\key{{\vl{v}}}{0,1}, \key{{\vl{v}}}{\vl{v},2})$.
	\end{enumerate}
\end{protocolbox}

\paragraph{Evaluation} 
\label{p:GCEvT3pcM}
Let $f(\vl{x})$ be the function to be evaluated. At this point, the function input is $\shrC{\cdot}$-shared. This renders $\shrG{\cdot}$-sharing for the input of the GC that corresponds to the function $f'\big(\mk{\vl{x}}, \pad{\vl{x}}{1}, \pad{\vl{x}}{2}, \pad{\vl{x}}{3} \big)$ which first combines the given boolean-shares to compute the actual input and then applies $f$ on it. Let $\GC_j$ denote the garbled circuit to be sent to $P_j \in \{P_1, P_2\}$ by garblers in $\PlSet{j}$. Sending of $\GC_j$ is overlapped  with the key transfer (during generation of $\shrC{\vl{x}}$), to save rounds, where garblers $\jsend$ $\GC_j$ to $P_j$. On receiving the $\GC$, evaluators evaluate their respective GCs and obtain the key corresponding to the output, say $\vl{z}$. This generates $\shrG{\vl{z}}$. 

\paragraph{Output phase} 
\label{p:GCOpT3pcM}
The goal of output computation is to compute the output $\vl{z}$ from $\shrG{\vl{z}}$.
To reconstruct $\vl{z}$ towards $P_j \in \{P_1, P_2\}$, garblers in $\PlSet{j}$ $\jsend$ the least significant bit $\vl{p}^j$ of $\key{\vl{z}}{0, j}$, referred to as the decoding information, to $P_j$. $P_j$ uses the received $\vl{p}^j$ to reconstruct $\vl{z}$ as $\vl{z} = \vl{p}^j \xor \vl{q}^j$, where $\vl{q}^j$ denotes the least significant bit of $\key{\vl{z}}{\vl{z}, j}$. $P_1, P_2$ then $\jsend$ $\vl{z}$ to $P_3$ completing the protocol. Reconstruction is lightweight and requires a single round towards $P_3$ while reconstruction towards $P_1, P_2$ can be overlapped with key transfer and does not incur extra rounds. The protocol appears in \boxref{fig:pirec3pcM}.

\begin{protocolbox}{$\pigrec(\Partyset, \shrG{\vl{z}})$}{Output computation: reconstruction of $\vl{z}$ in $\Tthis$.}{fig:pirec3pcM}
	\justify
	\detail{
		{\bf Input(s):} $\shrG{\vl{z}}$,~~{\bf Output:} $\vl{z}$.
	} \\
	\begin{enumerate}
		\item For an output wire $\vl{z}$, let $\vl{p}^j$ denote the least significant bit of $\key{{\vl{z}}}{0,j}$ and $\vl{q}^j$ denote the least significant bit of $\key{{\vl{z}}}{\vl{z},j}$ for $j \in \{1, 2\}$.
		\item {\em Reconstruction towards $P_j$}: Parties in $\PlSet{j}$ $\jsend$ $\vl{p}^j$ to $P_j$ who reconstructs $\vl{z} = \vl{p}^j \xor \vl{q}^j$.
		\item {\em Reconstruction towards $P_3$}: $P_1, P_2$ $\jsend$ $\vl{z}$ to $P_3$. 
	\end{enumerate}
\end{protocolbox}

\paragraph{Optimizations when deployed in mixed framework}
\label{p:GCMixT3pcM}
Working in the preprocessing model enables transfer of the (communication-intensive) GC and generating $\shrG{\cdot}$-shares of the input-independent shares of $\vl{x}$ (i.e. $\pad{\vl{x}}{}$) in the preprocessing. Thus, the online phase is very light and only requires one round to generate $\shrG{\cdot}$-shares  for the input-dependent data (i.e. ${\mk{\vl{x}}}$). Since evaluation is local, evaluators obtain $\shrG{\cdot}$-sharing of the GC output at the end of $1$ round. Moreover, we require the garbled output to be reconstructed towards both $P_1$ and $P_2$ in clear. Thus, the steps for reconstruction towards $P_3$ can be avoided in $\pigrec$ protocol~(\boxref{fig:pirec3pcM}).

\subsection{1 GC Variant}
\label{sec:GCO3pcM}
The garbling scheme here is similar to the 2GC variant except that now there exists only a single garbling instance. Parties in $\PlSet{1} = \{P_2, P_3\}$ act as the garblers while $P_1$ act as the evaluator. Looking ahead, in the mixed protocol framework, the output has to be reconstructed towards $P_1, P_2$. Reconstruction towards $P_1$ does not incur additional rounds since sending of decoding information can be overlapped with the key transfer. To reconstruct towards $P_2$, $P_1$ sends the least significant bit of $\key{{\vl{z}}}{\vl{z},1}$, denoted by $\vl{q}^1$, along with a hash of $\key{{\vl{z}}}{\vl{z},j}$ to $P_2$. Party $P_2$ accepts $\vl{q}^1$ if the hash is consistent with the respective key. This is fine since a corrupt $P_1$ cannot send an incorrect key due to the {\em authenticity} of the garbling scheme~\cite{CCS:BelHoaRog12}. On the other hand, if the hash is inconsistent, $P_2$ aborts for the case of security with abort or fairness. For robustness it accuses $P_1$ and $P_3$ is chosen as the $\TTP$.  

\section{Security proofs}
\label{sec:GCSec3pcM}
Without loss of generality, we prove the security of our robust framework. The case for fairness follows similarly, and we omit its details. We provide proofs in the $\{\FSETUP, \Func[\MultPre], \Func[\jsend]\}$-hybrid model, where $\FSETUP$~(\S\ref{sec:KeySetupprelims}), $\Func[\MultPre]$~(\boxref{fig:FMultPre3pcM}) and $\Func[\jsend]$~(\boxref{fig:JsendFunc}) denote the ideal functionality for the shared-key setup, preprocessing of multiplication~($\prot{\MultPre}$) and $\jsend$, respectively. 

The strategy for simulating the computation of function $f$ (represented by a circuit $\Ckt$) is as follows. The simulation begins with the simulator emulating the shared-key setup~($\FSETUP$) functionality and giving the respective keys to the adversary. This is followed by the input sharing phase in which $\Sim$ computes the input of $\Adv$, using the known keys, and sets the honest parties' inputs to be used in the simulation to $0$. $\Sim$ invokes the ideal functionality $\Func[GOD]$ on behalf of $\Adv$ using the extracted input and obtains the output $\vl{y}$. $\Sim$ now knows the inputs of $\Adv$ and can compute all the intermediate values for each building block. $\Sim$ proceeds with simulating each of the building blocks in the topological order. We provide the simulation for the case for corrupt $P_1$ and $P_3$. The case for corrupt $P_2$ is similar to that of $P_1$. 

For modularity, we provide the simulation steps for each building block separately. Carrying out these blocks in the topological order yields the simulation for the entire computation. If a $\TTP$ is identified during the simulation, the simulator stops and returns the function output to the adversary on behalf of the $\TTP$ as per $\Func[\jsend]$.

\paragraph{Ideal $\jsend$ Functionality}
The ideal $\jsend$ functionality for fairness security appears in \boxref{fig:JsendFFunc3pcM} and that for the robust setting appears in \boxref{fig:Funcjsend3pcM}.

\vspace{-1mm}
\begin{systembox}{$\Func[\jsend]$~(for fair security)}{Ideal functionality for $\jsend$ in $\Tthis$}{fig:JsendFFunc3pcM}
	\justify
	$\Func[\jsend]$ interacts with the parties in $\Partyset$ and the adversary $\Sim$. 
	\begin{myitemize}
		\item[\bf Step 1:] $\Func[\jsend]$ receives $(\INPUT,\vl{v}_s)$ from senders $P_s$ for $s \in \{i,j\}$, $(\INPUT,\bot)$ from receiver $P_k$. While sending the inputs, the adversary is also allowed to send a special $\abort$ command.
		\item[\bf Step 2:] Set $\msg_i = \msg_j = \msg_l = \bot$.
		\item[\bf Step 3:] If $\vl{v}_i = \vl{v}_j$, set $\msg_k = \vl{v}_i$. Else, set $\msg_k = \abort$.
		\item[\bf Step 4:] Send $(\OUTPUT, \msg_s)$ to $P_s$ for $s \in \{1,2, 3\}$.
	\end{myitemize}
\end{systembox}

\paragraph{Sharing Protocol~($\prot{\Sh}$, \boxref{fig:piSh3pcM})}
During the preprocessing, $\Sim_{\prot{\Sh}}^{P_1}$ emulates $\FSETUP$ and gives the respective keys to $\Adv$. The values commonly held with $\Adv$ are sampled using the respective keys, while others are sampled randomly. The details for the online phase are provided next. We omit the simulation for corrupt $P_3$ as it is similar to that of $P_1$.

\begin{simulatorbox}{$\Sim_{\prot{\Sh}}^{P_1}$}{Simulator $\Sim_{\prot{\Sh}}^{P_1}$ for corrupt $P_1$ }{fig:ShFSim13pcM}
	\justify 
	\algoHead{Online:}
	\begin{myitemize}
		\item[--] If dealer is $\Adv$, $\Sim_{\prot{\Sh}}^{P_1}$ receives $\mk{\vl{v}}$ from $\Adv$ on behalf of $P_2$. $\Sim_{\prot{\Sh}}^{P_1}$ computes $\Adv$'s input $\vl{v}$ as $\vl{v} = \mk{\vl{v}} - \sqr{\pd{\vl{v}}}_1 - \sqr{\pd{\vl{v}}}_2 - \sqr{\pd{\vl{v}}}_3$. It invokes $\Func[GOD]$ on input $(\INPUT, \vl{v})$ to obtain the function output $\vl{y}$. 
		\item[--] If dealer is $P_2$ or $P_3$, $\Sim_{\prot{\Sh}}^{P_1}$ sets $\vl{v} = 0$ and performs the protocol steps honestly. 
	\end{myitemize}
\end{simulatorbox}

Shares unknown to $\Adv$ are sampled randomly in the simulation, whereas in the real protocol, they are sampled using the pseudorandom function (PRF). The indistinguishability of the simulation thus follows by a reduction to the security of the PRF. The same holds for the rest of the blocks.

The simulation for the joint sharing protocol~($\prot{\JSh}$) is similar to that of the sharing protocol. The protocol's design is such that the simulator will always know the value to be sent as part of the joint sharing protocol. The communication is constituted by $\jsend$ calls and is emulated according to the simulation of $\Func[\jsend]$.

\paragraph{Multiplication Protocol~($\prot{\Mult}$~\boxref{fig:piMult3pcM})}

\begin{simulatorbox}{$\Sim_{\prot{\Mult}}^{P_1}$}{Simulator $\Sim_{\prot{\Mult}}^{P_1}$ for corrupt $P_1$ }{fig:MulFSim13pcM}
	\justify 
	\algoHead{Preprocessing:}  \vspace{-2mm}
	\begin{description}
		\item[--]  $\Sim_{\prot{\Mult}}^{P_1}$ emulates $\Func[\MultPre]$ for a corrupt $P_1$ and obtains $\gm{\vl{a}\vl{b}}{1}, \gm{\vl{a}\vl{b}}{2}$, and $\gm{\vl{a}\vl{b}}{3}$.
	\end{description}
	\justify 
	\vspace{-2mm}
	\algoHead{Online:}  \vspace{-2mm}
	\begin{description}
		\item[--] Computes $\vl{y}_1, \vl{y}_2, \vl{y}_3$ honestly. 
		\item[--] Emulates two instances of $\Func[\jsend]$ -- i) $\Adv$ as sender to send $\vl{y}_1$ to $P_2$, and ii) $\Adv$ as receiver to obtain $\vl{y}_2$ from $P_2$. 
		\item[--] Simulates joint sharing for a corrupt sender as discussed earlier.
	\end{description}
\end{simulatorbox}

\begin{simulatorsplitbox}{$\Sim_{\prot{\Mult}}^{P_3}$}{Simulator $\Sim_{\prot{\Mult}}^{P_3}$ for corrupt $P_3$ }{fig:MulFSim33pcM}
	\justify 
	\algoHead{Preprocessing:} \vspace{-2mm}
		\begin{description}
		\item[--]  $\Sim_{\prot{\Mult}}^{P_3}$ emulates $\Func[\MultPre]$ for a corrupt $P_3$ and obtains $\gm{\vl{a}\vl{b}}{1}, \gm{\vl{a}\vl{b}}{2}$, and $\gm{\vl{a}\vl{b}}{3}$.
	\end{description}
	\justify 
	\vspace{-2mm}
	\algoHead{Online:}  \vspace{-2mm}
	\begin{description}
		\item[--] Computes $\vl{y}_1, \vl{y}_2, \vl{y}_3$ honestly. 
		\item[--] Emulates two instances of $\Func[\jsend]$ with $\Adv$ as sender to send $\vl{y}_1$ to $P_2$ and $\vl{y}_2$ to $P_1$. 
		\item[--] Simulates joint sharing for a corrupt receiver as discussed earlier.
	\end{description}
\end{simulatorsplitbox}

\paragraph{Reconstruction Protocol~($\prot{\Rec}$, \boxref{fig:piRec3pcM})}
Using the input of $\Adv$ obtained during simulation of sharing protocol, $\Sim_{\prot{\Rec}}$ invokes $\Func[GOD]$ on behalf of $\Adv$ and obtains the function output $\vl{y}$ in clear. $\Sim_{\prot{\Rec}}$ calculates the missing share of $\Adv$ using $\vl{y}$ and the other shares. The missing share is then communicated to $\Adv$ by emulating the $\Func[\jsend]$ functionality.

\chapter{$\Fthis$: 4PC Fair and Robust Protocols}
\label{chap:layer1_4pc}
This chapter provides details for the Layer I blocks of our 4PC framework $\Fthis$. 
Some of the results in this chapter resulted in publications at NDSS'20~\cite{NDSS:ChaRacSur20} and NDSS'22~\cite{EPRINT:KPRS21}. 
Depending on the sensitivity of the application and the underlying data, we may want different levels of security. For this, we propose multiple variants of the framework, covering fairness~(\textbf{$\Fthis$}) and robustness~(\textbf{$\FthisA$}, \textbf{$\FthisB$}) guarantees. Comparison of $\Fthis$ with actively secure 4PC PPML frameworks, in terms of the communication for multiplication, is presented in \tabref{4pcCost}.

\begin{table}[htb!]
	\centering
	\resizebox{0.98\textwidth}{!}{
		\begin{NiceTabular}{r c r|r r|r r|c}
			\toprule
			\Block{2-1}{Work}
			& \Block[c]{2-1}{\#Active\\Parties}
			& \Block{2-1}{Security}
			& \multicolumn{2}{c}{Multiplication} 
			& \multicolumn{2}{c}{Multiplication with Truncation\tabularnote{$\ell$ - size of ring in bits, $x$ - number of bits for the fractional part in FPA semantics.}} 
			& \Block{2-1}{Conversions\tabularnote{A, B, G indicate support for arithmetic, boolean, and garbled worlds respectively.}} \\ \cmidrule{4-7}
			&  & 
			& Comm\textsubscript{pre} 
			& Comm\textsubscript{on}\tabularnote{`Comm' - communication, `pre' - preprocessing, `on' - online} 
			& Comm\textsubscript{pre} 
			& Comm\textsubscript{on} &  \\ 
			\midrule
			Mazloom et al.~\cite{USENIX:MLRG20} & 4 & Abort & $2\ell$ & $4\ell$ & $2\ell$ & $4\ell$ & A-B \\
			Trident~\cite{NDSS:ChaRacSur20} & 3 & Fair & $3\ell$ & $3\ell$ & $6\ell$ & $3\ell$ & A-B-G \\			
			\textbf{$\Fthis$} & 2 & Fair & $2\ell$ & $3\ell$ & $2\ell$ & $3\ell$ & A-B-G\\
			\midrule
			SWIFT~(4PC)~\cite{USENIX:KPPS21} & 2 & GOD & $3\ell$ & $3\ell$ & $4\ell$ & $3\ell$ & A-B\\
			Fantastic Four~\cite{EPRINT:DalEscKel20} & 3 & GOD & - & $6(\ell + \kappa$) &  $76(\ell+\kappa)+54x + 12$  & $9\ell + 6\kappa$  & A-B\\
			\textbf{$\FthisA$} & 2 & GOD & $2\ell$ & $3\ell$ & $2\ell$ & $3\ell$ & A-B-G\\
			\textbf{$\FthisB$} & 2 & GOD & $3\ell$ & $3\ell$ & $3\ell$ & $3\ell$ & A-B-G\\			
			\bottomrule
		\end{NiceTabular}
	}
	\caption{Comparison of malicious 4PC frameworks for PPML}\label{tab:4pcCost}
\end{table}

\section{Preliminaries and Definitions}
\label{sec:4pcPrelim}
We consider $4$ parties denoted by $\Partyset = \{ P_0, P_1, P_2, P_3 \}$ that are connected by pair-wise private and authentic channels in a synchronous network, and a static, active adversary that can corrupt at most 1 party.

\subsection{Sharing Semantics}
\label{sec:4pcsematics}
For the arithmetic and boolean sharing, we follow  a $(4, 1)$ replicated secret sharing~(RSS)~\cite{NDSS:ChaRacSur20}, where a value $\vl{v} \in \Z{\ell}$ is split into four shares. To leverage the benefits of the preprocessing paradigm, we associate meaning to the shares and demarcate the parties in terms of their roles.  Three of the shares of a $(4, 1)$ RSS can be generated in the preprocessing phase independent of the value to be shared, and their sum can be interpreted as a mask. The fourth share, dependent on $\vl{v}$,  can be computed in the online phase and can be treated as the masked value. We denote the three preprocessed shares as $\pad{\vl{v}}{1}, \pad{\vl{v}}{2}, \pad{\vl{v}}{3}$ and the mask as $\pad{\vl{v}}{} =  \pad{\vl{v}}{1} + \pad{\vl{v}}{2} + \pad{\vl{v}}{3}$. The masked value is denoted as $\mk{\vl{v}}$, and $\mk{\vl{v}}= \vl{v} +\pad{\vl{v}}{}$.  

\begin{table}[htb!]
	\centering
		\begin{NiceTabular}{r r r r r}[notes/para]
			\toprule
			Type  & $P_0$ & $P_1$ & $P_2$ & $P_3$\\
			\midrule
			$\sqr{\cdot}$-sharing 
			& $-$      & ${\vl{v}}^1$     & ${\vl{v}}^2$    & $-$\\
			$\spr{\cdot}$-sharing 
			& $-$
			& ${\vl{v}}^1$     
			& ${\vl{v}}^2$    & ${\vl{v}}^3$\\
			$\sgr{\cdot}$-sharing\tabularnote{$\vl{v} = \vl{v}^1 + \vl{v}^2 + \vl{v}^3$} 
			& $-$    &  $({\vl{v}}^1, {\vl{v}}^3)$ 
			& $({\vl{v}}^2, {\vl{v}}^3)$     & $({\vl{v}}^1, {\vl{v}}^2)$\\
			$\shr{\cdot}$-sharing\tabularnote{$\pad{\vl{v}}{} = \pad{\vl{v}}{1} + \pad{\vl{v}}{2}  + \pad{\vl{v}}{3}$, $\mk{\vl{v}} = \vl{v} + \pad{\vl{v}}{}$} 
			& $(\pad{\vl{v}}{1}, \pad{\vl{v}}{2}, \pad{\vl{v}}{3})$ 
			& $(\mk{\vl{v}}, \pad{\vl{v}}{1}, \pad{\vl{v}}{3})$  
			& $(\mk{\vl{v}}, \pad{\vl{v}}{2}, \pad{\vl{v}}{3})$ 
			& $(\mk{\vl{v}}, \pad{\vl{v}}{1}, \pad{\vl{v}}{2})$\\
			\bottomrule
		\end{NiceTabular}
	    \caption{Sharing semantics for a value $\vl{v} \in \Z{\ell}$ in \Fthis.}\label{tab:4pcsharing}
\end{table}

Next, we distinguish the four parties into two sets; the {\em eval} set $\SetE = \{P_1,P_2\}$ which is assigned the task of carrying out the computation, and is active throughout the online phase. The {\em helper} set $\SetD = \{P_0, P_3\}$, is used to assist $\SetE$ in verification, and so it is only active towards the end of the computation. Complying with the roles and RSS format, the distribution is done as follows: $P_0: \{\pad{\vl{v}}{1}, \pad{\vl{v}}{2},  \pad{\vl{v}}{3}\}, P_1: \{\pad{\vl{v}}{1},  \pad{\vl{v}}{3},  \mk{\vl{v}}\}, P_2: \{\pad{\vl{v}}{2},  \pad{\vl{v}}{3},  \mk{\vl{v}}\}$, and $P_3: \{\pad{\vl{v}}{1}, \pad{\vl{v}}{2}, \mk{\vl{v}}\}$. 
The shares are distributed among $\SetD$ such that $P_3$ gets $\mk{\vl{v}}$ whereas $P_0$ gets all the shares of $\pad{\vl{v}}{}$. In the preprocessing phase, $P_0$ computes a part of the data needed for verification~(cf. \boxref{fig:piMultiplication}) using its input independent shares, which is communicated to $P_3$. This enables a verification in the online, without $P_0$, for the fair protocols. 

The RSS sharing semantics is presented in \tabref{4pcsharing}, denoted by $\shr{\cdot}$, in a modular way with the help of three intermediate sharing semantics $\sqr{\cdot}, \spr{\cdot}$ and $\sgr{\cdot}$. All the sharings used are linear i.e. given sharings of values $\vl{v}_1,\ldots, \vl{v}_m$ and public constants $c_1,\ldots,c_m$, sharing of $\sum_{i=1}^m c_i \vl{v}_i$ can be computed non-interactively for an integer $m$.

\begin{notation} \label{notation:4pcconcise}
	(a) For the $\shr{\cdot}$-shares of $n$ values $\vl{a}_1,\ldots,\vl{a}_n$, $\gm{\vl{a}_1 \ldots \vl{a}_n}{} = \prod\limits_{i=1}^{n} \pad{\vl{a}_i}{}$ and $\mk{\vl{a}_1 \ldots \vl{a}_n}{} = \prod\limits_{i=1}^{n} \mk{\vl{a}_i}{}$ (b) We use superscripts ${\bf B}$, and ${\bf G}$ to denote sharing semantics in boolean, and garbled world, respectively-- $\shrB{\cdot}$,  $\shrG{\cdot}$. We omit the superscript for arithmetic world. 
\end{notation}

Sharing semantics for boolean sharing over $\Z{}$ is similar to arithmetic sharing except that addition is replaced with XOR. The semantics for garbled sharing are described in \S\ref{sec:4pcGCWorld} with the relevant context.

\subsubsection{$\FZero$ - Generating additive shares of zero}
\label{sec:4pcMFZero}
In $\Fthis$, we make use of a functionality $\FZero$ to enable parties $P_0, P_i$ obtain $Z_i$ for $i \in \{1,2,3\}$ such that $Z_1 + Z_2 + Z_3 = 0$. We observe that the functionality can be instantiated non-interactively using the pre-shared keys~(cf. \S\ref{sec:KeySetupprelims}). For this, parties in $\Partyset \setminus \{P_j\}$ sample random value $\vl{r}_j$ for $j \in \{1,2,3\}$. The shares are then defined as $Z_1 = \vl{r}_3 - \vl{r}_2, Z_2 = \vl{r}_1 - \vl{r}_3$ and $Z_3 = \vl{r}_2 - \vl{r}_1$.

\subsection{Joint-Send~($\jsend$) Primitive}
\label{sec:4pcjsend}
The Joint-Send~($\jsend$) primitive, for the case of security with fairness, allows to parties $P_i, P_j$ to relay a message $\vl{v}$ to a third party $P_k$ ensuring either the delivery of the message or $\abort$ in case of inconsistency. Towards this, $P_i$ sends $\vl{v}$ to $P_k$, while $P_j$ sends a hash of the same~($\Hash(\vl{v})$) to $P_k$. Party $P_k$ accepts the message if the hash values are consistent and $\abort$ otherwise. Note that the communication of the hash can be clubbed together for several instances and be deferred to the end of the protocol, amortizing the cost.

\paragraph{Joint-Send~($\jsend$) for robust protocols}
The $\jsend$ primitive~(\boxref{fig:4pcjsendrobust}), for the case of robustness, allows two senders $P_i, P_j$ to relay a common message, $\vl{v} \in \Z{\ell}$, to recipient $P_k$, either by ensuring successful delivery of $\vl{v}$, or by establishing a Trusted Third Party ($\TTP$) among the parties. The instantiation of $\jsend$ can be viewed as consisting of two phases ({\em send, verify}), where the {\em send} phase consists of $P_i$ sending $\vl{v}$ to $P_k$ and the rest of the protocol steps go to {\em verify} phase (which ensures correct {\em send} or $\TTP$ identification). This requires $1$ round of interaction and $\ell$ bits of communication. To leverage amortization, {\em verify} will be executed only once, at the end the computation, requiring $2$ rounds.

Note that the appropriate instantiation of $\jsend$ is used depending on the security guarantee. For simplicity, protocols where the fair and robust variants only differ in the instantiation of $\jsend$ used, we give a common construction for both.

\begin{notation}\label{notation_jsend}
	Protocol $\prot{\jsend}$ denotes the instantiation of Joint-Send~($\jsend$) primitive. We say that $P_i, P_j$ $\jsend$ $\vl{v}$ to $P_k$ when they invoke $\prot{\jsend}(P_i, P_j, \vl{v}, P_k)$.
\end{notation}

\begin{protocolbox}{$\prot{\jsend}(P_i, P_j,\vl{v},P_k)$}{Joint-Send for robust protocols in $\Fthis$.}{fig:4pcjsendrobust}
	\detail{
		{\bf Input(s):} $P_i, P_j : \vl{v}$, $P_k : \bot$,~~{\bf Output:} $P_i, P_j : \bot / \TTP$, $P_k : \vl{v} / \TTP$.
	}\\
	\justify
	$P_s \in \Partyset$ initializes an inconsistency bit $\bitb_s = 0$. If $P_s$ remains silent instead of sending $\bitb_s$ in any of the following rounds, the recipient sets $\bitb_s$ to $1$.
	
	\begin{enumerate} 
		\item[--] {\em Send: } $P_i$ sends $\vl{v}$ to $P_k$.
		
		\item[--] {\em Verify: } 
		\begin{enumerate}
			\item[--]  $P_j$ sends $\Hash(\vl{v})$ to $P_k$. 
			$P_k$ sets $\bitb_k = 1$ if the received values are inconsistent or if the value is not received. 
			\item[--] $P_k$ sends $\bitb_k$ to all parties. $P_s$ for $s \in \{i, j, l\}$ sets $\bitb_s = \bitb_k$.
			\item[--] $P_s$ for $s \in \{i, j, l\}$ mutually exchange their bits. $P_s$ resets $\bitb_s = \bitb^{\prime}$ where $\bitb^{\prime}$ denotes the bit which appears in majority among $\bitb_i, \bitb_j, \bitb_l$.
			\item[--] All parties set $\TTP = P_l$ if $\bitb^{\prime} = 1$, terminate otherwise.
		\end{enumerate}
	\end{enumerate} 
\end{protocolbox}

\begin{lemma}[Communication]
	\label{lemma:pijsend}
	Protocol $\prot{\jsend}$~(\boxref{fig:4pcjsendrobust}) requires an amortized communication of $\ell$ bits and $1$ round.
\end{lemma}
\begin{proof}
	In the protocol $\prot{\jsend}(P_i, P_j,\vl{v},P_k)$ for the fair variant, $P_i$ communicates $\vl{v}$ to $P_k$ requiring communication of $\ell$ bits and one round. The hash value communication from $P_j$ to $P_k$ can be clubbed for multiple instances with the same set of parties and hence the cost gets amortized. The analysis is similar for the robust case as well. Here, though the verification consists of multiple steps, the cost gets amortized over multiple instances.
\end{proof}

\section{Arithmetic / Boolean 4PC}
\label{sec:4pcFourPC}
This section covers the details of our 4PC protocol $\Fthis$ over an arithmetic ring $\Z{\ell}$. We begin by explaining the sharing protocol in \secref{share4pcM}, multiplication with abort in \secref{mult4pcM}, and the reconstruction in~\secref{rec4pcM}. Lastly, the details on elevating the security to fairness are presented in \secref{recfair4pcM} and to robustness in \secref{robust4pc}. 

\subsection{Sharing}
\label{sec:share4pcM}
Protocol $\prot{\Sh}$~(\boxref{fig:piSh}) enables $P_i$ to generate $\shr{\cdot}$-share of a value $\vl{v}$. During the preprocessing phase, $\pd{}$-shares are sampled non-interactively using the pre-shared keys~(cf. \S\ref{sec:KeySetupprelims}) in a way that $P_i$ will get the entire mask $\pd{\vl{v}}$. During the online phase, $P_i$ computes $\mk{\vl{v}} = \vl{v} + \pd{\vl{v}}$ and sends to $P_1, P_2, P_3$, which exchange the hash values to check for consistency. Parties abort in the fair protocol in case of inconsistency, whereas for robust security, parties proceed with a default value.

\begin{protocolbox}{$\prot{\Sh}(P_i, \vl{v})$}{$\shr{\cdot}$-sharing of a value $\vl{v}$ by party $P_i$ in $\Fthis$.}{fig:piSh}
	\detail{
		{\bf Input(s):} $P_i : \vl{v}$,~~{\bf Output:} $\shr{\vl{v}}$.
	}
	\justify
	\algoHead{Preprocessing:} Sample as follows: 
	$P_i, P_0, P_1, P_3: \pad{\vl{v}}{1}~~\Big|~~P_i, P_0, P_2, P_3: \pad{\vl{v}}{2}~~\Big|~~P_i, P_0, P_1, P_2: \pad{\vl{v}}{3}$ \\
	\justify
	\vspace{-4mm}
	\algoHead{Online:}
	\begin{enumerate} 
		\item $P_i$ computes $\mk{\vl{v}} = \vl{v} + \pd{\vl{v}}$ and sends to $P_1, P_2, P_3$.
		\item $P_1, P_2, P_3$ mutually exchange $\Hash(\mk{\vl{v}})$ and accept the sharing if there exists a majority. Else parties $\abort$ for the case of fairness and accepts a default value for the case of robust security. 
	\end{enumerate}       
\end{protocolbox}

\begin{lemma}[Communication]
	\label{lemma:pish}
	Protocol $\prot{\Sh}$~(\boxref{fig:piSh}) requires an amortized communication of at most $3\ell$ bits and $1$ round in the online phase.
\end{lemma}
\begin{proof}
	The preprocessing of $\prot{\Sh}$ is non-interactive as the parties sample non interactively using key setup $\Func[Setup]$~(\S\ref{sec:KeySetupprelims}). in the online phase, $P_i$ sends $\mk{\vl{v}}$ to $P_1, P_2, P_3$ resulting in 1 round and communication of at most $3\ell$ bits~($P_i = P_0$). The next round of hash exchange can be clubbed for several instances and the cost gets amortized over multiple instances.
\end{proof}

\subsubsection{Joint Sharing}
\label{sec:jsh4pcM}
Protocol $\prot{\JSh}$ enables parties $P_i, P_j$ to generate $\shr{\cdot}$-share of a value $\vl{v}$. The protocol is similar to $\prot{\Sh}$ except that $P_j$ ensures the correctness of the sharing performed by $P_i$. During the preprocessing, $\pd{}$-shares are sampled such that both $P_i, P_j$ will get the entire mask $\pd{\vl{v}}$. During the online phase, $P_i, P_j$ compute and $\jsend$ $\mk{\vl{v}} = \vl{v} + \pd{\vl{v}}$ to parties $P_1, P_2, P_3$.

For joint-sharing a value $\vl{v}$ possessed by $P_0$ along with another party in the preprocessing, the communication can be optimized further. The protocol steps based on the $(P_i, P_j)$ pair are summarised below:

\begin{small}
	\begin{enumerate} 
		\item[--] $(P_0, P_1): \Partyset \setminus \{P_2\}$ sample $\pad{\vl{v}}{1} \in_R \Z{\ell}$; Parties set $\pad{\vl{v}}{2} = \mk{\vl{v}} = 0$; $P_0, P_1$ $\jsend$ $\pad{\vl{v}}{3} = - \vl{v} - \pad{\vl{v}}{1}$ to $P_2$. 
		\item[--] $(P_0, P_2): \Partyset \setminus \{P_3\}$ sample $\pad{\vl{v}}{3} \in_R \Z{\ell}$; Parties set $\pad{\vl{v}}{1} = \mk{\vl{v}} = 0$; $P_0, P_2$ $\jsend$ $\pad{\vl{v}}{2} = - \vl{v} - \pad{\vl{v}}{3}$ to $P_3$. 
		\item[--] $(P_0, P_3): \Partyset \setminus \{P_1\}$ sample $\pad{\vl{v}}{2} \in_R \Z{\ell}$; Parties set $\pad{\vl{v}}{3} = \mk{\vl{v}} = 0$; $P_0, P_3$ $\jsend$ $\pad{\vl{v}}{1} = - \vl{v} - \pad{\vl{v}}{1}$ to $P_1$.
	\end{enumerate}
\end{small}

\subsection{Multiplication}
\label{sec:mult4pcM}
Given the shares of $\vl{a}, \vl{b}$, the goal of the multiplication protocol is to generate shares of $\vl{z} = \vl{ab}$. The protocol is designed such that parties $P_1, P_2$ obtain a masked version of the output $\vl{z}$, say $\vl{z} - \vl{r}$ in the online phase, and $P_0, P_3$ obtain the mask $\vl{r}$ in the preprocessing phase. Parties then generate $\shr{\cdot}$-sharing of these values by executing $\prot{\JSh}$, and locally compute $\shr{\vl{z} - \vl{r}} + \shr{\vl{r}}$ to obtain the final output. 

\paragraph{Online} 
Note that,

\begin{align}\label{eq:z+r}
	\vl{z} - \vl{r} &= \vl{a}\vl{b} - \vl{r} = (\mk{\vl{a}} - \pd{\vl{a}})(\mk{\vl{b}} - \pd{\vl{b}}) - \vl{r} \nonumber\\ 
	&= \mk{\vl{ab}} - \mk{\vl{a}}\pd{\vl{b}} - \mk{\vl{b}}\pd{\vl{a}} + \gm{\vl{a}\vl{b}}{} - \vl{r}
	~~\text{\footnotesize{(cf. notation~\ref{notation:4pcconcise})}}
\end{align}

In Eq~\ref{eq:z+r}, $P_1, P_2$ can compute $\mk{\vl{ab}}$ locally, and hence we are interested in computing $\vl{y} = (\vl{z - r}) - \mk{\vl{ab}}$. Let us view $\vl{y}$ as $\vl{y} = \vl{y}_1 + \vl{y}_2 + \vl{y}_3$, where $\vl{y}_1$ and $\vl{y}_2$ can be computed respectively by $P_1$ and $P_2$, and $\vl{y}_3$ consists of terms that can be computed by both $P_1, P_2$.
\begin{align}\label{eq:y}
	P_1: \vl{y}_1 &= - \pad{\vl{a}}{1} \mk{\vl{b}} - \pad{\vl{b}}{1} \mk{\vl{a}} + \sqr{\gm{\vl{a}\vl{b}}{} - \vl{r}}_1 \nonumber\\
	P_2: \vl{y}_2 &= - \pad{\vl{a}}{2} \mk{\vl{b}} - \pad{\vl{b}}{2} \mk{\vl{a}} + \sqr{\gm{\vl{a}\vl{b}}{} - \vl{r}}_2 \nonumber\\
	P_1, P_2: \vl{y}_3 &= - \pad{\vl{a}}{3} \mk{\vl{b}} - \pad{\vl{b}}{3} \mk{\vl{a}}
\end{align}

The preprocessing is set up such that $P_1, P_2$ receive an additive sharing~($\sqr{\cdot}$) of $\gm{\vl{a}\vl{b}}{} - \vl{r}$. Parties $P_1, P_2$ mutually exchange the missing share to reconstruct $\vl{y}$ and subsequently $\vl{z - r}$. 

\smallskip
\begin{protocolsplitbox}{$\prot{\Mult}(\vl{a}, \vl{b}, \isTr)$}{Multiplication with / without truncation in $\Fthis$.}{fig:piMultiplication}
	$\isTr$ is a bit denoting whether truncation is required ($\isTr =1$) or not ($\isTr=0$). \\
	\detail{
		{\bf Input(s):} $\shr{\vl{a}}, \shr{\vl{b}}$.\\
		{\bf Output:} $\shr{\vl{o}}$ where $\vl{o} = \vl{z}^{\vl{t}}$ if $\isTr = 1$ and $\vl{o} = \vl{z}$ if $\isTr = 0$ and $\vl{z} = \vl{ab}$.
	}
	\justify 
	\vspace{-2mm}
	\algoHead{Preprocessing:} 
	\begin{enumerate} 
		\item Locally compute the following:
		\begin{align*}
			P_0, P_1: \gm{\vl{a}\vl{b}}{1} &= \pad{\vl{a}}{1} \pad{\vl{b}}{3} + \pad{\vl{a}}{3} \pad{\vl{b}}{1} + \pad{\vl{a}}{3} \pad{\vl{b}}{3} \\
			P_0, P_2: \gm{\vl{a}\vl{b}}{2} &= \pad{\vl{a}}{2} \pad{\vl{b}}{3} + \pad{\vl{a}}{3} \pad{\vl{b}}{2} + \pad{\vl{a}}{2} \pad{\vl{b}}{2}\\
			P_0, P_3: \gm{\vl{a}\vl{b}}{3} &= \pad{\vl{a}}{1} \pad{\vl{b}}{2} + \pad{\vl{a}}{2} \pad{\vl{b}}{1} + \pad{\vl{a}}{1} \pad{\vl{b}}{1}
		\end{align*}
		\item $P_0, P_3$ and $P_j$ sample random ${\vl{u}}^j \in_R \Z{\ell}$ for $j \in \{1,2\}$. Let ${\vl{u}^1} + \vl{u}^2 = \gm{\vl{a}\vl{b}}{3} - \vl{r}$ for a random $\vl{r} \in_R \Z{\ell}$.  
		\item $P_0, P_3$ compute $\vl{r} = \gm{\vl{a}\vl{b}}{3} - {\vl{u}^1} - \vl{u}^2$ and set $\vl{q} = \vl{r}^{\vl{t}}$  if $\isTr = 1$, else set $\vl{q} = \vl{r}$. $P_0, P_3$ execute $\prot{\JSh}(P_0, P_3, \vl{q})$ to generate $\shr{\vl{q}}$.
		\item  $P_0, P_1, P_2$ sample random ${\vl{s}}_1, {\vl{s}}_2 \in_R \Z{\ell}$ and set ${\vl{s}} = {\vl{s}}_1 + {\vl{s}} _2$\footnote{For the fair protocol, it is enough for $P_0, P_1, P_2$ to sample ${\vl{s}}$ directly.}. 
		$P_0$ sends $\vl{w} = \gm{\vl{a}\vl{b}}{1} + \gm{\vl{a}\vl{b}}{2} + {\vl{s}}$ to $P_3$.
	\end{enumerate}
	\justify
	\vspace{-2mm}
	\algoHead{Online:} Let $\vl{y} = (\vl{z} - \vl{r}) - \mk{\vl{a}} \mk{\vl{b}}$.
	\begin{enumerate} 
		\item Locally compute the following:
		\begin{align*}
			P_1: \vl{y}_1 &= - \pad{\vl{a}}{1} \mk{\vl{b}} - \pad{\vl{b}}{1} \mk{\vl{a}} + \gm{\vl{a}\vl{b}}{1} + {\vl{u}}^1 \\
			P_2: \vl{y}_2 &= - \pad{\vl{a}}{2} \mk{\vl{b}} - \pad{\vl{b}}{2} \mk{\vl{a}} + \gm{\vl{a}\vl{b}}{2} + {\vl{u}}^2 \\
			P_1, P_2: \vl{y}_3 &= - \pad{\vl{a}}{3} \mk{\vl{b}} - \pad{\vl{b}}{3} \mk{\vl{a}}
		\end{align*}
		\item $P_1$ sends $\vl{y}_1$ to $P_2$, while $P_2$ sends $\vl{y}_2$ to $P_1$, and they locally compute $\vl{z} - \vl{r} = (\vl{y}_1 + \vl{y}_2 + \vl{y}_3) + \mk{\vl{a}} \mk{\vl{b}}$.
		\item If $\isTr = 1$, $P_1, P_2$ set $\vl{p} = (\vl{z} - \vl{r})^{\vl{t}}$, else $\vl{p} = \vl{z} - \vl{r}$. $P_1, P_2$ execute $\prot{\JSh}(P_1, P_2, \vl{p})$ to generate $\shr{\vl{p}}$. 
		\item Parties locally compute $\shr{\vl{o}} = \shr{\vl{p}} + \shr{\vl{q}}$. Here $\vl{o} = \vl{z}^{\vl{t}}$ if $\isTr = 1$ and $\vl{z}$ otherwise.
		\item {\em Verification:} $P_3$ computes $\vl{v} = - (\pad{\vl{a}}{1} + \pad{\vl{a}}{2}) \mk{\vl{b}} - (\pad{\vl{b}}{1} + \pad{\vl{b}}{2}) \mk{\vl{a}} + {\vl{u}^1} + \vl{u}^2 + \vl{w}$ and sends $\Hash(\vl{v})$ to $P_1$ and $P_2$. Parties $P_1, P_2$ $\abort$ iff $\Hash(\vl{v}) \ne \Hash(\vl{y}_1 + \vl{y}_2 + {\vl{s}})$.
	\end{enumerate}     
\end{protocolsplitbox}

\paragraph{Verification}
To ensure the correctness of the values exchanged, we use the assistance of $P_3$. Concretely, $P_3$ obtains $\vl{y}_1 + \vl{y}_2 + \vl{s}$, where $\vl{s}$ is a random mask known to $P_0, P_1, P_2$. For this $P_3$ needs $\gm{\vl{a}\vl{b}}{} + \vl{s}$, which it obtains from the preprocessing phase. The mask $\vl{s}$ is used to prevent the leakage from $\gm{\vl{a}\vl{b}}{}$ to $P_3$. $P_3$ computes a hash of $\vl{y}_1 + \vl{y}_2 + \vl{s}$ and sends it to $P_1, P_2$, which $\abort$ if it is inconsistent. 

\paragraph{Preprocessing} 
Parties should obtain the following values from the preprocessing phase:
\begin{equation*}
	{\sf i)}~~P_1, P_2: \sqr{\gm{\vl{ab}}{} - \vl{r}}~~\Big|~~
	{\sf ii)}~~P_0, P_3: \vl{r}~~\Big|~~
	{\sf iii)}~~P_3: \gm{\vl{a}\vl{b}}{} + \vl{s}
\end{equation*}

For ${\sf i)}$ and ${\sf ii)}$, let $\gm{\vl{ab}}{} = \gm{\vl{ab}}{1} + \gm{\vl{ab}}{2} + \gm{\vl{ab}}{3}$, where $P_0$ along with $P_i$ can compute $\gm{\vl{ab}}{i}$ for $i \in \{1, 2, 3\}$. For $P_1, P_2$, to form an additive sharing of $(\gm{\vl{ab}}{} - r)$, it suffices for them to define their share as $\gm{\vl{ab}}{i} + \sqr{\gm{\vl{ab}}{3} - r}$. 
Instead of sampling a random $\vl{r}$, $P_0, P_3$, along with $P_i$, sample the share for $\gm{\vl{ab}}{3} - \vl{r}$ as $\vl{u}^i$ for $i \in \{1, 2\}$.  $P_0, P_3$ compute $\vl{r}$ as $\gm{\vl{ab}}{3} - \vl{u}^1 - \vl{u}^2$.

For ${\sf iii)}$, $P_3$ needs $\vl{w} = \gm{\vl{ab}}{1} + \gm{\vl{ab}}{2} + \vl{s}$. To tackle this, $P_0, P_1, P_2$ sample $\vl{s}_1, \vl{s}_2$, and set $\vl{s} = \vl{s}_1 + \vl{s}_2$. $P_0, P_i$, for $i \in \{1, 2\}$, $\jsend$ $\gm{\vl{ab}}{i} + \vl{s}_i$ to $P_3$. This requires a communication of 2 elements. 
As an optimization, $P_0$ sends $\vl{w}$ to $P_3$. If $P_0$ is malicious, it might send a wrong value to $P_3$. However, in this case, every party in the online phase would be honest. And since $P_1, P_2$ do not use $\vl{w}$ in their computation, any error in $\vl{w}$ is bound to get caught in the verification phase.

\begin{lemma}[Communication]
	\label{lemma:piMultF}
	Protocol $\prot{\Mult}$~(\boxref{fig:piMultiplication})~(in $\Fthis$) requires $2 \ell$ bits of communication in the preprocessing phase, and $1$ round and $3 \ell$ bits of communication in the online phase.
\end{lemma}
\begin{proof}
	During preprocessing, sampling of values ${\vl{u}}^1, {\vl{u}}^2$ are performed non-interactively using $\Func[Setup]$. A communication of $\ell$ bits is required for the joint sharing of $\vl{q}$ by $P_0, P_3$ as explained in \S\ref{sec:jsh4pcM}. In addition, $P_0$ communicates $\vl{w}$ to $P_3$ requiring additional $\ell$ bits. 
	During online, two instances of $\prot{\jsend}$ are executed in parallel resulting in a communication of $2\ell$ bits and 1 round. This is followed by a joint sharing by $P_1,P_2$ for which an additional communication of $\ell$ bits are required. However, in joint sharing, the communication is from $P_1$ to $P_3$ and the same can be deferred till the verification stage. Thus the online round is retained as $1$ in an amortized sense. 
\end{proof}

\subsubsection{Truncation}
To accommodate truncation, the multiplication protocol is modified as follows. $P_1, P_2$ locally truncate $(\vl{z - r})$ and generate $\shr{\cdot}$-shares of it in the online phase. Similarly, $P_0, P_3$ truncate $\vl{r}$ in the preprocessing phase and generate its $\shr{\cdot}$-shares. Parties locally compute $\shr{\vl{z}^{\vl{t}}} = \shr{(\vl{z-r})^{\vl{t}}} + \shr{\vl{r}^{\vl{t}}}$.

\subsubsection{Multiplication with constant}
Multiplication by a constant in MPC is typically local. Given constant $\alpha$ and $\shr{\vl{v}}$, the $\shr{\cdot}$-shares of the product $\vl{y} = \alpha\vl{v}$ can be locally computed as per \eqref{eq:mutconst4pcM}. 
\begin{equation}\label{eq:mutconst4pcM}
	\mk{\vl{y}} = \alpha \mk{\vl{u}},~~~\pad{\vl{y}}{1} = \alpha \pad{\vl{v}}{1},~~~\pad{\vl{y}}{2} = \alpha \pad{\vl{v}}{2},~~~\pad{\vl{y}}{3} = \alpha \pad{\vl{v}}{3}
\end{equation}

However, in FPA, we need to perform a truncation on the output. For this, note that the product can be written as $\alpha\vl{v} = \beta^1 + \beta^2$ where $\beta^1 = \alpha.(\mk{\vl{v}} - \pad{\vl{v}}{3})$ and $\beta^2 = \alpha.(- \pad{\vl{v}}{1} - \pad{\vl{v}}{2})$.
 $P_1, P_2$ locally truncate $\beta^{1}$ and execute $\prot{\JSh}$, while $P_0, P_3$ do the same with $\beta^{2}$. 
 
\subsubsection{Special multiplication protocol~$\prot{\MultS}$}
\label{sec:4pcMMultSP}
Given the $\sgr{\cdot}$-shares of values $\vl{a}, \vl{b}$ with $P_0$ knowing the entire shares of both $\sgr{\vl{a}}$ and $\sgr{\vl{b}}$, protocol $\prot{\MultS}$~(\boxref{fig:piMultS}) computes $\sgr{\vl{z}}$ for $\vl{z} = \vl{ab}$.

\begin{protocolbox}{$\prot{\MultS}(\sgr{\vl{a}}, \sgr{\vl{b}})$}{Special multiplication of $\sgr{\cdot}$-shares in $\Fthis$.}{fig:piMultS}
	\detail{
		{\bf Input(s):} $\sgr{\vl{a}}, \sgr{\vl{b}}$,~~~{\bf Output:} $\sgr{\vl{z}}$ where $\vl{z} = \vl{ab}$.
	}
	\justify 
	\begin{enumerate} 
		\item Parties invoke $\FZero$~(\secref{4pcMFZero}) to enable $P_0, P_j$ obtain $Z_j$ for $j \in \{1,2,3\}$ such that $Z_1 + Z_2 + Z_3 = 0$. Then,
			\begin{align*}
				P_0, P_1 &~\jsend~\vl{(ab)}^{1}  = \vl{a}^1 \vl{b}^3 + \vl{a}^3 \vl{b}^1 + \vl{a}^3 \vl{b}^3 + Z_1~\text{to } P_2.\\
				P_0, P_2 &~\jsend~\vl{(ab)}^{2} = \vl{a}^2 \vl{b}^3 + \vl{a}^3 \vl{b}^2 + \vl{a}^2 \vl{b}^2 + Z_2~\text{to } P_3.\\
				P_0, P_3 &~\jsend~\vl{(ab)}^{3} = \vl{a}^1 \vl{b}^2 + \vl{a}^2 \vl{b}^1 + \vl{a}^1 \vl{b}^1 +  Z_3~\text{to } P_1.
			\end{align*}
			\item[--]  Set $\sgr{\vl{z}}$ as $\sgr{\vl{z}}^1 = \vl{(ab)}^{3},~~\sgr{\vl{z}}^2 = \vl{(ab)}^{2},~~\sgr{\vl{z}}^3 = \vl{(ab)}^{1}$.
		\end{enumerate}
	
\end{protocolbox}

\subsection{Reconstruction}
\label{sec:rec4pcM}
Protocol $\prot{\Rec}(\Partyset, \vl{v})$ (\boxref{fig:piRec}) enables parties in $\Partyset$ to compute $\vl{v}$, given its $\shr{\cdot}$-share and achieves security with abort. Note that each party misses one share to reconstruct the output, and the other 3 parties hold this share. 2 out of the 3 parties will $\jsend$ the missing share to the party that lacks it. Reconstruction towards a single party can be viewed as a special case.

\begin{protocolbox}{$\prot{\Rec}(\Partyset, \shr{\vl{v}})$}{Reconstruction (with abort security) of value $\vl{v}$ among $\Partyset$ in $\Fthis$.}{fig:piRec}
	\justify
	\detail{
		{\bf Input(s):} $\shr{\vl{v}}$,~~{\bf Output:} $\vl{v}$.
	}
	\begin{enumerate} 
		\item $P_1, P_0$ $\jsend$ $\pad{\vl{v}}{1}$ to $P_2$;~~~$P_2, P_0$ $\jsend$ $\pad{\vl{v}}{3}$ to $P_3$;\newline $P_3, P_0$ $\jsend$ $\pad{\vl{v}}{2}$ to $P_1$;~~~$P_1, P_2$ $\jsend$ $\mk{\vl{v}}$ to $P_0$.
		\item Parties compute $\vl{v} = \mk{\vl{v}} - \pad{\vl{v}}{1} - \pad{\vl{v}}{2} - \pad{\vl{v}}{3}$. 
	\end{enumerate}
\end{protocolbox}

\begin{lemma}[Communication]
	\label{lemma:pirec}
	Protocol $\prot{\Rec}$~(\boxref{fig:piRec}) requires an amortized communication of $4\ell$ bits and $1$ round in the online phase.
\end{lemma}
\begin{proof}
	The protocol involves 4 invocations of $\prot{\jsend}$ protocol and the communication follows from Lemma~\ref{lemma:pijsend}.
\end{proof}

\subsubsection{Achieving Fairness}
\label{sec:recfair4pcM}

Here, we show how to extend the security of $\Fthis$ from abort to fairness. Before proceeding with the output reconstruction, we must ensure that all the honest parties are alive after the verification phase. For this, all the parties maintain an {\em aliveness} bit, say $\bitb$, which is initialized to $\continue$. If the verification phase is not successful for a party, it sets $\bitb = \abort$. In the first round of reconstruction, the parties mutually exchange their $\bitb$ bit and accept the value that forms the majority. Since we have only one corruption, it is guaranteed that all the honest parties will agree on $\bitb$. If $\bitb = \continue$, the parties exchange their missing shares and accept the majority. As per the sharing semantics, every missing share is possessed by three parties, out of which there can be at most one corruption. As an optimization, for instances where many values are reconstructed, two out of the three parties can send the share while the third can send a hash of the same.

Looking ahead, a similar reconstruction will be used for the robust variants as well. However, there is no need to perform an explicit aliveness check as it will be taken care of by the verification of $\jsend$ instances.

\subsection{Achieving Robustness}
\label{sec:robust4pc}
In this section, we show how to extend the security of $\Fthis$ to robustness. We provide two variants with different trade-offs in the communication for multiplication -- i) $\FthisA$: It has the same amortized communication complexity as that of $\Fthis$ but requires verification in the preprocessing phase, and ii) $\FthisB$: It avoids the verification in $\FthisA$ but incurs a communication overhead of $1$ element in the preprocessing phase over $\Fthis$. 

\subsubsection{$\FthisA$}
\label{p:Tetrad1}
On a high level, we make two modifications to the multiplication protocol $\prot{\Mult}$~(\boxref{fig:piMultiplication}). In the preprocessing, communication comes from a $\prot{\JSh}$ in step 3 of the protocol, and the value $\vl{w}$ sent by $P_0$ to $P_3$, in step 4. To get robustness, the robust variant of $\prot{\JSh}$ is used. To ensure the correctness of $\vl{w}$, we introduce $\piVrfyP$~(\boxref{fig:piVrfyP0}). If $\piVrfyP$ fails, parties identify a $\TTP$ in the preprocessing phase itself. The second modification is in the online phase, which proceeds as that of $\prot{\Mult}$. If any $\abort$ happens, $P_0$ is assigned as the $\TTP$. Since $P_0$ does not participate in the online phase of the multiplication, and its communication in the preprocessing has been verified via $\piVrfyP$, this assignment is safe.

\medskip {\em Verifying the communication by $P_0$:}
In $\prot{\Mult}$~(\boxref{fig:piMultiplication}) protocol, $P_0$ computes and sends $\vl{w} = \gm{\vl{a}\vl{b}}{1} + \gm{\vl{a}\vl{b}}{2} + {\vl{s}}_1 + {\vl{s}}_2$ to $P_3$ with $P_0, P_1, P_2$ knowing ${\vl{s}}_1, {\vl{s}}_2$ in clear. Note that $\vl{w} = \vl{w}^1 + \vl{w}^2$ for $\vl{w}^1 = \gm{\vl{a}\vl{b}}{1} + {\vl{s}}_1$ and $\vl{w}^2 = \gm{\vl{a}\vl{b}}{2} + {\vl{s}}_2$. Also, $P_0$ along with $P_1, P_2$ and $P_3$ possess the values $\vl{w}^1, \vl{w}^2$ and $\vl{w}$ respectively. Checking the correctness of $\vl{w}$ reduces to verifying $\vl{w} = \vl{w}^1 + \vl{w}^2$.

To verify this relation for all $M$ multiplication gates in the circuit, i.e. $\{\vl{w}_j \iseq \vl{w}_j^1 + \vl{w}_j^2\}_{j \in [M]}$, one approach is to compute a random linear combination and verify the relation on the sum. 
While working over a field $\F_p$, this solution has an error probability $1/|\F_p|$, where $|\F_p|$ denotes the size of $\F_p$. However, this solution does not work naively over rings since not every element in the ring has an inverse, unlike fields. Concretely, the check can still pass with a probability of at most $1/2$~\cite{TCC:ACDEY19,CCS:BGIN19}. To reduce the cheating probability, the check is repeated $\csec$ times, thereby bounding the cheating probability by $1/2^{\csec}$. As an optimization, it is sufficient to choose the random combiners from $\{0,1\}$. Thus, parties need to sample only a binary string of $M$ bits using the shared key for one check. The formal verification protocol appears in \boxref{fig:piVrfyP0A}.

\begin{protocolbox}{$\piVrfyP(\{\sqr{{\vl{w}}_j}\}_{j=1}^{M})$}{Verifying $P_0$'s communication in the multiplication protocol of $\FthisA$: Approach 1}{fig:piVrfyP0A}
	\detail{
		{\bf Input(s):} $P_0, P_1: {\vl{w}}_j^1~~\Big|~~P_0, P_2: {\vl{w}}_j^2~~\Big|~~P_0, P_3: {\vl{w}}_j~~\Big|~$, for $j = 1, \ldots, M$.\\
		{\bf Output:} Whether ${\vl{w}}_j = {\vl{w}}_j^1 + {\vl{w}}_j^2$ or not, for $j = 1, \ldots, M$.
	}
	\justify 
	Repeat the following $\csec$ times, in parallel. 
	\begin{enumerate}[itemsep=0mm]
		\item Sample random values $\tau_1,\ldots, \tau_M \in \Z{\ell}$.
		\item Locally compute: $P_0, P_1: \vl{e}^1 = \sum_{j=1}^{M} \tau_j \vl{w}_j^1$; $P_0, P_2: \vl{e}^2 = \sum_{j=1}^{M} \tau_j \vl{w}_j^2$; $P_0, P_3: \vl{e} = \sum_{j=1}^{M} \tau_j \vl{w}_j$.
		\item $(P_0,P_1)$, $(P_0,P_2)$ and $(P_0,P_3)$ generate $\shr{\cdot}$-shares of $ \vl{e}^1,  \vl{e}^2$ and  $\vl{e}$ respectively using $\prot{\JSh}$. 
		\item Locally compute $\shr{\vl{g}} = \shr{\vl{e}} - \shr{\vl{e}^1} - \shr{\vl{e}^2}$. 
		\item Robustly reconstruct $\vl{g}$ and check if $\vl{g} \iseq 0$. 
	\end{enumerate}     
	If for all $\csec$ repetitions, $\vl{g} = 0$, then continue with rest of the computation. Else, $P_0$ is identified to be corrupt and $\TTP = P_1$.
\end{protocolbox}

Another approach, that avoids the repetition in the $\piVrfyP$ protocol above, is to perform a similar check over a Galois ring~\cite{TCC:ACDEY19,CCS:BGIN19}. To carry out the verification, the extended ring $\Z{\ell}/f(x)$ is used, which is the ring of all polynomials with coefficients in $\Z{\ell}$ modulo an irreducible polynomial $f$ of degree $d$ over $\Z{}$. Here, each element in $\Z{\ell}$ is lifted to a $d$-degree polynomial in $\Z{\ell}[x]/f(x)$ (which results in blowing up the communication by a factor $d$). Given this, to verify the $M$ values, further packing is performed. More concretely, assume that $d$ divides $M$ and $M = d \cdot q$. For $j = 1, \ldots, q $, public polynomial $g_j$ and shared polynomials $g^1_j$ and $g^2_j$ are defined for each set of $d$ values $\{\vl{w}, \vl{w}^1, \vl{w}^2\}$, all of which are then combined to check whether $\{\vl{w}_j \iseq \vl{w}_j^1 + \vl{w}_j^2\}_{j \in [M]}$. We describe the polynomial with respect to $j=1$ below.  

\begin{align*}
	g_1 = \vl{w}_{1} + X \cdot \vl{w}_{2} + \ldots + X^{d-1} \cdot \vl{w}_{d} \\
	g_1^1 = \vl{w}^1_{1} + X \cdot \vl{w}^1_{2} + \ldots + X^{d-1} \cdot \vl{w}^1_{d} \\
	g_1^2 = \vl{w}^2_{1} + X \cdot \vl{w}^2_{2} + \ldots + X^{d-1} \cdot \vl{w}^2_{d} 
\end{align*}

Now, parties sample random values $\vl{r}_1, \ldots, \vl{r}_{q} \in \Z{\ell}/f(x)$ and compute $g = \sum_{j=1}^{q} \vl{r}_j g_j $, $g^1 = \sum_{j=1}^{q} \vl{r}_j g^1_j$ and $g^2 = \sum_{j=1}^{q} \vl{r}_j g^2_j $. This is followed by robustly reconstructing $g - g^1 - g^2$ and verifying if this value is $0$. If not, $P_0$ is identified to be a corrupt and computation is carried out by a $\TTP$. The formal verification protocol appears in \boxref{fig:piVrfyP0}.

\smallskip
\begin{protocolbox}{$\piVrfyP(\{\sqr{{\vl{w}}_j}\}_{j=1}^{M})$}{Verifying $P_0$'s communication in the multiplication protocol of $\FthisA$: Approach 2}{fig:piVrfyP0}
	\detail{
		{\bf Input(s):} $P_0, P_1: {\vl{w}}_j^1~~\Big|~~P_0, P_2: {\vl{w}}_j^2~~\Big|~~P_0, P_3: {\vl{w}}_j~~\Big|~$, for $j = 1, \ldots, M$.\\
		{\bf Output:} Whether ${\vl{w}}_j = {\vl{w}}_j^1 + {\vl{w}}_j^2$ or not, for $j = 1, \ldots, M$.
	}
	\justify 
	\begin{enumerate} 
		\item Define the following polynomials over $\Z{\ell}/f(x)$ for $j \in [q]$ . 
		\begin{align*}
			g_j = \vl{w}_{1 + (j-1)d} + X \cdot \vl{w}_{2 + (j-1)d} + \ldots + X^{d-1} \cdot \vl{w}_{d + (j-1)d} \\
			g^1_j = \vl{w}^1_{1 + (j-1)d} + X \cdot \vl{w}^1_{2 + (j-1)d} + \ldots + X^{d-1} \cdot \vl{w}^1_{d + (j-1)d} \\
			g^2_j = \vl{w}^2_{1 + (j-1)d} + X \cdot \vl{w}^2_{2 + (j-1)d} + \ldots + X^{d-1} \cdot \vl{w}^2_{d + (j-1)d} 
		\end{align*}
		\item Parties generate random values $\vl{r}_1, \ldots, \vl{r}_q \in \Z{\ell}/f(x)$, and compute $g = \sum_{j=1}^{q} \vl{r}_j g_j $, $g^1 = \sum_{j=1}^{q} \vl{r}_j g^1_j$ and $g^2 = \sum_{j=1}^{q} \vl{r}_j g^2_j $.
		\item Parties execute $\prot{\JSh}(P_0, P_1, g^1)$, $\prot{\JSh}(P_0, P_2, g^2)$ and $\prot{\JSh}(P_0, P_3, g)$ to generate $\shr{g^1}, \shr{g^2}$ and $\shr{g}$, respectively.
		\item Parties robustly reconstruct $g - g^1 - g^2$ and check equality to $0$. If it is $0$, then parties continue with rest of the computation. Else, $P_0$ is identified to be corrupt and $\TTP = P_1$.
	\end{enumerate}     
\end{protocolbox}

\subsubsection{$\FthisB$}
\label{p:Tetrad2}
This variant~(\boxref{fig:piMultGOD}) avoids the verification of $P_0$ at the cost of communicating 1 extra ring element in the preprocessing, compared to $\FthisA$. Note that the communication cost of this protocol is similar to that of the one in SWIFT~\cite{USENIX:KPPS21}. We were unable to extend the latter's efficiently to support multi-input multiplication.
Hence, we design $\FthisB$ that has the same communication complexity as SWIFT but also supports multi-input multiplication, as well as truncation without any overhead. In order to get rid of $\piVrfyP$, the communication of $\vl{w}$ from $P_0$ to $P_3$ is split into 2 parts. ($P_0, P_1$) and ($P_0, P_2$) compute $\vl{w}$ in parts, and send them to $P_3$ using $\jsend$. 
This modification allows $P_3$ to compute $\vl{y}_1 + \vl{s}_1$ and $\vl{y}_2 + \vl{s}_2$ separately in the online phase. In addition, to enable $P_2$ to obtain $\vl{y}_1$, $P_1, P_3$ can $\jsend$ $\vl{y}_1 + \vl{s}_1$ to $P_2$. $P_1$ obtains $\vl{y}_2 + \vl{s}_2$ similarly.

The formal protocol for the robust multiplication in $\FthisB$, $\piMultR$, appears in \boxref{fig:piMultGOD}. The primary difference from the fair counterpart is that the communication of $\vl{w}$ from $P_0$ to $P_3$ in the preprocessing is now split into two parts.  $(P_0, P_1), (P_0, P_2)$ communicates $\vl{w}_1, \vl{w}_2$  respectively to $P_3$ via $\jsend$.

\begin{protocolsplitbox}{$\piMultR(\vl{a}, \vl{b}, \isTr)$}{Robust multiplication in $\FthisB$.}{fig:piMultGOD}
	$\isTr$ is a bit denoting whether truncation is required ($\isTr =1$) or not ($\isTr=0$). \\
	\detail{
		{\bf Input(s):} $\shr{\vl{a}}, \shr{\vl{b}}$.\\
		{\bf Output:} $\shr{\vl{o}}$ where $\vl{o} = \vl{z}^{\vl{t}}$ if $\isTr = 1$ and $\vl{o} = \vl{z}$ if $\isTr = 0$ and $\vl{z} = \vl{ab}$.
	}
	\justify 
	\vspace{-2mm}
	\algoHead{Preprocessing:} 
	\begin{enumerate} 
		\item Parties locally compute the following:
		\begin{align*}
			P_0, P_1: \gm{\vl{a}\vl{b}}{1} &= \pad{\vl{a}}{1} \pad{\vl{b}}{3} + \pad{\vl{a}}{3} \pad{\vl{b}}{1} + \pad{\vl{a}}{3} \pad{\vl{b}}{3} \\
			P_0, P_2: \gm{\vl{a}\vl{b}}{2} &= \pad{\vl{a}}{2} \pad{\vl{b}}{3} + \pad{\vl{a}}{3} \pad{\vl{b}}{2} + \pad{\vl{a}}{2} \pad{\vl{b}}{2} \\
			P_0, P_3: \gm{\vl{a}\vl{b}}{3} &= \pad{\vl{a}}{1} \pad{\vl{b}}{2} + \pad{\vl{a}}{2} \pad{\vl{b}}{1} + \pad{\vl{a}}{1} \pad{\vl{b}}{1}
		\end{align*}
		\item $P_0, P_3$ and $P_j$ sample random ${\vl{u}}^j \in_R \Z{\ell}$ for $j \in \{1,2\}$. Let ${\vl{u}^1} + \vl{u}^2 = \gm{\vl{a}\vl{b}}{3} - \vl{r}$ for a random $\vl{r} \in_R \Z{\ell}$.  
		\item $P_0, P_3$ compute $\vl{r} =  \gm{\vl{a}\vl{b}}{3} - {\vl{u}^1} - \vl{u}^2$ and set $\vl{q} = \vl{r}^{\vl{t}}$  if $\isTr = 1$, else set $\vl{q} = \vl{r}$. $P_0, P_3$ execute $\prot{\JSh}(P_0, P_3, \vl{q})$ to generate $\shr{\vl{q}}$.
		\item  $P_0, P_1, P_2$ sample random ${\vl{s}}_1, {\vl{s}}_2 \in_R \Z{\ell}$. $P_0, P_j$ $\jsend$ $\vl{w}_j = \gm{\vl{a}\vl{b}}{j} + {\vl{s}}_j$ to $P_3$  for $j \in \{1,2\}$.
	\end{enumerate}
	\justify
	\vspace{-2mm}
	\algoHead{Online:} Let $\vl{y} = (\vl{z} - \vl{r}) - \mk{\vl{a}} \mk{\vl{b}} + \vl{s}_1 + \vl{s}_2$ .
	\begin{enumerate} 
		\item Parties locally compute the following:
		\begin{align*}
			P_1, P_3: \vl{y}_1 + \vl{s}_1   &=  - \pad{\vl{a}}{1} \mk{\vl{b}} - \pad{\vl{b}}{1} \mk{\vl{a}} + {\vl{u}}^1 + {\vl{w}}_1 \\
			P_2, P_3: \vl{y}_2 + \vl{s}_2 &=  - \pad{\vl{a}}{2} \mk{\vl{b}} - \pad{\vl{b}}{2} \mk{\vl{a}} + {\vl{u}}^2 + {\vl{w}}_2 \\
			P_1, P_2: \vl{y}_3  &= - \pad{\vl{a}}{3} \mk{\vl{b}} - \pad{\vl{b}}{3} \mk{\vl{a}}
		\end{align*}
		\item $P_1, P_3$ $\jsend$ $\vl{y}_1 + \vl{s}_1$ to $P_2$, while $P_1, P_3$ $\jsend$ $\vl{y}_2 + \vl{s}_2$ to $P_1$. 
		\item $P_1, P_2$ locally compute $\vl{z} - \vl{r} = (\vl{y}_1 + \vl{y}_2 + \vl{y}_3) + \mk{\vl{a}} \mk{\vl{b}} - \vl{s}_1 - \vl{s}_2$. 
		\item If $\isTr = 1$, $P_1, P_2$ locally set $\vl{p} = (\vl{z} - \vl{r})^{\vl{t}}$, else $\vl{p} = \vl{z} - \vl{r}$. $P_1, P_2$ execute $\prot{\JSh}(P_1, P_2, \vl{p})$ to generate $\shr{\vl{p}}$. 
		\item Parties locally compute $\shr{\vl{o}} = \shr{\vl{p}} + \shr{\vl{q}}$. Here $\vl{o} = \vl{z}^{\vl{t}}$ if $\isTr = 1$ and $\vl{z}$ otherwise.
	\end{enumerate}     
\end{protocolsplitbox}

\begin{lemma}[Communication]
	\label{lemma:piMultR}
	Protocol $\piMultR$~(\boxref{fig:piMultGOD})~(in $\FthisB$) requires $3 \ell$ bits of communication in the preprocessing phase, and $1$ round and $3 \ell$ bits of communication in the online phase.
\end{lemma}
\begin{proof}
	During preprocessing, the sampling of values ${\vl{u}}^1, {\vl{u}}^2$ are performed non-interactively using $\Func[Setup]$. A communication of $\ell$ bits is required for the joint sharing of $\vl{q}$ by $P_0, P_3$ as explained in \S\ref{sec:jsh4pcM}. In addition, $P_0, P_j$ for $j \in \{1,2\}$ communicates $\vl{w}_j$ to $P_3$ via $\jsend$ requiring additional $2\ell$ bits. 
	The online phase is similar to the fair multiplication protocol~($\prot{\Mult}$) and the costs follow from Lemma~\ref{lemma:piMultF}. 
\end{proof}

\subsection{Multi-input Multiplication}
\label{sec:4pcMultiinputfair}
The goal of 3-input multiplication is to generate $\shr{\cdot}$-sharing of $\vl{z} = \vl{a} \vl{b} \vl{c}$ given $\shr{\vl{a}}, \shr{\vl{b}}, \shr{\vl{c}}$. For this parties proceed similar to the multiplication protocol (see \S\ref{sec:mult4pcM}), where they compute $\shr{\vl{z}} = \shr{\vl{z} - \vl{r}} + \shr{\vl{r}}$. Observe that 
\begin{align*} 
	\vl{z} - \vl{r}
	&= \vl{a} \vl{b} \vl{c} - \vl{r} = (\mk{\vl{a}} - \pd{\vl{a}})(\mk{\vl{b}} - \pd{\vl{b}})(\mk{\vl{c}} - \pd{\vl{c}}) - \vl{r} \\
	&= \mk{\vl{a}\vl{b}\vl{c}} - \mk{\vl{a}\vl{c}} \pd{\vl{b}} - \mk{\vl{b}\vl{c}} \pd{\vl{a}} - \mk{\vl{a}\vl{b}} \pd{\vl{c}} + \mk{\vl{a}} \gm{\vl{b} \vl{c}}{} +  \mk{\vl{b}} \gm{\vl{a} \vl{c}}{} + \mk{\vl{c}} \gm{\vl{a} \vl{b}}{} - \gm{\vl{a} \vl{b} \vl{c}}{} - \vl{r}
\end{align*}
%
Similar to the 2-input fair multiplication $\prot{\Mult}$~(\boxref{fig:piMultiplication}), the goal of the preprocessing phase is to generate additive shares of $\gm{\vl{ab}}{}, \gm{\vl{ac}}{}, \gm{\vl{bc}}{},  \gm{\vl{abc}}{} + \vl{r}$ among $P_1, P_2$.

Informally, the terms that $P_1, P_2$ cannot compute locally for the aforementioned $\gm{}{}$ values, can be computed by $P_0, P_3$, as evident from our sharing semantics. $P_0, P_3$ compute the missing terms and share them among $P_1, P_2$ in the preprocessing phase. $P_1, P_2$ proceed with online phase similar to $\prot{\Mult}$, to compute $\vl{z} - \vl{r}$. Thus the online complexity is retained as that of $\prot{\Mult}$ while the preprocessing communication is increased to 9 elements. The protocol appears in \boxref{fig:piMultT}.

\begin{protocolsplitbox}{$\prot{\MultT}(\vl{a}, \vl{b}, \vl{c}, \isTr)$}{3-input fair multiplication in $\Fthis$.}{fig:piMultT}
	$\isTr$ is a bit denoting whether truncation is required ($\isTr =1$) or not ($\isTr=0$). \\
	\detail{
		{\bf Input(s):} $\shr{\vl{a}}, \shr{\vl{b}}, \shr{\vl{c}}$.\\
		{\bf Output:} $\shr{\vl{o}}$ where $\vl{o} = \vl{z}^{\vl{t}}$ if $\isTr = 1$ and $\vl{o} = \vl{z}$ if $\isTr = 0$ and $\vl{z} = \vl{abc}$.
	}
	\justify 
	\vspace{-2mm}
	\algoHead{Preprocessing:} 
	\begin{enumerate} 
		\item Computation for $\gm{\vl{ab}}{}$: Invoke $\prot{\MultS}$~(\boxref{fig:piMultS}) on $\sgr{\padR{\vl{a}}}$ and $\sgr{\padR{\vl{b}}}$ to generate $\sgr{\gm{\vl{ab}}{}}$.
		\item Computation for $\gm{\vl{ac}}{}$: 
		\begin{myitemize}
			\item[--] Parties locally compute the following:
			\begin{align*}
				P_0, P_1: \gm{\vl{ac}}{1}  &= \pad{\vl{a}}{1} \pad{\vl{c}}{3} + \pad{\vl{a}}{3} \pad{\vl{c}}{1} + \pad{\vl{a}}{3} \pad{\vl{c}}{3}\\
				P_0, P_2: \gm{\vl{ac}}{2} &= \pad{\vl{a}}{2} \pad{\vl{c}}{3} + \pad{\vl{a}}{3} \pad{\vl{c}}{2} + \pad{\vl{a}}{2} \pad{\vl{c}}{2}\\
				P_0, P_3: \gm{\vl{ac}}{3} &= \pad{\vl{a}}{1} \pad{\vl{c}}{2} + \pad{\vl{a}}{2} \pad{\vl{c}}{1} + \pad{\vl{a}}{1} \pad{\vl{c}}{1}
			\end{align*}
			\item[--] $P_0, P_3$ and $P_1$ sample random ${\vl{u}}_{\vl{ac}}^1 \in_R \Z{\ell}$. $P_0, P_3$ compute and $\jsend$ ${\vl{u}}_{\vl{ac}}^2 = \gm{\vl{ac}}{3} - {\vl{u}}_{\vl{ac}}^1$ to $P_2$.  
			\item[--]  $P_0, P_1, P_2$ sample random ${\vl{s}}_{\vl{ac}} \in_R \Z{\ell}$. $P_0$ sends ${\vl{w}}_{\vl{ac}} = \gm{\vl{ac}}{1} + \gm{\vl{ac}}{2} + {\vl{s}}_{\vl{ac}}$ to $P_3$.
		\end{myitemize}
		\smallskip
		\item Computation for $\gm{\vl{bc}}{}$: Similar to Step 2 (for $\gm{\vl{ac}}{}$). $P_1, P_2$ obtain ${\vl{u}}_{\vl{bc}}^1, {\vl{u}}_{\vl{bc}}^2$ respectively such that ${\vl{u}}_{\vl{bc}}^1 + {\vl{u}}_{\vl{bc}}^2 = \gm{\vl{bc}}{3}$ . $P_3$ obtains ${\vl{w}}_{\vl{bc}}  = \gm{\vl{bc}}{1} + \gm{\vl{bc}}{2} + {\vl{s}}_{\vl{bc}}$. 
		\item Computation for $\gm{\vl{abc}}{}$: 
		\begin{myitemize}
			\item[--] Using $\gm{\vl{ab}}{}$~(Step 1), $\pad{\vl{c}}{}$, compute the following:
			\begin{align*}
				P_0, P_1: \gm{\vl{abc}}{1}  &= \gm{\vl{ab}}{1} \pad{\vl{c}}{3} + \gm{\vl{ab}}{3} \pad{\vl{c}}{1} + \gm{\vl{ab}}{3} \pad{\vl{c}}{3}\\
				P_0, P_2: \gm{\vl{abc}}{2} &= \gm{\vl{ab}}{2} \pad{\vl{c}}{3} + \gm{\vl{ab}}{3} \pad{\vl{c}}{2} + \gm{\vl{ab}}{2} \pad{\vl{c}}{2}\\
				P_0, P_3: \gm{\vl{abc}}{3} &= \gm{\vl{ab}}{1} \pad{\vl{c}}{2} + \gm{\vl{ab}}{2} \pad{\vl{c}}{1} + \gm{\vl{ab}}{1} \pad{\vl{c}}{1}
			\end{align*}
			\item[--] $P_0, P_3$ and $P_j$ sample random ${\vl{u}}_{\vl{abc}}^j \in_R \Z{\ell}$ for $j \in \{1,2\}$. Let ${\vl{u}}_{\vl{abc}}^1 + {\vl{u}}_{\vl{abc}}^2 = \gm{\vl{abc}}{3} + \vl{r}$ for $\vl{r} \in_R \Z{\ell}$.   
			\item[--]  $P_0, P_1, P_2$ sample random ${\vl{s}} \in_R \Z{\ell}$. $P_0$ sends ${\vl{w}}_{\vl{abc}} = \gm{\vl{abc}}{1} + \gm{\vl{abc}}{2} + {\vl{s}}$ to $P_3$.
		\end{myitemize}
		\smallskip
		\item $P_0, P_3$ compute $\vl{r} = {\vl{u}}_{\vl{abc}}^1 + {\vl{u}}_{\vl{abc}}^2 - \gm{\vl{abc}}{3}$ and set $\vl{q} = \vl{r}^{\vl{t}}$  if $\isTr = 1$, else set $\vl{q} = \vl{r}$. Execute $\prot{\JSh}(P_0, P_3, \vl{q})$ to generate $\shr{\vl{q}}$.
	\end{enumerate}
	\justify
	\vspace{-2mm}
	\algoHead{Online:} Let $\vl{y} = (\vl{z} - \vl{r}) - \mk{\vl{abc}}$.
	\begin{enumerate} 
		\item Parties locally compute the following:
		\begin{align*}
			P_1: \vl{y}_1 &= - \pad{\vl{a}}{1} \mk{\vl{bc}} - \pad{\vl{b}}{1} \mk{\vl{ac}} - \pad{\vl{c}}{1} \mk{\vl{ab}} 
			+ \gm{\vl{ab}}{1} \mk{\vl{c}} + (\gm{\vl{ac}}{1} + {\vl{u}}_{\vl{ac}}^1) \mk{\vl{b}} + (\gm{\vl{bc}}{1} + {\vl{u}}_{\vl{bc}}^1) \mk{\vl{a}} - (\gm{\vl{a}\vl{bc}}{1} + {\vl{u}}_{\vl{a}\vl{bc}}^1)\\
			P_2: \vl{y}_2 &= - \pad{\vl{a}}{2} \mk{\vl{bc}} - \pad{\vl{b}}{2} \mk{\vl{ac}} - \pad{\vl{c}}{2} \mk{\vl{ab}} 
			+ \gm{\vl{ab}}{2} \mk{\vl{c}} + (\gm{\vl{ac}}{2} + {\vl{u}}_{\vl{ac}}^2) \mk{\vl{b}} - (\gm{\vl{bc}}{2} + {\vl{u}}_{\vl{bc}}^2) \mk{\vl{a}} - (\gm{\vl{a}\vl{bc}}{2} + {\vl{u}}_{\vl{a}\vl{bc}}^2)\\
			P_1, P_2: \vl{y}_3 &= - \pad{\vl{a}}{3} \mk{\vl{bc}} - \pad{\vl{b}}{3} \mk{\vl{ac}} - \pad{\vl{c}}{3} \mk{\vl{ab}} + \gm{\vl{ab}}{3} \mk{\vl{c}}
		\end{align*}
		\item $P_1$ sends $\vl{y}_2$ to $P_2$, while $P_2$ sends $\vl{y}_1$ to $P_1$, and they locally compute $\vl{z} - \vl{r} = (\vl{y}_1 + \vl{y}_2 + \vl{y}_3) + \mk{\vl{abc}}$.
		\item If $\isTr = 1$, $P_1, P_2$ locally set $\vl{p} = (\vl{z} - \vl{r})^{\vl{t}}$, else $\vl{p} = \vl{z} - \vl{r}$. Execute $\prot{\JSh}(P_1, P_2, \vl{p})$ to generate $\shr{\vl{p}}$. 
		\item Parties locally compute $\shr{\vl{o}} = \shr{\vl{p}} + \shr{\vl{q}}$. Here $\vl{o} = \vl{z}^{\vl{t}}$ if $\isTr = 1$ and $\vl{z}$ otherwise.
		\item {\em Verification:} 
		\begin{myitemize}
			\item[--] Parties locally compute the following:
			\begin{align*}
				P_3: \vl{v} &= - (\pad{\vl{a}}{1} + \pad{\vl{a}}{2})\mk{\vl{bc}} - (\pad{\vl{b}}{1} + \pad{\vl{b}}{2} )\mk{\vl{ac}} - (\pad{\vl{c}}{1} + \pad{\vl{c}}{2}) \mk{\vl{ab}} +(\gm{\vl{ab}}{1} + \gm{\vl{ab}}{2})\mk{\vl{c}} + ({\vl{w}}_{\vl{ac}} + \gm{\vl{ac}}{3}) \mk{\vl{b}} \\
				&~~~+ ({\vl{w}}_{\vl{bc}} + \gm{\vl{bc}}{3}) \mk{\vl{a}} - ({\vl{w}}_{\vl{abc}} + \gm{\vl{abc}}{3} + \vl{r})\\
				P_1, P_2: \vl{v}' &=  \vl{y}_1 + \vl{y}_2 + {\vl{s}}_{\vl{ac}} \mk{\vl{b}} + {\vl{s}}_{\vl{bc}} \mk{\vl{a}} - {\vl{s}}
			\end{align*}
			\item[--] $P_3$ sends $\Hash(\vl{v})$ to $P_1, P_2$, who $\abort$ iff $\Hash(\vl{v}) \ne \Hash(\vl{v}')$.
		\end{myitemize}
	\end{enumerate}     
\end{protocolsplitbox}

\begin{lemma}[Communication]
	\label{lemma:pimultTf}
	Protocol $\prot{\MultT}$~(\boxref{fig:piMultT})~(in $\Fthis$) requires $9 \ell$ bits of communication in preprocessing, and $1$ round and $3 \ell$ bits of communication in the online phase.
\end{lemma}
\begin{proof}
	In the preprocessing, computation of $\gm{\vl{ab}}{}$ involves three instances of $\jsend$. Each of the computation of $\gm{\vl{ac}}{}, \gm{\vl{bc}}{}$ involves one instance of $\jsend$ and a communication from $P_0$ to $P_3$. The computation of $\gm{\vl{abc}}{}$ is similar to the preprocessing of fair multiplication protocol~(\boxref{fig:piMultiplication}). The communication pattern of the online phase is similar to that of the fair multiplication protocol. The costs follow from Lemma~\ref{lemma:piMultF} and Lemma~\ref{lemma:pijsend}.
\end{proof}

Analogously, $\piMultRT$ can be extended to support 3-input multiplication while costing $12$ elements communication in preprocessing. The protocol appears in \boxref{fig:piMultTrobust}. 

\begin{protocolbox}{$\piMultRT(\vl{a}, \vl{b}, \vl{c}, \isTr)$}{3-input robust multiplication in $\FthisB$.}{fig:piMultTrobust}
	$\isTr$ is a bit denoting whether truncation is required ($\isTr =1$) or not ($\isTr=0$). \\
	\detail{
		{\bf Input(s):} $\shr{\vl{a}}, \shr{\vl{b}}, \shr{\vl{c}}$.\\
		{\bf Output:} $\shr{\vl{o}}$ where $\vl{o} = \vl{z}^{\vl{t}}$ if $\isTr = 1$ and $\vl{o} = \vl{z}$ if $\isTr = 0$ and $\vl{z} = \vl{abc}$.
	}
	\justify 
	\vspace{-2mm}
	\algoHead{Preprocessing:} 
	\begin{enumerate} 
		\item Computation for $\gm{\vl{ab}}{}$: Invoke $\prot{\MultS}$~(\boxref{fig:piMultS}) on $\sgr{\padR{\vl{a}}}$ and $\sgr{\padR{\vl{b}}}$ to generate $\sgr{\gm{\vl{ab}}{}}$.
		\item Computation for $\gm{\vl{ac}}{}, \gm{\vl{bc}}{}$: Similar to Step 1~(for $\gm{\vl{ab}}{}$).
		\item Computation for $\gm{\vl{abc}}{}$: 
		\begin{myitemize}
			\item[--] Using $\gm{\vl{ab}}{}$~(Step 1), $\pad{\vl{c}}{}$, compute the following:
			\begin{align*}
				P_0, P_1: \gm{\vl{abc}}{1}  &= \gm{\vl{ab}}{1} \pad{\vl{c}}{3} + \gm{\vl{ab}}{3} \pad{\vl{c}}{1} + \gm{\vl{ab}}{3} \pad{\vl{c}}{3}\\
				P_0, P_2: \gm{\vl{abc}}{2} &= \gm{\vl{ab}}{2} \pad{\vl{c}}{3} + \gm{\vl{ab}}{3} \pad{\vl{c}}{2} + \gm{\vl{ab}}{2} \pad{\vl{c}}{2}\\
				P_0, P_3: \gm{\vl{abc}}{3} &= \gm{\vl{ab}}{1} \pad{\vl{c}}{2} + \gm{\vl{ab}}{2} \pad{\vl{c}}{1} + \gm{\vl{ab}}{1} \pad{\vl{c}}{1}
			\end{align*}
			\item[--] $P_0, P_3$ and $P_j$ sample random ${\vl{u}}_{\vl{abc}}^j \in_R \Z{\ell}$ for $j \in \{1,2\}$. Let ${\vl{u}}_{\vl{abc}}^1 + {\vl{u}}_{\vl{abc}}^2 = \gm{\vl{abc}}{3} + \vl{r}$ for $\vl{r} \in_R \Z{\ell}$.   
			\item[--]  $P_0, P_1, P_2$ sample random ${\vl{s}}_1, {\vl{s}}_2 \in_R \Z{\ell}$. $P_0, P_j$ $\jsend$ $\vl{w}^j = \gm{\vl{abc}}{j} + {\vl{s}}_j$ to $P_3$  for $j \in \{1,2\}$.
		\end{myitemize}
		\smallskip
		\item $P_0, P_3$ compute $\vl{r} = {\vl{u}}_{\vl{abc}}^1 + {\vl{u}}_{\vl{abc}}^2 - \gm{\vl{abc}}{3}$ and set $\vl{q} = \vl{r}^{\vl{t}}$  if $\isTr = 1$, else set $\vl{q} = \vl{r}$. Execute $\prot{\JSh}(P_0, P_3, \vl{q})$ to generate $\shr{\vl{q}}$.
	\end{enumerate}
	\justify
	\vspace{-2mm}
	\algoHead{Online:} Let $\vl{y} = (\vl{z} - \vl{r}) - \mk{\vl{abc}} - \vl{s}_1 - \vl{s}_2$.
	\begin{enumerate} 
		\item Parties locally compute the following:
		\begin{align*}
			P_0, P_1: \vl{y}_1 &= - \pad{\vl{a}}{1} \mk{\vl{bc}} - \pad{\vl{b}}{1} \mk{\vl{ac}} - \pad{\vl{c}}{1} \mk{\vl{ab}} 
			+ \gm{\vl{ab}}{1} \mk{\vl{c}} + \gm{\vl{ac}}{1} \mk{\vl{b}} + \gm{\vl{bc}}{1} \mk{\vl{a}} - ({\vl{u}}_{\vl{abc}}^1 + \vl{w}^1)\\
			P_0, P_2: \vl{y}_2 &= - \pad{\vl{a}}{2} \mk{\vl{bc}} - \pad{\vl{b}}{2} \mk{\vl{ac}} - \pad{\vl{c}}{2} \mk{\vl{ab}} 
			+ \gm{\vl{ab}}{2} \mk{\vl{c}} + \gm{\vl{ac}}{2} \mk{\vl{b}} + \gm{\vl{bc}}{2} \mk{\vl{a}}  - ({\vl{u}}_{\vl{abc}}^2 + \vl{w}^2)\\
			P_1, P_2: \vl{y}_3 &= - \pad{\vl{a}}{3} \mk{\vl{bc}} - \pad{\vl{b}}{3} \mk{\vl{ac}} - \pad{\vl{c}}{3} \mk{\vl{ab}} + \gm{\vl{ab}}{3} \mk{\vl{c}} + \gm{\vl{ac}}{3} \mk{\vl{b}} + \gm{\vl{bc}}{3} \mk{\vl{a}}
		\end{align*}
		\item $P_1, P_3$ $\jsend$ $\vl{y}_1$ to $P_2$, while $P_2, P_3$ $\jsend$ $\vl{y}_2$ to $P_1$. $P_1, P_2$ locally compute $\vl{z} - \vl{r} = (\vl{y}_1 + \vl{y}_2 + \vl{y}_3) + \mk{\vl{abc}} + \vl{s}_1 + \vl{s}_2$. 
		\item If $\isTr = 1$, $P_1, P_2$ set $\vl{p} = (\vl{z} - \vl{r})^{\vl{t}}$, else $\vl{p} = \vl{z} - \vl{r}$. Execute $\prot{\JSh}(P_1, P_2, \vl{p})$ to generate $\shr{\vl{p}}$. 
		\item Parties locally compute $\shr{\vl{o}} = \shr{\vl{p}} + \shr{\vl{q}}$. Here $\vl{o} = \vl{z}^{\vl{t}}$ if $\isTr = 1$ and $\vl{z}$ otherwise.
	\end{enumerate}     
\end{protocolbox}

\begin{lemma}[Communication]
	\label{lemma:pimultRTf}
	Protocol $\piMultRT$~(\boxref{fig:piMultTrobust})~(in $\FthisB$) requires $12 \ell$ bits of communication in preprocessing, and $1$ round and $3 \ell$ bits of communication in the online phase.
\end{lemma}
\begin{proof}
	In the preprocessing, computation of each of $\gm{\vl{ab}}{}, \gm{\vl{ac}}{}, \gm{\vl{bc}}{}$ involves three instances of $\jsend$. The computation of $\gm{\vl{abc}}{}$ is similar to the preprocessing of robust multiplication protocol~(\boxref{fig:piMultGOD}). The communication pattern of the online phase is similar to that of the robust multiplication protocol. The costs follow from Lemma~\ref{lemma:piMultR} and Lemma~\ref{lemma:pijsend}.
\end{proof}

To obtain $\shr{\cdot}$-sharing of $\vl{z}= \vl{a} \vl{b} \vl{c} \vl{d}$ given the $\shr{\cdot}$-sharing of $\vl{a}, \vl{b}, \vl{c}, \vl{d}$, we can write $\vl{z} - \vl{r}$ as 
\begin{align}\label{eq:4pcMmultF}
	\vl{z} - \vl{r} &= \vl{abcd} - \vl{r} = (\mk{\vl{a}} - \pd{\vl{a}})(\mk{\vl{b}} - \pd{\vl{b}})(\mk{\vl{c}} - \pd{\vl{c}})(\mk{\vl{d}} - \pd{\vl{d}}) - \vl{r} \nonumber\\ 
	&= \mk{\vl{abcd}} - \mk{\vl{abc}}\pd{\vl{d}} - \mk{\vl{abd}}\pd{\vl{c}} - \mk{\vl{acd}}\pd{\vl{b}} - \mk{\vl{bcd}}\pd{\vl{a}} 
	+ \mk{\vl{ab}}\gm{\vl{cd}}{} + \mk{\vl{ac}}\gm{\vl{bd}}{} + \mk{\vl{ad}}\gm{\vl{bc}}{} + \mk{\vl{bc}}\gm{\vl{ad}}{} \nonumber\\
	&~~~+ \mk{\vl{bd}}\gm{\vl{ac}}{} + \mk{\vl{cd}}\gm{\vl{ab}}{} 
	- \mk{\vl{a}}\gm{\vl{bcd}}{} - \mk{\vl{b}}\gm{\vl{acd}}{} - \mk{\vl{c}}\gm{\vl{abd}}{} - \mk{\vl{d}}\gm{\vl{abc}}{} + \gm{\vl{abcd}}{} - \vl{r}
	~~\text{\footnotesize{(cf. notation~\ref{notation:4pcconcise})}}
\end{align}
While the online phase proceeds similarly to the 2 and 3-input multiplication, in the preprocessing phase, the parties need to generate the additive shares of $\gm{\vl{ab}}{}$,$\gm{\vl{ac}}{}$,$\gm{\vl{ad}}{}$,$\gm{\vl{bc}}{}$,$\gm{\vl{bd}}{}$,$\gm{\vl{cd}}{}$,$\gm{\vl{abc}}{}$,$\gm{\vl{abd}}{}$,$\gm{\vl{acd}}{}$,$\gm{\vl{bcd}}{}$ and $\gm{\vl{abcd}}{} - \vl{r}$. This is computed similarly as in the case of 3-input multiplication as follows. 
Parties generate shares of $\gm{\vl{ac}}{}, \gm{\vl{ad}}{}, \gm{\vl{bc}}{}, \gm{\vl{bd}}{}$ similar to the generation of shares of $\gm{\vl{ac}}{}$ in the 3-input multiplication. For $\gm{\vl{ab}}{}, \gm{\vl{cd}}{}$, parties proceed similar to generation of shares of $\gm{\vl{ab}}{}$ in the 3-input multiplication, where the respective $\sgr{\cdot}$-shares are generated. This is followed by generation of shares of $\gm{\vl{abc}}{}, \gm{\vl{abd}}{}, \gm{\vl{acd}}{}, \gm{\vl{bcd}}{}, \gm{\vl{abcd}}{}$ following steps similar to the ones involved in generating $\gm{\vl{abcc}}{}$ in the 3-input multiplication. Since the protocol is very similar to the 3-input protocol, we omit the formal details.

\subsection{Supporting on-demand computations}
\label{p:nopre}
For on-demand applications where the underlying function to be computed is not known in advance, the preprocessing model is not desirable. We observe that the $\Fthis$ protocol can be modified by executing the preprocessing phase in the online phase itself, keeping the same overall communication cost. The formal protocol appears in \boxref{fig:piMultNoPre}. 

We provide the fair multiplication, $\piMultO$, for {\em function-independent} preprocessing in \boxref{fig:piMultNoPre}. The protocol incurs no overhead over the fair multiplication~($\prot{\Mult}$) in $\Fthis$. This is due to the design of $\prot{\Mult}$ where values ${\vl{u}}^1, {\vl{u}}^2$ are sampled non-interactively in the preprocessing. Thus the joint-sharing by $P_0, P_3$~(Step 5 (a) in \boxref{fig:piMultNoPre}) can be performed along with the communication among $P_1, P_2$~(Step 4 in \boxref{fig:piMultNoPre}) in the online. Moreover, the rest of the communication can be deferred till the verification stage and thus, the online round complexity is retained. The protocol for robust setting is similar.

\smallskip
\begin{protocolsplitbox}{$\piMultO(\vl{a}, \vl{b}, \isTr)$}{Fair multiplication without preprocessing in $\Fthis$.}{fig:piMultNoPre}
	$\isTr$ is a bit denoting whether truncation is required ($\isTr =1$) or not ($\isTr=0$). \\
	\detail{
		{\bf Input(s):} $\shr{\vl{a}}, \shr{\vl{b}}$.\\
		{\bf Output:} $\shr{\vl{o}}$ where $\vl{o} = \vl{z}^{\vl{t}}$ if $\isTr = 1$ and $\vl{o} = \vl{z}$ if $\isTr = 0$ and $\vl{z} = \vl{ab}$.
	}
	\justify 
	\vspace{-2mm}
	\algoHead{Online:} 
	\begin{enumerate} 
		\item Parties locally compute the following:
		\begin{align*}
			P_0, P_1: \gm{\vl{a}\vl{b}}{1} &= \pad{\vl{a}}{1} \pad{\vl{b}}{3} + \pad{\vl{a}}{3} \pad{\vl{b}}{1} + \pad{\vl{a}}{3} \pad{\vl{b}}{3} \\
			P_0, P_2: \gm{\vl{a}\vl{b}}{2} &= \pad{\vl{a}}{2} \pad{\vl{b}}{3} + \pad{\vl{a}}{3} \pad{\vl{b}}{2} + \pad{\vl{a}}{2} \pad{\vl{b}}{2} \\
			P_0, P_3: \gm{\vl{a}\vl{b}}{3} &= \pad{\vl{a}}{1} \pad{\vl{b}}{2} + \pad{\vl{a}}{2} \pad{\vl{b}}{1} + \pad{\vl{a}}{1} \pad{\vl{b}}{1}
		\end{align*}
		\item $P_0, P_3$ and $P_j$ sample random ${\vl{u}}^j \in_R \Z{\ell}$ for $j \in \{1,2\}$. Let ${\vl{u}^1} + \vl{u}^2 = \gm{\vl{a}\vl{b}}{3} - \vl{r}$ for a random $\vl{r} \in_R \Z{\ell}$.  
		\item Let $\vl{y} = (\vl{z} - \vl{r}) - \mk{\vl{a}} \mk{\vl{b}}$. Parties locally compute the following:
		\begin{align*}
			P_1: \vl{y}_1 &= - \pad{\vl{a}}{1} \mk{\vl{b}} - \pad{\vl{b}}{1} \mk{\vl{a}} + \gm{\vl{a}\vl{b}}{1} + {\vl{u}}^1 \\
			P_2: \vl{y}_2 &= - \pad{\vl{a}}{2} \mk{\vl{b}} - \pad{\vl{b}}{2} \mk{\vl{a}} + \gm{\vl{a}\vl{b}}{2} + {\vl{u}}^2 \\
			P_1, P_2: \vl{y}_3 &= - \pad{\vl{a}}{3} \mk{\vl{b}} - \pad{\vl{b}}{3} \mk{\vl{a}}
		\end{align*}
		\item $P_1$ sends $\vl{y}_1$ to $P_2$, while $P_2$ sends $\vl{y}_2$ to $P_1$.
		\item Parties proceed as follows:
		\begin{enumerate} 
			\item $P_0, P_3$: $\vl{r} = \gm{\vl{a}\vl{b}}{3} - {\vl{u}^1} - \vl{u}^2$; $\vl{q} = \vl{r}^{\vl{t}}$  if $\isTr = 1$, else $\vl{q} = \vl{r}$; Execute $\prot{\JSh}(P_0, P_3, \vl{q})$.
			\item $P_1, P_2$: $\vl{z} - \vl{r} = (\vl{y}_1 + \vl{y}_2 + \vl{y}_3) + \mk{\vl{a}} \mk{\vl{b}}$; $\vl{p} = (\vl{z} - \vl{r})^{\vl{t}}$ if $\isTr = 1$, else $\vl{p} = \vl{z} - \vl{r}$; Execute $\prot{\JSh}(P_1, P_2, \vl{p})$.
		\end{enumerate}
		\item Parties locally compute $\shr{\vl{o}} = \shr{\vl{p}} + \shr{\vl{q}}$. Here $\vl{o} = \vl{z}^{\vl{t}}$ if $\isTr = 1$ and $\vl{z}$ otherwise.
	\end{enumerate}
	\justify
	\vspace{-2mm}
	\algoHead{Verification:}
	\begin{enumerate} 
		\item  $P_0, P_1, P_2$ sample random ${\vl{s}} \in_R \Z{\ell}$. $P_0$ sends $\vl{w} = \gm{\vl{a}\vl{b}}{1} + \gm{\vl{a}\vl{b}}{2} + {\vl{s}}$ to $P_3$.
		\item $P_3$ computes $\vl{v} = - (\pad{\vl{a}}{1} + \pad{\vl{a}}{2}) \mk{\vl{b}} - (\pad{\vl{b}}{1} + \pad{\vl{b}}{2}) \mk{\vl{a}} + {\vl{u}^1} + \vl{u}^2 + \vl{w}$ and sends $\Hash(\vl{v})$ to $P_1$ and $P_2$. Parties $P_1, P_2$ $\abort$ iff $\Hash(\vl{v}) \ne \Hash(\vl{y}_1 + \vl{y}_2 + {\vl{s}})$.
	\end{enumerate}     
\end{protocolsplitbox}

\subsubsection{Comparison with  Fantastic Four~\cite{EPRINT:DalEscKel20}} 
\label{pa:fantasticfour}
We analyse the performance of Fantastic Four~\cite{EPRINT:DalEscKel20} where execution proceeds in segments~(cf. \S6.4,~\cite{EPRINT:DalEscKel20}). Elaborately, computation is carried out optimistically for each segment, followed by a verification phase before proceeding to the next segment. If verification fails, the current segment is recomputed via an active 3PC protocol. Subsequent segments also proceed with a 3PC execution until the verification fails again. In this case, a semi-honest 2PC with a helper is carried out for the current and rest of the segments. For analysis, we consider their best and worst-case execution cost. 

\begin{table}[htb!]
	\centering
		\begin{NiceTabular}{r r r c}
			\toprule
			\Block{2-1}{Work} & \multicolumn{2}{c}{Dot Product w/ Truncation} 
			& \Block{2-1}{\#Active\\Parties}\\ \cmidrule{2-3}
			& Preprocessing & Online & \\
			\midrule
			Fantastic Four: Case I    & $\ell$  & $9\ell$  & 4 \\
			Fantastic Four: Case II   & $76(\ell+\kappa)+54x + 12$  & $9\ell + 6\kappa$  & 3 \\
			\FthisA (on-demand)       & -        & $5\ell$  & 3 \\     
			\FthisB (on-demand)       & -        & $6\ell$  & 3 \\                  
			\bottomrule
		\end{NiceTabular}
	\caption{Comparison with Fantastic Four~\cite{EPRINT:DalEscKel20}}\label{tab:compFour}
\end{table}

Observe that the best case happens when the verification is always successful, which we call as {\em Case I}. In this case, the communication cost is that of the 4PC execution. Note that an adversary can {\em always} make the verification fail in the first segment itself. This results in executing the entire protocol (all segments) with their active 3PC, which accounts for their worst-case cost. We denote this as {\em Case II}. Their 3PC protocols are designed to work over the extended ring of size $\ell + \kappa$ bits. As evident from Tables 2, 3 of their paper, their 3PC is at least $10 \times$ more expensive than their 4PC in terms of both runtime and communication. Thus, the higher cost of 3PC defeats the purpose of having an additional honest party in the system. 

Observe that their protocols are designed to work with a function-independent preprocessing. Thus, for a fair comparison, we compare both cases against the on-demand variants of our robust protocols~(\FthisA, \FthisB). The results are summarised in~\tabref{compFour}. We remark that the values for their cases are obtained from Table 1 of their paper~\cite{EPRINT:DalEscKel20}.

\section{Garbled World}
\label{sec:4pcGCWorld}

In the applications we consider, the garbled circuit is used as an intermediary to evaluate certain functions where the input to the function as well as the output are in $\shr{\cdot}$-shared (or $\shrB{\cdot}$-shared) form.

Instantiating the garbled world using existing 4PC GC-based protocols~\cite{C:IKKP15, CCS:BJPR18} turn out to be overkill, as they are standalone protocols. For instance,~\cite{C:IKKP15} provides robust protocols by communicating $12$ GCs while~\cite{CCS:BJPR18} requires generating and exchanging commitments on the inputs to ensure input consistency. On the other hand, the inputs to our protocol are consistent~(due to $\shr{\cdot}$-sharing), and we do not need an explicit reconstruction, making it lighter overall. 

Towards this, we propose 2 GC protocols -- $\FthisT$ requiring communication of 2 GC evaluations and 1 online round, and $\FthisC$ requiring 1 GC and 2 rounds. Moreover, these protocols leverage the benefit of amortization which comes from using $\jsend$. 
The 2 GC variant has two parallel executions, each comprising of 3 garblers and 1 evaluator. $P_1, P_2$ act as evaluators in two independent executions and the parties in $\PlSet{1} = \{P_0, P_2, P_3\}$, $\PlSet{2} = \{P_0, P_1 ,P_3\}$ act as garblers, respectively. The 1 GC variant comprises of a single execution with $\PlSet{1}$ acting as garblers and $P_1$ as the evaluator. 
Leveraging an honest majority among the garblers and using $\jsend$, we only need semi-honest GC computation to get active security.

\subsection{2 GC Variant}
\label{sec:4pcgcworld1}

\myparagraph{Input Phase}
\label{p:gbip}
Given that the function  input $\vl{x}$ is already available as $\shrB{\vl{x}}$, the boolean values $\mk{\vl{x}}, \av{\vl{x}}, \pad{\vl{x}}{3}$, where $\av{\vl{x}} = \pad{\vl{x}}{1} \xor \pad{\vl{x}}{2}$ and $\vl{x} = \mk{\vl{x}} \xor \av{\vl{x}} \xor \pad{\vl{x}}{3}$, act as the {\em new} inputs for the garbled computation, and garbled sharing ($\shrG{\cdot}$) is generated for each of these values. The semantics of $\shrB{\cdot}$-sharing ensures that each of these shares ($\mk{\vl{x}}, \av{\vl{x}}, \pad{\vl{x}}{3}$) is available with two garblers in each garbling instance. The keys for the shares can either be sent (using $\jsend$) correctly to the evaluators or the inconsistency is detected. This key delivery essentially generates  $\shrG{\cdot}$-sharing for each of these three values  which enables GC evaluation. Thus, the goal of our input phase is to create the compound sharing, $\shrC{\vl{x}} = (\shrG{\mk{\vl{x}}}, \shrG{\av{\vl{x}}}, \shrG{\pad{\vl{x}}{3}})$ for every input $\vl{x}$ to the function to be evaluated via the GC. We first discuss the semantics for $\shrG{\cdot}$-sharing followed by steps for generating $\shrC{\cdot}$-sharing.

\myparagraph{Garbled sharing semantics}\label{gcsemantics}
A value $\vl{v} \in \Z{}$  is $\shrG{\cdot}$-shared (garbled shared) amongst  $\Partyset$ if $P_i \in \{P_0, P_3\}$ holds $\shrG{\vl{v}}_{i}= (\key{{\vl{v}}}{0,1}, \key{{\vl{v}}}{0,2})$, $P_1$ holds $\shrG{\vl{v}}_{1} = (\key{{\vl{v}}}{\vl{v},1}, \key{{\vl{v}}}{0,2})$ and $P_2$ holds $\shrG{\vl{v}}_{2} = (\key{{\vl{v}}}{0,1}, \key{{\vl{v}}}{\vl{v},2})$. Here, $\key{{\vl{v}}}{\vl{v}, j} = \key{{\vl{v}}}{0, j} \xor \vl{v} \Delta^{j}$ for $j \in \{1, 2\}$, and $\Delta^{j}$, which is known only to the garblers in $\PlSet{j}$, denotes the global offset with its least significant bit set to $1$ and is same for every wire in the circuit. 
A value $\vl{x} \in \Z{}$ is said to be $\shrC{\cdot}$-shared (compound shared) if each value  from $(\mk{\vl{x}}, \av{\vl{x}}, \pad{\vl{x}}{3})$, which are as defined above, is $\shrG{\cdot}$-shared. We write $\shrC{\vl{x}} = (\shrG{\mk{\vl{x}}},\shrG{\av{\vl{x}}},\shrG{\pad{\vl{x}}{3}})$. 

\paragraph{Generation of $\shrG{\vl{v}}$ and $\shrC{\vl{x}}$} 
Protocol $\pigsh(\Partyset, \vl{v})$~(\boxref{fig:pigsh4pc}) enables generation of $\shrG{\vl{v}}$ where two garblers in each garbling instance hold $\vl{v}$, and proceeds as follows. Consider the first garbling instance with evaluator $P_1$ where garblers $P_k, P_l$ hold $\vl{v}$. Garblers in $\PlSet{1}$ generate $\{\key{{{\vl{v}}}}{\bitb, 1}\}_{\bitb \in \{0, 1\}}$ which denotes the key for value $\bitb$ on wire $\vl{v}$, following the free-XOR technique~\cite{ICALP:KolSch08,C:KolMohRos14}. 
$P_k, P_l$ $\jsend$ $\key{{{\vl{v}}}}{\vl{v}, 1}$ to evaluator $P_1$. Similar steps carried out with respect to the second garbling instance, at the end of which, garblers in $\PlSet{2}$ possess $\{\key{\vl{v}}{\bitb, 2}\}_{\bitb \in \{0,1\}}$ while the evaluator $P_2$ holds $\key{\vl{v}}{\vl{v}, 2}$. Following this, the shares $\shrG{\vl{v}}_s$ held by $P_s \in \Partyset$ are defined as $\shrG{\vl{v}}_0 = \shrG{\vl{v}}_3 = (\key{\vl{v}}{0, 1}, \key{\vl{v}}{0, 2})$, $\shrG{\vl{v}}_1 = (\key{\vl{v}}{\vl{v}, 1}, \key{\vl{v}}{0, 2})$, $\shrG{\vl{v}}_2 = (\key{\vl{v}}{0, 1}, \key{\vl{v}}{\vl{v}, 2})$. 

\begin{protocolbox}{$\pigsh(\Partyset, \vl{v})$}{Generation of $\shrG{\vl{v}}$}{fig:pigsh4pc}
	\justify
	\begin{enumerate} 
		\item Garblers in $\PlSet{j}$ for $j \in \{1, 2\}$ generate keys $\key{{\vl{v}}}{0, j}, \key{{\vl{v}}}{1, j}$ for wire $\vl{v}$, using free-XOR technique.
		\item Let $P_k^j, P_l^j$ denote the garblers in the $j^{\text{th}}$ garbling instance, for $j \in \{1, 2\}$, who hold $\vl{v} \in \Z{}$. $P_k^j, P_l^j$ $\jsend$ $\key{\vl{v}}{\vl{v}, j}$ to evaluator $P_j$. 
		\item $P_i \in \{P_0, P_3\}$ sets $\shrG{\vl{v}}_i = (\key{{\vl{v}}}{0,1}, \key{{\vl{v}}}{0,2})$, $P_1$ sets $\shrG{\vl{v}}_{1} = (\key{{\vl{v}}}{\vl{v},1}, \key{{\vl{v}}}{0,2})$ and $P_2$ sets $\shrG{\vl{v}}_{2} = (\key{{\vl{v}}}{0,1}, \key{{\vl{v}}}{\vl{v},2})$.
	\end{enumerate}
\end{protocolbox}

To generate $\shrC{\vl{x}}$, we need a way to generate  $(\shrG{\mk{\vl{x}}}, \shrG{\av{\vl{x}}}, \shrG{\pad{\vl{x}}{3}})$, given $\shrB{\vl{x}}$. For this, $\pigsh$ is invoked for each of $\mk{\vl{x}}, \av{\vl{x}}, \pad{\vl{x}}{3}$.  

\myparagraph{Evaluation} 
\label{subsec:gbeval} 
Let $f(\vl{x})$ be the function to be evaluated. At this point, the function input is $\shrC{\cdot}$-shared. This renders $\shrG{\cdot}$-sharing for the input of the GC that corresponds to the function $f'\big({\mk{\vl{x}}}, {\av{\vl{x}}}, {\pad{\vl{x}}{3}} \big)$ which first combines the given boolean-shares to compute the actual input and then applies $f$ on it. Let $\GC_j$ denote the garbled circuit to be sent to $P_j \in \{P_1, P_2\}$ by garblers in $\PlSet{j}$. Sending of $\GC_j$ is overlapped  with the key transfer (during generation of $\shrC{\vl{x}}$), to save rounds, where garblers in $\{P_0, P_3\}$ $\jsend$ $\GC_j$ to $P_j$. On receiving the $\GC$, evaluators evaluate their respective GCs and obtain the key corresponding to the output, say $\vl{z}$. This generates $\shrG{\vl{z}}$. 

\myparagraph{Output phase} 
\label{subsec:gbop} 
The goal of output computation is to compute the output $\vl{z}$ from $\shrG{\vl{z}}$.
To reconstruct $\vl{z}$ towards $P_j \in \{P_1, P_2\}$, two garblers in $\PlSet{j}$ send the least significant bit $\vl{p}^j$ of $\key{\vl{z}}{0, j}$, referred to as the decoding information, to $P_j$. If the received values are consistent, $P_j$ uses the received $\vl{p}^j$ to reconstruct $\vl{z}$ as $\vl{z} = \vl{p}^j \xor \vl{q}^j$, where $\vl{q}^j$ denotes the least significant bit of $\key{\vl{z}}{\vl{z}, j}$; else $P_j$ aborts. 
To reconstruct $\vl{z}$ towards the garblers $P_g \in \{P_0, P_3\}$, one evaluator, say $P_1$ sends the least significant bit, $\vl{q}^1$, of $\key{{\vl{z}}}{\vl{z}, 1}$ along with $\h = \Hash(\key{{\vl{z}}}{\vl{z}, 1})$ to $P_g$, where $\Hash$ is a collision-resistant hash function. If a garbler received a consistent $(\vl{q}^1, \h)$ pair from $P_1$ such that there exists a $K \in \{\key{{\vl{z}}}{{0, 1}}, \key{{\vl{z}}}{{1, 1}}\}$ whose least significant bit is $\vl{q}^1$ and $\Hash(K) = \h$, then it uses $\vl{q}^1$ for reconstructing $\vl{z}$; else the garbler aborts the computation.
Note that a corrupt evaluator $P_1$ cannot create confusion among garblers in $\{P_0, P_3\}$ by sending the key that was not output by the GC owing to the authenticity of the garbling scheme. Reconstruction is lightweight and requires a single round for garblers while reconstruction towards evaluators can be overlapped with key transfer and does not incur extra rounds.
The protocol appears in \boxref{fig:pirec}.

\begin{protocolbox}{$\pigrec(\Partyset, \shrG{\vl{z}})$}{Output computation: reconstruction of $\vl{z}$}{fig:pirec}
	\justify
	\begin{myitemize}
		\item[-] For an output wire $\vl{z}$, let $\vl{p}^j$ denote the least significant bit of $\key{{\vl{z}}}{0,j}$ and $\vl{q}^j$ denote the least significant bit of $\key{{\vl{z}}}{\vl{z},j}$for $j \in \{1, 2\}$.
		\item[-] {\em Reconstruction towards $P_j \in \{P_1, P_2\}$}: Garblers $P_0, P_3$ in $\PlSet{j}$ $\jsend$ $\vl{p}^j$ to $P_j$. If $P_j$ received consistent values from $P_0, P_3$, it reconstructs $\vl{z}$ as $\vl{z} = \vl{p}^j \xor \vl{q}^j$.
		\item[-] {\em Reconstruction towards $P_g \in \{P_0, P_3\}$}: $P_1$ sends $\vl{q}^1$ and $\h = \Hash(\key{{\vl{z}}}{\vl{z},1})$ to $P_g$, where $\Hash$ is a collision-resistant hash function. 	
		$P_g$ uses the $\vl{q}^1$ received from $P_1$ for reconstructing $\vl{z}$ as $\vl{z} = \vl{p}^1 \xor \vl{q}^1$ if there exists a $K \in \{\key{{\vl{z}}}{{0,1}}, \key{{\vl{z}}}{{1,1}}\}$ whose least significant bit is $\vl{q}^1$ and $\Hash(K) = \h$. 
	\end{myitemize}
\end{protocolbox}

\myparagraph{Optimizations when deployed in mixed framework}
\label{subsec:gbopt}
Working in the preprocessing model enables transfer of the (communication-intensive) GC and generating $\shrG{\cdot}$-shares of the input-independent shares of $\vl{x}$ (i.e. $ {\av{\vl{x}}}, {\pad{\vl{x}}{3}}$) in the preprocessing phase. Thus, the online phase is very light and only requires one round to generate $\shrG{\cdot}$-shares  for the input-dependent data (i.e. ${\mk{\vl{x}}}$). Since evaluation is local, evaluators obtain $\shrG{\cdot}$-sharing of the GC output at the end of $1$ round. 

\myparagraph{Achieving fairness and robustness}
To ensure fairness, we require a fair reconstruction protocol that proceeds as follows. As described in \S\ref{sec:recfair4pcM}, parties first ensure that all parties are alive. 
If so, they proceed similar to the protocol in \boxref{fig:pirec}, except with the following differences. For reconstruction towards evaluators, {\em all} three respective garblers send it the decoding information. The evaluator selects the value appearing in the majority for reconstruction. For reconstruction towards garblers $P_0, P_3$, {\em both} the evaluators send the least significant bit of the output key together with its hash to the garbler. The presence of at least one honest evaluator guarantees that both garblers will be on the same page. 
The protocol appears in \boxref{fig:pifrec}.

\begin{protocolbox}{$\pigfrec(\Partyset, \shrG{\vl{z}})$}{Fair output computation: fair reconstruction of $\vl{z}$}{fig:pifrec}
	\justify
	\begin{myitemize}
		\item[-] Parties perform a bit exchange as described in \S\ref{sec:recfair4pcM} to ensure that all parties are alive. If all parties are alive, they proceed as follows.
		\item[-] For an output wire $\vl{z}$, let $\vl{p}^j$ denote the least significant bit of $\key{{\vl{z}}}{0,j}$ and $\vl{q}^j$ denote the least significant bit of $\key{{\vl{z}}}{\vl{z},j}$for $j \in \{1, 2\}$.
		\item[-] {\em Reconstruction towards $P_j \in \{P_1, P_2\}$}: Garblers in $\PlSet{j}$ send $\vl{p}^j$ to $P_j$. $P_j$ selects the value forming majority among these and reconstructs $\vl{z}$ as $\vl{z} = \vl{p}^j \xor \vl{q}^j$.
		\item[-] {\em Reconstruction towards $P_g \in \{P_0, P_3\}$}: $P_j \in \{P_1, P_2\}$ sends $\vl{q}^j$ and $\h^j = \Hash(\key{{\vl{z}}}{\vl{z},j})$ to $P_g$, where $\Hash$ is a collision-resistant hash function. 	
		$P_g$ uses the $\vl{q}^1$ received from $P_1$ for reconstructing $\vl{z}$ as $\vl{z} = \vl{p}^1 \xor \vl{q}^1$ if there exists a $K \in \{\key{{\vl{z}}}{{0,1}}, \key{{\vl{z}}}{{1,1}}\}$ whose least significant bit is $\vl{q}^1$ and $\Hash(K) = \h^1$. Else, it computes $\vl{z} = \vl{p}^2 \xor \vl{q}^2$.
	\end{myitemize}
\end{protocolbox}

The main difference from its fair counterpart is the use of a robust $\jsend$ primitive to achieve robustness. This guarantees that a $\TTP$ is identified if misbehaviour is detected, taking the computation to completion and delivering the output to all. 

\subsection{1 GC Variant}
\label{sec:4pcgcworld2}
The input $\vl{x} = \vl{x}_1 \xor \vl{x}_2$ for this variant consists of two shares, $\vl{x}_1 = \mk{\vl{x}} \xor \pad{\vl{x}}{2}$ and $\vl{x}_2 = \pad{\vl{x}}{1} \xor \pad{\vl{x}}{3}$, where $\mk{\vl{x}}, \pad{\vl{x}}{1}, \pad{\vl{x}}{2}, \pad{\vl{x}}{3}$ are as defined in $\shrB{\vl{x}}$. To ensure correct key transfer for the value $\vl{x}_2$ held by garbler $P_0$ and evaluator $P_1$, garblers $P_0, P_3$ commit to both keys for $\vl{x}_2$ towards $P_1$, while $P_0$ sends the opening to the key for $\vl{x}_2$. Then, $P_1$ verifies the consistency of the received commitments and the opening, as it possesses $\vl{x}_2$. The protocol appears in \boxref{fig:pigsh1}. 

\begin{protocolbox}{$\pigsh(P_i, P_j, \vl{v})$}{Generation of $\shrG{\vl{v}}$}{fig:pigsh1}
	\justify
	\begin{enumerate} 
		\item Garblers in $\PlSet{1}$ generate keys $\key{{\vl{v}}}{0}, \key{{\vl{v}}}{1}$ using free-XOR technique.
		\item If $(P_i, P_j) = (P_2, P_3) $: $P_i, P_j$ $\jsend$ $\key{\vl{v}}{\vl{v}}$ to $P_1$. 
		\item If $(P_i, P_j) = (P_0, P_1)$:
		\begin{myitemize}
			\item[-] $P_0, P_3$ compute commitments on $\key{{\vl{v}}}{0}, \key{{\vl{v}}}{1}$, and $\jsend$ the commitment to $P_1$. 
			\item[-] $P_0$ sends the opening of the commitment for $\key{{\vl{v}}}{\vl{v}}$ to $P_1$. 
			\item[-] $P_1$ verifies if the received opening information correctly decommits the commitment on $\key{\vl{v}}{\vl{v}}$, where $\vl{v}$ is held by $P_1$. Else it $\aborts$. 
		\end{myitemize}	
		\item Party $P_s \in \PlSet{1}$ sets $\shrG{\vl{v}}_s = \key{\vl{v}}{0}$, while $P_1$ sets $\shrG{\vl{v}}_1 = \key{\vl{v}}{\vl{v}}$.
	\end{enumerate}
\end{protocolbox}

The evaluation and output phases are similar to the 2GC variant, except there is only a single garbling instance now. Looking ahead, in the mixed protocol framework, the output has to be reconstructed towards $P_1, P_2$. Reconstruction towards $P_1$ does not incur additional rounds since sending of decoding information can be overlapped with the key transfer. However, unlike in the 2GC variant where reconstruction towards $P_2$ can be done similar to reconstruction towards $P_1$, in the 1GC variant, an additional round is required as $P_2$ is no longer an evaluator. This incurs one extra round as opposed to the 2GC variant.

\myparagraph{Achieving fairness}
To ensure fair reconstruction, as in \S\ref{sec:recfair4pcM}, parties first perform an aliveness check. Following this, they proceed towards a fair reconstruction of $\vl{z}$ from $\shrG{\vl{z}}$ as follows. 
First, reconstruction of $\vl{z}$ is carried out towards the garblers $P_g \in \PlSet{1}$. For this, $P_1$ sends $\vl{q}$ (least significant bit of \key{{\vl{z}}}{\vl{z}}) and $\h = \Hash(\key{{\vl{z}}}{\vl{z}})$ to $P_g$ as before. Now, if a garbler received a consistent $(\vl{q}, \h)$ pair from $P_1$ such that there exists a $K \in \{\key{{\vl{z}}}{{0}}, \key{{\vl{z}}}{{1}}\}$ whose least significant bit is $\vl{q}$ and $\Hash(K) = \h$, then it uses $\vl{q}$ for reconstructing $\vl{z}$, and sends $\vl{z}$ to its co-garblers. Else, a garbler accepts a $\vl{z}$ received from a co-garbler as the output. Thus, further dissemination of the output by garblers ensures that all parties are on the same page. 
If garblers receive the output, reconstruction of $\vl{z}$ is carried out towards $P_1$. For this, all garblers (who received the output) send the decoding information to $P_1$, who selects the majority value to reconstruct $\vl{z}$. 

\begin{protocolbox}{$\pigfrec(\shrG{\vl{z}})$}{Fair reconstruction of $\vl{z}$ from $\shrG{\vl{z}}$}{fig:pifrec1}
	\justify
	\begin{myitemize}
		\item[-] Parties perform a bit exchange as described in \S\ref{sec:recfair4pcM} to ensure that all parties are alive. If all parties are alive, they proceed as follows.
		\item[-] Let $\vl{p}, \vl{q}$ denote the least significant bit of $\key{\vl{z}}{0}, \key{\vl{z}}{\vl{z}}$, respectively.
	\end{myitemize}
	\begin{myitemize}
		\item[-] {\em Reconstruction towards garblers $P_g \in \PlSet{1}$}: $P_1$ sends $\vl{q}$ and $\h = \Hash(\key{{\vl{z}}}{\vl{z}})$ to $P_g \in \PlSet{1}$, where $\Hash$ is a collision-resistant hash function. $P_g$ does the following to reconstruct $\vl{z}$.
		\begin{itemize}
			\item[-] If $P_g$ received $(\vl{q}, \h)$ from $P_1$ such that there exists a $K \in \{\key{{\vl{z}}}{{0}}, \key{{\vl{z}}}{{1}}\}$ whose least significant bit is $\vl{q}$ and $\Hash(K) = \h$, set $\vl{z} = \vl{p} \xor \vl{q}$. $P_g$ sends $\vl{z}$ to its co-garblers. 
			\item[-] Else, if $P_g$ did not receive a consistent  $(\vl{q}, \h)$-pair from $P_1$ but received a $\vl{z}$ from its co-garbler in the following round, then accept $\vl{z}$ as the output.
		\end{itemize}
		\item[-] {\em Reconstruction towards $P_1$}: If garblers obtained the output, then they send $\vl{p}$ to $P_1$. $P_1$ selects the value forming majority among these and reconstructs $\vl{z}$ as $\vl{z} = \vl{p} \xor \vl{q}$.
	\end{myitemize}
\end{protocolbox}

\myparagraph{Achieving robustness}
To attain robustness, we list below the differences from the fair protocol that must be carried out. The first difference is the use of a robust variant of $\jsend$. Second, in input sharing protocol, where $\vl{x}_1$ is held by only garbler $P_0$, a corrupt $P_0$ may refrain from providing $P_1$ with the correct key (sent as the opening information for the commitment). To ensure robustness, if $P_1$ fails to receive the correct key from $P_0$, we let $P_1$ complain to all parties about this inconsistency by sending an inconsistency bit. All parties exchange this inconsistency bit among themselves and agree on the majority value. If all parties agree on the presence of inconsistency, then $P_0, P_1$ are identified to be in conflict, and $\TTP = P_2$ is set to carry out the rest of the computation. 
Finally, to ensure a robust reconstruction, the following approach is taken. Observe that the fair reconstruction provides robustness as long as evaluator $P_1$ is honest. When none of the garblers obtains the output in the fair protocol, it is guaranteed that evaluator $P_1$ is corrupt. Thus, in such a scenario, all parties take $P_1$ to be corrupt and proceed with $P_0$ as the $\TTP$.

\section{Security proofs}
\label{sec:security4pc}

Without loss of generality, we prove the security of our robust framework. The case for fairness follows similarly, and we omit its details. We provide proofs in the $\FSETUP, \Func[\jsend]$-hybrid model, where $\FSETUP$~(\S\ref{sec:KeySetupprelims}), $\Func[\jsend]$~(\boxref{fig:JsendFunc}) denote the ideal functionality for the shared-key setup and $\jsend$, respectively. 

The strategy for simulating the computation of function $f$ (represented by a circuit $\Ckt$) is as follows: Simulation begins with the simulator emulating the shared-key setup~($\FSETUP$) functionality and giving the respective keys to the adversary. This is followed by the input sharing phase in which $\Sim$ computes the input of $\Adv$, using the known keys, and sets the honest parties' inputs to be used in the simulation to $0$. $\Sim$ invokes the ideal functionality $\Func[GOD]$ on behalf of $\Adv$ using the extracted input and obtains the output $\vl{y}$. $\Sim$ now knows the inputs of $\Adv$ and can compute all the intermediate values for each building block. $\Sim$ proceeds with simulating each of the building blocks in the topological order. We provide the simulation for the case for corrupt $P_0, P_1$ and $P_3$. The case for corrupt $P_2$ is similar to that of $P_1$. 

For modularity, we provide the simulation steps for each building block separately. Carrying out these blocks in the topological order yields the simulation for the entire computation. If a $\TTP$ is identified during the simulation, the simulator stops and returns the function output to the adversary on behalf of the $\TTP$ as per $\Func[\jsend]$.

\paragraph{Ideal $\jsend$ Functionality}
The ideal $\jsend$ functionality for fairness security appears in \boxref{fig:JsendFFunc} and that for the robust setting appears in \boxref{fig:JsendFunc}.

\vspace{-1mm}
\begin{systembox}{$\Func[\jsend]$~(for fair security)}{Ideal functionality for $\jsend$ in $\Fthis$}{fig:JsendFFunc}
	\justify
	$\Func[\jsend]$ interacts with the parties in $\Partyset$ and the adversary $\Sim$. 
	\begin{myitemize}
		\item[\bf Step 1:] $\Func[\jsend]$ receives $(\INPUT,\vl{v}_s)$ from senders $P_s$ for $s \in \{i,j\}$, $(\INPUT,\bot)$ from receiver $P_k$ and fourth party $P_l$. While sending the inputs, the adversary is also allowed to send a special $\abort$ command.
		\item[\bf Step 2:] Set $\msg_i = \msg_j = \msg_l = \bot$.
		\item[\bf Step 3:] If $\vl{v}_i = \vl{v}_j$, set $\msg_k = \vl{v}_i$. Else, set $\msg_k = \abort$.
		\item[\bf Step 4:] Send $(\OUTPUT, \msg_s)$ to $P_s$ for $s \in \{0,1,2, 3\}$.
	\end{myitemize}
\end{systembox}

\vspace{-4mm}
\begin{systembox}{$\Func[\jsend]$~(for robust security)}{Ideal functionality for robust $\jsend$ in $\Fthis$.}{fig:JsendFunc}
	\justify
	$\Func[\jsend]$ interacts with the parties in $\Partyset$ and the adversary $\Sim$. 
	\begin{myitemize}
		\item[\bf Step 1:] $\Func[\jsend]$ receives $(\INPUT,\vl{v}_s)$ from senders $P_s$ for $s \in \{i,j\}$, $(\INPUT,\bot)$ from receiver $P_k$ and fourth party $P_l$, while it receives $(\SELECT,\ttp)$ from $\Sim$. Here $\ttp$ is a boolean value, with a $1$ indicating that $\TTP = P_l$ should be established. 
		\item[\bf Step 2:] If $\vl{v}_i =\vl{v}_j$ and $\ttp = 0$, or  if $\Sim$ has corrupted $P_l$\footnote{This condition is used to capture the fact that a corrupt $P_l$ cannot create an inconsistency in $\Func[\jsend]$ since the parties actively involved in $\Func[\jsend]$ would be honest}, set $\msg_i = \msg_j = \msg_l = \bot, \msg_k = \vl{v}_i$ and go to {\bf Step 4}.
		\item[\bf Step 3:] Else, set $\msg_i = \msg_j = \msg_k = \msg_l = P_l$.
		\item[\bf Step 4:] Send $(\OUTPUT, \msg_s)$ to $P_s$ for $s \in \{0,1,2, 3\}$.
	\end{myitemize}
\end{systembox}
\vspace{-4mm}

\paragraph{Sharing Protocol~($\prot{\Sh}$, \boxref{fig:piSh})}
During the preprocessing, $\Sim_{\prot{\Sh}}^{P_0}$ emulates $\FSETUP$ and gives the respective keys to $\Adv$. The values commonly held with $\Adv$ are sampled using the respective keys, while others are sampled randomly. The details for the online phase are provided next. We omit the simulation for corrupt $P_3$ as it is similar to that of $P_1,P_2$.
\begin{simulatorbox}{$\Sim_{\prot{\Sh}}^{P_0}$}{Simulator $\Sim_{\prot{\Sh}}^{P_0}$ for corrupt $P_0$ }{fig:ShFSim0}
	\justify 
	\algoHead{Online:}
	\begin{myitemize}
		\item[--] If dealer is $\Adv$, $\Sim_{\prot{\Sh}}^{P_0}$ receives $\mk{\vl{v}}$ from $\Adv$ on behalf of $P_1, P_2, P_3$. If the received values are consistent, $\Sim_{\prot{\Sh}}^{P_0}$ computes $\Adv$'s input $\vl{v}$ as $\vl{v} = \mk{\vl{v}} - \sqr{\pd{\vl{v}}}_1 - \sqr{\pd{\vl{v}}}_2 - \sqr{\pd{\vl{v}}}_3$, else sets $\vl{v}$ as the default value. It invokes $\Func[GOD]$ on input $(\INPUT, \vl{v})$ to obtain the function output $\vl{y}$. 
		\item[--] If dealer is $P_1, P_2$ or $P_3$, nothing to simulate as $P_0$ doesn't receive any value during the protocol.
	\end{myitemize}
\end{simulatorbox}

\begin{simulatorbox}{$\Sim_{\prot{\Sh}}^{P_1}$}{Simulator $\Sim_{\prot{\Sh}}^{P_1}$ for corrupt $P_1$ }{fig:ShFSim1}
	\justify 
	\algoHead{Online:}
	\begin{myitemize}
		\item[--] If dealer is $\Adv$, $\Sim_{\prot{\Sh}}^{P_1}$ receives $\mk{\vl{v}}$ from $\Adv$ on behalf of $P_2, P_3$. If the received values are consistent, $\Sim_{\prot{\Sh}}^{P_1}$ computes $\Adv$'s input $\vl{v}$ as $\vl{v} = \mk{\vl{v}} - \sqr{\pd{\vl{v}}}_1 - \sqr{\pd{\vl{v}}}_2 - \sqr{\pd{\vl{v}}}_3$, else sets $\vl{v}$ as the default value. It invokes $\Func[GOD]$ on input $(\INPUT, \vl{v})$ to obtain the function output $\vl{y}$. 
		\item[--] If dealer is $P_0, P_2$ or $P_3$, $\Sim_{\prot{\Sh}}^{P_1}$ sets $\vl{v} = 0$ and performs the protocol steps honestly. 
	\end{myitemize}
\end{simulatorbox}

Shares unknown to $\Adv$ are sampled randomly in the simulation, whereas in the real protocol, they are sampled using the pseudorandom function (PRF). The indistinguishability of the simulation thus follows by a reduction to the security of the PRF. The same holds for the rest of the blocks.

The simulation for the joint sharing protocol~($\prot{\JSh}$) is similar to that of the sharing protocol. The protocol's design is such that the simulator will always know the value to be sent as part of the joint sharing protocol. The communication is constituted by $\jsend$ calls and is emulated according to the simulation of $\Func[\jsend]$.

\paragraph{Multiplication Protocol~($\prot{\Mult}$  in $\FthisB$)}

\begin{simulatorbox}{$\Sim_{\prot{\Mult}}^{P_0}$}{Simulator $\Sim_{\prot{\Mult}}^{P_0}$ for corrupt $P_0$ }{fig:MulFSim0}
	\justify 
	\algoHead{Preprocessing:} \vspace{-2mm}
	\begin{description}
		\item[--]  Computes $\gm{\vl{a}\vl{b}}{1}, \gm{\vl{a}\vl{b}}{2}$, and $\gm{\vl{a}\vl{b}}{3}$ on behalf of $P_1, P_2, P_3$.
		\item[--] Samples $\vl{u}^1, \vl{u}^2$ using the respective keys with $\Adv$ and computes $\vl{r}$. The joint sharing of $\vl{q}$ is simulated as discussed earlier.
		\item[--] Receives $\vl{w}$ from $\Adv$ on behalf of $P_3$.
		\item[--] Simulating $\piVrfyP$: Joint sharing of $\vl{e_1}, \vl{e_2}, \vl{e}$ is simulated as discussed earlier. The rest of the steps are simulated honestly. This is possible since $\Sim_{\prot{\Mult}}^{P_0}$ knows the randomness and inputs that should be used by $\Adv$. 
	\end{description}
	\justify 
	\vspace{-2mm}
	\algoHead{Online:}
	$P_0$ has no communication in the online phase except the $\jsend$ instances which are emulated by $\Sim_{\prot{\Mult}}^{P_0}$.
\end{simulatorbox}

\begin{simulatorbox}{$\Sim_{\prot{\Mult}}^{P_1}$}{Simulator $\Sim_{\prot{\Mult}}^{P_1}$ for corrupt $P_1$ }{fig:MulFSim1}
	\justify 
	\algoHead{Preprocessing:}  \vspace{-2mm}
	\begin{description}
		\item[--]  Computes $\gm{\vl{a}\vl{b}}{1}, \gm{\vl{a}\vl{b}}{2}$, and $\gm{\vl{a}\vl{b}}{3}$ on behalf of $P_0, P_2, P_3$.
		\item[--] Samples $\vl{u}^1$ using the respective keys with $\Adv$. Samples a random $\vl{u}^2$ and computes $\vl{r}$. The joint sharing of $\vl{q}$ is simulated as discussed earlier.
		\item[--] Simulate the steps of $\piVrfyP$ honestly. 
	\end{description}
	\justify 
	\vspace{-2mm}
	\algoHead{Online:}  \vspace{-2mm}
	\begin{description}
		\item[--] Computes $\vl{y}_1 + \vl{s}_1, \vl{y}_2 + \vl{s}_2, \vl{y}_3$ honestly. 
		\item[--] Emulates two instances of $\Func[\jsend]$ -- i) $\Adv$ as sender to send $\vl{y}_1 + \vl{s}_1$ to $P_2$, and ii) $\Adv$ as receiver to obtain $\vl{y}_2 + \vl{s}_2$ from $P_2$. 
		\item[--] Simulates joint sharing as discussed earlier.
	\end{description}
\end{simulatorbox}

\begin{simulatorbox}{$\Sim_{\prot{\Mult}}^{P_3}$}{Simulator $\Sim_{\prot{\Mult}}^{P_3}$ for corrupt $P_3$ }{fig:MulFSim3}
	\justify 
	\algoHead{Preprocessing:} \vspace{-2mm}
	\begin{description}
		\item[--]  Computes $\gm{\vl{a}\vl{b}}{1}, \gm{\vl{a}\vl{b}}{2}$, and $\gm{\vl{a}\vl{b}}{3}$ on behalf of $P_0, P_1, P_2$.
		\item[--] Samples $\vl{u}^1, \vl{u}^2$ using the respective keys with $\Adv$ and computes $\vl{r}$. The joint sharing of $\vl{q}$ is simulated as discussed earlier.
		\item[--] Computes and sends $\vl{w}$ to $\Adv$ and simulate the steps of $\piVrfyP$ honestly.
	\end{description}
	\justify 
	\vspace{-2mm}
	\algoHead{Online:}  \vspace{-2mm}
	\begin{description}
		\item[--] Computes $\vl{y}_1 + \vl{s}_1, \vl{y}_2 + \vl{s}_2, \vl{y}_3$ honestly. 
		\item[--] Emulates two instances of $\Func[\jsend]$ with $\Adv$ as sender to exchange $\vl{y}_1 + \vl{s}_1, \vl{y}_2 + \vl{s}_2$ among $P_1, P_2$. 
		\item[--] Simulates joint sharing as discussed earlier.
	\end{description}
\end{simulatorbox}

\paragraph{Reconstruction Protocol~($\prot{\Rec}$, \boxref{fig:piRec})}
Using the input of $\Adv$ obtained during simulation of sharing protocol, $\Sim_{\prot{\Rec}}$ invokes $\Func[GOD]$ on behalf of $\Adv$ and obtains the function output $\vl{y}$ in clear. $\Sim_{\prot{\Rec}}$ calculates the missing share of $\Adv$ using $\vl{y}$ and the other shares. The missing share is then communicated to $\Adv$ by emulating the $\Func[\jsend]$ functionality. 

\chapter{$\TWthis$: 2PC Semi-honest Protocols}
\label{chap:layer1_2pc}
This chapter provides details for the Layer I blocks of our 2PC framework $\TWthis$. Some of the results in this chapter resulted in a publication at USENIX Security'21~\cite{USENIX:PSSY21}\footnote{This is joint work with Thomas Schneider and Hossein Yalame of TU Darmstadt. All co-authors contributed to the fruitful discussions that resulted in this publication. Ajith Suresh designed the new sharing scheme for two-party computation, provided new conversions between different MPC protocols, and benchmarked the protocols. Hossein Yalame designed the new circuits for parallel-prefix adder, comparison, and equality test based on multi-input AND gates and provided the depth-optimized variant of AES.}. 
Comparison of $\TWthis$ with passively secure 2PC PPML framework of~\cite{SP:MohZha17}, in terms of the communication for multiplication,  is presented in \tabref{2pcSCost}.

\begin{table}[htb!]
	\centering
	\resizebox{0.98\textwidth}{!}{
		\begin{NiceTabular}{r c r|r r|r r|c}
			\toprule
			\Block{2-1}{Work}
			& \Block[c]{2-1}{\#Active\\Parties}
			& \Block{2-1}{Security}
			& \multicolumn{2}{c}{Multiplication} 
			& \multicolumn{2}{c}{Multiplication with Truncation\tabularnote{$\ell$ - size of ring in bits, $\csec$ - computational security parameter.}} 
			& \Block{2-1}{Conversions\tabularnote{A, B, G indicate support for arithmetic, boolean, and garbled worlds respectively.}} \\ \cmidrule{4-7}
			&  & 
			& Comm\textsubscript{pre} 
			& Comm\textsubscript{on}\tabularnote{`Comm' - communication, `pre' - preprocessing, `on' - online} 
			& Comm\textsubscript{pre} 
			& Comm\textsubscript{on} &  \\ 
			\midrule
			\cite{SP:MohZha17} & 2 & Semi-honest & $2\ell(\csec + \ell)$ & $4\ell$ & $2\ell(\csec + \ell)$ & $4\ell$ & A-B-G\\		
			\textbf{$\TWthis$} & 2 & Semi-honest & $2\ell(\csec + \ell)$ & $2\ell$ & $2\ell(\csec + \ell)$ & $2\ell$ & A-B-G\\		
			\bottomrule
		\end{NiceTabular}
    }
	\caption{Comparison of semi-honest 2PC PPML frameworks}\label{tab:2pcSCost}
\end{table}

\section{Preliminaries and Definitions}
\label{sec:2pcSPrelim}
In our framework, we have two parties $\Partyset = \{P_1, P_2 \}$ who are connected by a bidirectional synchronous channel (e.g. instantiated via TLS over TCP/IP), and a static, semi-honest adversary that can corrupt at most one party. This framework is similar to that of the three-party framework $\TSthis$ except for the absence of helper party $P_0$.

\subsection{Sharing Semantics}
\label{sec:2pcSsematics}
For the arithmetic and boolean sharing, we follow  masked evaluation technique, where a value $\vl{v} \in \Z{\ell}$ is split into three shares. Two of the shares~($\pad{\vl{v}}{1}, \pad{\vl{v}}{2})$ can be generated in the preprocessing phase independent of the value to be shared, and their sum can be interpreted as a mask~($\pad{\vl{v}}{}$). The third share, dependent on $\vl{v}$,  can be computed in the online phase and can be treated as the masked value $\mk{\vl{v}}= \vl{v} +\pad{\vl{v}}{}$. 

\begin{table}[htb!]
        \centering
		\begin{NiceTabular}{r r r}[notes/para]
			\toprule
			Sharing Type  & $P_1$ & $P_2$\\
			\midrule
			$\sqr{\cdot}$-sharing\tabularnote{$\vl{v} = \vl{v}^1 + \vl{v}^2$} 
			& ${\vl{v}}^1$     & ${\vl{v}}^2$\\
			$\shr{\cdot}$-sharing\tabularnote{$\pad{\vl{v}}{} = \pad{\vl{v}}{1} + \pad{\vl{v}}{2}$, $\mk{\vl{v}} = \vl{v} + \pad{\vl{v}}{}$}   
			& $(\mk{\vl{v}}, \pad{\vl{v}}{1})$  
			& $(\mk{\vl{v}}, \pad{\vl{v}}{2})$ \\
			\bottomrule
		\end{NiceTabular}
	\caption{Semantics for $\vl{v} \in \Z{\ell}$ in \TWthis.}\label{tab:2pcSsharing}
\end{table}

The sharing semantics is presented in \tabref{2pcSsharing}, denoted by $\shr{\cdot}$, along with the semantics for $\sqr{\cdot}$-sharing. Both the sharings used are linear i.e. given sharings of $\vl{v}_1,\ldots, \vl{v}_m$ and public constants $c_1,\ldots,c_m$, sharing of $\sum_{i=1}^m c_i \vl{v}_i$ can be computed non-interactively for an integer $m$.

\begin{notation} \label{notation:2pcSconcise}
	(a) For the $\shr{\cdot}$-shares of $n$ values $\vl{a}_1,\ldots,\vl{a}_n$, $\gm{\vl{a}_1 \ldots \vl{a}_n}{} = \prod\limits_{i=1}^{n} \pad{\vl{a}_i}{}$ and $\mk{\vl{a}_1 \ldots \vl{a}_n}{} = \prod\limits_{i=1}^{n} \mk{\vl{a}_i}{}$ (b) We use superscripts ${\bf B}$, and ${\bf G}$ to denote sharing semantics in boolean, and garbled world, respectively-- $\shrB{\cdot}$,  $\shrG{\cdot}$. We omit the superscript for arithmetic world. 
\end{notation}

Sharing semantics for boolean sharing over $\Z{}$ is similar to arithmetic sharing except that addition is replaced with XOR. The semantics for garbled sharing are described in \S\ref{sec:2pcSGCWorld} with the relevant context.

\subsection{Oblivious Transfer~(OT)}
\label{subsec:OTs}
In a 1-out-of-$n$ Oblivious Transfer~\cite{STOC:ImpRud89,JC:NaoPin05} (OT) over $\ell$-bit messages, the sender $S$ inputs $n$ messages $(x_1, \ldots, x_n)$ each of length $\ell$ bits, while the receiver $R$ inputs the choice $c \in \{1,\ldots,n\}$. $R$ receives $x_c$ as output while $S$ receives $\bot$ as output. The privacy guarantee is that $S$ learns nothing about $c$, while $R$ learns nothing about the inputs of $S$ other than $x_c$. We use $\OTN{m}{\ell}{n}$ to denote $m$ instances of 1-out-of-$n$ OT on $\ell$ bit inputs. 

OT is a fundamental building block for MPC~\cite{STOC:Kilian88} and requires expensive public-key cryptography~\cite{STOC:ImpRud89}. The technique of OT Extension~\cite{C:IKNP03,CCS:ALSZ13,C:KolKum13,NDSS:PatSarSur17} allows us to generate many OTs from a small number (equal to the security parameter) of base OTs at the expense of symmetric-key operations alone. This reduces the cost of OT mainly to highly efficient symmetric-key primitives. Concretely, the OT Extension of \cite{CCS:ALSZ13} generates around 1 million $\OTN{1}{\ell}{2}$ per second with passive security. An orthogonal line of work considered pre-computation of OT~\cite{C:Beaver95}, where all the cryptographic operations can be shifted to a setup phase, independent of the function to be evaluated. This technique enables a very efficient online phase for protocols that use OT. In the semi-honest setting, the state-of-the-art solution for OT extension~\cite{CCS:ALSZ13} has communication $\csec + 2\ell$ bits per OT for $\OTN{1}{\ell}{2}$ where $\kappa$ denotes the computational security parameter.

A correlated OT ($\COT{}{}$)~\cite{CCS:ALSZ13} is a variant of the traditional OT  where the sender's input messages are correlated. In a $\COT{}{}$, the sender inputs a correlation function $f()$ and obtains the message pair ($x_0 \in_R \{0,1\}^\ell, x_1 = f(x_0)$) as the output. The receiver, on the other hand, inputs her choice $c$ and obtains $x_c$ as output. We use $\COT{m}{\ell}$ to denote $m$ instances of 1-out-of-2 correlated OT on $\ell$ bit inputs. In the semi-honest setting, $\COT{1}{\ell}$ has communication $\csec + \ell$ bits~\cite{CCS:ALSZ13}.

\subsection{Homomorphic Encryption~(HE)}
\label{subsec:HEs}
The homomorphic property allows us to compute a ciphertext from a set of ciphertexts such that the plaintext underlying the former is a function of the underlying plaintexts of the latter. Towards this, one party called client generates a key-pair $(\mathsf{pk}, \mathsf{sk})$ for the HE scheme and sends $\mathsf{pk}$ to the other party called server. To perform a secure computation operation, the client encrypts its data using $\mathsf{pk}$ and sends this to the server. Now the server can locally compute the ciphertext corresponding to the operation and return the encrypted result to the client. The client can now decrypt the received ciphertext using her private key $\mathsf{sk}$. An \emph{additively HE} allows us to generate the ciphertext corresponding to the sum of the underlying plaintexts by doing operations on the ciphertexts. Prominent examples of additively HE schemes are Paillier~\cite{EC:Paillier99}, DGK~\cite{DamgardGK08} and RLWE-AHE~\cite{CANS:RatSchShu19}. On the other hand, fully homomorphic encryption schemes allow arbitrary computations under the encryption but are less efficient. See~\cite{AcarAUC18} for a more detailed description.

\section{Arithmetic / Boolean 2PC}
\label{sec:2pcSFourPC}

This section covers the details of our 2PC semi-honest protocol $\TWthis$ over an arithmetic ring $\Z{\ell}$. The protocol primarily consists of the following primitives -- i) Sharing~(\secref{share2pcS}), ii) Multiplication~(\secref{mult2pcS}), and iii) Reconstruction~(\secref{rec2pcS}). 

\subsection{Sharing}
\label{sec:share2pcS}
Protocol $\prot{\Sh}$~(\boxref{fig:piSh2pcS}) enables $P_i$ to generate $\shr{\cdot}$-share of a value $\vl{v}$. During the preprocessing phase, $\pd{}$-shares are sampled non-interactively using the pre-shared keys~(cf. \S\ref{sec:KeySetupprelims}) in a way that $P_i$ will get the entire mask $\pd{\vl{v}}$. During the online phase, $P_i$ computes $\mk{\vl{v}} = \vl{v} + \pd{\vl{v}}$ and sends to $P_1, P_2$. 
For the special case when parties want to generate $\shr{\vl{v}}$ in the preprocessing, the protocol can be made non-interactive. W.l.o.g. consider the case when $P_i = P_1$. Parties set $\mk{\vl{v}} = 0$. $P_1, P_2$ sample $\pad{\vl{v}}{2}$ non-interactively while $P_1$ sets $\pad{\vl{v}}{1} = - (\vl{v} + \pad{\vl{v}}{2})$. The case for $P_i = P_2$ is similar.

\begin{protocolbox}{$\prot{\Sh}(P_i, \vl{v})$}{$\shr{\cdot}$-sharing of a value $\vl{v}$ by party $P_i$ in $\TWthis$.}{fig:piSh2pcS}
	\detail{
		{\bf Input(s):} $P_i : \vl{v}$,~~{\bf Output:} $\shr{\vl{v}}$.
	}
	\justify
	\algoHead{Preprocessing:} 
	Sample as follows: $P_i, P_1: \pad{\vl{v}}{1}$,~~$P_i, P_2: \pad{\vl{v}}{2}$.
	\justify
	\vspace{-2mm}
	\algoHead{Online:} $P_i$ computes $\mk{\vl{v}} = \vl{v} + \pd{\vl{v}}$ and sends to $P_1, P_2$.  
\end{protocolbox}

\begin{lemma}[Communication]
	\label{lemma:pish2pcS}
	Protocol $\prot{\Sh}$~(\boxref{fig:piSh2pcS}) requires a communication of at most $\ell$ bits and $1$ round in the online phase.
\end{lemma}
\begin{proof}
	The preprocessing of $\prot{\Sh}$ is non-interactive as the parties sample non interactively using key setup $\Func[Key]$~(\S\ref{sec:KeySetupprelims}). In the online phase, $P_i$ sends $\mk{\vl{v}}$ to either $P_1$ or $P_2$ (depending upon $P_i$) resulting in 1 round and communication of $\ell$ bits.
\end{proof}

\subsubsection{Joint Sharing}
\label{sec:jsh2pcS}
Protocol $\prot{\JSh}$ enables parties $P_1,P_2$ to generate $\shr{\cdot}$-share of a value $\vl{v}$ known to both of them non-interactively. For this, parties set $\pad{\vl{v}}{1} = \pad{\vl{v}}{2} = 0$ and $\mk{\vl{v}} = \vl{v}$.

\subsection{Multiplication}
\label{sec:mult2pcS}
Given the shares of $\vl{a}, \vl{b}$, the goal of the multiplication protocol is to generate shares of $\vl{z} = \vl{ab}$. The protocol is designed such that $P_i$ for $i \in \{1,2\}$ obtain $\vl{z}^i$ in the online phase such that $\vl{z} = \vl{z}^1 + \vl{z}^2$. Parties then compute $\shr{\vl{z}}$ as $\shr{\vl{z}^1} + \shr{\vl{z}^2}$ to obtain the final output. 

\paragraph{Online} 
Note that,

\begin{align}\label{eq:2pcSmult}
	\vl{z} = \vl{ab} = (\mk{\vl{a}} - \pd{\vl{a}})(\mk{\vl{b}} - \pd{\vl{b}}) 
	= \mk{\vl{ab}} - \mk{\vl{a}}\pd{\vl{b}} - \mk{\vl{b}}\pd{\vl{a}} + \gm{\vl{ab}}{}
	~~\text{\footnotesize{(cf. notation~\ref{notation:2pcSconcise})}}
\end{align}

Let $\vl{z} = \vl{z}_1 + \vl{z}_2$, where $\vl{z}_1$ and $\vl{z}_2$ can be computed respectively by $P_1$ and $P_2$.

\begin{align}\label{eq:2pcSmulty}
	P_1: \vl{z}_1 &= \mk{\vl{ab}} - \pad{\vl{a}}{1} \mk{\vl{b}} - \pad{\vl{b}}{1} \mk{\vl{a}} + \sqr{\gm{\vl{ab}}{}}_1 \nonumber\\
	P_2: \vl{z}_2 &= ~~~~~~- \pad{\vl{a}}{2} \mk{\vl{b}} - \pad{\vl{b}}{2} \mk{\vl{a}} + \sqr{\gm{\vl{ab}}{}}_2 
\end{align}

During preprocessing, parties rely on $\prot{\MultPre}$ to generate an additive sharing~($\sqr{\cdot}$) of $\gm{\vl{ab}}{}$. We note that Turbospeedz~\cite{ACNS:BenNieOmr19} achieves same online cost as that of ours, but with a more expensive preprocessing. We provide more details in \S\ref{sec:comp2pcS}.

\smallskip
\begin{protocolbox}{$\prot{\Mult}(\vl{a}, \vl{b}, \isTr)$}{Multiplication with / without truncation in $\TWthis$.}{fig:piMult2pcS}
	$\isTr$ is a bit denoting whether truncation is required ($\isTr =1$) or not ($\isTr=0$). \\
	\detail{
		{\bf Input(s):} $\shr{\vl{a}}, \shr{\vl{b}}$.\\
		{\bf Output:} $\shr{\vl{o}}$ where $\vl{o} = \vl{z}^{\vl{t}}$ if $\isTr = 1$ and $\vl{o} = \vl{z}$ if $\isTr = 0$ and $\vl{z} = \vl{ab}$.
	}
	\justify 
	\vspace{-2mm}
	\algoHead{Preprocessing:} Execute $\prot{\MultPre}$ on $\sqr{\pad{\vl{a}}{}}$ and $\sqr{\pad{\vl{b}}{}}$ to generate $\sqr{\gm{\vl{ab}}{}}$.
	\justify
	\vspace{-2mm}
	\algoHead{Online:} 
	\begin{enumerate}
		\item Compute: $P_1: \vl{z}_1 = \mk{\vl{ab}} - \pad{\vl{a}}{1} \mk{\vl{b}} - \pad{\vl{b}}{1} \mk{\vl{a}} + \sqr{\gm{\vl{ab}}{}}_1,~~ 
			P_2: \vl{z}_2 = - \pad{\vl{a}}{2} \mk{\vl{b}} - \pad{\vl{b}}{2} \mk{\vl{a}} + \sqr{\gm{\vl{ab}}{}}_2$
		\item If $\isTr = 1$, $P_i$ sets $\vl{p}_i = \vl{z}_i^{\vl{t}}$, else $\vl{p}_i = \vl{z}_i$ where $i \in \{1,2\}$.
		Execute $\prot{\Sh}(P_i, \vl{p}_i)$ to generate $\shr{\vl{p}_i}$. 
		\item Compute $\shr{\vl{o}} = \shr{\vl{p}_1} + \shr{\vl{p}_2}$. Here $\vl{o} = \vl{z}^{\vl{t}}$ if $\isTr = 1$ and $\vl{z}$ otherwise.
	\end{enumerate}     
\end{protocolbox}

\paragraph{Preprocessing} 
We now provide the details for instantiating $\prot{\MultPre}$ using two of the well-known primitives: i) Oblivious Transfer~(OT) as used in~\cite{NDSS:DemSchZoh15,CCS:KelOrsSch16} and ii) Homomorphic Encryption~(HE) as used in~\cite{CCS:HKSSW10,C:DPSZ12,CANS:RatSchShu19}. These two approaches have been rallied against each other in terms of practical efficiency in the past, and fair competition is still going on. In our work, we make only black-box access to these primitives, and hence any improvement in any of them will directly impact the overall efficiency of the setup phase of our protocols. 

Note that $\gm{\vl{ab}}{} = (\pad{\vl{a}}{1} + \pad{\vl{a}}{2})(\pad{\vl{b}}{1} + \pad{\vl{b}}{2}) = \pad{\vl{a}}{1}\pad{\vl{b}}{1} + \pad{\vl{a}}{1}\pad{\vl{b}}{2} + \pad{\vl{a}}{2}\pad{\vl{b}}{1} + \pad{\vl{a}}{2}\pad{\vl{b}}{2}$. Here $P_i$ for $i \in \{1,2\}$ can locally compute $\pad{\vl{a}}{i}\pad{\vl{b}}{i}$ and hence the problem reduces to computing $\pad{\vl{a}}{1}\pad{\vl{b}}{2}$ and  $\pad{\vl{a}}{2}\pad{\vl{b}}{1}$.

\smallskip
{\em OT based $\prot{\MultPre}$:}
In our OT-based approach, we use Correlated OTs ($\COT{}{}$)~\cite{CCS:ALSZ13} where the sender inputs a correlation function $f(\cdot)$ to $\COT{}{}$ and obtains $(m_0, m_1)$, where $m_0$ is a random element and $m_1 = f(m_0)$. We use $\COT{n}{\ell}$ to represent $n$ parallel instances of 1-out-of-2 Correlated OTs on $\ell$ bit input strings.

To compute $\sqr{\pad{\vl{a}}{1}\pad{\vl{b}}{2}}$, the parties execute $\COT{\ell}{\ell}$ with $P_1$ being the sender and $P_2$ being the receiver. For the $j$-th instance of $\COT{}{}$ where $j \in \{0,\ldots,\ell-1\}$, $P_1$ inputs the correlation $f_j(x) = x + 2^j \pad{\vl{a}}{1}$ and obtains $(m_{j,0} = r_j, m_{j,1} = r_j + 2^j \pad{\vl{a}}{1})$. $P_2$ inputs choice bit $b_j$ as the $j$-th bit of $\pad{\vl{b}}{2}$ and obtains $m_{j,b_j}$ as output. Now the $\sqr{\cdot}$-shares are defined as $\sqr{\pad{\vl{a}}{1}\pad{\vl{b}}{2}}_1 = \sum_{j=0}^{\ell-1} (-r_j)$ and $\sqr{\pad{\vl{a}}{1}\pad{\vl{b}}{2}}_2 = \sum_{j=0}^{\ell-1} m_{j,b_j}$. Computation of $\pad{\vl{a}}{2}\pad{\vl{b}}{1}$ proceeds similarly with the role of the parties reversed.

In the OT-based approach~\cite{NDSS:DemSchZoh15,CCS:KelOrsSch16}, the technique of OT extension~\cite{CCS:ALSZ13,C:KolKum13,NDSS:PatSarSur17} can be used. One instance of $\prot{\MultPre}$ requires two instances of $\COT{\ell}{\ell}$ where each instance has communication $\ell (\csec + \ell)$ bits. Over a $64$-bit ring, this corresponds to $3072$ bytes. 
Recently, \cite{CCS:BCGIKRS19} came up with a very efficient OT extension technique named {Silent OT Extension} which claims to outperform state-of-the-art solutions for performing $\prot{\MultPre}$. Since our protocol makes black-box calls to $\prot{\MultPre}$, it can directly benefit from the performance improvements of~\cite{CCS:BCGIKRS19}.

\smallskip
{\em HE-based $\prot{\MultPre}$:}
In a HE based solution, $P_1$, using its public key $\mathsf{pk_1}$, encrypts its messages $\pad{\vl{a}}{1}, \pad{\vl{b}}{1}$ in independent ciphertexts and sends the ciphertexts to $P_2$. In parallel, $P_2$ computes the ciphertexts corresponding to $\pad{\vl{a}}{2}, \pad{\vl{b}}{2}$ and a random element $\vl{r} \in_R \Z{\ell}$ using $\mathsf{pk_1}$. Upon receiving the ciphertexts from $P_1$, $P_2$ computes the ciphertext corresponding to $\vl{v} = \pad{\vl{a}}{1}\pad{\vl{b}}{2} + \pad{\vl{a}}{2}\pad{\vl{b}}{1} - \vl{r}$ using the homomorphic property of the underlying HE. $P_2$ then sends encryption of $\vl{v}$ to $P_1$ who then decrypts it using its secret key $\mathsf{sk_1}$. Note that ($\vl{v}$, $\vl{r}$) forms an additive sharing of the desired value: $\pad{\vl{a}}{1}\pad{\vl{b}}{2} + \pad{\vl{a}}{2}\pad{\vl{b}}{1} = \vl{v} + \vl{r}$.

Recently, Ring LWE-based AHE~\cite{CANS:RatSchShu19} was shown to outperform the solutions based on OT for generating multiplication triples. The authors observed that the plaintext space is much larger than the range of the values being encrypted. Thus they used the technique of {\em ciphertext packing}, using Microsoft SEAL library, where ciphertexts corresponding to multiple plaintexts are packed into a single ciphertext. This optimizes the number of ciphertexts being sent back and the number of decryptions on $P_1$'s side. In \cite{CANS:RatSchShu19}, the amortized communication cost for performing one instance of $\prot{\MultPre}$ over a $64$-bit ring with a security level of 128 bits is $448$ bytes, which is a $7\times$ improvement over the best OT-based solutions~\cite{NDSS:DemSchZoh15} available at that time.

\begin{lemma}[Communication]
	\label{lemma:piMult2pcS}
	Protocol $\prot{\Mult}$~(\boxref{fig:piMult2pcS})~(in $\TWthis$) requires $2\ell (\csec + \ell)$ bits of communication in the preprocessing, and $1$ round and $2 \ell$ bits of communication in the online phase.
\end{lemma}
\begin{proof}
	During the preprocessing, as part of $\prot{\MultPre}$, we use 2 instances of correlated OTs ($\COT{}{}$)~\cite{CCS:ALSZ13} which incur a communication of $\ell+\csec$ bits per $\COT{}{}$ on $\ell$-bit strings, where $\csec$ is the computational security parameter. During the online phase, each of $P_1$ and $P_2$ executes one instance of $\prot{\Sh}$ and the cost follows from Lemma~\ref{lemma:pish2pcS}. 
\end{proof}

\subsubsection{Truncation}
To accommodate truncation, following $\TSthis$, $P_i$ for $i \in \{1,2\}$ locally truncates $\vl{z}_i$ before executing the sharing in the online of $\prot{\Mult}$~(\boxref{fig:piMult2pcS}). The correctness follows from  \cite{SP:MohZha17}.

\subsubsection{Multiplication with constant}
Multiplication by a constant in MPC is typically local. Given constant $\alpha$ and $\shr{\vl{v}}$, the $\shr{\cdot}$-shares of the product $\vl{y} = \alpha\vl{v}$ can be locally computed as per \eqref{eq:mutconst2pcS}. 
\begin{equation}\label{eq:mutconst2pcS}
	\mk{\vl{y}} = \alpha \mk{\vl{u}},~~~\pad{\vl{y}}{1} = \alpha \pad{\vl{v}}{1},~~~\pad{\vl{y}}{2} = \alpha \pad{\vl{v}}{2}
\end{equation}

However, in FPA, we need to perform a truncation on the output. Let $\alpha\vl{v} = \beta^1 + \beta^2$ where $\beta^1 = \alpha.(\mk{\vl{v}} -\pad{\vl{v}}{1})$ and $\beta^2 =  - \alpha.\pad{\vl{v}}{2}$. $P_i$ for $i \in \{1,2\}$ locally truncates $\beta^i$ and executes the sharing
protocol $\prot{\Sh}$ on the truncated value. Parties locally compute $\shr{\alpha\vl{v}} = \shr{\beta^1} + \shr{\beta^2}$ to obtain the final result.

\subsection{Reconstruction}
\label{sec:rec2pcS}
$\prot{\Rec}(\Partyset, \vl{v})$ enables parties to compute $\vl{v}$, given its $\shr{\cdot}$-share. For this, $P_1$ sends $\pad{\vl{v}}{1}$ to $P_2$ and $P_2$ sends $\pad{\vl{v}}{2}$ to $P_1$. Parties locally compute $\vl{v} = \mk{\vl{v}} - \pad{\vl{v}}{1} - \pad{\vl{v}}{2}$. Reconstruction towards a single party can be viewed as a special case.

\begin{lemma}[Communication]
	\label{lemma:pirec2pcS}
	Protocol $\prot{\Rec}$ requires a communication of $2\ell$ bits and $1$ round.
\end{lemma}

\subsection{Multi-input Multiplication}
\label{sec:multT2pcS}

\subsubsection{3-input multiplication}
To compute $\shr{\cdot}$-shares of $\vl{z} = \vl{abc}$, note that  

\begin{align}\label{eq:2pcSmultT}
	\vl{z} &= \vl{abc} = (\mk{\vl{a}} - \pd{\vl{a}})(\mk{\vl{b}} - \pd{\vl{b}})(\mk{\vl{c}} - \pd{\vl{c}}) \nonumber\\ 
	&= \mk{\vl{abc}} - \mk{\vl{ac}}\pd{\vl{b}} - \mk{\vl{bc}}\pd{\vl{a}} - \mk{\vl{ab}}\pd{\vl{c}} + \mk{\vl{a}}\gm{\vl{bc}}{} + \mk{\vl{b}}\gm{\vl{ac}}{} + \mk{\vl{c}}\gm{\vl{ab}}{} - \gm{\vl{abc}}{}
	~~\text{\footnotesize{(cf. notation~\ref{notation:2pcSconcise})}}
\end{align}

Similar to $\prot{\Mult}$,  parties rely on $\prot{\MultPre}$ to generate an additive sharing~($\sqr{\cdot}$) of $\gm{\vl{ab}}{}, \gm{\vl{bc}}{}$ and $\gm{\vl{ac}}{}$. Parties then generate $\sqr{\gm{\vl{abc}}{}}$ using another instance of $\prot{\MultPre}$ with inputs $\gm{\vl{ab}}{}$ and $\pad{\vl{c}}{}$.

\begin{lemma}[Communication]
	\label{lemma:piMultT2pcS}
	Protocol $\prot{\MultT}$~(in $\TWthis$) requires $8\ell(\csec + \ell)$ bits of communication in the preprocessing, and $1$ round and $2 \ell$ bits of communication in the online phase.
\end{lemma}
\begin{proof}
	The preprocessing involves four instances of $\prot{\MultPre}$ each costing a communication of $2\ell(\csec + \ell)$ bits.
	The online phase is similar to $\prot{\Mult}$ and the costs follow from Lemma~\ref{lemma:piMult2pcS}. 
\end{proof}

\subsubsection{4-input multiplication}
For the case of 4-input multiplication with $\vl{z} = \vl{abcd}$, note that  

\begin{align}\label{eq:2pcSmultF}
	\vl{z} &= \vl{abcd} - \vl{r} = (\mk{\vl{a}} - \pd{\vl{a}})(\mk{\vl{b}} - \pd{\vl{b}})(\mk{\vl{c}} - \pd{\vl{c}})(\mk{\vl{d}} - \pd{\vl{d}}) \nonumber\\ 
	&= \mk{\vl{abcd}} - \mk{\vl{abc}}\pd{\vl{d}} - \mk{\vl{abd}}\pd{\vl{c}} - \mk{\vl{acd}}\pd{\vl{b}} - \mk{\vl{bcd}}\pd{\vl{a}} 
	+ \mk{\vl{ab}}\gm{\vl{cd}}{} + \mk{\vl{ac}}\gm{\vl{bd}}{} + \mk{\vl{ad}}\gm{\vl{bc}}{} + \mk{\vl{bc}}\gm{\vl{ad}}{} \nonumber\\
	&~~~+ \mk{\vl{bd}}\gm{\vl{ac}}{} + \mk{\vl{cd}}\gm{\vl{ab}}{} 
	- \mk{\vl{a}}\gm{\vl{bcd}}{} - \mk{\vl{b}}\gm{\vl{acd}}{} - \mk{\vl{c}}\gm{\vl{abd}}{} - \mk{\vl{d}}\gm{\vl{abc}}{} + \gm{\vl{abcd}}{}
	~~\text{\footnotesize{(cf. notation~\ref{notation:2pcSconcise})}}
\end{align}

Here the parties need to generate $\sqr{\cdot}$-shares of $\gm{\vl{ab}}{},\gm{\vl{ac}}{},\gm{\vl{ad}}{},\gm{\vl{bc}}{},\gm{\vl{bd}}{},\gm{\vl{cd}}{},\gm{\vl{abc}}{},\gm{\vl{abd}}{},\gm{\vl{acd}}{},\gm{\vl{bcd}}{}$ and $\gm{\vl{abcd}}{}$. This is computed similarly as in 3-input multiplication and the protocol is denoted as $\prot{\MultF}$.

\begin{lemma}[Communication]
	\label{lemma:piMultF2pcS}
	Protocol $\prot{\MultF}$~(in $\TWthis$) requires $22\ell(\csec + \ell)$ bits of communication in the preprocessing, and $1$ round and $2 \ell$ bits of communication in the online phase.
\end{lemma}

\subsubsection{Comparison with the LUT-based protocol of~\cite{NDSS:DKSSZZ17}}
\label{sec:lookupComp2pcS}
We compare our multi-input AND gate protocols with~\cite{NDSS:DKSSZZ17} for two, three and four inputs. \cite{NDSS:DKSSZZ17} proposed two variants -- i) OP-LUT - optimized online communication of $2N$ bits, and ii) SP-LUT - optimized total communication of $2\kappa + 2^N$ bits. The concrete details are given in \tabref{Multi2pcS}.

\begin{table}[htb!]
	\centering
	\resizebox{0.7\textwidth}{!}{
		\begin{NiceTabular}{c|r|r|r|c}
			\toprule
			\Block{2-1}{Gate}  & \Block{2-1}{Protocol} 
			& \multicolumn{1}{c|}{Preprocessing} & \multicolumn{2}{c}{Online} \\ \cmidrule{3-5}
			& & Communication & Communication & Rounds \\
			\midrule
			\Block[c]{3-1}{AND\\$\vl{z} = \vl{ab}$}    
			& OP-LUT        & $206$	  & $4$	         & {\bf 1}	\\
			& SP-LUT    	& $190$	   & $6$          & {\bf 1}	\\
			& $\TWthis$	    & $134$	   & $\bf{2}$	& {\bf 1}   \\
			\midrule
			\Block[c]{3-1}{AND3\\$\vl{z} = \vl{abc}$}   
			& OP-LUT       & $285$	  & $6$	        & {\bf 1}	\\
			& SP-LUT	   & $221$	   & $11$        & {\bf 1}	\\
			& $\TWthis$	   & $250$	  &  {\bf 2}	 & {\bf 1}   \\
			\midrule
			\Block[c]{3-1}{AND4\\$\vl{z} = \vl{abcd}$}    
			& OP-LUT      & $492$	  & $8$	        & {\bf 1}	\\
			& SP-LUT	  & $236$     & $20$       & {\bf 1}	\\
			& $\TWthis$	  & $412$	  &  {\bf 2}	 & {\bf 1}   \\
			\bottomrule
		\end{NiceTabular}
	}
	\caption{Comparison of $\TWthis$ and~\cite{NDSS:DKSSZZ17}~(OP-LUT and SP-LUT). Communication is provided in bits. Best values for the online phase are marked in bold.}\label{tab:Multi2pcS}
\end{table}

\subsection{Comparison with Turbospeedz~\cite{ACNS:BenNieOmr19} and~\cite{FC:OhaNui20}}
\label{sec:comp2pcS}

Here, we compare our 2PC protocol with Turbospeedz~\cite{ACNS:BenNieOmr19} and \cite{FC:OhaNui20}. 

\subsubsection{Comparison with Turbospeedz~\cite{ACNS:BenNieOmr19}}
For the 2-input multiplication, Turbospeedz~\cite{ACNS:BenNieOmr19} presented a protocol that reduces the online communication of SPDZ-style protocols from $4$ to $2$ ring elements using a function-dependent preprocessing. Turbospeedz first executes a SPDZ-like preprocessing where random multiplication triples are generated. These triples are then associated with the multiplication gates using additional values that they call "external values"~(cf. \cite{ACNS:BenNieOmr19}, \S3.2). On the contrary, we obtain the preprocessing data directly and hence save communication of 4 ring elements and storage of 5 ring elements compared with Turbospeedz. \tabref{comp2pc} provides the communication and storage required for the 2-input multiplication protocol of ABY~\cite{NDSS:DemSchZoh15}, Turbospeedz~\cite{ACNS:BenNieOmr19} and $\TWthis$.

\begin{table}[htb!]
	\centering
	\resizebox{0.8\textwidth}{!}{
		\begin{NiceTabular}{r r r r r}
			\toprule
			Phase & Parameter & ABY~\cite{NDSS:DemSchZoh15} & Turbospeedz~\cite{ACNS:BenNieOmr19} & $\TWthis$\\ 
			\midrule 
			\multirow{2}{*}{\bf Preprocessing}
			& Storage       & $3\ell$             & $9\ell$              & $4\ell$\\ \cmidrule{2-5}     
			& Communication & $|\Triple|$         & $|\Triple| + 4\ell$  & $|\Triple|$ \\ 
			\midrule 
			\multirow{2}{*}{\bf Online}
			& Storage       & $5\ell$             & $5\ell$              & $\mathbf{3\ell}$\\ \cmidrule{2-5}     
			& Communication & $4\ell$             & $\mathbf{2\ell}$              & $\mathbf{2\ell}$ \\ 
			\midrule
			\multirow{2}{*}{\bf Total}
			& Storage       & $8\ell$             & $14\ell$             & $7\ell$\\ \cmidrule{2-5}     
			& Communication & $|\Triple| + 4\ell$ & $|\Triple| + 6\ell$  & $|\Triple| + 2\ell$ \\ 
			\bottomrule
		\end{NiceTabular}
	}
	\caption{Comparison of  $\TWthis$ with ABY~\cite{NDSS:DemSchZoh15} and Turbospeedz~\cite{ACNS:BenNieOmr19} in terms of storage and communication for a single multiplication. All values are given in bits. $|\Triple|$ denotes the communication required to generate a multiplication triple. Best values for the online phase are marked in bold.} \label{tab:comp2pc}
\end{table}

For the multi-input multiplication (fan-in of $N$), the tree-based method (multiplying $N$ elements by taking two at a time) requires $\log_2(N)$ rounds for both ABY~\cite{NDSS:DemSchZoh15} and Turbospeedz~\cite{ACNS:BenNieOmr19}, while it requires communication of $4(N-1)$ ring elements for ABY and $2(N-1)$ elements for Turbospeedz in the online phase.

\subsubsection{Comparison with \cite{FC:OhaNui20}}
Recently, \cite{FC:OhaNui20} proposed round-efficient solutions for multi-input multiplication using a preprocessing for which the communication cost grows exponentially with the fan-in of the multiplication gate. However, for an $N$-input multiplication, \cite{FC:OhaNui20} requires an online communication of $2N-2$ ring elements. On the contrary, $\TWthis$ requires only an online communication of $2$ ring elements, and the preprocessing cost remains the same as that of \cite{FC:OhaNui20}. Note that since the preprocessing cost grows exponentially with the number of inputs to the multiplication gate, \cite{FC:OhaNui20} considered only up to 5-input multiplication gates in their work. 

\paragraph{$\prot{\Mult}$ when input parties are the computing parties}
For the case of a two-input multiplication gate, \cite{FC:OhaNui20} considered a special case where the input parties are the computing parties~(cf. \cite{FC:OhaNui20}, \S3.4). For this case, \cite{FC:OhaNui20} proposed a protocol for which the online communication is $2$ ring elements. For the same setting, we observe that our solution results in a protocol with zero online communication. To see this, recall the online phase of our multiplication protocol $\prot{\Mult}$~(\boxref{fig:piMult2pcS}). The modified protocol is as follows: During the online phase, party $P_i$ for $i \in \{1,2\}$ locally computes $\vl{z}_i$ such that $\vl{z}_1 + \vl{z}_2 = \vl{z}$. 
Now to generate $\shr{\vl{z}}$, parties locally set $\pad{\vl{z}}{1} = - \vl{z}_1, \pad{\vl{z}}{2} = - \vl{z}_2$ and $\mk{\vl{z}} = 0$. It is easy to see that $\vl{z} = \mk{\vl{z}} -  \pad{\vl{z}}{1} - \pad{\vl{z}}{2}$.

\section{Garbled World}
\label{sec:2pcSGCWorld}

The GC world comprises a single execution with $P_1$ acting as garbler and $P_2$ as the evaluator.

\paragraph{Input Phase}
\label{p:GCIpT2pcS}
Given that the function  input $\vl{x}$ is already available as $\shrB{\vl{x}}$, the boolean values $\av{\vl{x}} = \mk{\vl{x}} \xor \pad{\vl{x}}{1}, \pad{\vl{x}}{2}$ act as the {\em new} inputs for the garbled computation, and garbled sharing ($\shrG{\cdot}$) is generated for each of these values. The $\shrG{\cdot}$-shares thus generated defines the compound sharing, $\shrC{\vl{x}} = (\shrG{\av{\vl{x}}}, \shrG{\pad{\vl{x}}{2}})$ for every input $\vl{x}$ to the function to be evaluated via the GC. We first discuss the semantics for $\shrG{\cdot}$-sharing followed by steps for generating $\shrC{\cdot}$-sharing.

\paragraph{Garbled sharing semantics}
\label{p:GCSemT2pcS}
A value $\vl{v} \in \Z{}$  is $\shrG{\cdot}$-shared (garbled shared) amongst $\Partyset$ if $P_1$ holds $\shrG{\vl{v}}_{1}= \key{{\vl{v}}}{0}$ and $P_2$ holds $\shrG{\vl{v}}_{2} = \key{{\vl{v}}}{\vl{v}}$. Here, $\key{{\vl{v}}}{\vl{v}} = \key{{\vl{v}}}{0} \xor \vl{v} \Delta$, and $\Delta$, which is known only to the garbler $P_1$, denotes the global offset with its least significant bit set to $1$ and is same for every wire in the circuit. 
A value $\vl{x} \in \Z{}$ is said to be $\shrC{\cdot}$-shared (compound shared) if each value  from $(\av{\vl{x}}, \pad{\vl{x}}{2})$ is $\shrG{\cdot}$-shared. We write $\shrC{\vl{x}} = (\shrG{\av{\vl{x}}},\shrG{\pad{\vl{x}}{2}})$. 

\paragraph{Generation of $\shrG{\vl{v}}$ and $\shrC{\vl{x}}$} 
Protocol $\pigsh(\Partyset, \vl{v})$ enables generation of $\shrG{\vl{v}}$ given $\vl{v}$. Garbler $P_1$ generates $\{\key{{{\vl{v}}}}{\bitb}\}_{\bitb \in \{0, 1\}}$ which denotes the key for value $\bitb$ on wire $\vl{v}$, following the free-XOR technique~\cite{ICALP:KolSch08,C:KolMohRos14}. If the value $\vl{v}$ is known to $P_1$, it sends $\key{{\vl{v}}}{\vl{v}}$ to $P_2$. For the case when the evaluator $P_2$ knows $\vl{v}$, parties engage in a $\COT{1}{\kappa}$ with $P_1$ being the sender and $P_2$ being the receiver. Here $P_1$ inputs the correlation function $f_R(\vl{y}) = \vl{y} \xor \Delta$ and obtains $(\key{{\vl{v}}}{0}, \key{{\vl{v}}}{\vl{v}} = \key{{\vl{v}}}{0} \xor \Delta)$ while $P_2$ inputs $\vl{v}$ as choice bit and receives $\key{{\vl{v}}}{\vl{v}}$ as the output.
To generate $\shrC{\vl{x}}$, $\pigsh$ is invoked for each of $\av{\vl{x}}$ and $\pad{\vl{x}}{2}$.  

\paragraph{Evaluation} 
\label{p:GCEvT2pcS}
Let $f(\vl{x})$ be the function to be evaluated. At this point, the function input is $\shrC{\cdot}$-shared. This renders $\shrG{\cdot}$-sharing for the input of the GC that corresponds to the function $f'\big({\av{\vl{x}}}, {\pad{\vl{x}}{2}} \big)$ which first combines the given boolean-shares to compute the actual input and then applies $f$ on it. Let $\GC$ denotes the garbled circuit to be sent to $P_2$ by garbler $P_1$. Sending of $\GC$ is overlapped  with the key transfer (during generation of $\shrC{\vl{x}}$), to save rounds, where $P_1$ sends $\GC$ to $P_2$. On receiving the $\GC$, $P_2$ evaluate it and obtain the key corresponding to the output, say $\vl{z}$. This generates $\shrG{\vl{z}}$. 

\paragraph{Output phase} 
\label{p:GCOpT2pcS}
The goal of output computation is to compute the output $\vl{z}$ from $\shrG{\vl{z}}$. To reconstruct $\vl{z}$ towards $P_2$, $P_1$ sends the least significant bit $\vl{p}$ of $\key{\vl{z}}{0}$, referred to as the decoding information, to $P_2$. $P_2$ uses the received $\vl{p}$ to reconstruct $\vl{z}$ as $\vl{z} = \vl{p} \xor \vl{q}$, where $\vl{q}$ denotes the least significant bit of $\key{\vl{z}}{\vl{z}}$. $P_2$ then sends $\vl{z}$ to $P_1$ completing the protocol.

\section{Security proofs}
\label{sec:GCSec2pcS}
The simulation for the semi-honest 2PC case is straightforward in the $\{\FSETUP, \Func[\MultPre]\}$-hybrid model. Here $\FSETUP$~(\S\ref{sec:KeySetupprelims}) denotes the ideal functionality for the shared-key setup and $\Func[\MultPre]$ denotes the ideal functionality for the multiplication preprocessing $\prot{\MultPre}$. The strategy for simulating the computation of function $f$ (represented by a circuit $\Ckt$) is as follows. The simulation begins with the simulator emulating the shared-key setup~($\FSETUP$) functionality and giving the respective keys to the adversary $\Adv$. Since $\Sim$ is given the input and output of the $\Adv$, it can compute all the intermediate values of the circuit $\Ckt$ in clear.

For the input sharing of value $\vl{v}$, $\Sim$ receives the  $\mk{\vl{v}}$ from $\Adv$ on behalf of the honest parties. Similarly, for the inputs of honest parties, $\Sim$ interacts with the $\Adv$ with the inputs set to $0$. The simulated view is indistinguishable from the ideal view due to the privacy of the underlying sharing scheme. The linear gates involve no communication, while simulation of the multiplication protocol is straightforward. Moreover, simulation for the joint sharing~($\prot{\JSh}$) instances is similar to that of the sharing protocol. The protocol's design is such that $\Sim$ will always know the value to be sent as part of the joint sharing protocol. Finally, for the reconstruction towards $\Adv$, $\Sim$ calculates the missing share of $\Adv$ using $\vl{y}$ and the other shares. The missing share is then communicated to $\Adv$ as per the reconstruction protocol. 

\part{Layer II: Building Blocks}
\label{part:layer2}

\chapter*{Introduction to Layer II}
\label{chap:layer2_intro}
\begin{quote} \small
	In this part, we provide the details of the Layer II blocks of our three-layer architecture~(\boxref{fig:architecture}). To begin with, we provide a high-level overview of the building blocks next. Moreover, for some of the blocks, such as matrix multiplication and non-linear activation functions, the constructions are generic and instantiated with the protocols from the corresponding framework. We provide a detailed description for those blocks and omit the same from the specific chapters to avoid repetition.
\end{quote}

\paragraph{Scalar Dot Product~($\prot{\dotp}$)}
Scalar Dot Product forms the fundamental building block for most of the ML algorithms and hence designing efficient constructions for the same are of utmost importance. Given the $\shr{\cdot}$-shares of $d$-length vectors $\vct{a}, \vct{b}$, dot product protocol $\prot{\dotp}$ computes the $\shr{\cdot}$-shares of $\vl{z}$ with $\vl{z} = \vct{a} \band \vct{b} = \sum_{i = 1}^{\vl{d}} {\vl{a}}_i {\vl{b}}_i$. One trivial way is to invoke the multiplication protocol  corresponding to each of the $d$ underlying multiplications. This would result in communication linear in the vector size $d$. In this thesis, we propose methods to make the online communication independent of the vector size for all our settings. Moreover, the communication in the preprocessing phase is also made independent of the vector size for the case of three and four-party settings.

\paragraph{Matrix Operations and Convolutions}  
Linear matrix operations, such as addition of two matrices $\Mat{A}, \Mat{B}$ to generate matrix $\Mat{C} = \Mat{A} + \Mat{B}$, can be computed by extending the scalar operations (addition, in this case) with respect to each element of the matrix. Matrix multiplication, on the other hand, can be expressed as a collection of dot products, where the element in the $i^{\text{th}}$ row and $j^{\text{th}}$ column of $\Mat{C} = \Mat{A} \times \Mat{B}$, where $\Mat{A}, \Mat{B}$ are matrices of dimension ${\vl{p}} \times {\vl{q}}$, ${\vl{q}} \times {\vl{r}}$, respectively, can be computed as a dot product of the $i^{\text{th}}$ row of $\Mat{A}$ and the $j^{\text{th}}$ column of $\Mat{B}$. Thus, computing $\Mat{C}$ of dimension ${\vl{p}} \times {\vl{r}}$ requires ${\vl{p}} {\vl{r}}$ dot products on vectors of length $\vl{q}$. This improves the cost of matrix multiplication over the naive approach which requires ${\vl{p}} {\vl{q}} {\vl{r}}$ multiplications. 

We abuse notation and follow the $\shr{\cdot}$-sharing semantics for matrices. For $\Mat{X}^{u\times v}$, we have $\mk{\Mat{X}} = \Mat{X} \bigoplus \sqr{\pad{\Mat{X}}{1}} \bigoplus \sqr{\pad{\Mat{X}}{2}} \bigoplus \sqr{\pad{\Mat{X}}{3}}$ for the case of active frameworks~($\Tthis,\Fthis$) and $\mk{\Mat{X}} = \Mat{X} \bigoplus \sqr{\pad{\Mat{X}}{1}} \bigoplus \sqr{\pad{\Mat{X}}{2}}$ for the case of passive frameworks~($\TSthis,\TWthis$) . Here $\mk{\Mat{X}}$, $\sqr{\pad{\Mat{X}}{1}}$, $\sqr{\pad{\Mat{X}}{2}}$, and $\sqr{\pad{\Mat{X}}{3}}$ are matrices of dimension $u \times v$, and $\bigoplus$ denote the matrix addition operation. Looking ahead $\bigominus, \MatMul$ will  be used to denote matrix subtraction and multiplication operation, respectively. 

{\em Convolutions:} Convolutions form an important building block in several neural network architectures and can be represented as matrix multiplications, as explained in the example below. Consider a 2-dimensional convolution ($\cv$) of a $3 \times 3$ input matrix $\Mat{X}$ with a kernel $\Mat{K}$ of size $2 \times 2$. This can be represented as a matrix multiplication as follows. 

\[ 
\cv \left(
\begin{bmatrix}
	{\vl{x}}_1  &  {\vl{x}}_2  & {\vl{x}}_3 \\
	{\vl{x}}_4  &  {\vl{x}}_5  & {\vl{x}}_6 \\
	{\vl{x}}_7  &  {\vl{x}}_8  & {\vl{x}}_9      
\end{bmatrix}
,
\begin{bmatrix}
	{\vl{k}}_1  &  {\vl{k}}_2      \\
	{\vl{k}}_3  &  {\vl{k}}_4      
\end{bmatrix} 
\right)
= 
\begin{bmatrix}
	{\vl{x}}_1  &  {\vl{x}}_2  & {\vl{x}}_4 & {\vl{x}}_5 \\
	{\vl{x}}_2  &  {\vl{x}}_3  & {\vl{x}}_5 & {\vl{x}}_6 \\
	{\vl{x}}_4  &  {\vl{x}}_5  & {\vl{x}}_7 & {\vl{x}}_8 \\   
	{\vl{x}}_5  &  {\vl{x}}_6  & {\vl{x}}_8 & {\vl{x}}_9     
\end{bmatrix}
\begin{bmatrix}
	{\vl{k}}_1  \\
	{\vl{k}}_2  \\
	{\vl{k}}_3  \\
	{\vl{k}}_4      
\end{bmatrix} 
\]

Generally, convolving a $f \times f$ kernel over a $w \times h$ input with $p \times p$ padding using $s \times s$ stride having $i$ input channels and $o$ output channels, is equivalent to performing a matrix multiplication on matrices of dimension $(w^{\prime} \cdot h^{\prime}) \times (i \cdot f \cdot f)$ and $(i \cdot f \cdot f) \times (o)$ where $w^{\prime} = \dfrac{w-f+2p}{s} + 1$ and $h^{\prime} = \dfrac{h-f+2p}{s} + 1$. We refer readers to \cite{PoPETS:WagGupCha19,ConvStanford} for more details.

\paragraph{Secure Comparison~($\prot{\bitext}$)}
Comparing two arithmetic values is one of the major hurdles in realizing efficient secure ML algorithms. Given arithmetic shares $\shr{\vl{a}}, \shr{\vl{b}}$, parties wish to check whether $\vl{a} > \vl{b}$. To compute $\vl{a} > \vl{b}$ in the FPA representation, given its $\shr{\cdot}$-sharing, $\prot{\bitext}$ uses the technique of extracting the most significant bit~($\msb$) of the value $\vl{v} = \vl{a} - \vl{b}$~\cite{CCS:MohRin18, NDSS:PatSur20, USENIX:KPPS21}.  

To compute the $\msb$, we use two variants - i) the communication optimized parallel prefix adder~(PPA) circuit from ABY3~\cite{CCS:MohRin18} ($2 (\ell - 1)$ AND gates, $\log \ell$ depth), and ii) the round optimized bit extraction circuit from ABY2~\cite{USENIX:PSSY21}. The circuit of ABY2 uses multi-input AND gates and has a multiplicative depth of $\log_4(\ell)$. These circuits take two $\ell$-bit values in boolean sharing as the input and output the result in boolean sharing form. 

\paragraph{Bit to Arithmetic~($\prot{\bitA}$) / Bit Injection~($\prot{\bitinj}$)}
The bit to arithmetic protocol, $\prot{\bitA}$, enables computing the arithmetic sharing~($\shr{\cdot}$) of a bit $\bitb$ given its boolean sharing $\shrB{\bitb}$. Let $\arval{\bitb}$ denotes the value of $\bitb \in \bitset$ over the arithmetic ring $\Z{\ell}$. Then for $\bitb = \bitb_1 \xor \bitb_2$, note that $\arval{\bitb} = (\arval{\bitb_1} - \arval{\bitb_2})^2$. Similarly, $\prot{\dbitA}$ protocol computes the arithmetic sharing of $\bitb_1 \bitb_2$ given the boolean sharings $\shrB{\bitb_1}$ and $\shrB{\bitb_2}$.

Given the boolean sharing of bit $\bitb$ and the arithemetic sharing of a value $\vl{v}$, the bit injection protocol, $\prot{\bitinj}$, enables computing the arithmetic sharing corresponding to the value $\bitb \vl{v}$. Similarly, $\prot{\dbitinj}$ computes the arithmetic sharing of $\bitb_1 \bitb_2 \vl{v}$ given $\shrB{\bitb_1}$, $\shrB{\bitb_2}$ and $\shr{\vl{v}}$.

\paragraph{Equality Test~($\prot{\eql}$)}
Given $\shr{\vl{a}}, \shr{\vl{b}}$, the goal of the Equality Testing ($\prot{\eql}$) protocol is to check whether $\vl{a} \iseq \vl{b}$ or not. An equivalent formulation of the problem \cite{BogdanovNTW12, FC:OhaNui20} is to check if all the bits of $\vl{a} - \vl{b}$ are $0$ or not. This simple primitive is crucial in building efficient protocol for applications like Circuit-based Private Set Intersection~\cite{EC:PSTY19,USENIX:PSSZ15,EC:PSWW18}, the Table Lookup Protocol from~\cite{NDSS:DKSSZZ17}, and Data Mining~\cite{BogdanovNTW12}. 

On a high level, the protocol starts with the parties computing the boolean shares of two value $\vl{v}_1, \vl{v}_2$ using the $\shr{\cdot}$-shares of $\vl{a}$ and $\vl{b}$. The values $\vl{v}_1, \vl{v}_2$ are computed such that $\vl{v}_1 = \vl{v}_2$ implies $\vl{a} = \vl{b}$. For instance, in $\TSthis$, parties set $\vl{v}_1 = (\mk{\vl{a}} - \pad{\vl{a}}{1}) - (\mk{\vl{b}} - \pad{\vl{b}}{1})$ and $\vl{v}_2 = \pad{\vl{a}}{2} - \pad{\vl{b}}{2}$. Note that the value $\vl{v}_i$ can be locally computed by party $P_i$ for $i \in \{1,2\}$ and hence can generate the boolean shares. 

The parties then locally compute the boolean shares of $\vl{v} = \vl{v}_1 \xor \vl{v}_2$. If $\vl{v}_1 = \vl{v}_2$, then all the bits of $\vl{v}$ should be $0$. Or in other words, all the bits of $\overline{\vl{v}}$ should be $1$. This can be checked by computing an AND of all the bits of $\overline{\vl{v}}$. For this, the parties use 4-input AND gates and a tree structure, where $4$ bits are taken at a time and the AND of them is computed in one go. This approach improves the round complexity by a factor of two~($\log_2(\ell)$ to $\log_4(\ell)$ for $\ell$-bit inputs) over the traditional approach using 2-input AND gates.
Parties can use the $\prot{\bitA}$ protocol to generate the arithmetic equivalent of the result in shared form.

\paragraph{Piecewise-polynomial functions}
Piece-wise polynomial functions  are constructed as a series of constant polynomials $f_1, \ldots, f_m$ with public coefficients and $ c_1 < \ldots < c_m$ such that, 
\begin{align*}
	f(y) = 
	\begin{cases}
		0, & y< c_1 \\
		f_1, & c_1 \leq y < c_2 \\
		\ldots & \\
		f_m, & c_m \leq y
	\end{cases}
\end{align*}
For computing $f$, we first compute a set of bits $\bitb_1, \ldots, \bitb_m$ such that $\bitb_i = 1$ if $y  \geq c_i$ and $0$ otherwise. $f$ can be computed as, $f(y) = \sum_{i=1}^{m} \bitb_i \cdot (f_{i} - f_{i-1})$, where $f_0 = 0$ and $f_m = 1$. Given the arithmetic shares~($\shr{\cdot}$) of $y$, one can obtain the boolean shares~($\shrB{\cdot}$) of the bits $\bitb_1, \ldots, \bitb_m$ using secure comparison. The bit injection protocol is then used to compute the $\shr{\cdot}$-shares of $\bitb_i \cdot (f_{i} - f_{i-1})$. Note that $f(y)$ can be viewed as a sum of $m$ bit injections, and parties can add up the shares locally to obtain the final result. In $\prot{\piecewise}$, we optimize the communication further and show how to make the online communication independent of $m$.

\paragraph{Non-Linear Activation functions}
We use the following three widely used activation functions -- (i) Rectified Linear Unit~($\relu$), (ii) Sigmoid~($\sig$), and (iii) Softmax~($\sftmx$). 

{\em (i) ReLU~($\relu$):}
The $\relu$ function, $\relu(\vl{v}) = \max(0, \vl{v})$, can be written as
\begin{align*}
	\relu(\vl{v}) = 
	\begin{cases}
		0, & \vl{v} <0 \\
		\vl{v} & 0 \leq \vl{v}
	\end{cases}
\end{align*}
Thus, it can be viewed as $\relu(\vl{v}) = \bar{\bitb} \cdot \vl{v}$, where bit $\bitb = 1$ if $\vl{v} < 0$ and $0$ otherwise. Here $\bar{\bitb}$ denotes the complement of $\bitb$. 
Given $\shr{\vl{v}}$, parties first extract the sign of $\vl{v}$ using the bit extraction protocol $\prot{\bitext}$. The desired result can then be obtained using an invocation of the bit injection protocol $\prot{\bitinj}$.

{\em (ii) Sigmoid~($\sig$):}
The sigmoid function on value $\vl{v}$ is given as $\ln(\frac{1}{1 + e^{-v}})$. However, computing the exact function is expensive in MPC and hence, we use the following MPC-friendly variant of the Sigmoid function~\cite{SP:MohZha17,CCS:MohRin18}:
\begin{align*}
	\sig(\vl{v}) = \left\{
	\begin{array}{lll}
		0                  & \quad \vl{v} < -\frac{1}{2} \\
		\vl{v} + \frac{1}{2} & \quad - \frac{1}{2} \leq \vl{v} \leq \frac{1}{2} \\
		1                  & \quad \vl{v} > \frac{1}{2}
	\end{array}
	\right.
\end{align*}
Thus, $\sig(\vl{v}) = 1 - \bitb_1 \left( \vl{v} + \frac{1}{2} \right) + \bitb_2 \left( \vl{v} - \frac{1}{2} \right)$, where $\bitb_1 = 1$ if $\vl{v} < - \frac{1}{2}$ and $0$ otherwise, and $\bitb_2 = 1$ if $\vl{v} < \frac{1}{2}$ and $0$ otherwise. Note that this can be viewed as an instance of a piecewise polynomial function.

{\em (iii) Softmax~($\sftmx$):}
Given a set of values, the softmax function is used to compute a probability distribution among the values such that each output is between 0 and 1, and all the outputs sum up to 1. This function is used at the output layer of the neural networks in Layer III of our architecture.
For a set of $d$ values, $\vl{v}_1, \dots, \vl{v}_d$, the softmax on the $i$th value $\vl{v}_i$ is given as $\frac{e^{-\vl{v}_i}}{\sum_{j=1}^{d} e^{-\vl{v}_j}}$. Since the actual function is not MPC-friendly, we use the approximate variant of the same proposed by SecureML~\cite{SP:MohZha17} and is defined as $\sftmx(\vl{v}_i) = \frac{\relu(\vl{v}_i)}{\sum_{j=1}^{d} \relu(\vl{v}_j)}$. In order to perform the division, we switch from arithmetic to garbled world and then use a division garbled circuit.

\paragraph{Oblivious Selection}
Given $\shr{\cdot}$-shares of $\vl{x}_0, \vl{x}_1 \in \Z{\ell}$ and $\shrB{\bitb}$ where $\bitb \in \bitset$, oblivious selection~($\prot{\obv}$) enables parties to generate re-randomized $\shr{\cdot}$-shares of $\vl{z} = \vl{x}_{\bitb}$. The protocol is similar in spirit to the Oblivious Transfer primitive. Note that $\vl{z}$ can be written as $\vl{z} = \bitb (\vl{x}_1 - \vl{x}_0) + \vl{x}_0$. To compute $\shr{\cdot}$-sharing of $ \bitb (\vl{x}_1 - \vl{x}_0)$, parties use an instance of piecewise polynomial protocol~$\prot{\piecewise}$ with $m=1$. The $\shr{\cdot}$-share of $\vl{z}$ can then be obtained by adding the output of $\prot{\piecewise}$ with $\shr{\vl{x}_0}$.

\paragraph{Maximum / Minimum among two and three values}
\label{sec:MinMax}
The $\prot{\maxTW}$ protocol is used to compute the maximum among two values $\vl{x}_1, \vl{x}_2$ in a secure manner given $\shr{\vl{x}_1}$ and $\shr{\vl{x}_2}$. For this, the parties execute the secure comparison protocol on $\shr{\vl{x}_1}, \shr{\vl{x}_2}$ to obtain $\shrB{\bitb} = \shrB{ \vl{x}_1 > \vl{x}_2}$. Note that $\prot{\maxTW}(\vl{x}_1, \vl{x}_2) = \bitb \cdot(\vl{x}_1 - \vl{x}_2) + \vl{x}_2$ and can be computed using an instance of oblivious selection protocol $\prot{\obv}$. The $\prot{\minTW}$ protocol proceeds similarly except that $\prot{\minTW}(\vl{x}_1, \vl{x}_2) = \bitb \cdot(\vl{x}_2 - \vl{x}_1) + \vl{x}_1$.

Given $\shr{\cdot}$-shares of  $\vl{x}_1, \vl{x}_2, \vl{x}_3$, the goal of the $\prot{\maxT}$ protocol is to find the maximum value among the three. For this, first securely compare the pairs $(\vl{x}_1, \vl{x}_2), (\vl{x}_1, \vl{x}_3)$ and $(\vl{x}_2, \vl{x}_3)$ using the secure comparison protocol and obtain  $\shrB{\bitb_1}, \shrB{\bitb_2}$ and $\shrB{\bitb_3}$ respectively. Here $\bitb_1 = 1$ if $\vl{x}_1 > \vl{x}_2$ and $0$ otherwise. $\bitb_2$ and $\bitb_3$  are defined likewise . Now the maximum among the three, denoted by $\vl{y}$, can be written as $\vl{y} = \bitb_1\cdot\bitb_2\cdot\vl{x}_1 + \overline{\bitb_1}\cdot\bitb_3\cdot\vl{x}_2 + \overline{\bitb_2}\cdot\overline{\bitb_3}\cdot\vl{x}_3$. To compute this, parties can use $\prot{\dbitinj}$ to obtain each term in the expression for $\vl{y}$ and can locally add them to obtain the desired result. As an optimization, we can combine the communication in the online phase corresponding to all three executions of the $\prot{\dbitinj}$ protocol into one. The protocol for $\prot{\minT}$, which computes the minimum among the three values can be obtained by slightly modifying the protocol for $\prot{\maxT}$. The difference lies in the expression for computing the minimum which will now be $\vl{y} = \overline{\bitb_1}\cdot\overline{\bitb_2}\cdot\vl{x}_1 + \bitb_1\cdot\overline{\bitb_3}\cdot\vl{x}_2 + \bitb_2\cdot\bitb_3\cdot\vl{x}_3$.

\paragraph{ArgMin/ ArgMax}
\label{sec:ArgMinMax}
Protocol $\prot{\agmin}$ (\boxref{fig:piargmin}) allows parties to compute the index of the smallest element in a vector $\vct{x} = (\vl{x}_1, \ldots, \vl{x}_m)$ of $m$ elements, where $\vct{x}$ is $\shr{\cdot}$-shared, i.e. each element $\vl{x}_i \in \Z{\ell}$ of $\vct{x}$ is $\shr{\cdot}$-shared. The protocol outputs a $\shrB{\cdot}$-shared bit vector $\vct{b}$ of size $m$ which has a $1$ at the index associated with the minimum value in $\vct{x}$, and $0$ elsewhere. 
We follow the standard tree-based approach~\cite{SP:DEFKSV19} to recursively find the minimum value in $\vct{x}$ while also updating $\vct{b}$ to reflect the index of this smallest element. Each bit of $\vct{b}$ is initialized to 1. The elements of $\vct{x}$ are grouped into pairs and securely compared to find their pairwise minimum. Using this information, $\vct{b}$ is updated such that $\bitb_j$'s are reset to $0$ for $\vl{x}_j$'s $\in \vct{x}$ which do not form the minimum in their respective pair; the other bits in $\vct{b}$ still equal $1$. The protocol recurses on the remaining elements $\vl{x}_j \in \vct{x}$, which were the pairwise minimums. Eventually, only one $\bitb_j \in \vct{b}$ equals $1$, indicating that $\vl{x}_j$ is the minimum, with index $j$. Computing $\prot{\agmax}$ can be done similarly.
The formal protocol appears in \boxref{fig:piargmin}. 

\begin{protocolbox}{$\prot{\agmin} (\shr{\vct{x}})$}{Protocol to find index of smallest element in $\vct{x}$}{fig:piargmin}
	\justify
	Let $\vct{b}$ be the bit vector of size $m$, where $m$ equals the size of $\vct{x}$. Parties execute the following steps in the respective preprocessing and online phases. 
	\begin{enumerate} 
		\item If $m = 2$, do the following.
		\begin{enumerate}
			\item  $\shrB{\vl{d}_1} = \prot{\bitext}(\shr{\vl{x}_1}, \shr{\vl{x}_2})$~~;~~
			$\shrB{\vl{d}_2} = 1 \xor \shrB{\vl{d}_1}$~~;~~
			$\shr{\vl{y}} = \prot{\obv}(\shr{\vl{x}_2}, \shr{\vl{x}_1}, \shrB{\vl{d}_1})$.
			\item  Return $(\shrB{\vl{d}_1}, \shrB{\vl{d}_2}, \shr{\vl{y}})$.
		\end{enumerate} 
		\item Else, if $m = 3$, do the following
		\begin{enumerate}
			\item  $\shrB{\vl{d}_1^{\prime}} = \prot{\bitext}(\shr{\vl{x}_1}, \shr{\vl{x}_2})$~~;~~
			$\shr{\vl{y}^{\prime}} = \prot{\obv}(\shr{\vl{x}_2}, \shr{\vl{x}_1}, \shrB{\vl{d}_1^{\prime}})$.
			\item  $\shrB{\vl{d}_2^{\prime}} = \prot{\bitext}(\shr{\vl{y}^{\prime}}, \shr{\vl{x}_3})$~~;~~
			$\shr{\vl{y}} = \prot{\obv}(\shr{\vl{x}_3}, \shr{\vl{y}^{\prime}},\shrB{\vl{d}_2^{\prime}})$.
			\item  $\shrB{\vl{d}_1} = \prot{\Mult}(\shrB{\vl{d}_1^{\prime}}, \shrB{\vl{d}_2^{\prime}})$~~;~~
			$\shrB{\vl{d}_2} = \shrB{\vl{d}_2^{\prime}} \xor \shrB{\vl{d}_1}$~~;~~
			$\shrB{\vl{d}_3} = 1 \xor \shrB{\vl{d}_1^{\prime}} \xor \shrB{\vl{d}_2^{\prime}}$. 
			\item  Return $(\shrB{\vl{d}_1}, \shrB{\vl{d}_2}, \shrB{\vl{d}_3}, \shr{\vl{y}})$.
		\end{enumerate}
		\item Else, let $\vct{x_1} = (\vl{x}_1, \ldots, \vl{x}_{\lfloor m/2 \rfloor})$ and $\vct{x_2} = (\vl{x}_{\lfloor m/2 \rfloor + 1}, \ldots, \vl{x}_m)$.
		\begin{enumerate}
			\item   $\big(\shrB{\vl{d}_1}, \ldots, \shrB{\vl{d}_{\lfloor m/2 \rfloor}}, \shr{\vl{y}_1} \big) = \prot{\agmin}(\shr{\vct{x_1}})$.
			\item   $\big(\shrB{\vl{d}_{\lfloor m/2 \rfloor + 1}}, \ldots, \shrB{\vl{d}_m}, \shr{\vl{y}_2} \big) = \prot{\agmin}(\shr{\vct{x_2}})$.
			\item   $\shrB{\vl{d}} = \prot{\bitext}(\shr{\vl{y}_1}, \shr{\vl{y}_2})$~~;~~
			$\shr{\vl{y}} = \prot{\obv}(\shr{\vl{y}_2}, \shr{\vl{y}_1}, \shrB{\vl{d}})$.
			\item  $\shrB{\bitb_j} = \prot{\Mult}(\shrB{\vl{d}}, \shrB{\vl{d}_j})$ for $j \in \{1, \ldots, \lfloor m/2 \rfloor \}$.
			\item  $\shrB{\bitb_j} = \prot{\Mult}(1 \xor \shrB{\vl{d}}, \shrB{\vl{d}_j})$ for $j \in \{ \lfloor m/2 \rfloor + 1, \ldots, m \}$.
			\item  Return $\big( \shrB{\bitb_1},  \ldots, \shrB{\bitb_m}, \shr{\vl{y}} \big)$.
		\end{enumerate}
	\end{enumerate}		
\end{protocolbox}

To begin with, parties initialize $\bitb_j = 1$ for $\bitb_j \in \vct{b}$ by locally setting $\mk{{\bitb}_j} = 1$ and $\pad{{\bitb}_j}{1} = \pad{{\bitb}_j}{2} = \pad{{\bitb}_j}{3} = 0$. The minimum, $\vl{y}_{ij}$, of two elements, $\vl{x}_i,\vl{x}_j$ can be computed as: one invocation of bit extraction protocol to obtain $\shrB{\cdot}$-sharing of $\bitb_{ij}$, where $\bitb_{ij} = 1$ if $\vl{x}_i < \vl{x}_j$, and $\bitb_{ij} = 0$ otherwise; one invocation of oblivious selection protocol $\prot{\obv}(\vl{x}_j, \vl{x}_i, \bitb_{ij})$, which outputs $\shr{\cdot}$-shares of $\vl{y}_{ij} = \vl{x}_j$ if $\bitb_{ij} = 0$, and $\vl{y}_{ij} = \vl{x}_i$, otherwise. 
To update $\vct{b}$ to reflect the pairwise minimums, we view the elements $\vl{x}_j \in \vct{x}$ as the leaves of a binary tree, in a bottom-up manner. For two elements in a pair, say $(\vl{x}_i, \vl{x}_j)$, whose pairwise minimum is $\vl{y}_{ij}$, we let $\vl{y}_{ij}$ be the root node with $\vl{x}_i$ as its left child and $\vl{x}_j$ as its right child. Now, to update $\vct{b}$, parties multiply $\bitb_{ij}$ with the bits in $\vct{b}$ associated with the {\em left-reachable leaf nodes}, which comprise of all the leaf nodes (elements of $\vct{x}$) that are reachable through the left child of the root. 
Similarly, parties multiply $1 \xor \bitb_{ij}$ with the bits in $\vct{b}$ associated with the {\em right-reachable leaf nodes}, which comprise of all the leaf nodes (elements of $\vct{x}$) that are reachable through the right child of the root. Thus, if $\bitb_{ij} = 1$ indicating that $\vl{x}_i < \vl{x}_j$, $\bitb_i$ remains $1$ as it gets multiplied by $\bitb_{ij} = 1$ while $\bitb_j$ gets reset to $0$ as it gets multiplied by $1 \xor \bitb_{ij} = 0$. The case for $\bitb_{ij} = 0$ holds for similar reasons.
Given the values $\vl{y}_{ij}$ for the next level, and the updated $\vct{b}$, the steps are applied recursively until the minimum element is obtained. 

The protocol $\prot{\agmax}$ which allows the parties to compute the index of the largest element in a $\shr{\cdot}$-shared vector $\vct{x} = (\vl{x}_1, \ldots, \vl{x}_m)$, is similar to $\prot{\agmin}$ with the following difference. To find the maximum among two elements $(\shr{\vl{x}_i}, \shr{\vl{x}_j})$, parties run the bit extraction protocol to obtain $\shrB{\bitb_{ij}}$ as before, followed by $\prot{\obv}(\vl{x}_i, \vl{x}_j, \bitb_{ij})$, which outputs $\shr{\cdot}$-shares of $\vl{y}_{ij} = \vl{x}_i$ if $\bitb_{ij} = 0$, and $\vl{y}_{ij} = \vl{x}_j$, otherwise. Now, $\vct{b}$ is updated in each level by multiplying $1 \xor \bitb_{ij}$ with the bits in $\vct{b}$ associated with the {\em left-reachable leaf nodes} (as described before) and multiplying $\bitb_{ij}$ with the bits in $\vct{b}$ associated with the {\em right-reachable leaf nodes}.

\paragraph{Mixed-world Conversions}
The protocols for mixed-world conversions enable efficient transitions among the arithmetic, boolean, and garbled worlds. The efficiency lift of our framework  compared to existing frameworks stands on the following useful observation--  a large portion of computation in most of the MPC-based PPML framework is done over the arithmetic and boolean world; they switch to the garbled world to perform the non-linear operations~(e.g. softmax) that are expensive in the arithmetic/boolean world and switch back to the arithmetic/boolean world immediately after.  We leverage this phenomenon to construct {\em end-to-end} conversion techniques such as Arithmetic-Garbled-Arithmetic. The standard approach until now was to perform a {\em piece-wise} combination of {\em Arithmetic to Garbled} followed by a {\em Garbled to Arithmetic} conversion. End-to-end conversions benefit from not having to generate a full-fledged garbled-shared output after the computation. Instead, these conversions aim to produce a "partial" garbled-shared output that is enough to lead to an arithmetic sharing of the output. This results in end-to-end conversions of the form "x-Garbled-x" where x can be either arithmetic or boolean that need just a single round for our garbled world as opposed to the two in the existing works~\cite{CCS:MohRin18, NDSS:ChaRacSur20}.

\chapter{$\TSthis$: Semi-honest Blocks}
\label{chap:layer2_3pcsemi}
This chapter provides details for the Layer II blocks of our 2PC framework $\TSthis$. Details for the Layer I blocks are provided in chapter~\ref{chap:layer1_3pcsemi}.

\section{Building Blocks}
\label{sec:layerIIblocks3pcS}

\subsection{Dot Product~(Scalar Product)}
\label{sec:3pcSDotp}
Given $\shr{\vct{a}}, \shr{\vct{b}}$ with $|\vct{a}| = |\vct{b}| = \vl{d}$, protocol $\prot{\dotp}$~(\boxref{fig:piDotP3pcS}) computes $\shr{\vl{z}}$ such that $\vl{z} = (\vct{a} \band \vct{b})^{\vl{t}}$ if truncation is enabled, else $\vl{z} = \vct{a} \band \vct{b}$. For this, we combine the partial products from the multiplication protocol across $\vl{d}$ multiplications and communicate them in a single shot. This makes the communication cost of the dot product independent of the vector size.

\begin{protocolsplitbox}{$\prot{\dotp}(\vct{a}, \vct{b}, \isTr)$}{Dot Product with / without Truncation in $\TSthis$.}{fig:piDotP3pcS}
	$\isTr$ is a bit denoting whether truncation is required ($\isTr =1$) or not ($\isTr=0$). \\
	\detail{
		{\bf Input(s):} $\shr{\vct{a}}, \shr{\vct{b}}$.\\
		{\bf Output:} $\shr{\vl{o}}$ where $\vl{o} = \vl{z}^{\vl{t}}$ if $\isTr = 1$ and $\vl{o} = \vl{z}$ if $\isTr = 0$ and $\vl{z} = \vct{a} \band \vct{b} = \sum_{i = 1}^{\vl{d}} {\vl{a}}_i {\vl{b}}_i$.
	}
	\justify 
	\vspace{-2mm}
	\algoHead{Preprocessing:}  Let $\gm{\vct{a}\vct{b}}{} = \sum_{i=1}^{\vl{d}} \gm{\vl{a}_i \vl{b}_i}{}$.
	\begin{enumerate}
		\item $P_0, P_j$ sample ${\vl{u}}^j \in_R \Z{\ell}$ for $j \in \{1,2\}$. Let ${\vl{u}^1} + \vl{u}^2 = \gm{\vct{a}\vct{b}}{}  - \vl{r}$ for $\vl{r} \in_R \Z{\ell}$.  
		\item Party $P_0$: Computes $\vl{r} = \gm{\vct{a}\vct{b}}{} - {\vl{u}^1} - \vl{u}^2$. If $\isTr = 1$, sets $\vl{q} = \vl{r}^{\vl{t}}$, else $\vl{q} = \vl{r}$.\newline Executes $\prot{\Sh}(P_0, \vl{q})$ to generate $\shr{\vl{q}}$.
	\end{enumerate}
	\justify
	\vspace{-2mm}
	\algoHead{Online:} Let $\vl{y} = (\vl{z} - \vl{r}) - \sum_{i=1}^{\vl{d}} \mk{\vl{a}_i \vl{b}_i}$.
	\begin{enumerate}
		\item Compute: $P_1: \vl{y}_1 = \sum_{i=1}^{\vl{d}} ( - \pad{\vl{a}_i}{1} \mk{\vl{b}_i} - \pad{\vl{b}_i}{1} \mk{\vl{a}_i}) + {\vl{u}}^1,~~
		P_2: \vl{y}_2 = \sum_{i=1}^{\vl{d}} (- \pad{\vl{a}_i}{2} \mk{\vl{b}_i} - \pad{\vl{b}_i}{2} \mk{\vl{a}_i}) + {\vl{u}}^2$
		\item $P_1$ sends $\vl{y}_1$ to $P_2$, while $P_2$ sends $\vl{y}_2$ to $P_1$, and they locally compute $\vl{z} - \vl{r} = \vl{y}_1 + \vl{y}_2 + \sum_{i=1}^{\vl{d}} \mk{\vl{a}_i \vl{b}_i}$.
		\item $P_1, P_2$: If $\isTr = 1$, set $\vl{p} = (\vl{z} - \vl{r})^{\vl{t}}$, else $\vl{p} = \vl{z} - \vl{r}$. Execute $\prot{\JSh}(P_1, P_2, \vl{p})$ to generate $\shr{\vl{p}}$. 
		\item Compute $\shr{\vl{o}} = \shr{\vl{p}} + \shr{\vl{q}}$. Here $\vl{o} = \vl{z}^{\vl{t}}$ if $\isTr = 1$ and $\vl{z}$ otherwise.
	\end{enumerate}     
\end{protocolsplitbox}

\begin{lemma}[Communication]
	\label{lemma:3pcSpidotpf}
	Protocol $\prot{\dotp}$~(\boxref{fig:piDotP3pcS})~(in $\TSthis$) requires $\ell$ bits of communication in preprocessing, and $1$ round and $2 \ell$ bits of communication in the online phase.
\end{lemma}

 \subsection{Bit Extraction}
 \label{sec:3pcSBitExt}

To compute most significant bit~($\msb$) of the value $\vl{v}$, note that $\vl{v} = \mk{\vl{v}} + (- \pad{\vl{v}}{})$ as per the sharing semantics~(cf.~\tabref{3pcSsharing}). $P_0$ generates the boolean sharing of $- \pad{\vl{v}}{}$ during the preprocessing, while $P_1, P_2$ generate $\shrB{\mk{\vl{v}}}$ during the online phase using joint sharing protocol. Parties compute the result by evaluating the bit extraction circuit~\cite{CCS:MohRin18, USENIX:PSSY21}.

 \subsection{Bit to Arithmetic}
 \label{sec:3pcSBit2A}
Protocol $\prot{\bitA}(\shrB{\bitb})$~(\boxref{fig:3pcSpiBitA}) enables computing $\shr{\bitb}$ of a bit $\bitb$ given its boolean sharing $\shrB{\bitb}$. Let $\arval{\bitb}$ denotes the value of $\bitb \in \bitset$ over the arithmetic ring $\Z{\ell}$. Using our sharing semantics, 
\begin{equation}
	\label{eq:3pcSbitA}
	\arval{\bitb} = \arval{(\mk{\bitb} \xor \pad{\bitb}{})} = \arval{\mk{\bitb}} + \padR{\bitb}(1-2\arval{\mk{\bitb}})
\end{equation} 

\begin{protocolbox}{$\prot{\bitA}(\shrB{\bitb})$}{Bit to Arithmetic conversion in $\TSthis$.}{fig:3pcSpiBitA}
	\detail{
		{\bf Input(s):} $\shrB{\bitb}$,~~~{\bf Output:} $\shr{\vl{y}} = \shr{\arval{\bitb}}$.
	}
	\justify 
	\vspace{-2mm}
	\algoHead{Preprocessing:} $P_0, P_1$ sample random $\sqr{\padR{\bitb}}_1 \in \Z{\ell}$. $P_0$ sends $\sqr{\padR{\bitb}}_2 = \padR{\bitb} - \sqr{\padR{\bitb}}_1$ to $P_2$.
	\justify
	\vspace{-2mm}
	\algoHead{Online:} 
	\begin{enumerate} 
		\item Locally compute: 
		$P_1: \vl{y}_1  =  \arval{\mk{\bitb}} + \sqr{\padR{\bitb}}_1 (1 - 2\arval{\mk{\bitb}})~~\Big|~~
		P_2:  \vl{y}_2  =  \sqr{\padR{\bitb}}_2 (1 - 2\arval{\mk{\bitb}})$
		\item $P_i$ for $i \in \{1,2\}$ executes $\prot{\Sh}$ on $\vl{y}_i$ to generate the respective $\shr{\cdot}$-shares.
		\item Compute $\shr{\vl{y}} = \shr{\vl{y}_1} + \shr{\vl{y}_2}$.
	\end{enumerate}     
\end{protocolbox}

During the preprocessing, $P_0$ generates $\sqr{\cdot}$-sharing of $\padR{\bitb}$. The online phase consists of each $P_1$ and $P_2$ locally computing an additive sharing of $\arval{\bitb}$, generating the corresponding $\shr{\cdot}$-sharing using $\prot{\Sh}$, and locally adding the shares to obtain $\shr{\bitb}$.

\begin{lemma}[Communication]
	\label{lemma:3pcSpibitA}
	Protocol $\prot{\bitA}$~(\boxref{fig:3pcSpiBitA}) requires $\ell$ bits of communication in preprocessing, and $1$ round and $2 \ell$ bits of communication in the online phase.
\end{lemma}
\begin{proof}
	During preprocessing, generation of $\sqr{\padR{\bitb}}$ involves communication of $\ell$ bits from $P_0$ to $P_2$. The online phase involves two instances of arithmetic sharing protocol  in parallel, resulting in $1$ round and a communication of $2\ell$ bits. 
\end{proof}

\subsubsection{Bit to Arithmetic:II}
\label{sec:3pcSdBit2A}
Similar to $\prot{\bitA}$ protocol, given the boolean sharings $\shrB{\bitb_1}, \shrB{\bitb_2}$, protocol $\prot{\dbitA}$ computes the arithmetic sharing of $\arval{(\bitb_1 \bitb_2)}$. Let $\Delta_{\bitb_1}$, $\Delta_{\bitb_2}$  denote the value $(1-2\arval{\mk{\bitb_1}})$, $(1-2\arval{\mk{\bitb_2}})$ respectively. Using \eqref{eq:3pcSbitA}, we can write

\begin{align}
	\label{eq:3pcSdbitA}
	\arval{(\bitb_1 \bitb_2)} &= \arval{(\mk{\bitb_1} \xor \pad{\bitb_1}{})} \arval{(\mk{\bitb_2} \xor \pad{\bitb_2}{})} 
	= (\arval{\mk{\bitb_1}} + \padR{\bitb_1}\Delta_{\bitb_1}) (\arval{\mk{\bitb_2}} + \padR{\bitb_2}\Delta_{\bitb_2}) \nonumber \\
    &= \arval{\mk{\bitb_1}}\arval{\mk{\bitb_2}} + \padR{\bitb_1}\arval{\mk{\bitb_2}}\Delta_{\bitb_1} + \padR{\bitb_2}\arval{\mk{\bitb_1}}\Delta_{\bitb_2} + \arval{(\pad{\bitb_1}{}\pad{\bitb_2}{})}\Delta_{\bitb_1}\Delta_{\bitb_2}
\end{align} 

During preprocessing, the $\sqr{\cdot}$-shares of $\padR{\bitb_1}$, $\padR{\bitb_2}$ and $ \arval{(\pad{\bitb_1}{}\pad{\bitb_2}{})}$ are computed similar to that of $\prot{\bitA}$~(\boxref{fig:3pcSpiBitA}). The online phase is similar to that of $\prot{\bitA}$ protocol.

\begin{lemma}[Communication]
	\label{lemma:3pcSpidbitA}
	Protocol $\prot{\dbitA}$ requires $3\ell$ bits of communication in preprocessing, and $1$ round and $2 \ell$ bits of communication in the online phase.
\end{lemma}

\subsection{Bit Injection}
\label{sec:3pcSBitInj}
Given the boolean sharing of a bit $\bitb$, denoted as $\shrB{\bitb}$, and the arithmetic sharing of $\vl{v} \in \Z{\ell}$, protocol $\prot{\bitinj}$ computes $\shr{\cdot}$-sharing of $\arval{\bitb}\vl{v}$. Let $\Delta_{\bitb}$ denote the value $(1-2\arval{\mk{\bitb}})$.
Similar to $\prot{\bitA}$, 
\begin{align}\label{eq:3pcSbitinj}
	\arval{\bitb} \vl{v} &= \arval{(\mk{\bitb} \xor \pad{\bitb}{})}(\mk{\vl{v}} - \pd{\vl{v}}{}) = (\arval{\mk{\bitb}} + \padR{\bitb}\Delta_{\bitb})(\mk{\vl{v}} - \pd{\vl{v}}{}) \nonumber \\
	&= \arval{\mk{\bitb}}\mk{\vl{v}} - \arval{\mk{\bitb}}\pd{\vl{v}}{} + \padR{\bitb}\mk{\vl{v}}\Delta_{\bitb} - \padR{\bitb}\pd{\vl{v}}{}\Delta_{\bitb}
\end{align} 

During the preprocessing, $P_0$ generates the $\sqr{\cdot}$-shares of $\padR{\bitb}$ and $\padR{\bitb}\pd{\vl{v}}{}$ similar to $\prot{\bitA}$ protocol. During the online phase, $P_1$ and $P_2$ compute an additive sharing of $\arval{\bitb} \vl{v}$ and execute $\prot{\Sh}$ on them to generate the respctive $\shr{\cdot}$-shares. 

\begin{lemma}[Communication]
	\label{lemma:3pcSbitinj}
	Protocol $\prot{\bitinj}$ requires $2\ell$ bits of communication in preprocessing, and $1$ round and $2 \ell$ bits of communication in the online phase.
\end{lemma}

 \subsubsection{Sum of Bit Injections}
 \label{sec:3pcSSumBitInj}
Given $m$ pair of values in the shared form, $\{\shrB{\bitb_i}, \shr{\vl{v}_i}\}_{i \in [m]}$, the goal of $\prot{\bitinjS}$ is to compute the $\shr{\cdot}$-share of $\vl{z} = \sum_{i=1}^{m} \arval{\bitb_i} \cdot \vl{v_i}$. For this, parties execute the preprocessing corresponding to $m$ bit injections of the form $\arval{\bitb_i} \cdot \vl{v_i}$. 

In the online phase, each of $P_1$ and $P_2$ locally compute an additive sharing of $\vl{z}_i$, corresponding to $\arval{\bitb_i} \cdot \vl{v_i}$ first. Instead of generating the $\shr{\cdot}$-sharing for each of the $m$ terms, parties locally add the shares and execute $\prot{\Sh}$ on the result. Concretely, parties locally compute $\vl{z}^j = \sum_{i=1}^{m} \vl{z}_i^j$ for $j \in \{1,2\}$ and execute $\prot{\Sh}$ on $\vl{z}^j$ to obtain its $\shr{\cdot}$-sharing. This results in an online communication independent of $m$.

\begin{lemma}[Communication]
	\label{lemma:3pcSSumBitInj}
	Protocol $\prot{\bitinjS}$ requires $m \cdot 2\ell$ bits of communication in preprocessing, and $1$ round and $2 \ell$ bits of communication in the online phase.
\end{lemma}

\subsubsection{Bit Injection:II}
\label{sec:3pcSBitInjII}
Similar to $\prot{\bitinj}$ protocol, given $\shrB{\bitb_1}, \shrB{\bitb_2}$ and $\shr{\vl{v}}$, protocol $\prot{\dbitA}$ computes the arithmetic sharing of $\arval{(\bitb_1 \bitb_2)}\vl{v}$. Let $\Delta_{\bitb_1}$, $\Delta_{\bitb_2}$  denote the value $(1-2\arval{\mk{\bitb_1}})$, $(1-2\arval{\mk{\bitb_2}})$ respectively. Using \eqref{eq:3pcSdbitA} and \eqref{eq:3pcSbitinj}, we can write

\begin{align}
	\label{eq:3pcSdbitinj}
	\arval{(\bitb_1 \bitb_2)} \vl{v} &= \arval{(\mk{\bitb_1} \xor \pad{\bitb_1}{})} \arval{(\mk{\bitb_2} \xor \pad{\bitb_2}{})} (\mk{\vl{v}} - \pd{\vl{v}}{}) \nonumber \\ 
	&= (\arval{\mk{\bitb_1}} + \padR{\bitb_1}\Delta_{\bitb_1}) (\arval{\mk{\bitb_2}} + \padR{\bitb_2}\Delta_{\bitb_2})(\mk{\vl{v}} - \pd{\vl{v}}{}) \nonumber \\
	&= \arval{\mk{\bitb_1}}\arval{\mk{\bitb_2}}\mk{\vl{v}}  + \padR{\bitb_1}\arval{\mk{\bitb_2}}\mk{\vl{v}} \Delta_{\bitb_1} + \padR{\bitb_2}\arval{\mk{\bitb_1}}\mk{\vl{v}}\Delta_{\bitb_2} + \arval{(\pad{\bitb_1}{}\pad{\bitb_2}{})}\mk{\vl{v}}\Delta_{\bitb_1}\Delta_{\bitb_2} \nonumber \\
	&~~~- \pd{\vl{v}}{}\arval{\mk{\bitb_1}}\arval{\mk{\bitb_2}} - \padR{\bitb_1}\pd{\vl{v}}{}\arval{\mk{\bitb_2}}\Delta_{\bitb_1} - \padR{\bitb_2}\pd{\vl{v}}{}\arval{\mk{\bitb_1}}\Delta_{\bitb_2} - \arval{(\pad{\bitb_1}{}\pad{\bitb_2}{})}\pd{\vl{v}}{}\Delta_{\bitb_1}\Delta_{\bitb_2}
\end{align} 

During the preprocessing, $P_0$ generates the $\sqr{\cdot}$-shares of  $\padR{\bitb_1}$, $\padR{\bitb_2}$, $\padR{\bitb_1}\pd{\vl{v}}$, $\padR{\bitb_2}\pd{\vl{v}}$, $\arval{(\pad{\bitb_1}{}\pad{\bitb_2}{})}$ and $\arval{(\pad{\bitb_1}{}\pad{\bitb_2}{})}\pd{\vl{v}}$ similar to $\prot{\bitA}$ protocol. The online phase is similar to that of $\prot{\bitinj}$ protocol.

\begin{lemma}[Communication]
	\label{lemma:3pcSdbitinj}
	Protocol $\prot{\dbitinj}$ requires $6\ell$ bits of communication in preprocessing, and $1$ round and $2 \ell$ bits of communication in the online phase.
\end{lemma}

\subsection{Equality Test~($\prot{\eql}$)}
To check whether $\vl{a} \iseq \vl{b}$ or not, given $\shr{\vl{a}}, \shr{\vl{b}}$, $\prot{\eql}$ proceeds with parties locally computing $\shr{\vl{y}} = \shr{\vl{a}} - \shr{\vl{b}}$. According to our sharing semantics, $\vl{y}$ can be written as $\vl{y} = \vl{y}_1 - \vl{y}_2$ where $\vl{y}_1 = \mk{\vl{y}}$ and $\vl{y}_2 = \pad{\vl{y}}{}$. $P_0$ generates $\shrB{\vl{y}_2}$ during the preprocessing while $P_1,P_2$ generate $\shrB{\vl{y}_1}$ in the online using $\prot{\JSh}$. Note that $\vl{a} = \vl{b}$ implies $\vl{y}_1 = \vl{y}_2$ and hence all the bits of $\vl{v} = \overline{(\vl{y}_1 \xor \vl{y}_2)}$ should be $1$. As mentioned in the introduction of Part II~(\ref{chap:layer2_intro}), parties use four input AND gates and a tree structure, where $4$ bits are taken at a time and the AND of them is computed in one go. 

\section{Mixed Protocol Framework}
\label{sec:3pcSixFrame}
\tabref{3pcSConv} compares our sharing conversions with ABY3~\cite{CCS:MohRin18}. For uniformity, we consider a function, {\sf F}, to be computed on an $\ell$-bit inputs $\vl{x}, \vl{y}$ using a garbled circuit (GC) in the mixed framework, which gives an $\ell$-bit output $\vl{z} = \mathsf{F(\vl{x}, \vl{y})}$, where $\ell$ denotes the ring size in bits. Let  $\Grb{F}$  denote the corresponding GC. In the table, $\Grb{Sn}$ denotes a ${\sf n}$-input garbled subtraction circuit; $\Grb{An}$ denotes ${\sf n}$-input garbled addition circuit; $\GrbD{}$ denotes the garbled circuit with decoding information; $\Grb{n_1\times1,\ldots,n_m \times m}$ denotes ${\sf n_i}$ instances of GC $\Grb{i}$ for $i \in \{1,\dots,{\sf m}\}$ and $\Size{\Grb{n_1\times1,\ldots,n_m \times m}}$ denotes its size. 

\begin{table}[htb!]
	\centering
	\resizebox{0.95\textwidth}{!}{
		\begin{NiceTabular}{rr|rrr|rrrr}
			\toprule 
			\Block{2-1}{Variant\tabularnote{Notations: $\ell$ - size of ring in bits, $\kappa$ - computational security parameter, 'pre' - preprocessing, 'on' - online.}} 
			& \Block{2-1}{Conversion\tabularnote{'A' - arithmetic, 'B' - boolean, 'G' - Garbled.}}
			& \Block[c]{1-3}{ABY3~\cite{CCS:MohRin18}} & &
			& \Block[c]{1-4}{$\TSthis$} & & & \\ \cmidrule{3-9}
			& & Comm.\textsubscript{pre} & Comm.\textsubscript{on}  & Rounds\textsubscript{on} 
			& \Block[c]{1-2}{Comm.\textsubscript{pre}} & & Comm.\textsubscript{on}  & Rounds\textsubscript{on} \\
			\midrule
			\Block{4-1}{2 GC}
			& A-G-A & $2 \ell \kappa + 2\Size{\GrbD{2 \times A2,S2,F}}$ & $10\ell\kappa$ & \Block{4-1}{$2$}
			& \Block[r]{4-1}{($6 \ell \kappa + \ell$)\\+}  & $2\Size{\GrbD{2 \times S2, A2, F}}$ 
			& \Block{4-1}{$4\ell\kappa$} & \Block{4-1}{$1$} \\
			& A-G-B &  $2\Size{\Grb{2 \times A2,F}}$             & $8\ell\kappa + 2\ell$  &  & & $2\Size{\GrbD{2 \times S2,F}}$ &  & \\
			& B-G-A &  $2 \ell \kappa + 2\Size{\GrbD{S2,F}}$  & $10\ell\kappa$          &  & & $2\Size{\GrbD{A2,F}}$    &  & \\
			& B-G-B &  $2\Size{\Grb{F}}$                                & $8\ell\kappa + 2\ell$ &  & & $2\Size{\GrbD{F}}$                   &  & \\
			\midrule
			%
			\Block{4-1}{1 GC}
			& A-G-A & $\ell \kappa + \Size{\GrbD{2 \times A2,S2,F}}$ & $5\ell\kappa$ & \Block{4-1}{$2$}
			& \Block[r]{4-1}{($3 \ell \kappa + \ell$)\\+}  & $\Size{\GrbD{2 \times S2, A2, F}}$ 
			& \Block{4-1}{$2\ell\kappa + \ell$} & \Block{4-1}{$2$} \\
			& A-G-B &  $\Size{\Grb{2 \times A2,F}}$             & $4\ell\kappa + \ell$  &  & & $\Size{\GrbD{2 \times S2,F}}$ &  & \\
			& B-G-A &  $\ell \kappa + \Size{\GrbD{S2,F}}$  & $5\ell\kappa$             &  & & $\Size{\GrbD{A2,F}}$    &  & \\
			& B-G-B &  $\Size{\Grb{F}}$                                & $4\ell\kappa + \ell$  &  & & $\Size{\GrbD{F}}$                   &  & \\
			\midrule
			\Block{2-1}{Others\tabularnote{$\vl{u_1} = \vl{n_2} + 4\vl{n_3} + 11\vl{n_4}$, $\vl{u_2} = \vl{n_2} + \vl{n_3} + \vl{n_4}$ denote the number of AND gates in the optimized adder circuit~\cite{USENIX:PSSY21} with 2, 3, 4 inputs, respectively. For $\ell = 64$, $\vl{n_2} = 216, \vl{n_3} = 184, \vl{n_4}=179$.}}
			& A-B & $-$  & $3\ell + 3 \ell \log \ell$ & $1 + \log \ell$
			& \Block[r]{1-2}{$\vl{u_1} + \ell$} &  & $2\vl{u_2}$ & $\log_4 \ell$ \\
			& B-A &  $-$  & $3\ell + 3 \ell \log \ell$ & $1 + \log \ell$
			& \Block[r]{1-2}{$\ell^2$} &  & $2\ell$ & $1$ \\
			%
			\bottomrule
		\end{NiceTabular}
	}
	\caption{Mixed protocol conversions of ABY3~\cite{CCS:MohRin18} and $\TSthis$.}\label{tab:3pcSConv}
\end{table}

\subsection{Conversions involving Garbled World} 
\label{sec:3pcSconv2gc}
Assume the GC is required to compute a function $f$ on inputs $\vl{x}, \vl{y} \in \Z{\ell}$ and let the output be $f(\vl{x}, \vl{y})$. All the conversions described are for the 2 GC variant. Conversions for the 1 GC variant are straightforward, hence we omit the details.

\paragraph{Case I: Boolean-Garbled-Boolean}
Since the inputs to the GC are available in boolean form, say $\shrB{\vl{x}}, \shrB{\vl{y}}$, parties generate $\shrC{\vl{x}}, \shrC{\vl{y}}$ by invoking the garbled sharing protocol $\pigsh$.
Additionally, $P_0$ samples $\vl{R} \in \Z{\ell}$ to mask the function output, $f(\vl{x}, \vl{y})$, and generate $\shrB{\vl{R}}$ and $\shrG{\vl{R}}$. Garblers $P_g \in \{P_0, P_2\}$ garble the circuit which computes $\vl{z} = f(\vl{x}, \vl{y}) \xor \vl{R}$, and send the GC along with the decoding information to evaluator $P_1$. Analogous steps are performed for evaluator $P_2$. Upon GC evaluation and output decoding, evaluators obtain $\vl{z} = f(\vl{x}, \vl{y}) \xor \vl{R}$, and jointly boolean share $\vl{z}$ to generate $\shrB{\vl{z}}$. Parties then compute $\shrB{f(\vl{x}, \vl{y})} = \shrB{\vl{z}} \xor \shrB{\vl{R}}$.  

\paragraph{Case II: Boolean-Garbled-Arithmetic}
This is similar to {\em Case I} except that the circuit which computes $\vl{z} = f(\vl{x}, \vl{y}) + \vl{R}$ is garbled instead. Boolean sharing of $\vl{z}$ is replaced with arithmetic, followed by computing $\shr{f(\vl{x}, \vl{y})} = \shr{\vl{z}} - \shr{\vl{R}}$.

\paragraph{Cases III \& IV: Input in Arithmetic Sharing} 
The function to be computed $f(\vl{x}, \vl{y})$, is modified as $f^{\prime}(\mk{\vl{x}}, \pad{\vl{x}}{}, \mk{\vl{y}}, \pad{\vl{y}}{}) = f(\mk{\vl{x}}-\pad{\vl{x}}{}, \mk{\vl{y}}-\pad{\vl{y}}{})$ where inputs $\vl{x}, \vl{y}$ are replaced by the pairs $\{\mk{\vl{x}}, \pad{\vl{x}}{}\}, \{\mk{\vl{y}}, \pad{\vl{y}}{}\}$. The circuit to be garbled thus, corresponds to the function $f^{\prime}$. Parties generate $\shrG{\mk{\vl{x}}}, \shrG{\pad{\vl{x}}{}}, \shrG{\mk{\vl{y}}}, \shrG{\pad{\vl{y}}{}}$ via $\pigsh$, following which, parties proceed with the rest of the computation whose steps are similar to {\em Case I}, and {\em II}, depending on the requirement on the output sharing.

\subsection{Other Conversions} 
\label{sec:3pcSotherconv}

\paragraph{Arithmetic to Boolean} 
To convert arithmetic sharing of $\vl{v} \in \Z{\ell}$ to boolean sharing, observe that $\vl{v} = \vl{v}_1 + \vl{v}_2$ where $\vl{v}_1 = \mk{\vl{v}}$ is possessed by parties $P_1, P_2$, while $\vl{v}_2 = - \pad{\vl{v}}{}$ is possessed by $P_0$. Thus, $\shrB{\vl{v}}$  can be computed as $\shrB{\vl{v}} = \shrB{\vl{v}_1} + \shrB{\vl{v}_2}$. For this, $P_0$ can generate $\shrB{\vl{v}_2}$ in the preprocessing, and $\shrB{\vl{v}_1}$ can be generated in the online by $P_1, P_2$ using joint sharing protocol. The protocol appears in \boxref{fig:3pcSpiab}. Boolean addition, when instantiated using the adder of ABY2.0~\cite{USENIX:PSSY21}, requires $\log_4(\ell)$ rounds.

\begin{protocolbox}{$\piab$}{Arithmetic to Boolean Conversion in $\TSthis$.}{fig:3pcSpiab}
	\justify
	\algoHead{Preprocessing:} $P_0$ generates $\shrB{\vl{v}_2}$ using $\prot{\Sh}$, where $\vl{v}_2 = -\pad{\vl{v}}{}$.
	\justify
	\vspace{-2mm}
	\algoHead{Online:}
	\begin{enumerate} 
		\item $P_1, P_2$ execute joint boolean sharing to generate $\shrB{\vl{v}_1}$, where $\vl{v}_1 = \mk{\vl{v}}$.
		\item Parties obtain $\shrB{\vl{v}} = \shrB{\vl{v}_1} + \shrB{\vl{v}_2}$ using a boolean adder circuit.
	\end{enumerate}
\end{protocolbox}

\paragraph{Boolean to Arithmetic} 
To convert a boolean sharing of $\vl{v} \in \Z{\ell}$ into an arithmetic sharing, note that 
\begin{small}
	\begin{align*}
		\vl{v} = \sum_{i=0}^{\ell - 1} 2^{i} \vl{v}[i] = \sum_{i=0}^{\ell - 1} 2^{i} (\pad{\vl{v}[i]}{} \xor \mk{\vl{v}[i]}) 
		= \sum_{i=0}^{\ell - 1} 2^{i} \left( \arval{\mk{\vl{v}[i]}} +  \padR{\vl{v}[i]} (1 - 2\arval{\mk{\vl{v}[i]}}) \right) 
	\end{align*}
\end{small}
where $\padR{\vl{v}[i]}, \arval{\mk{\vl{v}[i]}}$ denote the arithmetic value of bits ${\pd{\vl{v}[i]}}, {\mk{\vl{v}[i]}}$ over the ring $\Z{\ell}$. 
For each bit $\vl{v}[i]$ of $\vl{v}$, $P_0$ generates the $\sqr{\cdot}$-shares of $\arval{{{\pd{\vl{v}[i]}}}}$ in the preprocessing, similar to $\prot{\bitA}$~(\boxref{fig:3pcSpiBitA}). During the online phase, additive shares for each bit $\vl{v}[i]$ are locally computed similar to $\prot{\bitA}$. Parties then multiply the $i$th share with $2^i$ and locally add up to obtain an additive sharing of $\vl{v}$. The rest of the steps are similar to $\prot{\bitA}$, and the formal protocol appears in \boxref{fig:3pcSpiba}. 

\begin{protocolbox}{$\piba(\Partyset, \shrB{\vl{v}})$}{Boolean to Arithmetic Conversion in $\TSthis$.}{fig:3pcSpiba}
	Let $\vl{v}[i]$ denote the $i$th bit of $\vl{v}$. Let $\vl{p}_i = \arval{\mk{\vl{v}[i]}}$, and $\vl{q}_i = \padR{\vl{v}[i]}$. \\
	\justify 
	\vspace{-2mm}
	\algoHead{Preprocessing:} 
	\begin{enumerate} 
		\item For $i \in \{0, 1, \ldots, \ell-1 \} $, execute the preprocessing of $\prot{\bitA}$ (\boxref{fig:3pcSpiBitA}) for each bit $\vl{v}[i]$, to generate $\sqr{{\vl{q}_i}} = (\sqr{\vl{q}_i}_1 , \sqr{\vl{q}_i}_2 )$.
	\end{enumerate}
	\justify
	\vspace{-2mm}
	\algoHead{Online:} Let $\vl{y}_i = \arval{(\vl{v}[i])}$ and $\vl{y}$ denotes the arithmetic equivalent of $\vl{v}$.
	\begin{enumerate}
		\item Locally compute:
		\begin{align*}
			P_1: \vl{y}^1 = \sum_{i=0}^{\ell-1} 2^i \vl{y}_i^1  &=  \sum_{i=0}^{\ell-1} 2^i (\vl{p}_i + \sqr{\vl{q}_i}_1 (1 - 2\vl{p}_i)) \\
			P_2: \vl{y}^2 = \sum_{i=0}^{\ell-1} 2^i \vl{y}_i^2 &=  \sum_{i=0}^{\ell-1} 2^i ( \sqr{\vl{q}_i}_2 (1 - 2\vl{p}_i))  
		\end{align*}
		\item $P_j$ for $j \in \{1,2\}$ executes $\prot{\Sh}$ on $\vl{y}^j$ to generate the respective $\shr{\cdot}$-shares.
		\item Compute $\shr{\vl{y}} = \shr{\vl{y}^1} + \shr{\vl{y}^2}$.
	\end{enumerate}     
\end{protocolbox}

\chapter{$\Tthis$: 3PC Fair and Robust Blocks}
\label{chap:layer2_3pcmal}
This chapter provides details for the Layer II blocks of our 2PC framework $\Tthis$. Details for the Layer I blocks are provided in chapter~\ref{chap:layer1_3pcmal}. The robust constructions of the blocks are detailed in this chapter, and the fair variants can be derived easily.

\section{Building Blocks}
\label{sec:layerIIblocks3pcM}

\subsection{Dot Product~(Scalar Product)}
\label{sec:3pcMDotp}
Given $\shr{\vct{a}}, \shr{\vct{b}}$ with $|\vct{a}| = |\vct{b}| = \vl{d}$, protocol $\prot{\dotp}$~(\boxref{fig:piDotP3pcM}) computes $\shr{\vl{z}}$ such that $\vl{z} = (\vct{a} \band \vct{b})^{\vl{t}}$ if truncation is enabled, else $\vl{z} = \vct{a} \band \vct{b}$. For this, we combine the partial products from the multiplication protocol across $\vl{d}$ multiplications and communicate them in a single shot. This makes the communication cost of the dot product independent of the vector size.

\smallskip
\begin{protocolsplitbox}{$\prot{\dotp}(\vct{a}, \vct{b}, \isTr)$}{Dot Product with / without Truncation in $\Tthis$.}{fig:piDotP3pcM}
	$\isTr$ is a bit denoting whether truncation is required ($\isTr =1$) or not ($\isTr=0$). \\
	\detail{
		{\bf Input(s):} $\shr{\vct{a}}, \shr{\vct{b}}$.\\
		{\bf Output:} $\shr{\vl{o}}$ where $\vl{o} = \vl{z}^{\vl{t}}$ if $\isTr = 1$ and $\vl{o} = \vl{z}$ if $\isTr = 0$ and $\vl{z} = \vct{a} \band \vct{b} = \sum_{i = 1}^{\vl{d}} {\vl{a}}_i {\vl{b}}_i$.
	}
	\justify 
	\vspace{-2mm}
	\algoHead{Preprocessing:} Let $\pad{\vct{a}}{} = \{\pad{\vl{a}_i}{}\}_{i \in [\vl{d}]}$, $\pad{\vct{b}}{} = \{\pad{\vl{b}_i}{}\}_{i \in [\vl{d}]}$ and $\gm{\vct{a}\vct{b}}{} = \sum_{i=1}^{\vl{d}} \gm{\vl{a}_i \vl{b}_i}{}$.
	\begin{enumerate}
		\item Invoke $\Func[\dotpPre]$ on $\sgr{\pad{\vct{a}}{}}$ and $\sgr{\pad{\vct{b}}{}}$ to obtain $\sgr{\gm{\vct{a}\vct{b}}{}}$.
		\item If $\isTr = 0$:
		\begin{enumerate}
			\item Local computation of $\sgr{\vl{r}}$: $\Partyset \setminus \{P_2\} \leftarrow_R {\vl{r}}^1$;~~~$\Partyset \setminus \{P_1\} \leftarrow_R {\vl{r}}^2$;~~~$\Partyset \setminus \{P_3\} \leftarrow_R {\vl{r}}^3$.
			\item Local computation of $\shr{\vl{r}}$: $\pad{\vl{r}}{1} = - {\vl{r}}^1$,~~$\pad{\vl{r}}{2} = - {\vl{r}}^2$,~~$\pad{\vl{r}}{3} = - {\vl{r}}^3$,~~$\mk{\vl{r}} = 0$. Set $\shr{\vl{q}} = \shr{\vl{r}}$. 
		\end{enumerate}
		\item If $\isTr = 1$, invoke $\prot{\trgen}$~(\boxref{fig:trgen3pcM}) to generate $(\sgr{\vl{r}}, \shr{\vl{r}^{\vl{t}}})$. Set $\shr{\vl{q}} =  \shr{\vl{r}^{\vl{t}}}$.
		\item Locally compute $\sgr{(\gm{\vct{a}\vct{b}}{} - \vl{r})} = \sgr{\gm{\vct{a}\vct{b}}{}} - \sgr{\vl{r}}$.
	\end{enumerate}
	\justify
	\vspace{-2mm}
	\algoHead{Online:} Let $\vl{y} = (\vl{z} - \vl{r}) -  \sum_{i=1}^{\vl{d}} \mk{\vl{a}_i \vl{b}_i}$.
	\begin{enumerate}
		\item Parties locally compute the following:
		\begin{align*}
			P_1, P_3: \vl{y}_1  &=  \sum_{i=1}^{\vl{d}} (- \pad{\vl{a}_i}{1} \mk{\vl{b}_i} - \pad{\vl{b}_i}{1} \mk{\vl{a}_i}) + {(\gm{\vct{a}\vct{b}}{} - \vl{r})}^1 \\
			P_2, P_3: \vl{y}_2 &=  \sum_{i=1}^{\vl{d}} (- \pad{\vl{a}_i}{2} \mk{\vl{b}_i} - \pad{\vl{b}_i}{2} \mk{\vl{a}_i}) + {(\gm{\vct{a}\vct{b}}{} - \vl{r})}^2 \\
			P_1, P_2: \vl{y}_3  &=  \sum_{i=1}^{\vl{d}} (- \pad{\vl{a}_i}{3} \mk{\vl{b}_i} - \pad{\vl{b}_i}{3} \mk{\vl{a}_i}) + {(\gm{\vct{a}\vct{b}}{} - \vl{r})}^3
		\end{align*}
		\item $P_1, P_3$ $\jsend$ $\vl{y}_1$ to $P_2$, while $P_2, P_3$ $\jsend$ $\vl{y}_2$ to $P_1$. They locally compute $\vl{z} - \vl{r} = (\vl{y}_1 + \vl{y}_2 + \vl{y}_3) +  \sum_{i=1}^{\vl{d}} \mk{\vl{a}_i \vl{b}_i}$.
		\item $P_1, P_2$: If $\isTr = 1$, set $\vl{p} = (\vl{z} - \vl{r})^{\vl{t}}$, else $\vl{p} = \vl{z} - \vl{r}$. Execute $\prot{\JSh}(P_1, P_2, \vl{p})$ to generate $\shr{\vl{p}}$. 
		\item Compute $\shr{\vl{o}} = \shr{\vl{p}} + \shr{\vl{q}}$. Here $\vl{o} = \vl{z}^{\vl{t}}$ if $\isTr = 1$ and $\vl{z}$ otherwise.
	\end{enumerate}     
\end{protocolsplitbox}

Analogous to the multiplication protocol, dot product offloads one call to a robust dot product protocol $\prot{\MultPre}$ to the preprocessing. By extending techniques of \cite{C:BBCGI19, CCS:BGIN19}, we give an instantiation for the dot product protocol used in our preprocessing whose (amortized) communication cost is constant, thereby making our preprocessing cost also {\em independent} of $\vl{d}$. The ideal world functionality $\Func[\dotpPre]$ for realizing $\prot{\dotpPre}$ is presented in \boxref{fig:FdotpPre3pcM}.

\paragraph{Instantiating $\Func[\dotpPre]$:}
A trivial way to instantiate $\prot{\dotpPre}$ is to treat a dot product operation as $\vl{d}$ multiplications. However, this results in a communication cost that is linearly dependent on the feature size. Instead, we instantiate $\prot{\dotpPre}$ by a semi-honest dot product protocol followed by a verification phase to check the correctness. For the verification phase, we extend the techniques of \cite{C:BBCGI19, CCS:BGIN19} to provide support for verification of dot product tuples. Setting the verification phase parameters appropriately gives a $\prot{\dotpPre}$ whose (amortized) communication cost is independent of the feature size. We will provide the details next.

\begin{systembox}{$\Func[\dotpPre]$}{Ideal functionality for $\prot{\dotpPre}$ in $\Tthis$.}{fig:FdotpPre3pcM}
\justify
$\Func[\dotpPre]$ interacts with the parties in $\Partyset$ and the adversary $\Sim$. $\Func[\dotpPre]$ receives $\sgr{\cdot}$-shares of $\vct{u} = \{ \vl{u}_1, \ldots, \vl{u}_{\vl{d}} \}$, $\vct{v} = \{ \vl{v}_1, \ldots, \vl{v}_{\vl{d}} \}$ from the parties. Let $P^{\star}$ denotes the party corrupted by $\Sim$. $\Func[\dotpPre]$ receives $(\vl{w}_i, \vl{w}_j)$ from $\Sim$ as its share for $\sgr{\vl{w}}$ where $\vl{w} = \vct{u} \band \vct{v}$. $\Func[\dotpPre]$ proceeds as follows:
\begin{enumerate}
\item Reconstructs $\vct{u}, \vct{v}$ using the shares received from honest parties and compute $\vl{w} = \vct{u} \band \vct{v}$.
\item Computes the third share $\vl{w}_k = \vl{w} - \vl{w}_i -\vl{w}_j$ and sets $\sgr{\vl{w}}_1 = (\vl{w}_1, \vl{w}_3), \sgr{\vl{w}}_2 = (\vl{w}_2, \vl{w}_3),  \sgr{\vl{w}}_3 = (\vl{w}_1, \vl{w}_2)$.
\item Send $(\OUTPUT, \sgr{\vl{w}}_s)$ to $P_s \in \Partyset$.
\end{enumerate}
\end{systembox}

To realize $\Func[\dotpPre]$, the approach is to run a semi-honest dot product protocol followed by a verification phase to check the correctness of the output. 
For verification, the work of \cite{C:BBCGI19} provides techniques to verify the correctness of $\nm$ multiplication triples (and degree-two relations) at the cost of $\Order(\sqrt{\nm})$ extended ring elements, albeit with abort security. While \cite{CCS:BGIN19} improves their techniques to provide {\em robust} verification for multiplication, we show how to extend the techniques in \cite{CCS:BGIN19} to robustly verify the correctness of $\nm$ dot product tuples (dot product being a degree two relation), with vectors of dimension $\vl{d}$, at a cost of $\Order(\sqrt{\vl{d} \nm})$ extended ring elements. Thus, the cost to realize one instance of $\Func[\dotpPre]$ can be brought down to only the cost of a semi-honest dot product computation (which is $3$ ring elements and independent of the vector dimension), where the cost due to verification can be amortized away by setting $\vl{d}, \nm$ appropriately.

Given vectors $\vct{u} = (\vl{u}_1, \ldots, \vl{u}_{\vl{d}}), \vct{v} = (\vl{v}_1, \ldots, \vl{v}_{\vl{d}})$, the semi-honest dot product protocol proceeds as follows. The parties, using the shared key setup, non-interactively generate \emph{3-out-of-3} additive shares of zero using $\FZero$~(\S\ref{sec:3pcMFZero}), i.e $P_i$ has $\sz_i$, such that $\sz_1 + \sz_2 +\sz_3 = 0$. Then, parties proceed with generating the $\sgr{\cdot}$-shares of $\vl{w} = \vct{u} \band \vct{v}$ as: 
\begin{align} \label{eq:dotpver1}
	P_1~\text{ computes and sends }~\vl{y}_1 &=  \sz_1 + \sum_{j = 1}^{\vl{d}} ( \vl{u}_{j}^1 \vl{v}_{j}^3 + \vl{u}_{j}^3 \vl{v}_{j}^1 + \vl{u}_{j}^3 \vl{v}_{j}^3)~\text{ to }~P_2 \nonumber \\
	P_2~\text{ computes and sends }~\vl{y}_2 &=  \sz_2 + \sum_{j = 1}^{\vl{d}} ( \vl{u}_{j}^2 \vl{v}_{j}^3 + \vl{u}_{j}^3 \vl{v}_{j}^2 + \vl{u}_{j}^2 \vl{v}_{j}^2)~\text{ to }~P_3 \nonumber \\
	P_3~\text{ computes and sends }~\vl{y}_3 &=  \sz_3 + \sum_{j = 1}^{\vl{d}} ( \vl{u}_{j}^1 \vl{v}_{j}^2 + \vl{u}_{j}^2 \vl{v}_{j}^1 + \vl{u}_{j}^1 \vl{v}_{j}^1)~\text{ to }~P_1
\end{align}

Now, to complete the $\sgr{\cdot}$-sharing of $\vl{w}$, parties locally set $\vl{w}^1 = \vl{y}_3$, $\vl{w}^2 = \vl{y}_2$ and $\vl{w}^3 = \vl{y}_1$. To check the correctness of the computation $\sgr{\vl{w}} = \sgr{\vct{u} \band \vct{v}}$, each $P_i \in \Partyset$ needs to prove that the $\vl{y}_i$ it sent in the semi-honest protocol satisfies \ref{eq:dotpver1}. Without loss of generality, consider the case when $P_i = P_1$. Then, it has to prove 
\begin{align} \label{eq:dotpver2}
	\sz_1 + \sum_{j = 1}^{\vl{d}} ( \vl{u}_{j}^1 \vl{v}_{j}^3 + \vl{u}_{j}^3 \vl{v}_{j}^1 + \vl{u}_{j}^3 \vl{v}_{j}^3) - \vl{y}_1 = 0
\end{align}
This difference in the expected message that should be sent (computed using $P_1$'s correct input shares) and the actual message sent by $P_1$ is captured by a circuit $\cc$, defined below. 

\begin{align} \label{eq:cc} 
	\cc \left(\{\vl{u}_{j}^1, \vl{u}_{j}^3, \vl{v}_{j}^1, \vl{v}_{j}^3 \}_{j=1}^{\vl{d}}, \sz_1, \vl{w}_1 \right) 
	= \sz_1 + \sum_{j = 1}^{\vl{d}} ( \vl{u}_{j}^1 \vl{v}_{j}^3 + \vl{u}_{j}^3 \vl{v}_{j}^1 + \vl{u}_{j}^3 \vl{v}_{j}^3) - \vl{y}_1
\end{align}

Here, $\cc$ takes as input $\ic = 4 \vl{d} + 2$ values: $\sgr{\cdot}$-shares of $\vct{u}, \vct{v}$ held by $P_1$, i.e. $\{\vl{u}_{j}^1, \vl{u}_{j}^3, \vl{v}_{j}^1, \vl{v}_{j}^3 \}_{j=1}^{\vl{d}}$, the additive share of zero, $\sz_1$, that $P_1$ holds, and the additive share $\vl{y}_1$ sent by $P_1$. For correct computation with respect to $P_1$, we require the difference in the expected message and the actual message to be $0$, i.e.,
\begin{align} \label{eq:dotpver3}
	\cc \left(\{\vl{u}_{j}^1, \vl{u}_{j}^3, \vl{v}_{j}^1, \vl{v}_{j}^3 \}_{j=1}^{\vl{d}}, \sz_1, \vl{w}_1 \right) = 0
\end{align}

We now explain how to verify the correctness for $\nm$ dot product tuples assuming that the operations are carried out over a prime-order field. The verification can be extended to support operations over rings following the techniques of \cite{C:BBCGI19, CCS:BGIN19}.
To verify the correctness for $\nm$ dot product tuples, $\{\vct{u}_k, \vct{v}_k, \vl{w}_k\}_{k=1}^{\nm}$ where $\vl{w}_k = \vct{u}_k \band \vct{v}_k$, the output of $\cc$ (which is the difference in the expected and actual message sent) for each of the corresponding dot product tuple must be $0$. 
To check correctness of all dot products {\em at once}, it suffices to check if a random linear combination of the output of each $\cc$ (for each dot product) is $0$. This is because the random linear combination of the differences will be $0$ with high probability if $\vl{w}_k = \vct{u}_k \band \vct{v}_k$ for each $k \in \{1, \ldots, \nm\}$. 
We remark that the definition of $c(\cdot)$ in \cite{CCS:BGIN19} enables the verification of only multiplication triples. With the re-definition of $\cc$ as in \ref{eq:cc}, we can now verify the correctness of dot products while the rest of the verification steps remain similar to that in \cite{CCS:BGIN19}. We elaborate on the details next.

A verification circuit, constructed as follows, enables $P_i$ to prove the correctness of the additive share of $\vl{w}$ that it sent, for $\nm$ instances of dot product at once. Note that the proof system is designed for the distributed-verifier setting where the proof generated by $P_i$ will be shared among $P_{i-1}, P_{i+1}$, who can together verify its correctness. First, a sub-circuit $\cg$ is defined as follows: group $\nL$ small $\cc$ circuits and take a random linear combination of the values on their output wires. Since each $\cc$ circuit takes  $\ic = 4\vl{d}+2$ inputs as described earlier, $\cg$ takes in $\ic \nL$ inputs. Precisely, $\cg$ is defined as follows:
\begin{align*}
	\cg(x_1, \ldots, x_{\ic \nL}) = \sum_{k=1}^{\nL} \rt_k \cdot \cc(x_{(k-1) \ic + 1}, \ldots, x_{(k-1) \ic + \ic})
\end{align*}

Since there are total $\nm$ dot products to be verified, there will be $\nM = \nm /\nL$ sub-circuits $\cg$. Looking ahead, this grouping technique enables obtaining a sub-linear communication cost for verification because the communication cost turns out to be $\Order(\ic \nL + \nM)$ and setting $\ic \nL = \nM$ gives the desired result. The sub-circuits $\cg$ make up the circuit $\cG$ which outputs a random linear combination of the values on the output wires of each $\cg$, i.e:
\begin{align*}
	\cG(x_1, \ldots, x_{\ic \nm}) = \sum_{k=1}^{\nM} \re_k \cdot \cg(x_{(k-1) \ic \nL + 1}, \ldots, x_{(k-1) \ic \nL + \ic \nL}) 
\end{align*}
Here, $\rt_k$ and $\re_k$ are randomly sampled (non-interactively) by all parties.
To prove correctness, $P_i$ needs to prove that $\cG$ outputs $0$. For this, $P_i$ defines $f_1\ldots,f_{\ic \nL}$ random polynomials of degree $\nM$, one for each input wire of $\cg$. For $\ell \in \{1,\ldots,\nM\}$ and $j \in \{1,\ldots,\ic \nL\}$,  $f_j(0)$ is chosen randomly and $f_j(\ell) = x_{(\ell-1)\ic+j}$ (i.e the $j\text{th}$ input of the $\ell\text{th}$ $\cg$ gate). 
$P_i$ further defines a $2\nM$ degree polynomial $p(\cdot)$ on the output wires of $\cg$, i.e $p(\cdot) = \cg (f_1, \ldots, f_{\ic \nL})$ where $p(\ell)$ for $\ell \in \{1, \ldots, \nM\}$ is the output of the $\ell$th $\cg$ gate. The additional $\nM+1$ points required to interpolate the $2\nM$ degree polynomial $p$, are obtained by evaluating $f_1, \ldots, f_{\ic \nL}$ on $\nM+1$ additional points, followed by an application of $\cg$ circuit. The proof generated by $P_i$ consists of $f_1(0), \ldots, f_{\ic \nL}(0)$ and the coefficients of $p$. Recall that since we are in the distributed-verifier setting, the prover $P_i$ additively shares the proof with $P_{i-1}, P_{i+1}$. Note here, that shares of $f_1(0), \ldots, f_{\ic \nL}(0)$ can be generated non-interactively. 

To verify the proof, verifiers $P_{i-1}, P_{i+1}$ need to check if the output of $\cG$  is $0$. This can be verified by computing the output of $\cG$ as $b= \sum_{\ell=1}^{\nM} \re_{\ell} \cdot p(\ell)$ and checking if $b = 0$, where $\re_{\ell}$'s are non-interactively sampled by all after the proof is sent. If $p$ is defined correctly, then this is indeed a random linear combination of the outputs of all the $\cg$-circuits. This necessitates the second check to verify the correctness of $p$ as per its definition i.e $p(\cdot) = \cg (f_1(\cdot), \ldots, f_{\ic \nL}(\cdot))$. This is performed by checking if $p(r) = \cg (f_1(r), \ldots, f_{\ic \nL}(r))$ for a random $r \notin \{1, \ldots, \nM\}$ (for privacy to hold) sampled non-interactively by all after the proof is sent. 
For the first check, verifiers can locally compute additive shares of $b$ (using the additive shares of coefficients of $p$ obtained as part of the proof) and reconstruct $b$ to check for equality with $0$.
For the second, verifiers locally compute additive shares of $p(r)$ using the shares of coefficients of $p$, and shares of $f_1(r), \ldots, f_{\ic \nL}(r)$ by interpolating $f_1, \ldots, f_{\ic \nL}$ using ($P_i$'s) inputs to the $\cc$-circuits which are implicitly additively shared between them (owing to the replicated sharing property). Verifiers exchange these values among themselves, reconstruct it and check if $p(r) = \cg (f_1(r), \ldots, f_{\ic \nL}(r))$.
Note that the messages computed and exchanged by the verifiers depend only on the proof sent by $P_i$ and the random values ($r, \re$) sampled by all. $ P_i$ can independently compute these messages. Thus, to prevent a verifier from falsely rejecting a correct proof, we use $\jsend$ to exchange these messages.
To optimize the communication cost further, it suffices if a single verifier computes the output of verification. 

\paragraph{Setting the parameters:} The proof sent by $P_i$ consists of the constant terms  $f_j(0)$ for $j \in \{1, \ldots, \ic \nL\}$ and $2\nM+1$ coefficients of $p$. The former can be can be generated non-interactively. Hence, $P_i$ needs to communicate $2\nM+1$ elements to the verifiers (one of which can be performed non-interactively). The message sent by the verifier consists of the additive share of $\sum_{\ell=1}^{\nM} \re_{\ell} \cdot p(\ell)$ (for the first check) and $f_1(r), \ldots , f_{\ic \nL}(r), p(r)$ (for the second check). Thus, the verifier communicates $\ic \nL + 2$ elements. As the proof is executed three times, each time with one party acting as the prover and the other two acting as the verifiers, overall, each party communicates $\ic \nL + 2 \nM + 3$ elements. Setting $\ic \nL = 2 \nM$ and $\nM = \frac{m}{\nL} $ results in the total communication required for verifying $m$ dot products to be $\Order(\sqrt{\vl{d} \nm})$. Thus, verifying a single dot product has an amortized cost of $\Order \left( \sqrt{\frac{\vl{d}}{\nm}} \right)$ which can be made very small by appropriately setting the values of $\vl{d}, \nm$. Thus, the (amortized) cost of a maliciously secure dot product protocol can be made equal to that of a semi-honest dot product protocol, which is $3$ ring elements. 

To support verification over rings \cite{CCS:BGIN19}, verification operations are carried out on the extended ring $\Z{\ell}/f(x)$, which is the ring of all polynomials with coefficients in $\Z{\ell}$ modulo a polynomial $f$, of degree $d$, irreducible over $\Z{ }$. Each element in $\Z{\ell}$ is lifted to a $d$-degree polynomial in $\Z{\ell}[x]/f(x)$ (which results in blowing up the communication by a factor $d$). Thus, the per party communication amounts to $(\ic \nL + 2 \nM + 3)d$ elements of $\Z{\ell}$ for verifying $\nm$ dot products of vector size $\vl{d}$ where $\ic = 4 \vl{d} + 2$. Further, the probability of a cheating prover is bounded by $\frac{2^{(\ell - 1) d} \cdot 2 \nM + 1}{2^{\ell d} - \nM}$ (cf. Theorem 4.7 of \cite{CCS:BGIN19}). 
Thus, if $\gamma$ is such that $2^{\gamma} \geq 2\nM$, then the cheating probability is 
\begin{align*}
	\frac{2^{(\ell - 1) d} \cdot 2 \nM + 1}{2^{\ell d} - \nM} \leq \frac{2^{(\ell - 1) d} \cdot 2^{\gamma} + 1}{2^{\ell d} - M} \approx 2^{-(d - \gamma)}
\end{align*}
We note that both, \cite{CCS:BGIN19} and our technique require a communication cost of $\Order(\sqrt{\nm \vl{d}})$ ring elements for verifying $\nm$ dot products of vector size $\vl{d}$. This is because multiplication is a special case of dot product with $\vl{d} = 1$. However, since our verification is for dot products, we can get away with performing only $\nm$ semi-honest dot products whose cost is equivalent to computing $\nm$ semi-honest multiplications, whereas \cite{CCS:BGIN19} requires to execute $\nm \vl{d}$ multiplications (as their technique can only verify correctness of multiplications), resulting in a dot product cost dependent on the vector size.
Concretely, to get $40$ bits of statistical security and for a vector size of $2^{10}$ (CIFAR-10~\cite{CIFAR10} dataset), the parameters mentioned above can be set as given in \tabref{zkver}. 

\begin{table}[htb!]
	\centering
		\begin{NiceTabular}{r | r | r | r | r }[notes/para]
			\toprule
			$\nm$\tabularnote{$\#$dot products to be verified} 
			&  $\nM$\tabularnote{$\# \cg$ sub-circuits} 
			& $\gamma$ 
			& $d$\tabularnote{degree of extension} 
			& Cost (per dot product) \\
			\midrule 
			$2^{20}$ & $2^{16}$ & $17$  & $57$ & $7.125$ \\ \midrule 
			$2^{30}$ & $2^{21}$ & $22$  & $62$ & $0.242$ \\ \midrule 
			$2^{40}$ & $2^{26}$ & $27$  & $67$ & $0.008$  \\ \midrule
			$2^{50}$ & $2^{31}$ & $32$  & $72$ & $0.0002$  \\ \bottomrule 
		\end{NiceTabular}
	\caption{Cost of verification in terms of the number of ring elements communicated per dot product, and parameters for vector size $\vl{d} = 2^{10}$ and $40$ bits of statistical security.\label{tab:zkver}}
\end{table}

It is possible to further bring down the communication cost required for verifying $\nm$ dot product tuples to $\Order(\log (\vl{d} \nm))$ at the expense of requiring more rounds by further extending the technique of \cite{C:BBCGI19}, which we leave as an exercise.
We refer readers to \cite{CCS:BGIN19} for formal details. 

\begin{lemma}[Communication]
	\label{lemma:3pcMpidotpf}
	Protocol $\prot{\dotp}$~(\boxref{fig:piDotP3pcM})~(in $\Tthis$) requires $3\ell$ bits of communication in preprocessing, and $1$ round and $3 \ell$ bits of communication in the online phase.
\end{lemma}

 \subsection{Bit Extraction}
 \label{sec:3pcMBitExt}

To compute most significant bit~($\msb$) of the value $\vl{v}$, note that $\vl{v} = \vl{v}_1 + \vl{v}_2 + \vl{v}_3$ for $\vl{v}_1 = \mk{\vl{v}} - \pad{\vl{v}}{3}$, $\vl{v}_2 = - \pad{\vl{v}}{1}$ and $\vl{v}_3 = - \pad{\vl{v}}{2}$ as per the sharing semantics~(cf.~\tabref{3pcMsharing}). Parties generate the boolean sharing of $\vl{v}_1, \vl{v}_2, \vl{v}_3$ using joint sharing protocol. 
It has been shown in ABY3~\cite{CCS:MohRin18} that $\vl{v} = 2c + s$ where {\sf FA}$(\vl{v}_1[i], \vl{v}_2[i], \vl{v}_3[i]) \rightarrow (c[i],s[i])$ for $i \in \{0,\ldots,\ell-1\}$. Here {\sf FA} denotes a Full Adder circuit while $s$ and $c$ denote the sum and carry bits respectively.
To summarize, parties execute $\ell$ instances of {\sf FA} in parallel to compute $\shrB{c}$ and $\shrB{s}$. The {\sf FA}'s are executed independently and require one round of communication. The final result is then computed as $\msb(2\shrB{c}+\shrB{s})$ by evaluating the bit extraction circuit~\cite{CCS:MohRin18, USENIX:PSSY21}.

 \subsection{Bit to Arithmetic}
 \label{sec:3pcMBit2A}
Protocol $\prot{\bitA}(\shrB{\bitb})$~(\boxref{fig:3pcMpiBitA}) enables computing $\shr{\bitb}$ of a bit $\bitb$ given its boolean sharing $\shrB{\bitb}$. Let $\arval{\bitb}$ denotes the value of $\bitb \in \bitset$ over the arithmetic ring $\Z{\ell}$. Using our sharing semantics,
\begin{equation}
	\label{eq:3pcMbitA}
		\arval{\bitb} = \arval{(\mk{\bitb} \xor \pad{\bitb}{})} = \arval{\mk{\bitb}} + \padR{\bitb}(1-2\arval{\mk{\bitb}})
\end{equation}

\begin{protocolbox}{$\prot{\bitA}(\shrB{\bitb})$}{Bit to Arithmetic conversion in $\Tthis$.}{fig:3pcMpiBitA}
	Let $\vl{u} = \padR{\bitb}$ and $\vl{v} = \arval{\mk{\bitb}}$.\\
	\detail{
		{\bf Input(s):} $\shrB{\bitb}$,~~~{\bf Output:} $\shr{\vl{y}} = \shr{\arval{\bitb}}$.
	}
	\justify 
	\vspace{-2mm}
	\algoHead{Preprocessing:} 
	\begin{enumerate} 
		\item $(P_1, P_3)$, $(P_2, P_3)$ and $(P_1, P_2)$ locally generate $\sgr{\cdot}$-shares of $\arval{(\pad{\bitb}{1})}$, $\arval{(\pad{\bitb}{2})}$ and $\arval{(\pad{\bitb}{3})}$ respectively~(\tabref{jsh3pcM}).
		\item Compute the $\sgr{\cdot}$-shares of $\arval{(\pad{\bitb}{1})}\arval{(\pad{\bitb}{2})}$ using $\prot{\MultPre}$.
		\item Locally compute $\sgr{\sigma}  = \sgr{\arval{(\pad{\bitb}{1})}} + \sgr{\arval{(\pad{\bitb}{2})}} -2\sgr{\arval{(\pad{\bitb}{1})}\arval{(\pad{\bitb}{2})}}$.
		\item Compute the $\sgr{\cdot}$-shares of $\arval{\sigma_1}\arval{(\pad{\bitb}{3})}$ using $\prot{\MultPre}$.
		\item Locally compute $\sgr{\vl{u}}  = \sgr{\sigma} + \sgr{\arval{(\pad{\bitb}{3})}} -2\sgr{\sigma\arval{(\pad{\bitb}{3})}}$.
	\end{enumerate}
	\justify
	\vspace{-2mm}
	\algoHead{Online:} Let $\vl{y} = \arval{\bitb}$.
	\begin{enumerate} 
		\item Locally compute the following:
		\begin{align*}
			P_1, P_3: \vl{y}_1  =  \vl{v} + \vl{u}^1 (1 - 2\vl{v})~~\Big|~~
			P_2, P_3: \vl{y}_2 =  \vl{u}^2 (1 - 2\vl{v})~~\Big|~~
			P_1, P_2: \vl{y}_3  =  \vl{u}^3 (1 - 2\vl{v})
		\end{align*}
		\item $(P_1, P_3), (P_2, P_3), (P_1, P_2)$ execute $\prot{\JSh}$ on $\vl{y}_1, \vl{y}_2, \vl{y}_3$ to generate the respective $\shr{\cdot}$-shares.
		\item Compute $\shr{\vl{y}} = \shr{\vl{y}_1} + \shr{\vl{y}_2} + \shr{\vl{y}_3}$.
	\end{enumerate}     
\end{protocolbox}

During preprocessing, parties locally generate $\sgr{\cdot}$-shares of $\arval{(\pad{\bitb}{1})}$, $\arval{(\pad{\bitb}{2})}$ and $\arval{(\pad{\bitb}{3})}$ similar to $\prot{\JSh}$~(\tabref{jsh3pcM}, ignore $\mk{}$ values). Then, $\sgr{\arval{\sigma}}$ can be computed in the preprocessing using two instances of $\prot{\MultPre}$ as given in \eqref{eq:3pcMbitASigma}.

\begin{align}
	\label{eq:3pcMbitASigma}
	\arval{\sigma_1} &= \arval{(\pad{\bitb}{1} \xor \pad{\bitb}{2})} = \arval{(\pad{\bitb}{1})} + \arval{(\pad{\bitb}{2})} -2\arval{(\pad{\bitb}{1})}\arval{(\pad{\bitb}{2})} \nonumber \\
	\arval{\sigma} &= \arval{(\sigma_1 \xor \pad{\bitb}{3})} = \arval{\sigma_1} + \arval{(\pad{\bitb}{3})} -2\arval{\sigma_1}\arval{(\pad{\bitb}{3})}
\end{align} 

The online phase consists of each pair of parties $(P_1, P_3)$, $(P_2, P_3)$ and $(P_1, P_2)$ locally computing an additive sharing of $\arval{\bitb}$ using \eqref{eq:3pcMbitA}, generating the corresponding $\shr{\cdot}$-sharing using $\prot{\JSh}$, and locally adding the shares to obtain $\shr{\arval{\bitb}}$.  

\begin{lemma}[Communication]
	\label{lemma:3pcMpibitA}
	Protocol $\prot{\bitA}$~(\boxref{fig:3pcMpiBitA}) requires $6 \ell$ bits of communication in preprocessing, and $1$ round and $3 \ell$ bits of communication in the online phase.
\end{lemma}
\begin{proof}
	During the preprocessing, generation of $\sgr{\cdot}$-shares of $\arval{(\pad{\bitb}{1})}$, $\arval{(\pad{\bitb}{2})}$ and $\arval{(\pad{\bitb}{3})}$ is local. Two instances of $\prot{\MultPre}$ are executed in the preprocessing incurring a communication of $6\ell$ bits. The online phase involves three instances of arithmetic joint sharing protocol  in parallel, resulting in $1$ round and a communication of $3\ell$ bits. 
\end{proof}

\subsubsection{Bit to Arithmetic:II}
\label{sec:3pcMdBit2A}
Similar to $\prot{\bitA}$ protocol, given the boolean sharings $\shrB{\bitb_1}, \shrB{\bitb_2}$, protocol $\prot{\dbitA}$ computes the arithmetic sharing of $\arval{(\bitb_1 \bitb_2)}$. Let $\Delta_{\bitb_1}$, $\Delta_{\bitb_2}$  denote the value $(1-2\arval{\mk{\bitb_1}})$, $(1-2\arval{\mk{\bitb_2}})$ respectively. Using \eqref{eq:3pcMbitA}, we can write

\begin{align}
	\label{eq:3pcMdbitA}
	\arval{(\bitb_1 \bitb_2)} &= \arval{(\mk{\bitb_1} \xor \pad{\bitb_1}{})} \arval{(\mk{\bitb_2} \xor \pad{\bitb_2}{})} 
	= (\arval{\mk{\bitb_1}} + \padR{\bitb_1}\Delta_{\bitb_1}) (\arval{\mk{\bitb_2}} + \padR{\bitb_2}\Delta_{\bitb_2}) \nonumber \\
    &= \arval{\mk{\bitb_1}}\arval{\mk{\bitb_2}} + \padR{\bitb_1}\arval{\mk{\bitb_2}}\Delta_{\bitb_1} + \padR{\bitb_2}\arval{\mk{\bitb_1}}\Delta_{\bitb_2} + \arval{(\pad{\bitb_1}{}\pad{\bitb_2}{})}\Delta_{\bitb_1}\Delta_{\bitb_2}
\end{align} 

During preprocessing, the $\sgr{\cdot}$-shares of $\padR{\bitb_1}$ and $\padR{\bitb_2}$ are computed similar to that of $\prot{\bitA}$~(\boxref{fig:3pcMpiBitA}). Parties then compute the $\sgr{\cdot}$-shares of $\arval{(\pad{\bitb_1}{}\pad{\bitb_2}{})}$ using another instance of $\prot{\MultPre}$. The online phase is similar to that of $\prot{\bitA}$ protocol.

\begin{lemma}[Communication]
	\label{lemma:3pcMpidbitA}
	Protocol $\prot{\dbitA}$ requires $15\ell$ bits of communication in preprocessing, and $1$ round and $3 \ell$ bits of communication in the online phase.
\end{lemma}

\subsection{Bit Injection}
\label{sec:3pcMBitInj}
Given the boolean sharing of a bit $\bitb$, denoted as $\shrB{\bitb}$, and the arithmetic sharing of $\vl{v} \in \Z{\ell}$, protocol $\prot{\bitinj}$ computes $\shr{\cdot}$-sharing of $\arval{\bitb}\vl{v}$. Let $\Delta_{\bitb}$ denote the value $(1-2\arval{\mk{\bitb}})$.
Similar to $\prot{\bitA}$, 
\begin{align}\label{eq:3pcMbitinj}
	\arval{\bitb} \vl{v} &= \arval{(\mk{\bitb} \xor \pad{\bitb}{})}(\mk{\vl{v}} - \pd{\vl{v}}{}) = (\arval{\mk{\bitb}} + \padR{\bitb}\Delta_{\bitb})(\mk{\vl{v}} - \pd{\vl{v}}{}) \nonumber \\
	&= \arval{\mk{\bitb}}\mk{\vl{v}} - \arval{\mk{\bitb}}\pd{\vl{v}}{} + \padR{\bitb}\mk{\vl{v}}\Delta_{\bitb} - \padR{\bitb}\pd{\vl{v}}{}\Delta_{\bitb}
\end{align} 

During the preprocessing, parties generates the $\sgr{\cdot}$-shares of $\padR{\bitb}$ similar to $\prot{\bitA}$ protocol. This is followed by generating the $\sgr{\cdot}$-shares of $\padR{\bitb}\pd{\vl{v}}{}$ using $\prot{\MultPre}$. The online phase is similar to that of $\prot{\bitA}$ protocol.

\begin{lemma}[Communication]
	\label{lemma:3pcMbitinj}
	Protocol $\prot{\bitinj}$ requires $9\ell$ bits of communication in preprocessing, and $1$ round and $3 \ell$ bits of communication in the online phase.
\end{lemma}

 \subsubsection{Sum of Bit Injections}
 \label{sec:3pcMSumBitInj}
Given $m$ pair of values in the shared form, $\{\shrB{\bitb_i}, \shr{\vl{v}_i}\}_{i \in [m]}$, the goal of $\prot{\bitinjS}$ is to compute the $\shr{\cdot}$-share of $\vl{z} = \sum_{i=1}^{m} \arval{\bitb_i} \cdot \vl{v_i}$. For this, parties execute the preprocessing corresponding to $m$ bit injections of the form $\arval{\bitb_i} \cdot \vl{v_i}$. 

In the online phase, parties locally compute an additive sharing of $\vl{z}_i$, corresponding to $\arval{\bitb_i} \cdot \vl{v_i}$ first. Instead of generating the $\shr{\cdot}$-sharing for each of the $m$ terms, parties locally add the shares and execute $\prot{\JSh}$ on the result. This results in an online communication independent of $m$.

\begin{lemma}[Communication]
	\label{lemma:3pcMSumBitInj}
	Protocol $\prot{\bitinjS}$ requires $m \cdot 9\ell$ bits of communication in preprocessing, and $1$ round and $3 \ell$ bits of communication in the online phase.
\end{lemma}

\subsubsection{Bit Injection:II}
\label{sec:3pcMBitInjII}
Similar to $\prot{\bitinj}$ protocol, given $\shrB{\bitb_1}, \shrB{\bitb_2}$ and $\shr{\vl{v}}$, protocol $\prot{\dbitA}$ computes the arithmetic sharing of $\arval{(\bitb_1 \bitb_2)}\vl{v}$. Let $\Delta_{\bitb_1}$, $\Delta_{\bitb_2}$  denote the value $(1-2\arval{\mk{\bitb_1}})$, $(1-2\arval{\mk{\bitb_2}})$ respectively. Using \eqref{eq:3pcMdbitA} and \eqref{eq:3pcMbitinj}, we can write

\begin{align}
	\label{eq:3pcMdbitinj}
	\arval{(\bitb_1 \bitb_2)} \vl{v} &= \arval{(\mk{\bitb_1} \xor \pad{\bitb_1}{})} \arval{(\mk{\bitb_2} \xor \pad{\bitb_2}{})} (\mk{\vl{v}} - \pd{\vl{v}}{}) \nonumber \\ 
	&= (\arval{\mk{\bitb_1}} + \padR{\bitb_1}\Delta_{\bitb_1}) (\arval{\mk{\bitb_2}} + \padR{\bitb_2}\Delta_{\bitb_2})(\mk{\vl{v}} - \pd{\vl{v}}{}) \nonumber \\
	&= \arval{\mk{\bitb_1}}\arval{\mk{\bitb_2}}\mk{\vl{v}}  + \padR{\bitb_1}\arval{\mk{\bitb_2}}\mk{\vl{v}} \Delta_{\bitb_1} + \padR{\bitb_2}\arval{\mk{\bitb_1}}\mk{\vl{v}}\Delta_{\bitb_2} + \arval{(\pad{\bitb_1}{}\pad{\bitb_2}{})}\mk{\vl{v}}\Delta_{\bitb_1}\Delta_{\bitb_2} \nonumber \\
	&~~~- \pd{\vl{v}}{}\arval{\mk{\bitb_1}}\arval{\mk{\bitb_2}} - \padR{\bitb_1}\pd{\vl{v}}{}\arval{\mk{\bitb_2}}\Delta_{\bitb_1} - \padR{\bitb_2}\pd{\vl{v}}{}\arval{\mk{\bitb_1}}\Delta_{\bitb_2} - \arval{(\pad{\bitb_1}{}\pad{\bitb_2}{})}\pd{\vl{v}}{}\Delta_{\bitb_1}\Delta_{\bitb_2}
\end{align} 

During preprocessing, the $\sgr{\cdot}$-shares of $\padR{\bitb_1}$ and $\padR{\bitb_2}$ are computed similar to that of $\prot{\bitA}$~(\boxref{fig:3pcMpiBitA}). Parties then compute the $\sgr{\cdot}$-shares of $\arval{(\pad{\bitb_1}{}\pad{\bitb_2}{})}$, $\padR{\bitb_1}\pd{\vl{v}}$, $\padR{\bitb_2}\pd{\vl{v}}$ and $\arval{(\pad{\bitb_1}{}\pad{\bitb_2}{})}\pd{\vl{v}}$ using four instances of $\prot{\MultPre}$. The online phase is similar to that of $\prot{\bitA}$ protocol.

\begin{lemma}[Communication]
	\label{lemma:3pcMdbitinj}
	Protocol $\prot{\dbitinj}$ requires $24\ell$ bits of communication in preprocessing, and $1$ round and $2 \ell$ bits of communication in the online phase.
\end{lemma}

\subsection{Truncation Pair Generation~($\prot{\trgen}$)}
\label{sec:trgen3pcM}
Protocol $\prot{\trgen}$~(\boxref{fig:trgen3pcM}) allows parties to generate a truncation pair of the form $(\sgr{\vl{r}}, \shr{\vl{r}^{\vl{t}}})$ for a random $\vl{r} \in_R \Z{\ell}$. Analogous to the approach of ABY3~\cite{CCS:MohRin18}, parties non-interactively generate the boolean sharing of an $\ell$-bit value $\vl{r}$ first. Parties then discard the shares for the lower $x$ bit positions to obtain the boolean shares of the truncated value denoted by $\vl{r}^{\vl{t}}$. To obtain the arithmetic shares of the truncation pair, we do not rely on the approach of ABY3 as it requires more rounds. Instead, we implicitly perform a {\em boolean to arithmetic} conversion using techniques from bit to arithmetic protocol $\prot{\bitA}$.

\begin{protocolbox}{$\prot{\trgen}$}{Truncation pair generation in $\Tthis$.}{fig:trgen3pcM}
	\justify 
	Let $i \in \{0,\ldots, \ell-1\}$  and $j \in \{0,\ldots, \ell-1 - x\}$. Here $x$ denotes the precision in FPA semantics.
	\begin{enumerate} 
		\item $P_s, P_3$ for $s \in \{1,2\}$ sample $\ell$-bits, denoted by $\vl{r}_s[i]$.
		\item Define $\ell$-bit value $\vl{r} = \vl{r}_1 \xor \vl{r}_2$. i.e. $\vl{r}[i] = \vl{r}_1[i] \xor \vl{r}_2[i]$.
		\item $P_s, P_3$ for $j \in \{1,2\}$ execute $\prot{\JSh}$ on $\arval{(\vl{r}_s[i])}$ to generate the respective $\shr{\cdot}$-shares.
		\item Locally compute $\sgr{\cdot}$-shares of $\arval{(\vl{r}_1[i])}$ and $\arval{(\vl{r}_2[i])}$\footnote{by discarding the $\mk{}$ value that is set to $0$ as per \tabref{jsh3pcM}}. 
		\item Define $\ell$-sized vectors $\vct{a}, \vct{b}$ as: $\vl{a}_j = 2^{i+1} \arval{(\vl{r}_1[i])}$ and $\vl{b}_i = \arval{(\vl{r}_2[i])}$. 
		\item Define $(\ell - x)$-sized vectors $\vct{c}, \vct{d}$ as: $\vl{c}_j = 2^{j+1} \arval{(\vl{r}_1[j+x])}$ and $\vl{b}_j = \arval{(\vl{r}_2[j+x])}$.
		\item Locally compute $\sgr{\vct{a}}, \sgr{\vct{b}}, \sgr{\vct{c}}, \sgr{\vct{d}}$.
		\item Compute the $\sgr{\cdot}$-shares of $\vl{x} = \vct{a} \band \vct{b}$ and $\vl{y} = \vct{c} \band \vct{d}$ using $\prot{\dotpPre}$ protocol\footnote{$\shr{\vl{y}}$ can be computed by locally setting $\mk{\vl{y}} = 0$}.
		\item Locally compute $\sgr{\vl{r}} = \sum_{i=0}^{\ell-1} 2^i  ( \sgr{\arval{(\vl{r}_1[i])}} + \sgr{\arval{(\vl{r}_2[i])}} ) - \sgr{\vl{x}}$.
		\item Locally compute $\shr{\vl{r}^\vl{t}} = \sum_{j=0}^{\ell-1 - x} 2^j  ( \shr{\arval{(\vl{r}_1[j+x])}} + \shr{\arval{(\vl{r}_2[j+x])}} ) - \shr{\vl{y}}$.
	\end{enumerate}     
\end{protocolbox}

Concretely, $P_1, P_3$ sample an $\ell$-bit value $\vl{r}_1$ while $P_2, P_3$ sample $\vl{r}_2$. For the $i^{th}$ bit position, define $\vl{r}[i] = \vl{r}_1[i] \xor \vl{r}_2[i]$ for $i \in \{0,\ldots, \ell-1\}$. For $\vl{r}$ defined as above, we have $\vl{r}^{\vl{t}}[j] = \vl{r}_1[j+x] \xor \vl{r}_2[j+x]$ for $j \in \{0,\ldots, \ell-1-x\}$. Further,
\begin{small}
\begin{align}
	\label{eq:3pcmtrA}
	\vl{r}  &= \sum_{i=0}^{\ell-1} 2^i \vl{r}[i] = \sum_{i=0}^{\ell-1} 2^i \left( \vl{r}_1[i] \xor \vl{r}_2[i] \right) =   \sum_{i=0}^{\ell-1} 2^i \left( \arval{(\vl{r}_1[i])} + \arval{(\vl{r}_2[i])} - 2 \arval{(\vl{r}_1[i])} \cdot \arval{(\vl{r}_2[i])} \right) \nonumber \\
	         & = \sum_{i=0}^{\ell-1} 2^i \left( \arval{(\vl{r}_1[i])} + \arval{(\vl{r}_2[i])} \right) - \sum_{i=0}^{\ell-1} \left(  \left( 2^{i+1}  \arval{(\vl{r}_1[i])} \right) \cdot \arval{(\vl{r}_2[i])} \right)
\end{align}
\end{small}

Similarly, for  $\vl{r}^{\vl{t}}$, 
\begin{small}
\begin{align}
	\label{eq:3pcmtrB}
	\vl{r}^{\vl{t}}  = \sum_{j=0}^{\ell-1-x} 2^j \left( \arval{(\vl{r}_1[j+x])} + \arval{(\vl{r}_2[j+x])} \right) - \sum_{j=0}^{\ell-1-x} \left(  \left( 2^{j+1}  \arval{(\vl{r}_1[j+x])} \right) \cdot \arval{(\vl{r}_2[j+x])} \right)
\end{align}
\end{small}

Given the boolean shares, parties can evaluate \eqref{eq:3pcmtrA} and \eqref{eq:3pcmtrB} using two instances of $\prot{\dotpPre}$ as shown in~\boxref{fig:trgen3pcM}.

\begin{lemma}[Communication]
	\label{lemma:3pcMtrgen}
	Protocol $\prot{\trgen}$~(\boxref{fig:trgen3pcM}) requires $6\ell$ bits of communication.
\end{lemma}

\subsection{Equality Test~($\prot{\eql}$)}
To check whether $\vl{a} \iseq \vl{b}$ or not, given $\shr{\vl{a}}, \shr{\vl{b}}$, $\prot{\eql}$ proceeds with parties locally computing $\shr{\vl{y}} = \shr{\vl{a}} - \shr{\vl{b}}$. According to our sharing semantics, $\vl{y}$ can be written as $\vl{y} = \vl{y}_1 - \vl{y}_2$ where $\vl{y}_1 = \mk{\vl{y}} -  \pad{\vl{y}}{3}$ and $\vl{y}_2 =  \pad{\vl{y}}{1} + \pad{\vl{y}}{2}$. 

During preprocessing, $(P_1,P_3)$ and $(P_2,P_3)$ generate the $\shrB{\cdot}$-shares of $\pad{\vl{y}}{1}$ and $\pad{\vl{y}}{2}$ respectively using $\prot{\JSh}$. Parties then compute $\shrB{\vl{y}_2}$ using a boolean adder~(PPA) circuit. 
During the online phase, $P_1,P_2$ generate $\shrB{\vl{y}_1}$ using $\prot{\JSh}$. Note that $\vl{a} = \vl{b}$ implies $\vl{y}_1 = \vl{y}_2$ and hence all the bits of $\vl{v} = \overline{(\vl{y}_1 \xor \vl{y}_2)}$ should be $1$. As mentioned in the introduction of Part II~(\ref{chap:layer2_intro}), parties use four input AND gates and a tree structure, where $4$ bits are taken at a time and the AND of them is computed in one go. 

\section{Mixed Protocol Framework}
\label{sec:3pcMixFrame}
\tabref{3pcMConv} compares our sharing conversions with ABY3~\cite{CCS:MohRin18}. For uniformity, we consider a function, {\sf F}, to be computed on an $\ell$-bit inputs $\vl{x}, \vl{y}$ using a garbled circuit (GC) in the mixed framework, which gives an $\ell$-bit output $\vl{z} = \mathsf{F(\vl{x}, \vl{y})}$, where $\ell$ denotes the ring size in bits. Let $\Grb{F}$  denote the corresponding GC. In the table, $\Grb{Sn}$ denotes a {\sf n}-input garbled subtraction circuit; $\Grb{An}$ denotes {\sf n}-input garbled addition circuit; $\GrbD{}$ denotes the garbled circuit with decoding information; $\Grb{n_1\times1,\ldots,n_m \times m}$ denotes ${\sf n_i}$ instances of GC $\Grb{i}$ for $i \in \{1,\dots,{\sf m}\}$ and $\Size{\Grb{n_1\times1,\ldots,n_m \times m}}$ denotes its size. 

\begin{table}[htb!]
	\centering
	\resizebox{0.95\textwidth}{!}{
		\begin{NiceTabular}{rr|rrrr|rrrr}
			\toprule 
			\Block{2-1}{Variant\tabularnote{Notations: $\ell$ - size of ring in bits, $\kappa$ - computational security parameter, 'pre' - preprocessing, 'on' - online.}} 
			& \Block{2-1}{Conversion\tabularnote{'A' - arithmetic, 'B' - boolean, 'G' - Garbled.}}
			& \Block[c]{1-4}{ABY3~\cite{CCS:MohRin18}} & & &
			& \Block[c]{1-4}{$\Tthis$} & & & \\ \cmidrule{3-10}
			& & \Block[c]{1-2}{Comm.\textsubscript{pre}} & & Comm.\textsubscript{on}  & Rounds\textsubscript{on} 
			& \Block[c]{1-2}{Comm.\textsubscript{pre}} & & Comm.\textsubscript{on}  & Rounds\textsubscript{on} \\
			\midrule
			\Block{4-1}{2 GC}
			& A-G-A & \Block[r]{4-1}{($2 \ell \kappa$)\\+}  & $2\Size{\GrbD{2 \times A3,S3,F}}$ 
			& \Block{4-1}{$10\ell\kappa$} & \Block{4-1}{$2$}
			& \Block[r]{4-1}{($12 \ell \kappa$)\\+}  & $2\Size{\GrbD{2 \times A3, A2, F}}$ 
			& \Block{4-1}{$4\ell\kappa + \ell$} & \Block{4-1}{$1$} \\
			& A-G-B &  & $2\Size{\Grb{2 \times A3,F}}$   &  & &  & $2\Size{\GrbD{2 \times A3,F}}$ &  & \\
			& B-G-A &  & $2\Size{\GrbD{S3,F}}$              &  & &  & $2\Size{\GrbD{A2,F}}$              &  & \\
			& B-G-B &  & $2\Size{\Grb{F}}$                      &  & &  & $2\Size{\GrbD{F}}$                   &  & \\
			\midrule
			%
			\Block{4-1}{1 GC}
			& A-G-A & \Block[r]{4-1}{($\ell \kappa$)\\+}  & $\Size{\GrbD{2 \times A3,S3,F}}$ 
			& \Block{4-1}{$5\ell\kappa$} & \Block{4-1}{$2$}
			& \Block[r]{4-1}{($6 \ell \kappa$)\\+}  & $\Size{\GrbD{2 \times A3, A2, F}}$ 
			& \Block{4-1}{$2\ell\kappa + 2\ell$} & \Block{4-1}{$2$} \\
			& A-G-B &  & $\Size{\Grb{2 \times A3,F}}$   &  & &  & $\Size{\GrbD{2 \times A3,F}}$ &  & \\
			& B-G-A &  & $\Size{\GrbD{S3,F}}$              &  & &  & $\Size{\GrbD{A2,F}}$    &  & \\
			& B-G-B &  & $\Size{\Grb{F}}$                      &  & &  & $\Size{\GrbD{F}}$                   &  & \\
			\midrule
			\Block{2-1}{Others\tabularnote{$\vl{u_1} = 3\vl{n_2} + 12\vl{n_3} + 33\vl{n_4}$, $\vl{u_2} = \vl{n_2} + \vl{n_3} + \vl{n_4}$ denote the number of AND gates in the optimized adder circuit~\cite{USENIX:PSSY21} with 2, 3, 4 inputs, respectively. For $\ell = 64$, $\vl{n_2} = 216, \vl{n_3} = 184, \vl{n_4}=179$.}}
			& A-B & \Block[r]{1-2}{$12\ell + 12\ell \log \ell$} &  & $9\ell + 9 \ell \log \ell$ & $1 + \log \ell$
			& \Block[r]{1-2}{$\vl{u_1} + 6\ell + 6\ell\log \ell$} &  & $3\vl{u_2}$ & $\log_4 \ell$ \\
			& B-A &  \Block[r]{1-2}{$12\ell + 12\ell \log \ell$} &  & $9\ell + 9 \ell \log \ell$ & $1 + \log \ell$
			& \Block[r]{1-2}{$6\ell^2$} &  & $3\ell$ & $1$ \\
			%
            \bottomrule
\end{NiceTabular}
}
\caption{Mixed protocol conversions of ABY3~\cite{CCS:MohRin18} and $\Tthis$.}\label{tab:3pcMConv}
\end{table}

\subsection{Conversions involving Garbled World} 
\label{sec:3pcMconv2gc}
Assume the GC is required to compute a function $f$ on inputs $\vl{x}, \vl{y} \in \Z{\ell}$ and let the output be $f(\vl{x}, \vl{y})$. All the conversions described are for the 2 GC variant. Conversions for the 1 GC variant are straightforward, hence we omit the details.

\paragraph{Case I: Boolean-Garbled-Boolean}
Since the inputs to the GC are available in boolean form, say $\shrB{\vl{x}}, \shrB{\vl{y}}$, parties generate $\shrC{\vl{x}}, \shrC{\vl{y}}$ by invoking the garbled sharing protocol $\pigsh$.
$(P_1, P_3)$ sample $\vl{R}_1 \in \Z{\ell}$ to mask the function output, $f(\vl{x}, \vl{y})$, and generate $\shrB{\vl{R}_1}$ and $\shrG{\vl{R}_1}$. Similarly, $(P_2, P_3)$ sample $\vl{R}_2 \in \Z{\ell}$ and generate $\shrB{\vl{R}_2}$ and $\shrG{\vl{R}_2}$.
Garblers $P_g \in \{P_1, P_3\}$ garble the circuit which computes $\vl{z} = f(\vl{x}, \vl{y}) \xor \vl{R}_1  \xor \vl{R}_2$, and send the GC along with the decoding information to evaluator $P_1$. Analogous steps are performed for evaluator $P_2$. Upon GC evaluation and output decoding, evaluators obtain $\vl{z} = f(\vl{x}, \vl{y}) \xor \vl{R}_1  \xor \vl{R}_2$, and jointly boolean share $\vl{z}$ to generate $\shrB{\vl{z}}$. Parties then compute $\shrB{f(\vl{x}, \vl{y})} = \shrB{\vl{z}} \xor \shrB{\vl{R}_1} \xor \shrB{\vl{R}_2}$.  

\paragraph{Case II: Boolean-Garbled-Arithmetic}
This is similar to {\em Case I} except that the circuit which computes $\vl{z} = f(\vl{x}, \vl{y}) + \vl{R}_1  + \vl{R}_2$ is garbled instead. Boolean sharing of $\vl{z}$ is replaced with arithmetic, followed by computing $\shr{f(\vl{x}, \vl{y})} = \shr{\vl{z}} - \shr{\vl{R}_1} - \shr{\vl{R}_2}$.

\paragraph{Cases III \& IV: Input in Arithmetic Sharing}
The function to be computed $f(\vl{x}, \vl{y})$, is modified as $f^{\prime}(\mk{\vl{x}}, \pad{\vl{x}}{1}, \pad{\vl{x}}{2}, \pad{\vl{x}}{3}, \mk{\vl{y}}, \pad{\vl{y}}{1}, \pad{\vl{y}}{2}, \pad{\vl{y}}{3}) = f(\mk{\vl{x}}-\pad{\vl{x}}{1}-\pad{\vl{x}}{2}-\pad{\vl{x}}{3}, \mk{\vl{y}}-\pad{\vl{y}}{1}-\pad{\vl{y}}{2}-\pad{\vl{y}}{3})$ where inputs $\vl{x}, \vl{y}$ are replaced by the sets $\{\mk{\vl{x}}, \pad{\vl{x}}{1}, \pad{\vl{x}}{2}, \pad{\vl{x}}{3}\}$, $\{\mk{\vl{y}}, \pad{\vl{y}}{1}, \pad{\vl{y}}{2}, \pad{\vl{y}}{3}\}$. The circuit to be garbled thus, corresponds to the function $f^{\prime}$. Parties generate $\shrG{\cdot}$-shares via $\pigsh$, following which, parties proceed with the rest of the computation whose steps are similar to {\em Case I}, and {\em II}, depending on the requirement on the output sharing. For the instance with $P_1$ as the evaluator,  function $f^{\prime}$ can be further optimized as $f(\av{\vl{x}}-\pad{\vl{x}}{1}-\pad{\vl{x}}{3}, \av{\vl{y}}-\pad{\vl{y}}{1}-\pad{\vl{y}}{3})$ with $\av{\vl{x}} = \mk{\vl{x}} - \pad{\vl{x}}{2}$ and $\av{\vl{y}} = \mk{\vl{y}} - \pad{\vl{y}}{2}$. Similar optimization can be done for the other garbling instance as well.

\subsection{Other Conversions} 
\label{sec:3pcMotherconv}

\paragraph{Arithmetic to Boolean}
To convert arithmetic sharing of $\vl{v} \in \Z{\ell}$ to boolean, observe that $\vl{v} = \vl{v}_1 + \vl{v}_2$ where $\vl{v}_1 = \mk{\vl{v}}$ and $\vl{v}_2 = - \pad{\vl{v}}{}$. Thus, $\shrB{\vl{v}}$  can be computed as $\shrB{\vl{v}} = \shrB{\vl{v}_1} + \shrB{\vl{v}_2}$. For this, parties generate $\shrB{\vl{v}_2}$ in the preprocessing, and $\shrB{\vl{v}_1}$ can be generated in the online locally by setting $\mk{\vl{v}_1} = \vl{v}_1$ and $\pad{\vl{v}_1}{} = \pad{\vl{v}_2}{} = \pad{\vl{v}_3}{} = 0$. The protocol appears in \boxref{fig:3pcMpiab}. Boolean addition, when instantiated using the adder of ABY2.0~\cite{USENIX:PSSY21}, requires $\log_4(\ell)$ rounds.

\begin{protocolbox}{$\piab$}{Arithmetic to Boolean Conversion in $\Tthis$.}{fig:3pcMpiab}
	Let $\vl{v}_1 = \mk{\vl{v}}$ and $\vl{v}_2 = - \pad{\vl{v}}{}$.
	\justify
	\algoHead{Preprocessing:} 
	\begin{enumerate} 
		\item Non-interactively generate $\shrB{\cdot}$-shares of $\vl{u}_i = - \pad{\vl{v}}{i}$ for $i \in \{1,2,3\}$ using $\prot{\JSh}$~(\secref{jsh3pcM}).
		\item Evaluate {\sf FA}$(\vl{v}_1[i], \vl{v}_2[i], \vl{v}_3[i]) \rightarrow (c[i],s[i])$ for $i \in \{0,\ldots,\ell-1\}$ to generate $\shrB{c[i]}$ and $\shrB{s[i]}$.
		\item Compute $2\shrB{c}+\shrB{s}$ using a boolean adder circuit~\cite{CCS:MohRin18, USENIX:PSSY21}.
	\end{enumerate}
	\justify
	\vspace{-2mm}
	\algoHead{Online:}
	\begin{enumerate} 
		\item Locally generate $\shrB{\vl{v}_1}$ as $\mk{\vl{v}_1} = \vl{v}_1$ and $\pad{\vl{v}_1} = \pad{\vl{v}_2} = \pad{\vl{v}_3} = 0$.
		\item Compute $\shrB{\vl{v}} = \shrB{\vl{v}_1} + \shrB{\vl{v}_2}$ using a boolean adder circuit~\cite{USENIX:PSSY21}.
	\end{enumerate}
\end{protocolbox}

To generate $\shrB{\vl{v}_2}$, let $\vl{v}_2 = \vl{u}_1 + \vl{u}_2 + \vl{u}_3$ where $\vl{u}_i = - \pad{\vl{v}}{i}$ for $i \in \{1,2,3\}$. Parties non-interactively generate the $\shrB{\cdot}$-shares of  $\vl{u}_1, \vl{u}_2, \vl{u}_3$ using joint sharing protocol~(\secref{jsh3pcM}). For a full adder circuit {\sf FA}$(\vl{v}_1[i], \vl{v}_2[i], \vl{v}_3[i]) \rightarrow (c[i],s[i])$ for $i \in \{0,\ldots,\ell-1\}$, it has been shown in ABY3~\cite{CCS:MohRin18} that $\vl{v}_2 = 2c + s$ where $s$ and $c$ denote the sum and carry bits respectively. Parties execute $\ell$ instances of {\sf FA} in parallel to compute $\shrB{c}$ and $\shrB{s}$. The {\sf FA}'s are executed independently and require one round of communication. The final result is then computed as $2\shrB{c}+\shrB{s}$ by evaluating a boolean adder circuit~\cite{CCS:MohRin18, USENIX:PSSY21}.

\paragraph{Boolean to Arithmetic} 
To convert a boolean sharing of $\vl{v} \in \Z{\ell}$ into an arithmetic sharing, note that 
\begin{small}
	\begin{align*}
		\vl{v} = \sum_{i=0}^{\ell - 1} 2^{i} \vl{v}[i] = \sum_{i=0}^{\ell - 1} 2^{i} (\pad{\vl{v}[i]}{} \xor \mk{\vl{v}[i]}) 
		= \sum_{i=0}^{\ell - 1} 2^{i} \left( \arval{\mk{\vl{v}[i]}} +  \padR{\vl{v}[i]} (1 - 2\arval{\mk{\vl{v}[i]}}) \right) 
	\end{align*}
\end{small}
where $\arval{{\pd{\vl{v}[i]}}}, \arval{\mk{\vl{v}[i]}}$ denote the arithmetic value of bits ${\pd{\vl{v}[i]}}, {\mk{\vl{v}[i]}}$ over the ring $\Z{\ell}$. 
For each bit $\vl{v}[i]$ of $\vl{v}$, parties generate the $\sgr{\cdot}$-shares of $\arval{{{\pd{\vl{v}[i]}}}}$ in the preprocessing, similar to $\prot{\bitA}$~(\boxref{fig:3pcMpiBitA}). During the online phase, additive shares for each bit $\vl{v}[i]$ are locally computed similar to $\prot{\bitA}$. Parties then multiply the $i$th share with $2^i$ and locally add up to obtain an additive sharing of $\vl{v}$. The rest of the steps are similar to $\prot{\bitA}$, and the formal protocol appears in \boxref{fig:3pcMpiba}. 

\begin{protocolbox}{$\piba(\Partyset, \shrB{\vl{v}})$}{Boolean to Arithmetic Conversion in $\Tthis$.}{fig:3pcMpiba}
	Let $\vl{v}[i]$ denote the $i$th bit of $\vl{v}$. Let $\vl{p}_i = \arval{\mk{\vl{v}[i]}}$, and $\vl{q}_i = \padR{\vl{v}[i]}$. \\
	\justify 
	\vspace{-2mm}
	\algoHead{Preprocessing:} 
	\begin{enumerate} 
		\item For $i \in \{0, 1, \ldots, \ell-1 \} $, execute the preprocessing of $\prot{\bitA}$ (\boxref{fig:3pcMpiBitA}) for each bit $\vl{v}[i]$, to generate $\sgr{{\vl{q}_i}} = (\vl{q}_i^1, \vl{q}_i^2, \vl{q}_i^3)$.
	\end{enumerate}
	\justify
	\vspace{-2mm}
	\algoHead{Online:} Let $\vl{y}_i = \arval{(\vl{v}[i])}$ and $\vl{y}$ denotes the arithmetic equivalent of $\vl{v}$.
	\begin{enumerate}
		\item Locally compute the following:
		\begin{align*}
			P_1, P_3: \vl{y}^1 = \sum_{i=0}^{\ell-1} 2^i \vl{y}_i^1  &=  \sum_{i=0}^{\ell-1} 2^i (\vl{p}_i + \vl{q}_i^1 (1 - 2\vl{p}_i)) \\
			P_2, P_3: \vl{y}^2 = \sum_{i=0}^{\ell-1} 2^i \vl{y}_i^2 &=  \sum_{i=0}^{\ell-1} 2^i (\vl{q}_i^2 (1 - 2\vl{p}_i))  \\
			P_1, P_2: \vl{y}^3 = \sum_{i=0}^{\ell-1} 2^i \vl{y}_i^3  &=  \sum_{i=0}^{\ell-1} 2^i (\vl{q}_i^3 (1 - 2\vl{p}_i)) 
		\end{align*}
		\item $(P_1, P_3), (P_2, P_3), (P_1, P_2)$ execute $\prot{\JSh}$ on $\vl{y}^1, \vl{y}^2, \vl{y}^3$ to generate the respective $\shr{\cdot}$-shares.
		\item Locally compute $\shr{\vl{y}} = \shr{\vl{y}^1} + \shr{\vl{y}^2} + \shr{\vl{y}^3}$.
	\end{enumerate}     
\end{protocolbox}

\chapter{$\Fthis$: 4PC Fair and Robust Protocols}
\label{chap:layer2_4pc}
This chapter provides details for the Layer II blocks of our 2PC framework $\Fthis$. Details for the Layer I blocks are provided in chapter~\ref{chap:layer1_4pc}.

\section{Building Blocks}
\label{sec:layerIIblocks4pc}

\subsection{Dot Product~(Scalar Product)}
\label{sec:4pcDotp}
Given $\shr{\vct{a}}, \shr{\vct{b}}$ with $|\vct{a}| = |\vct{b}| = \vl{d}$, protocol $\prot{\dotp}$~(\boxref{fig:piDotP4pcM}) computes $\shr{\vl{z}}$ such that $\vl{z} = (\vct{a} \band \vct{b})^{\vl{t}}$ if truncation is enabled, else $\vl{z} = \vct{a} \band \vct{b}$.
For this, we combine the partial products from the multiplication protocol across $\vl{d}$ multiplications and communicate them in a single shot. This makes the communication cost of the dot product independent of the vector size. The protocols for robust setting follows similarly from $\FthisA$ and $\FthisB$. 

\medskip
\begin{protocolsplitbox}{$\prot{\dotp}(\vct{a}, \vct{b}, \isTr)$}{Dot Product with / without Truncation in $\Fthis$.}{fig:piDotP4pcM}
	Let $\isTr$ be a bit that denotes whether truncation is required ($\isTr =1$) or not~($\isTr=0$). \\
	\detail{
		{\bf Input(s):} $\shr{\vct{a}}, \shr{\vct{b}}$.\\
		{\bf Output:} $\shr{\vl{o}}$ where $\vl{o} = \vl{z}^{\vl{t}}$ if $\isTr = 1$ and $\vl{o} = \vl{z}$ if $\isTr = 0$ and $\vl{z} = \vct{a} \band \vct{b} = \sum_{i = 1}^{\vl{d}} {\vl{a}}_i {\vl{b}}_i$.
	}
	\justify 
	\vspace{-2mm}
	\algoHead{Preprocessing:} 
	\begin{enumerate} 
		\item Locally compute the following:
		\begin{align*}
			P_0, P_1: \gm{\vct{a}\vct{b}}{1} &= \sum_{i=1}^{\vl{d}} (\pad{\vl{a}_i}{1} \pad{\vl{b}_i}{3} + \pad{\vl{a}_i}{3} \pad{\vl{b}_i}{1} + \pad{\vl{a}_i}{3} \pad{\vl{b}_i}{3})\\
			P_0, P_2: \gm{\vct{a}\vct{b}}{2} &= \sum_{i=1}^{\vl{d}} (\pad{\vl{a}_i}{2} \pad{\vl{b}_i}{3} + \pad{\vl{a}_i}{3} \pad{\vl{b}_i}{2} + \pad{\vl{a}_i}{2} \pad{\vl{b}_i}{2})\\
			P_0, P_3: \gm{\vct{a}\vct{b}}{3} &= \sum_{i=1}^{\vl{d}} (\pad{\vl{a}_i}{1} \pad{\vl{b}_i}{2} + \pad{\vl{a}_i}{2} \pad{\vl{b}_i}{1} + \pad{\vl{a}_i}{1} \pad{\vl{b}_i}{1})
		\end{align*}
		\item $P_0, P_3$ and $P_j$ sample random ${\vl{u}}^j \in_R \Z{\ell}$ for $j \in \{1,2\}$. Let ${\vl{u}^1} + \vl{u}^2 = \gm{\vct{a}\vct{b}}{3} - \vl{r}$ for a random $\vl{r} \in_R \Z{\ell}$.  
		\item $P_0, P_3$ compute $\vl{r} =  \gm{\vct{a}\vct{b}}{3} - {\vl{u}^1} - \vl{u}^2$ and set $\vl{q} = \vl{r}^{\vl{t}}$  if $\isTr = 1$, else set $\vl{q} = \vl{r}$. $P_0, P_3$ execute $\prot{\JSh}(P_0, P_3, \vl{q})$ to generate $\shr{\vl{q}}$.
		\item  $P_0, P_1, P_2$ sample random ${\vl{s}}_1, {\vl{s}}_2 \in_R \Z{\ell}$ and set ${\vl{s}} = {\vl{s}}_1 + {\vl{s}} _2$\footnote{For the fair protocol, it is enough for $P_0, P_1, P_2$ to sample ${\vl{s}}$ directly.}. 
		$P_0$ sends $\vl{w} = \gm{\vct{a}\vct{b}}{1} + \gm{\vct{a}\vct{b}}{2} + {\vl{s}}$ to $P_3$.
	\end{enumerate}
	\justify
	\vspace{-2mm}
	\algoHead{Online:} Let $\vl{y} = (\vl{z} - \vl{r}) - \sum_{i=1}^{\vl{d}} \mk{\vl{a}_i \vl{b}_i}$.
	\begin{enumerate} 
		\item Locally compute the following:
		\begin{align*}
			P_1: \vl{y}_1 &= \sum_{i=1}^{\vl{d}} (- \pad{\vl{a}_i}{1} \mk{\vl{b}_i} - \pad{\vl{b}_i}{1} \mk{\vl{a}_i}) + \gm{\vct{a}\vct{b}}{1} + {\vl{u}}^1\\
			P_2: \vl{y}_2 &= \sum_{i=1}^{\vl{d}} (- \pad{\vl{a}_i}{2} \mk{\vl{b}_i} - \pad{\vl{b}_i}{2} \mk{\vl{a}_i}) + \gm{\vct{a}\vct{b}}{2} + {\vl{u}}^2\\
			P_1, P_2: \vl{y}_3 &= \sum_{i=1}^{\vl{d}} (- \pad{\vl{a}_i}{3} \mk{\vl{b}_i} - \pad{\vl{b}_i}{3} \mk{\vl{a}_i})
		\end{align*}
		\item $P_1$ sends $\vl{y}_1$ to $P_2$, while $P_2$ sends $\vl{y}_2$ to $P_1$, and they locally compute $\vl{z} - \vl{r} = (\vl{y}_1 + \vl{y}_2 + \vl{y}_3) + \sum_{i=1}^{\vl{d}} \mk{\vl{a}_i \vl{b}_i}$.
		\item If $\isTr = 1$, $P_1, P_2$ set $\vl{p} = (\vl{z} - \vl{r})^{\vl{t}}$, else $\vl{p} = \vl{z} - \vl{r}$. $P_1, P_2$ execute $\prot{\JSh}(P_1, P_2, \vl{p})$ to generate $\shr{\vl{p}}$. 
		\item Parties locally compute $\shr{\vl{o}} = \shr{\vl{p}} + \shr{\vl{q}}$. Here $\vl{o} = \vl{z}^{\vl{t}}$ if $\isTr = 1$ and $\vl{z}$ otherwise.
		\item {\em Verification:} $P_3$ computes $\vl{v} = \sum_{i=1}^{\vl{d}} (- (\pad{\vl{a}_i}{1} + \pad{\vl{a}_i}{2}) \mk{\vl{b}_i} - (\pad{\vl{b}_i}{1} + \pad{\vl{b}_i}{2}) \mk{\vl{a}_i})+ {\vl{u}^1} + \vl{u}^2 + \vl{w}$ and sends $\Hash(\vl{v})$ to $P_1$ and $P_2$. Parties $P_1, P_2$ $\abort$ iff $\Hash(\vl{v}) \ne \Hash(\vl{y}_1 + \vl{y}_2 + {\vl{s}})$.
	\end{enumerate}     
\end{protocolsplitbox}

\begin{lemma}[Communication]
	\label{lemma:4pcMpidotpf}
	Protocol $\prot{\dotp}$~(\boxref{fig:piDotP4pcM})~(in $\Fthis$) requires $2 \ell$ bits of communication in preprocessing, and $1$ round and $3 \ell$ bits of communication in the online phase.
\end{lemma}

\begin{lemma}[Communication]
	\label{lemma:4pcMpidotpr}
	Protocol $\prot{\dotp}$~(in $\FthisB$) requires $3 \ell$ bits of communication in preprocessing, and $1$ round and $3 \ell$ bits of communication in the online phase.
\end{lemma}

 \subsection{Bit Extraction}
 \label{sec:4pcBitExt}

To compute most significant bit~($\msb$) of the value $\vl{v}$, note that $\vl{v} = (\mk{\vl{v}} - \pad{\vl{v}}{3}) + (- \pad{\vl{v}}{1} - \pad{\vl{v}}{2})$ as per the sharing semantics~(cf.~\tabref{4pcsharing}). $P_0, P_3$ execute $\protB{\JSh}$ on $(- \pad{\vl{v}}{1} - \pad{\vl{v}}{2})$ during the preprocessing, while $P_0, P_3$ execute $\protB{\JSh}$ on $(\mk{\vl{v}} - \pad{\vl{v}}{3})$ during the online phase to generate the respective boolean sharing. Parties finally compute the result by evaluating the bit extraction circuit~\cite{CCS:MohRin18, USENIX:PSSY21}.

 \subsection{Bit to Arithmetic}
 \label{sec:4pcBit2A}
Protocol $\prot{\bitA}(\shrB{\bitb})$~(\boxref{fig:4pcMpiBitA}) enables computing $\shr{\bitb}$ of a bit $\bitb$ given its boolean sharing $\shrB{\bitb}$. Let $\arval{\bitb}$ denotes the value of $\bitb \in \bitset$ over the arithmetic ring $\Z{\ell}$. Then for $\bitb = \bitb_1 \xor \bitb_2$, note that $\arval{\bitb} = (\arval{\bitb_1} - \arval{\bitb_2})^2$. Let $\bitb_1 = \mk{\vl{\bitb}} \xor \pad{\vl{v}}{3}$ and $\bitb_2 = \pad{\vl{v}}{1} \xor \pad{\vl{v}}{2}$. To compute $\shr{\bitb}$, a pair of parties can generate the arithmetic sharing corresponding to $\arval{\bitb_1}$ and $\arval{\bitb_2}$ by executing $\prot{\JSh}$. $\shr{\bitb}$ can be computed by invoking $\prot{\Mult}$ once with inputs $\vl{x} = \vl{y} = \arval{\bitb_1}- \arval{\bitb_2}$.

We obtain a communication-optimized variant by trading off computation in the preprocessing. For this, note that 
\begin{equation}
	\label{eq:4pcbitA}
	\arval{\bitb} = \arval{(\mk{\bitb} \xor \pad{\bitb}{})} = \arval{\mk{\bitb}} + \padR{\bitb}(1-2\arval{\mk{\bitb}})
\end{equation} 
Let $\vl{v} = \arval{\mk{\bitb}}$ and $\vl{u} = \padR{\bitb}$. During the preprocessing, $P_0$ generates $\sgr{\cdot}$-sharing of $\vl{u}$ and a check is executed to verify the correctness. The online phase consists of each pair of parties $(P_1, P_3)$, $(P_2, P_3)$ and $(P_1, P_2)$ locally computing an additive sharing of $\arval{\bitb}$, generating the corresponding $\shr{\cdot}$-sharing using $\prot{\JSh}$, and locally adding the shares to obtain $\shr{\arval{\bitb}}$.  

For verifying the $\sgr{\cdot}$-sharing of $\vl{u}$ by $P_0$, we let $P_3$ obtain the bit $(\pad{\bitb}{} \xor \vl{r}_{\bitb})$ as well as its arithmetic equivalent $\arval{(\pad{\bitb}{} \xor \vl{r}_{\bitb})}$ in clear. Here $\vl{r}_{\bitb}$ denotes a random bit known to $P_0, P_1, P_2$. $P_3$ checks if both the received values are equivalent and raise a complaint if they are inconsistent. To catch a corrupt $P_0$ from sharing a wrong $\vl{u}$ value, parties use the $\sgr{\cdot}$-shares of $\vl{u}$ to compute $\arval{(\pad{\bitb}{} \xor \vl{r}_{\bitb})}$. Moreover, the verification steps are designed in such a way that every value communicated can be locally computed by at least two parties. This enables to use $\jsend$ for communication and hence the desired security guarantee is achieved.

\begin{protocolbox}{$\prot{\bitA}(\shrB{\bitb})$}{Bit to Arithmetic conversion in $\Fthis$.}{fig:4pcMpiBitA}
	Let $\vl{u} = \padR{\bitb}$ and $\vl{v} = \arval{\mk{\bitb}}$.\\
	\detail{
		{\bf Input(s):} $\shrB{\bitb}$,~~~{\bf Output:} $\shr{\vl{y}} = \shr{\arval{\bitb}}$.
	}
	\justify 
	\vspace{-2mm}
	\algoHead{Preprocessing:} 
	\begin{enumerate} 
		\item Generation of $\sgr{\vl{u}}$: $P_0, P_3, P_i$ for $i \in \{1,2\}$ sample $\vl{u}^i$. $P_0$ sends $\vl{u}^3 = \vl{u} - \vl{u}^1 - \vl{u}^2$ to $P_1, P_2$.
		\item $P_0, P_1, P_2$ sample random $\vl{r}_{\bitb} \in \bitset$ and $\vl{r} \in \Z{\ell}$.
		\item $P_1, P_2$ $\jsend$ $\pad{\bitb}{3} \xor \vl{r}_{\bitb}$ to $P_3$. $P_3$ locally sets $\pad{\bitb}{} \xor \vl{r}_{\bitb} = (\pad{\bitb}{1} \xor \pad{\bitb}{2}) \xor (\pad{\bitb}{3} \xor \vl{r}_{\bitb})$.
		\item Parties compute: $P_1, P_0: \vl{w}_1 = \arval{\vl{r}_{\bitb}} + (\vl{u}^1 + \vl{u}^3) (1 - 2 \arval{\vl{r}_{\bitb}}) + \vl{r},~~P_2, P_0: \vl{w}_2 = (\vl{u}^2) (1 - 2 \arval{\vl{r}_{\bitb}}) - \vl{r}$. 
		\item $P_1, P_0$ $\jsend$ $\vl{w}_1$ to $P_3$, while $P_2, P_0$ $\jsend$ $\Hash(\vl{w}_2)$ to $P_3$.
		\item $P_3$ sets $\flag = \continue$ if $\Hash(\arval{(\pad{\bitb}{} \xor \vl{r}_b)} - \vl{w}_1) = \Hash(\vl{w}_2)$, else $\flag = \abort$. $P_3$ sends $\flag$ to $P_0, P_1, P_2$. Parties mutually exchange the flag and accept the value that forms the majority.
		\item For robust setting, if $\flag = \abort$, then $\TTP = P_1$ (or $P_2$).
	\end{enumerate}
	\justify
	\vspace{-2mm}
	\algoHead{Online:} Let $\vl{y} = \arval{\bitb}$.
	\begin{enumerate} 
		\item Locally compute the following:
		\begin{align*}
			P_1, P_3: \vl{y}_1  =  \vl{v} + \vl{u}^1 (1 - 2\vl{v})~~\Big|~~
			P_2, P_3: \vl{y}_2 =  \vl{u}^2 (1 - 2\vl{v})~~\Big|~~
			P_1, P_2: \vl{y}_3  =  \vl{u}^3 (1 - 2\vl{v})
		\end{align*}
		\item $(P_1, P_3), (P_2, P_3), (P_1, P_2)$ execute $\prot{\JSh}$ on $\vl{y}_1, \vl{y}_2, \vl{y}_3$ to generate the respective $\shr{\cdot}$-shares.
		\item Compute $\shr{\vl{y}} = \shr{\vl{y}_1} + \shr{\vl{y}_2} + \shr{\vl{y}_3}$.
	\end{enumerate}     
\end{protocolbox}

\begin{lemma}[Communication]
	\label{lemma:4pcMpibitA}
	Protocol $\prot{\bitA}$~(\boxref{fig:4pcMpiBitA}) requires $3 \ell + 1$ bits of communication in preprocessing, and $1$ round and $3 \ell$ bits of communication in the online phase.
\end{lemma}
\begin{proof}
	During preprocessing, generation of $\sgr{\vl{u}}$ involves communication of $\ell$ bits from $P_0$ to each of $P_1, P_2$. As part of verification, two instances of $\jsend$ are executed, one on $1$ bit and other on $\ell$ bits. The communication for hash gets amortized over multiple instances. The online phase involves three instances of joint sharing protocol resulting in $1$ rounds and a communication of $3\ell$ bits. The costs follow from Lemma~\ref{lemma:pijsend}.
\end{proof}

\subsubsection{Bit to Arithmetic:II}
\label{sec:4pcdBit2A}
Similar to $\prot{\bitA}$ protocol, given the boolean sharings $\shrB{\bitb_1}, \shrB{\bitb_2}$, protocol $\prot{\dbitA}$ computes the arithmetic sharing of $\arval{(\bitb_1 \bitb_2)}$. Let $\Delta_{\bitb_1}$, $\Delta_{\bitb_2}$  denote the value $(1-2\arval{\mk{\bitb_1}})$, $(1-2\arval{\mk{\bitb_2}})$ respectively. Using \eqref{eq:4pcbitA}, we can write

\begin{align}
	\label{eq:4pcdbitA}
	\arval{(\bitb_1 \bitb_2)} &= \arval{(\mk{\bitb_1} \xor \pad{\bitb_1}{})} \arval{(\mk{\bitb_2} \xor \pad{\bitb_2}{})} 
	= (\arval{\mk{\bitb_1}} + \padR{\bitb_1}\Delta_{\bitb_1}) (\arval{\mk{\bitb_2}} + \padR{\bitb_2}\Delta_{\bitb_2}) \nonumber \\
    &= \arval{\mk{\bitb_1}}\arval{\mk{\bitb_2}} + \padR{\bitb_1}\arval{\mk{\bitb_2}}\Delta_{\bitb_1} + \padR{\bitb_2}\arval{\mk{\bitb_1}}\Delta_{\bitb_2} + \arval{(\pad{\bitb_1}{}\pad{\bitb_2}{})}\Delta_{\bitb_1}\Delta_{\bitb_2}
\end{align} 

During preprocessing, the $\sgr{\cdot}$-shares of $\padR{\bitb_1}$ and $\padR{\bitb_2}$ are computed similar to that of $\prot{\bitA}$~(\boxref{fig:4pcMpiBitA}). Once the $\sgr{\cdot}$-shares are generated, parties invoke the $\prot{\MultS}$~(\boxref{fig:piMultS}) on $\sgr{\padR{\bitb_1}}$ and $\sgr{\padR{\bitb_2}}$ to generate the $\sgr{\cdot}$-shares of $\arval{(\pad{\bitb_1}{}\pad{\bitb_2}{})}$. The online phase is similar to that of $\prot{\bitA}$ protocol.

\begin{lemma}[Communication]
	\label{lemma:4pcMpidbitA}
	Protocol $\prot{\dbitA}$ requires $9 \ell + 2$ bits of communication in preprocessing, and $1$ round and $3 \ell$ bits of communication in the online phase.
\end{lemma}

\subsection{Bit Injection}
\label{sec:4pcBitInj}
Given the boolean sharing of a bit $\bitb$, denoted as $\shrB{\bitb}$, and the arithmetic sharing of $\vl{v} \in \Z{\ell}$, protocol $\prot{\bitinj}$ computes $\shr{\cdot}$-sharing of $\arval{\bitb}\vl{v}$. Let $\Delta_{\bitb}$ denote the value $(1-2\arval{\mk{\bitb}})$.
Similar to $\prot{\bitA}$, 
\begin{align}\label{eq:4pcbitinj}
	\arval{\bitb} \vl{v} &= \arval{(\mk{\bitb} \xor \pad{\bitb}{})}(\mk{\vl{v}} - \pd{\vl{v}}{}) = (\arval{\mk{\bitb}} + \padR{\bitb}\Delta_{\bitb})(\mk{\vl{v}} - \pd{\vl{v}}{}) \nonumber \\
	&= \arval{\mk{\bitb}}\mk{\vl{v}} - \arval{\mk{\bitb}}\pd{\vl{v}}{} + \padR{\bitb}\mk{\vl{v}}\Delta_{\bitb} - \padR{\bitb}\pd{\vl{v}}{}\Delta_{\bitb}
\end{align} 

During the preprocessing, $P_0$ generates the $\sgr{\cdot}$-shares of $\padR{\bitb}$ similar to $\prot{\bitA}$ protocol. Parties then invoke the $\prot{\MultS}$~(\boxref{fig:piMultS}) on $\sgr{\padR{\bitb}}$ and $\sgr{\pad{\vl{v}}{}}$ to generate the $\sgr{\cdot}$-shares of $\padR{\bitb}\pd{\vl{v}}{}$. During the online phase, $(P_1, P_3)$, $(P_2, P_3)$ and $(P_1, P_2)$ compute an additive sharing of $\arval{\bitb} \vl{v}$ using \eqref{eq:4pcbitinj} and execute $\prot{\JSh}$ on them to generate the respective $\shr{\cdot}$-shares. Parties locally add the shares to obtain the output.

\begin{lemma}[Communication]
	\label{lemma:bitinj}
	Protocol $\prot{\bitinj}$ requires $6\ell + 1$ bits of communication in preprocessing, and $1$ round and $3 \ell$ bits of communication in the online phase.
\end{lemma}

 \subsubsection{Sum of Bit Injections}
 \label{sec:4pcSumBitInj}
Given $m$ pair of values in the shared form, $\{\shrB{\bitb_i}, \shr{\vl{v}_i}\}_{i \in [m]}$, the goal of $\prot{\bitinjS}$ is to compute the $\shr{\cdot}$-share of $\vl{z} = \sum_{i=1}^{m} \arval{\bitb_i} \cdot \vl{v_i}$. For this, parties execute the preprocessing corresponding to $m$ bit injections of the form $\arval{\bitb_i} \cdot \vl{v_i}$. 

In the online phase, each pair of parties $(P_1, P_3)$, $(P_2, P_3)$ and $(P_1, P_2)$ locally compute an additive sharing of $\vl{z}_i$, corresponding to $\arval{\bitb_i} \cdot \vl{v_i}$ first. Instead of generating the $\shr{\cdot}$-sharing for each of the $m$ terms, parties locally add the shares and execute $\prot{\JSh}$ on the result. Concretely, parties locally compute $\vl{z}^j = \sum_{i=1}^{m} \vl{z}_i^j$ for $j \in \{1,2,3\}$ and execute $\prot{\JSh}$ on $\vl{z}^j$ to obtain its $\shr{\cdot}$-sharing. Finally, parties locally add up the shares similar to $\prot{\bitinj}$ protocol. This results in an online communication independent of $m$.

\begin{lemma}[Communication]
	\label{lemma:4pcSumBitInj}
	Protocol $\prot{\bitinjS}$ requires $m \cdot (6\ell + 1)$ bits of communication in preprocessing, and $1$ round and $3 \ell$ bits of communication in the online phase.
\end{lemma}

\subsubsection{Bit Injection:II}
\label{sec:4pcBitInjII}
Similar to $\prot{\bitinj}$ protocol, given $\shrB{\bitb_1}, \shrB{\bitb_2}$ and $\shr{\vl{v}}$, protocol $\prot{\dbitA}$ computes the arithmetic sharing of $\arval{(\bitb_1 \bitb_2)}\vl{v}$. Let $\Delta_{\bitb_1}$, $\Delta_{\bitb_2}$  denote the value $(1-2\arval{\mk{\bitb_1}})$, $(1-2\arval{\mk{\bitb_2}})$ respectively. Using \eqref{eq:4pcdbitA} and \eqref{eq:4pcbitinj}, we can write

\begin{align}
	\label{eq:4pcdbitinj}
	\arval{(\bitb_1 \bitb_2)} \vl{v} &= \arval{(\mk{\bitb_1} \xor \pad{\bitb_1}{})} \arval{(\mk{\bitb_2} \xor \pad{\bitb_2}{})} (\mk{\vl{v}} - \pd{\vl{v}}{}) \nonumber \\ 
	&= (\arval{\mk{\bitb_1}} + \padR{\bitb_1}\Delta_{\bitb_1}) (\arval{\mk{\bitb_2}} + \padR{\bitb_2}\Delta_{\bitb_2})(\mk{\vl{v}} - \pd{\vl{v}}{}) \nonumber \\
	&= \arval{\mk{\bitb_1}}\arval{\mk{\bitb_2}}\mk{\vl{v}}  + \padR{\bitb_1}\arval{\mk{\bitb_2}}\mk{\vl{v}} \Delta_{\bitb_1} + \padR{\bitb_2}\arval{\mk{\bitb_1}}\mk{\vl{v}}\Delta_{\bitb_2} + \arval{(\pad{\bitb_1}{}\pad{\bitb_2}{})}\mk{\vl{v}}\Delta_{\bitb_1}\Delta_{\bitb_2} \nonumber \\
	&~~~- \pd{\vl{v}}{}\arval{\mk{\bitb_1}}\arval{\mk{\bitb_2}} - \padR{\bitb_1}\pd{\vl{v}}{}\arval{\mk{\bitb_2}}\Delta_{\bitb_1} - \padR{\bitb_2}\pd{\vl{v}}{}\arval{\mk{\bitb_1}}\Delta_{\bitb_2} - \arval{(\pad{\bitb_1}{}\pad{\bitb_2}{})}\pd{\vl{v}}{}\Delta_{\bitb_1}\Delta_{\bitb_2}
\end{align} 

During preprocessing, the $\sgr{\cdot}$-shares of $\padR{\bitb_1}, \padR{\bitb_2}$ and $\arval{(\pad{\bitb_1}{}\pad{\bitb_2}{})}$ are computed similar to that of $\prot{\dbitA}$. Once the $\sgr{\cdot}$-shares are generated, parties invoke the $\prot{\MultS}$~(\boxref{fig:piMultS}) on $\arval{(\pad{\bitb_1}{}\pad{\bitb_2}{})}$ and $\sgr{\pd{\vl{v}}}$ to generate the $\sgr{\cdot}$-shares of $\arval{(\pad{\bitb_1}{}\pad{\bitb_2}{})}\pd{\vl{v}}$. Similarly, parties compute $\sgr{\padR{\bitb_1}\pd{\vl{v}}}$ and $\sgr{\padR{\bitb_2}\pd{\vl{v}}}$ using two instances of $\prot{\MultS}$.
The online phase is similar to that of $\prot{\bitinj}$ protocol.

\begin{lemma}[Communication]
	\label{lemma:4pcMdbitinj}
	Protocol $\prot{\dbitinj}$ requires $18\ell + 2$ bits of communication in preprocessing, and $1$ round and $3 \ell$ bits of communication in the online phase.
\end{lemma}

\subsection{Equality Test~($\prot{\eql}$)}
To check whether $\vl{a} \iseq \vl{b}$ or not, given $\shr{\vl{a}}, \shr{\vl{b}}$, $\prot{\eql}$ proceeds with parties locally computing $\shr{\vl{y}} = \shr{\vl{a}} - \shr{\vl{b}}$. According to our sharing semantics, $\vl{y}$ can be written as $\vl{y} = \vl{y}_1 - \vl{y}_2$ where $\vl{y}_1 = \mk{\vl{y}} - \pad{\vl{y}}{3}$ and $\vl{y}_2 = \pad{\vl{y}}{1} + \pad{\vl{y}}{2}$. Parties $(P_1,P_2)$ and $(P_0, P_3)$  generate $\shrB{\vl{y}_1}$ and $\shrB{\vl{y}_2}$ resepctively using the joint sharing protocol $\prot{\JSh}$. Note that $\vl{a} = \vl{b}$ implies $\vl{y}_1 = \vl{y}_2$ and hence all the bits of $\vl{v} = \overline{(\vl{y}_1 \xor \vl{y}_2)}$ should be $1$. As mentioned in the introduction of Part II~(\ref{chap:layer2_intro}), parties use four input AND gates and a tree structure, where $4$ bits are taken at a time and the AND of them is computed in one go. 

\section{Mixed Protocol Framework}
\label{sec:4pcMixFrame}
\tabref{4pcMConv} compares our sharing conversions with Trident~\cite{NDSS:ChaRacSur20}. For uniformity, we consider a function, {\sf F}, to be computed on an $\ell$-bit inputs $\vl{x}, \vl{y}$ using a garbled circuit (GC) in the mixed framework, which gives an $\ell$-bit output $\vl{z} = \mathsf{F(\vl{x}, \vl{y})}$, where $\ell$ denotes the ring size in bits. Let $\Grb{F}$  denote the corresponding GC. In the table, $\Grb{Sn}$ denotes a {\sf n}-input garbled subtraction circuit; $\Grb{An}$ denotes {\sf n}-input garbled addition circuit; $\GrbD{}$ denotes the garbled circuit with decoding information; $\Grb{n_1\times1,\ldots,n_m \times m}$ denotes ${\sf n_i}$ instances of GC $\Grb{i}$ for $i \in \{1,\dots,{\sf m}\}$ and $\Size{\Grb{n_1\times1,\ldots,n_m \times m}}$ denotes its size. 
\begin{table}[htb!]
	\centering
	\resizebox{0.95\textwidth}{!}{
		\begin{NiceTabular}{rr|rrrr|rrrr}
			\toprule 
			\Block{2-1}{Variant\tabularnote{Notations: $\ell$ - size of ring in bits, $\kappa$ - computational security parameter, 'pre' - preprocessing, 'on' - online.}} 
			& \Block{2-1}{Conversion\tabularnote{'A' - arithmetic, 'B' - boolean, 'G' - Garbled.}}
			& \Block[c]{1-4}{Trident~\cite{NDSS:ChaRacSur20}} & & &
			& \Block[c]{1-4}{$\Fthis$} & & & \\ \cmidrule{3-10}
			& & \Block[c]{1-2}{Comm.\textsubscript{pre}} & & Comm.\textsubscript{on}  & Rounds\textsubscript{on} 
			& \Block[c]{1-2}{Comm.\textsubscript{pre}} & & Comm.\textsubscript{on}  & Rounds\textsubscript{on} \\
			\midrule
			\Block{4-1}{2 GC}
			& A-G-A & \Block[r]{4-1}{($6 \ell \kappa + \ell$)\\+}  & $2\Size{\GrbD{2 \times S2,F}}$ 
			& \Block{4-1}{$4\ell\kappa + 2\ell$} & \Block{4-1}{$2$}
			& \Block[r]{4-1}{($6 \ell \kappa + \ell$)\\+}  & $2\Size{\GrbD{2 \times S2, F}}$ 
			& \Block{4-1}{$4\ell\kappa + \ell$} & \Block{4-1}{$1$} \\
			& A-G-B &  & $2\Size{\Grb{S2,F}}$   &  & &   & $2\Size{\Grb{S2,F}}$ &  & \\
			& B-G-A &  & $2\Size{\GrbD{S2,F}}$              &  & &  & $2\Size{\GrbD{S2,F}}$              &  & \\
			& B-G-B &  & $2\Size{\Grb{F}}$                      &  & &  & $2\Size{\Grb{F}}$                   &  & \\
			\midrule
			%
			\Block{4-1}{1 GC}
			& A-G-A & \Block[r]{4-1}{($3\ell \kappa + \ell$)\\+}  & $\Size{\GrbD{2 \times S2,F}}$ 
			& \Block{4-1}{$2\ell\kappa + 3\ell$} & \Block{4-1}{$2$}
			& \Block[r]{4-1}{($3 \ell \kappa + \ell$)\\+}  & $\Size{\GrbD{2 \times S2, F}}$ 
			& \Block{4-1}{$2\ell\kappa + 2\ell$} & \Block{4-1}{$2$} \\
			& A-G-B &  & $\Size{\Grb{S2,F}}$   &  & &  & $\Size{\Grb{S2,F}}$ &  & \\
			& B-G-A &  & $\Size{\GrbD{S2,F}}$              &  & &  & $\Size{\GrbD{S2,F}}$    &  & \\
			& B-G-B &  & $\Size{\Grb{F}}$                      &  & &  & $\Size{\Grb{F}}$                   &  & \\
			\midrule
			\Block{2-1}{Others\tabularnote{$\vl{u_1} = 2\vl{n_2} + 8\vl{n_3} + 22\vl{n_4}$, $\vl{u_2} = \vl{n_2} + \vl{n_3} + \vl{n_4}$ denote the number of AND gates in the optimized adder circuit~\cite{USENIX:PSSY21} with 2, 3, 4 inputs, respectively. For $\ell = 64$, $\vl{n_2} = 216, \vl{n_3} = 184, \vl{n_4}=179$.}}
			& A-B & \Block[r]{1-2}{$2\ell + 3\ell \log \ell$} &  & $\ell + 3 \ell \log \ell$ & $1 + \log \ell$
			& \Block[r]{1-2}{$\vl{u_1} + \ell$} &  & $3\vl{u_2} + \ell$ & $\log_4 \ell$ \\
			& B-A &  \Block[r]{1-2}{$3\ell^2 \ell$} &  & $3 \ell$ & $1$
			& \Block[r]{1-2}{$3\ell^2 + \ell$} &  & $3\ell$ & $1$ \\
			%
			\bottomrule
		\end{NiceTabular}
	}
	\caption{Mixed protocol conversions of Trident~\cite{NDSS:ChaRacSur20} and $\Fthis$.}\label{tab:4pcMConv}
\end{table}

\subsection{Conversions involving Garbled World} 
\label{sec:4pcconv2gc}
Assume the GC is required to compute a function $f$ on inputs $\vl{x}, \vl{y} \in \Z{\ell}$ and let the output be $f(\vl{x}, \vl{y})$. All the conversions described are for the 2 GC variant. Conversions for the 1 GC variant are straightforward, hence we omit the details.

\paragraph{Case I: Boolean-Garbled-Boolean}
Since the inputs to the GC are available in boolean form, say $\shrB{\vl{x}}, \shrB{\vl{y}}$, parties generate $\shrC{\vl{x}}, \shrC{\vl{y}}$ by invoking the garbled sharing protocol $\pigsh$.
Additionally, parties $P_0, P_3$ sample $\vl{R} \in \Z{\ell}$ to mask the function output, $f(\vl{x}, \vl{y})$, and generate $\shrB{\vl{R}}$ (using the joint sharing protocol) and $\shrG{\vl{R}}$. Garblers $P_g \in \{P_0, P_2, P_3\}$ garble the circuit which computes $\vl{z} = f(\vl{x}, \vl{y}) \xor \vl{R}$, and send the GC along with the decoding information to evaluator $P_1$. Analogous steps are performed for evaluator $P_2$. Upon GC evaluation and output decoding, evaluators obtain $\vl{z} = f(\vl{x}, \vl{y}) \xor \vl{R}$, and jointly boolean share $\vl{z}$ to generate $\shrB{\vl{z}}$. Parties then compute $\shrB{f(\vl{x}, \vl{y})} = \shrB{\vl{z}} \xor \shrB{\vl{R}}$.  

\paragraph{Case II: Boolean-Garbled-Arithmetic}
This is similar to {\em Case I} except that the circuit which computes $\vl{z} = f(\vl{x}, \vl{y}) + \vl{R}$ is garbled instead. Boolean sharing of $\vl{z}$ is replaced with arithmetic, followed by computing $\shr{f(\vl{x}, \vl{y})} = \shr{\vl{z}} - \shr{\vl{R}}$.

\paragraph{Cases III \& IV: Input in Arithmetic Sharing} 
The function to be computed $f(\vl{x}, \vl{y})$, is modified as $f^{\prime}(\mk{\vl{x}}, \av{\vl{x}}, \pad{\vl{x}}{3}, \mk{\vl{y}}, \av{\vl{y}}, \pad{\vl{y}}{3}) = f(\mk{\vl{x}}-\av{\vl{x}}-\pad{\vl{x}}{3}, \mk{\vl{y}}-\av{\vl{y}}-\pad{\vl{y}}{3})$ where inputs $\vl{x}, \vl{y}$ are replaced by the triples $\{\mk{\vl{x}}, \av{\vl{x}}, \pad{\vl{x}}{3}\}, \{\mk{\vl{y}}, \av{\vl{y}}, \pad{\vl{y}}{3}\}$ and $\av{\vl{x}} = \pad{\vl{x}}{1} + \pad{\vl{x}}{2}$ and $\av{\vl{y}} = \pad{\vl{y}}{1} + \pad{\vl{y}}{2}$. The circuit to be garbled thus, corresponds to the function $f^{\prime}$. Parties generate $\shrG{\mk{\vl{x}}}, \shrG{\av{\vl{x}}}, \shrG{\pad{\vl{x}}{3}}, \allowbreak \shrG{\mk{\vl{y}}}, \shrG{\av{\vl{y}}}, \shrG{\pad{\vl{y}}{3}}$ via $\pigsh$, following which, parties proceed with the rest of the computation whose steps are similar to {\em Case I}, and {\em II}, depending on the requirement on the output sharing.

\subsection{Other Conversions} 
\label{sec:4pcotherconv}

\paragraph{Arithmetic to Boolean} 
To convert arithmetic sharing of $\vl{v} \in \Z{\ell}$ to boolean sharing, observe that $\vl{v} = \vl{v}_1 + \vl{v}_2$ where $\vl{v}_1 = \mk{\vl{v}} - \pad{\vl{v}}{3}$ is possessed by parties $P_1, P_2$, while $\vl{v}_2 = -(\pad{\vl{v}}{1} + \pad{\vl{v}}{2})$ is possessed by parties $P_0, P_3$. Thus, $\shrB{\vl{v}}$  can be computed as $\shrB{\vl{v}} = \shrB{\vl{v}_1} + \shrB{\vl{v}_2}$, where $\shrB{\vl{v}_2}$ can be generated in the preprocessing phase, and $\shrB{\vl{v}_1}$ can be generated in the online phase by the respective parties executing joint boolean sharing protocol. The protocol appears in \boxref{fig:4pcMpiab}. Boolean addition, when instantiated using the adder of ABY2.0~\cite{USENIX:PSSY21}, requires $\log_4(\ell)$ rounds.

\begin{protocolbox}{$\piab$}{Arithmetic to Boolean Conversion in $\Fthis$.}{fig:4pcMpiab}
	\justify
	\algoHead{Preprocessing:} $P_0, P_3$ execute joint boolean sharing to generate $\shrB{\vl{v}_2}$,  where $\vl{v}_2 =-(\pad{\vl{v}}{1} + \pad{\vl{v}}{2})$.
	\justify
	\vspace{-2mm}
	\algoHead{Online:}
	\begin{enumerate} 
		\item $P_1, P_2$ execute joint boolean sharing to generate $\shrB{\vl{v}_1}$, where $\vl{v}_1 = \mk{\vl{v}} - \pad{\vl{v}}{3}$.
		\item Parties obtain $\shrB{\vl{v}} = \shrB{\vl{v}_1} + \shrB{\vl{v}_2}$ using a boolean adder circuit.
	\end{enumerate}
\end{protocolbox}

\paragraph{Boolean to Arithmetic} 
To convert a boolean sharing of $\vl{v} \in \Z{\ell}$ into an arithmetic sharing, note that 
\begin{small}
	\begin{align*}
		\vl{v} = \sum_{i=0}^{\ell - 1} 2^{i} \vl{v}[i] = \sum_{i=0}^{\ell - 1} 2^{i} (\pad{\vl{v}[i]}{} \xor \mk{\vl{v}[i]}) 
		= \sum_{i=0}^{\ell - 1} 2^{i} \left( \arval{\mk{\vl{v}[i]}} +  \padR{\vl{v}[i]} (1 - 2\arval{\mk{\vl{v}[i]}}) \right) 
	\end{align*}
\end{small}
where $\arval{{\pd{\vl{v}[i]}}}, \arval{\mk{\vl{v}[i]}}$ denote the arithmetic value of bits ${\pd{\vl{v}[i]}}, {\mk{\vl{v}[i]}}$ over the ring $\Z{\ell}$. 
For each bit $\vl{v}[i]$ of $\vl{v}$, parties generate the $\sgr{\cdot}$-shares of $\arval{{{\pd{\vl{v}[i]}}}}$ in the preprocessing, similar to $\prot{\bitA}$~(\boxref{fig:4pcMpiBitA}). During the online phase, additive shares for each bit $\vl{v}[i]$ are locally computed similar to $\prot{\bitA}$. Parties then multiply the $i$th share with $2^i$ and locally add up to obtain an additive sharing of $\vl{v}$. The rest of the steps are similar to $\prot{\bitA}$, and the formal protocol appears in \boxref{fig:4pcMpiba}. 

\begin{protocolbox}{$\piba(\Partyset, \shrB{\vl{v}})$}{Boolean to Arithmetic Conversion in $\Fthis$.}{fig:4pcMpiba}
	Let $\vl{v}[i]$ denote the $i$th bit of $\vl{v}$. Let $\vl{p}_i = \arval{\mk{\vl{v}[i]}}$, and $\vl{q}_i = \padR{\vl{v}[i]}$. \\
	\justify 
	\vspace{-2mm}
	\algoHead{Preprocessing:} 
	\begin{enumerate} 
		\item For $i \in \{0, 1, \ldots, \ell-1 \} $, execute the preprocessing of $\prot{\bitA}$ (\boxref{fig:4pcMpiBitA}) for each bit $\vl{v}[i]$, to generate $\sgr{{\vl{q}_i}} = (\vl{q}_i^1, \vl{q}_i^2, \vl{q}_i^3)$.
	\end{enumerate}
	\justify
	\vspace{-2mm}
	\algoHead{Online:} Let $\vl{y}_i = \arval{(\vl{v}[i])}$ and $\vl{y}$ denotes the arithmetic equivalent of $\vl{v}$.
	\begin{enumerate}
		\item Locally compute the following:
		\begin{align*}
			P_1, P_3: \vl{y}^1 = \sum_{i=0}^{\ell-1} 2^i \vl{y}_i^1  &=  \sum_{i=0}^{\ell-1} 2^i (\vl{p}_i + \vl{q}_i^1 (1 - 2\vl{p}_i)) \\
			P_2, P_3: \vl{y}^2 = \sum_{i=0}^{\ell-1} 2^i \vl{y}_i^2 &=  \sum_{i=0}^{\ell-1} 2^i (\vl{q}_i^2 (1 - 2\vl{p}_i))  \\
			P_1, P_2: \vl{y}^3 = \sum_{i=0}^{\ell-1} 2^i \vl{y}_i^3  &=  \sum_{i=0}^{\ell-1} 2^i (\vl{q}_i^3 (1 - 2\vl{p}_i)) 
		\end{align*}
		\item $(P_1, P_3), (P_2, P_3), (P_1, P_2)$ execute $\prot{\JSh}$ on $\vl{y}^1, \vl{y}^2, \vl{y}^3$ to generate the respective $\shr{\cdot}$-shares.
		\item Locally compute $\shr{\vl{y}} = \shr{\vl{y}^1} + \shr{\vl{y}^2} + \shr{\vl{y}^3}$.
	\end{enumerate}     
\end{protocolbox}

We remark that the protocol $\piba$ can be used to efficiently generate edaBits~\cite{C:EGKRS20} in our setting. For this, the parties non-interactively generate the boolean sharing for $\ell$-bits and perform the $\piba$ conversion to obtain the equivalent arithmetic value.

\chapter{$\TWthis$: 2PC Semi-honest Blocks}
\label{chap:layer2_2pc}
This chapter provides details for the Layer II blocks of our 2PC framework $\TWthis$. Details for the Layer I blocks are provided in chapter~\ref{chap:layer1_2pc}.

\section{Building Blocks}
\label{sec:layerIIblocks2pcS}

\subsection{Dot Product~(Scalar Product)}
\label{sec:2pcSDotp}
Given $\shr{\vct{a}}, \shr{\vct{b}}$ with $|\vct{a}| = |\vct{b}| = \vl{d}$, protocol $\prot{\dotp}$~(\boxref{fig:piDotP2pcS}) computes $\shr{\vl{z}}$ such that $\vl{z} = (\vct{a} \band \vct{b})^{\vl{t}}$ if truncation is enabled, else $\vl{z} = \vct{a} \band \vct{b}$. The protocol is similar to the multiplication protocol $\prot{\Mult}$~(\boxref{fig:piMult2pcS}) except that the parties combine the partial products in the online phase across $\vl{d}$ multiplications and communicate them in a single shot. This makes the communication cost of the dot product in the online phase independent of the vector size.

\begin{protocolsplitbox}{$\prot{\dotp}(\vct{a}, \vct{b}, \isTr)$}{Dot Product with / without Truncation in $\TWthis$.}{fig:piDotP2pcS}
	$\isTr$ is a bit denoting whether truncation is required ($\isTr =1$) or not ($\isTr=0$). \\
	\detail{
		{\bf Input(s):} $\shr{\vct{a}}, \shr{\vct{b}}$.\\
		{\bf Output:} $\shr{\vl{o}}$ where $\vl{o} = \vl{z}^{\vl{t}}$ if $\isTr = 1$ and $\vl{o} = \vl{z}$ if $\isTr = 0$ and $\vl{z} = \vct{a} \band \vct{b} = \sum_{j = 1}^{\vl{d}} {\vl{a}}_j {\vl{b}}_j$.
	}
	\justify 
	\vspace{-2mm}
	\algoHead{Preprocessing:}  Execute $\prot{\MultPre}$ on $\sqr{\pad{\vl{a}_j}{}}$ and $\sqr{\pad{\vl{b}_j}{}}$ to generate $\sqr{\gm{\vl{a_j b_j}}{}}$ for $j \in [\vl{d}]$.
	\justify
	\vspace{-2mm}
	\algoHead{Online:}
	\begin{enumerate}
		\item Locally compute: 
		\begin{align*}
		P_1: \vl{z}_1 &=  \sum_{j=1}^{\vl{d}} (\mk{\vl{a_j b_j}} - \pad{\vl{a}_j}{1} \mk{\vl{b}_j} - \pad{\vl{b}_j}{1} \mk{\vl{a}_j} + \sqr{\gm{\vl{a_j b_j}}{}}_1)\\ 
		P_2: \vl{z}_2 &=  \sum_{j=1}^{\vl{d}} (- \pad{\vl{a}_j}{2} \mk{\vl{b}_j} - \pad{\vl{b}_j}{2} \mk{\vl{a}_j} + \sqr{\gm{\vl{a_j b_j}}{}}_2)\\
		\end{align*}
		\item If $\isTr = 1$, $P_i$ sets $\vl{p}_i = \vl{z}_i^{\vl{t}}$, else $\vl{p}_i = \vl{z}_i$ where $i \in \{1,2\}$.
		Execute $\prot{\Sh}(P_i, \vl{p}_i)$ to generate $\shr{\vl{p}_i}$. 
		\item Compute $\shr{\vl{o}} = \shr{\vl{p}_1} + \shr{\vl{p}_2}$. Here $\vl{o} = \vl{z}^{\vl{t}}$ if $\isTr = 1$ and $\vl{z}$ otherwise.
	\end{enumerate}     
\end{protocolsplitbox}

\begin{lemma}[Communication]
	\label{lemma:2pcSpidotpf}
	Protocol $\prot{\dotp}$~(\boxref{fig:piDotP2pcS})~(in $\TWthis$) requires $2 \vl{d} \ell (\csec + \ell)$ bits of communication in the preprocessing, and $1$ round and $2 \ell$ bits of communication in the online phase.
\end{lemma}

 \subsection{Bit Extraction}
 \label{sec:2pcSBitExt}

To compute most significant bit~($\msb$) of the value $\vl{v}$, note that $\vl{v} = (\mk{\vl{v}} - \pad{\vl{v}}{1}) + (- \pad{\vl{v}}{2})$ as per the sharing semantics~(cf.~\tabref{2pcSsharing}). $P_2$ generates the boolean sharing of $- \pad{\vl{v}}{2}$ during the preprocessing, while $P_1$ generates $\shrB{(\mk{\vl{v}} - \pad{\vl{v}}{1})}$ during the online phase using sharing protocol. Parties compute the result by evaluating the bit extraction circuit~\cite{CCS:MohRin18, USENIX:PSSY21}.

 \subsection{Bit to Arithmetic}
 \label{sec:2pcSBit2A}
Protocol $\prot{\bitA}(\shrB{\bitb})$~(\boxref{fig:2pcSpiBitA}) enables computing $\shr{\bitb}$ of a bit $\bitb$ given its boolean sharing $\shrB{\bitb}$. Let $\arval{\bitb}$ denotes the value of $\bitb \in \bitset$ over the arithmetic ring $\Z{\ell}$. Using our sharing semantics, 
\begin{equation}
	\label{eq:2pcSbitA}
	\arval{\bitb} = \arval{(\mk{\bitb} \xor \pad{\bitb}{})} = \arval{\mk{\bitb}} + \padR{\bitb}(1-2\arval{\mk{\bitb}})
\end{equation} 

During the preprocessing, parties interactively generate $\sqr{\cdot}$-sharing of $\padR{\bitb}$ using steps similar to that of $\prot{\MultPre}$. The online phase consists of each $P_1$ and $P_2$ locally computing an additive sharing of $\arval{\bitb}$, generating the corresponding $\shr{\cdot}$-sharing using $\prot{\Sh}$, and locally adding the shares to obtain $\shr{\bitb}$.

Now we describe how to generate $\sqr{\padR{\bitb}}$ in the preprocessing. Since $\pad{\bitb}{} = \pad{\bitb}{1} \xor \pad{\bitb}{2}$, we can write $\padR{\bitb} = \padR{\bitb_1} + \padR{\bitb_2} - 2 \padR{\bitb_1} \padR{\bitb_2}$.
Parties execute $\COT{1}{\ell}$ with $P_1$ being the sender and $P_2$ being the receiver. $P_1$ inputs the correlation $f(x) = x + \padR{\bitb_1}$ and obtains $(m_{0} = r, m_{1} = r + \padR{\bitb_1})$. $P_2$ inputs $c = \pad{\bitb}{2}$ as the choice bit and obtains $m_{c}$ as output. Now the $\sqr{\cdot}$-shares are defined as $\sqr{\padR{\bitb_1} \padR{\bitb_2}}_1 = -r$ and $\sqr{\padR{\bitb_1} \padR{\bitb_2}}_2 = m_{\pad{\bitb}{2}}$.

\begin{protocolbox}{$\prot{\bitA}(\shrB{\bitb})$}{Bit to Arithmetic conversion in $\TWthis$.}{fig:2pcSpiBitA}
	\detail{
		{\bf Input(s):} $\shrB{\bitb}$,~~~{\bf Output:} $\shr{\vl{y}} = \shr{\arval{\bitb}}$.
	}
	\justify 
	\vspace{-2mm}
	\algoHead{Preprocessing:} 
		\begin{enumerate} 
		\item Generating $\sqr{\cdot}$-shares of $\padR{\bitb_1} \padR{\bitb_2}$:
		\begin{enumerate}
			\item  Execute $\COT{1}{\ell}$ with $P_1$ being the sender with input $f(x) = x + \padR{\bitb_1}$ and $P_2$ being the receiver with input $c = \pad{\bitb}{2}$.
			\item $P_1$ obtains $(m_{0} = r, m_{1} = r + \padR{\bitb_1})$ while $P_2$ obtains $m_{c}$.
			\item Set $\sqr{\padR{\bitb_1} \padR{\bitb_2}}_1 = -r$ and $\sqr{\padR{\bitb_1} \padR{\bitb_2}}_2 = m_{c}$.
		\end{enumerate}
		\item $P_i$ for $i \in \{1,2\}$ locally computes $\sqr{\padR{\bitb}}_i = \padR{\bitb_i} - 2 \sqr{\padR{\bitb_1} \padR{\bitb_2}}_i$.
	\end{enumerate}    
	\justify
	\vspace{-2mm}
	\algoHead{Online:} 
	\begin{enumerate} 
		\item Locally compute: 
		$P_1: \vl{y}_1  =  \arval{\mk{\bitb}} + \sqr{\padR{\bitb}}_1 (1 - 2\arval{\mk{\bitb}})~~\Big|~~
		P_2:  \vl{y}_2  =  \sqr{\padR{\bitb}}_2 (1 - 2\arval{\mk{\bitb}})$
		\item $P_i$ for $i \in \{1,2\}$ executes $\prot{\Sh}$ on $\vl{y}_i$ to generate the respective $\shr{\cdot}$-shares.
		\item Compute $\shr{\vl{y}} = \shr{\vl{y}_1} + \shr{\vl{y}_2}$.
	\end{enumerate}     
\end{protocolbox}

\begin{lemma}[Communication]
	\label{lemma:2pcSpibitA}
	Protocol $\prot{\bitA}$~(\boxref{fig:2pcSpiBitA}) requires $\csec + \ell$ bits of communication in preprocessing, and $1$ round and $2 \ell$ bits of communication in the online phase.
\end{lemma}
\begin{proof}
	During preprocessing, generation of $\sqr{\padR{\bitb}}$ involves one instance of $\COT{1}{\ell}$. The online phase involves two instances of arithmetic sharing protocol  in parallel, resulting in $1$ round and a communication of $2\ell$ bits. 
\end{proof}

\subsubsection{Bit to Arithmetic:II}
\label{sec:2pcSdBit2A}
Similar to $\prot{\bitA}$ protocol, given the boolean sharings $\shrB{\bitb_1}, \shrB{\bitb_2}$, protocol $\prot{\dbitA}$ computes the arithmetic sharing of $\arval{(\bitb_1 \bitb_2)}$. Let $\Delta_{\bitb_1}$, $\Delta_{\bitb_2}$  denote the value $(1-2\arval{\mk{\bitb_1}})$, $(1-2\arval{\mk{\bitb_2}})$ respectively. Using \eqref{eq:2pcSbitA}, we can write

\begin{align}
	\label{eq:2pcSdbitA}
	\arval{(\bitb_1 \bitb_2)} &= \arval{(\mk{\bitb_1} \xor \pad{\bitb_1}{})} \arval{(\mk{\bitb_2} \xor \pad{\bitb_2}{})} 
	= (\arval{\mk{\bitb_1}} + \padR{\bitb_1}\Delta_{\bitb_1}) (\arval{\mk{\bitb_2}} + \padR{\bitb_2}\Delta_{\bitb_2}) \nonumber \\
    &= \arval{\mk{\bitb_1}}\arval{\mk{\bitb_2}} + \padR{\bitb_1}\arval{\mk{\bitb_2}}\Delta_{\bitb_1} + \padR{\bitb_2}\arval{\mk{\bitb_1}}\Delta_{\bitb_2} + \arval{(\pad{\bitb_1}{}\pad{\bitb_2}{})}\Delta_{\bitb_1}\Delta_{\bitb_2}
\end{align} 

During preprocessing, the $\sqr{\cdot}$-shares of $\padR{\bitb_1}$, and $\padR{\bitb_2}$ are computed similar to that of $\prot{\bitA}$~(\boxref{fig:2pcSpiBitA}). In parallel, parties execute $\prot{\MultPre}$ on the boolean $\sqr{\cdot}$-shares of $\pad{\bitb_1}{}$ and $\pad{\bitb_2}{}$ to generate $\sqr{\gm{\bitb_1 \bitb_2}{}} = \sqr{\pad{\bitb_1}{}\pad{\bitb_2}{}}$ in boolean form. Once $\sqr{\gm{\bitb_1 \bitb_2}{}}$ is generated, parties compute the $\sqr{\cdot}$-shares of its arithmetic equivalent similar to that of $\prot{\bitA}$.
The online phase is similar to that of $\prot{\bitA}$ protocol.

\begin{lemma}[Communication]
	\label{lemma:2pcSpidbitA}
	Protocol $\prot{\dbitA}$ requires $5\csec + 3\ell + 2$ bits of communication in preprocessing, and $1$ round and $2 \ell$ bits of communication in the online phase.
\end{lemma}

\subsection{Bit Injection}
\label{sec:2pcSBitInj}
Given the boolean sharing of a bit $\bitb$, denoted as $\shrB{\bitb}$, and the arithmetic sharing of $\vl{v} \in \Z{\ell}$, protocol $\prot{\bitinj}$ computes $\shr{\cdot}$-sharing of $\arval{\bitb}\vl{v}$. Let $\Delta_{\bitb}$ denote the value $(1-2\arval{\mk{\bitb}})$.
Similar to $\prot{\bitA}$, 
\begin{align}\label{eq:2pcSbitinj}
	\arval{\bitb} \vl{v} &= \arval{(\mk{\bitb} \xor \pad{\bitb}{})}(\mk{\vl{v}} - \pad{\vl{v}}{}) = (\arval{\mk{\bitb}} + \padR{\bitb}\Delta_{\bitb})(\mk{\vl{v}} - \pad{\vl{v}}{}) \nonumber \\
	&= \arval{\mk{\bitb}}\mk{\vl{v}} - \arval{\mk{\bitb}}\pad{\vl{v}}{} + \padR{\bitb}\mk{\vl{v}}\Delta_{\bitb} - \padR{\bitb}\pad{\vl{v}}{}\Delta_{\bitb}
\end{align} 

During the preprocessing, parties generate the $\sqr{\cdot}$-shares of $\padR{\bitb}$ similar to $\prot{\bitA}$ protocol. 
To compute $\padR{\bitb}\pad{\vl{v}}{}$, one naive method is to multiply $\padR{\bitb}$ and $\pad{\vl{v}}{}$ using $\prot{\MultPre}$. The cost can be reduced further as follows. Note that
\begin{align}\label{eq:2pcSbitinjB}
	\padR{\bitb}\pad{\vl{v}}{} &= (\padR{\bitb_1} + \padR{\bitb_2} - 2 \padR{\bitb_1} \padR{\bitb_2})(\pad{\vl{v}}{1} + \pad{\vl{v}}{2}) \nonumber \\
	&=  \padR{\bitb_1}\pad{\vl{v}}{1} + \padR{\bitb_1}\pad{\vl{v}}{2} + \padR{\bitb_2}\pad{\vl{v}}{1} + \padR{\bitb_2}\pad{\vl{v}}{2} - 2 \padR{\bitb_1} \padR{\bitb_2}\pad{\vl{v}}{1}  - 2 \padR{\bitb_1} \padR{\bitb_2}\pad{\vl{v}}{2}
\end{align} 
Here $P_1$ can locally compute $\padR{\bitb_1}\pad{\vl{v}}{1}$ while $P_2$ can compute $\padR{\bitb_2}\pad{\vl{v}}{2}$. The $\sqr{\cdot}$-shares for the remaining four terms can be generated using four instances of $\COT{1}{\ell}$ similar to $\prot{\bitA}$ resulting in a communication of $4(\csec + \ell)$ bits. For instance, to compute $\sqr{\cdot}$-shares of $\padR{\bitb_1} \padR{\bitb_2}\pad{\vl{v}}{1}$, parties engage in an instance of $\COT{1}{\ell}$ with $P_1$ as sender with input $ \padR{\bitb_1}\pad{\vl{v}}{1}$ and $P_2$ as receiver with choice bit $\pad{\bitb_2}{}$.

During the online phase, $P_1$ and $P_2$ compute an additive sharing of $\arval{\bitb} \vl{v}$ and execute $\prot{\Sh}$ on them to generate the respective $\shr{\cdot}$-shares. 

\begin{lemma}[Communication]
	\label{lemma:2pcSbitinj}
	Protocol $\prot{\bitinj}$ requires $5(\csec +\ell)$ bits of communication in preprocessing, and $1$ round and $2 \ell$ bits of communication in the online phase.
\end{lemma}

 \subsubsection{Sum of Bit Injections}
 \label{sec:2pcSSumBitInj}
Given $m$ pair of values in the shared form, $\{\shrB{\bitb_i}, \shr{\vl{v}_i}\}_{i \in [m]}$, the goal of $\prot{\bitinjS}$ is to compute the $\shr{\cdot}$-share of $\vl{z} = \sum_{i=1}^{m} \arval{\bitb_i} \cdot \vl{v_i}$. For this, parties execute the preprocessing corresponding to $m$ bit injections of the form $\arval{\bitb_i} \cdot \vl{v_i}$. 

In the online phase, each of $P_1$ and $P_2$ locally compute an additive sharing of $\vl{z}_i$, corresponding to $\arval{\bitb_i} \cdot \vl{v_i}$ first. Instead of generating the $\shr{\cdot}$-sharing for each of the $m$ terms, parties locally add the shares and execute $\prot{\Sh}$ on the result. Concretely, parties locally compute $\vl{z}^j = \sum_{i=1}^{m} \vl{z}_i^j$ for $j \in \{1,2\}$ and execute $\prot{\Sh}$ on $\vl{z}^j$ to obtain its $\shr{\cdot}$-sharing. This results in an online communication independent of $m$.

\begin{lemma}[Communication]
	\label{lemma:2pcSSumBitInj}
	Protocol $\prot{\bitinjS}$ requires $5m (\csec +\ell)$ bits of communication in preprocessing, and $1$ round and $2 \ell$ bits of communication in the online phase.
\end{lemma}

\subsubsection{Bit Injection:II}
\label{sec:2pcSBitInjII}
Similar to $\prot{\bitinj}$ protocol, given $\shrB{\bitb_1}, \shrB{\bitb_2}$ and $\shr{\vl{v}}$, protocol $\prot{\dbitA}$ computes the arithmetic sharing of $\arval{(\bitb_1 \bitb_2)}\vl{v}$. Let $\Delta_{\bitb_1}$, $\Delta_{\bitb_2}$  denote the value $(1-2\arval{\mk{\bitb_1}})$, $(1-2\arval{\mk{\bitb_2}})$ respectively. Using \eqref{eq:2pcSdbitA} and \eqref{eq:2pcSbitinj}, we can write

\begin{align}
	\label{eq:2pcSdbitinj}
	\arval{(\bitb_1 \bitb_2)} \vl{v} &= \arval{(\mk{\bitb_1} \xor \pad{\bitb_1}{})} \arval{(\mk{\bitb_2} \xor \pad{\bitb_2}{})} (\mk{\vl{v}} - \pad{\vl{v}}{}) \nonumber \\ 
	&= (\arval{\mk{\bitb_1}} + \padR{\bitb_1}\Delta_{\bitb_1}) (\arval{\mk{\bitb_2}} + \padR{\bitb_2}\Delta_{\bitb_2})(\mk{\vl{v}} - \pad{\vl{v}}{}) \nonumber \\
	&= \arval{\mk{\bitb_1}}\arval{\mk{\bitb_2}}\mk{\vl{v}}  + \padR{\bitb_1}\arval{\mk{\bitb_2}}\mk{\vl{v}} \Delta_{\bitb_1} + \padR{\bitb_2}\arval{\mk{\bitb_1}}\mk{\vl{v}}\Delta_{\bitb_2} + \arval{(\pad{\bitb_1}{}\pad{\bitb_2}{})}\mk{\vl{v}}\Delta_{\bitb_1}\Delta_{\bitb_2} \nonumber \\
	&~~~- \pad{\vl{v}}{}\arval{\mk{\bitb_1}}\arval{\mk{\bitb_2}} - \padR{\bitb_1}\pad{\vl{v}}{}\arval{\mk{\bitb_2}}\Delta_{\bitb_1} - \padR{\bitb_2}\pad{\vl{v}}{}\arval{\mk{\bitb_1}}\Delta_{\bitb_2} - \arval{(\pad{\bitb_1}{}\pad{\bitb_2}{})}\pad{\vl{v}}{}\Delta_{\bitb_1}\Delta_{\bitb_2}
\end{align} 

During preprocessing, the $\sqr{\cdot}$-shares of $\padR{\bitb_1}, \padR{\bitb_2}$ and $\arval{(\pad{\bitb_1}{}\pad{\bitb_2}{})}$ are computed similar to that of $\prot{\dbitA}$. Once the $\sqr{\cdot}$-shares are generated, parties compute $\sgr{\padR{\bitb_1}\pd{\vl{v}}}$ and $\sgr{\padR{\bitb_2}\pd{\vl{v}}}$ using steps similar to $\prot{\bitinj}$. 
Using the boolean shares of $\sqr{\pad{\bitb_1}{}\pad{\bitb_2}{}}$ computed as part of  $\sqr{\arval{(\pad{\bitb_1}{}\pad{\bitb_2}{})}}$ and the $\sqr{\cdot}$-shares of $\pad{\vl{v}}{}$, parties compute   the $\sqr{\cdot}$-shares of $\arval{(\pad{\bitb_1}{}\pad{\bitb_2}{})}\pad{\vl{v}}{}$ similar to protocol $\prot{\bitinj}$.
The online phase is similar to that of $\prot{\bitinj}$ protocol.

\begin{lemma}[Communication]
	\label{lemma:2pcSdbitinj}
	Protocol $\prot{\dbitinj}$ requires $14\csec + 12\ell + 2$ bits of communication in preprocessing, and $1$ round and $2 \ell$ bits of communication in the online phase.
\end{lemma}

\subsection{Equality Test~($\prot{\eql}$)}
To check whether $\vl{a} \iseq \vl{b}$ or not, given $\shr{\vl{a}}, \shr{\vl{b}}$, $\prot{\eql}$ proceeds with parties locally computing $\shr{\vl{y}} = \shr{\vl{a}} - \shr{\vl{b}}$. According to our sharing semantics, $\vl{y}$ can be written as $\vl{y} = \vl{y}_1 - \vl{y}_2$ where $\vl{y}_1 = \mk{\vl{y}} - \pad{\vl{y}}{1}$ and $\vl{y}_2 = \pad{\vl{y}}{2}$. $P_2$ generates $\shrB{\vl{y}_2}$ during the preprocessing while $P_1$ generates $\shrB{\vl{y}_1}$ in the online using $\prot{\Sh}$. Note that $\vl{a} = \vl{b}$ implies $\vl{y}_1 = \vl{y}_2$ and hence all the bits of $\vl{v} = \overline{(\vl{y}_1 \xor \vl{y}_2)}$ should be $1$. As mentioned in the introduction of Part II~(\ref{chap:layer2_intro}), parties use four input AND gates and a tree structure, where $4$ bits are taken at a time and the AND of them is computed in one go. 

\section{Mixed Protocol Framework}
\label{sec:2pcSixFrame}
\tabref{2pcSConv} compares our sharing conversions with ABY~\cite{NDSS:DemSchZoh15}. For uniformity, we consider a function, {\sf F}, to be computed on an $\ell$-bit inputs $\vl{x}, \vl{y}$ using a garbled circuit (GC) in the mixed framework, which gives an $\ell$-bit output $\vl{z} = \mathsf{F(\vl{x}, \vl{y})}$, where $\ell$ denotes the ring size in bits. Let  $\Grb{F}$  denote the corresponding GC. In the table, $\Grb{Sn}$ denotes a ${\sf n}$-input garbled subtraction circuit; $\Grb{An}$ denotes ${\sf n}$-input garbled addition circuit; $\GrbD{}$ denotes the garbled circuit with decoding information; $\Grb{n_1\times1,\ldots,n_m \times m}$ denotes ${\sf n_i}$ instances of GC $\Grb{i}$ for $i \in \{1,\dots,{\sf m}\}$ and $\Size{\Grb{n_1\times1,\ldots,n_m \times m}}$ denotes its size. 

\begin{table}[htb!]
	\centering
	\resizebox{0.98\textwidth}{!}{
		\begin{NiceTabular}{rr|rrr|rrrr}
			\toprule 
			\Block{2-1}{Variant\tabularnote{Notations: $\ell$ - size of ring in bits, $\kappa$ - computational security parameter, 'pre' - preprocessing, 'on' - online.}} 
			& \Block{2-1}{Conversion\tabularnote{'A' - arithmetic, 'B' - boolean, 'G' - Garbled.}}
			& \Block[c]{1-3}{ABY~\cite{NDSS:DemSchZoh15}} & &
			& \Block[c]{1-4}{$\TWthis$} & & & \\ \cmidrule{3-9}
			& & Comm.\textsubscript{pre} & Comm.\textsubscript{on}  & Rounds\textsubscript{on} 
			& \Block[c]{1-2}{Comm.\textsubscript{pre}} & & Comm.\textsubscript{on}  & Rounds\textsubscript{on} \\
			\midrule
			\Block{4-1}{1 GC}
			& A-G-A & $14\ell \kappa + \Size{\Grb{2 \times A2,F}}$ & $6\ell\kappa + (\ell^2 + 7\ell)/2$ & $4$
			& \Block[r]{4-1}{($3 \ell \kappa + 2\ell$)\\+}  & $\Size{\GrbD{2 \times S2, A2, F}}$ 
			& \Block{4-1}{$2\ell\kappa + \ell$} & \Block{4-1}{$2$} \\
			& A-G-B & $12\ell \kappa + \Size{\Grb{F}}$ & $6\ell\kappa + 2\ell$ & $2$  
			& & $\Size{\GrbD{2 \times S2,F}}$ &  & \\
			& B-G-A & $14\ell \kappa + \Size{\Grb{F}}$ & $4\ell\kappa + (\ell^2 + 7\ell)/2$  & $4$  
			& & $\Size{\GrbD{A2,F}}$    &  & \\
			& B-G-B & $12\ell \kappa + \Size{\Grb{F}}$ & $4\ell\kappa + 2\ell$  & $2$  
			& & $\Size{\GrbD{F}}$                   &  & \\
			\midrule
			\Block{2-1}{Others\tabularnote{$\vl{u_1} = \vl{n_2} + 4\vl{n_3} + 11\vl{n_4}$, $\vl{u_2} = \vl{n_2} + \vl{n_3} + \vl{n_4}$ denote the number of AND gates in the optimized adder circuit~\cite{USENIX:PSSY21} with 2, 3, 4 inputs, respectively. For $\ell = 64$, $\vl{n_2} = 216, \vl{n_3} = 184, \vl{n_4}=179$.}}
			& A-B & $2\ell\log \ell (\kappa + \ell)$  & $4\ell\log \ell$ & $\log \ell$
			& \Block[r]{1-2}{$2\vl{u_1} (\kappa + \ell)$} &  & $2\vl{u_2} + \ell$ & $1 + \log_4 \ell$ \\
			& B-A &  $2\ell\kappa$  & $(\ell^2 + 3\ell) / 2$ & $2$
			& \Block[r]{1-2}{$\ell\kappa + \ell^2$} &  & $2\ell$ & $1$ \\
			%
			\bottomrule
		\end{NiceTabular}
	}
	\caption{Mixed protocol conversions of ABY~\cite{NDSS:DemSchZoh15} and $\TWthis$.}\label{tab:2pcSConv}
\end{table}

\subsection{Conversions involving Garbled World} 
\label{sec:2pcSconv2gc}
Assume the GC is required to compute a function $f$ on inputs $\vl{x}, \vl{y} \in \Z{\ell}$ and let the output be $f(\vl{x}, \vl{y})$. 

\paragraph{Case I: Boolean-Garbled-Boolean}
Since the inputs to the GC are available in boolean form, say $\shrB{\vl{x}}, \shrB{\vl{y}}$, parties generate $\shrC{\vl{x}}, \shrC{\vl{y}}$ by invoking the garbled sharing protocol $\pigsh$.
Additionally, $P_1$ samples $\vl{R} \in \Z{\ell}$ to mask the function output, $f(\vl{x}, \vl{y})$, and generate $\shrB{\vl{R}}$ and $\shrG{\vl{R}}$. $P_g = P_1$ garbles the circuit which computes $\vl{z} = f(\vl{x}, \vl{y}) \xor \vl{R}$, and sends the GC along with the decoding information to evaluator $P_2$. Upon GC evaluation and output decoding, $P_2$ obtains $\vl{z} = f(\vl{x}, \vl{y}) \xor \vl{R}$, and boolean share $\vl{z}$ to generate $\shrB{\vl{z}}$. Parties then compute $\shrB{f(\vl{x}, \vl{y})} = \shrB{\vl{z}} \xor \shrB{\vl{R}}$.  

\paragraph{Case II: Boolean-Garbled-Arithmetic}
This is similar to {\em Case I} except that the circuit which computes $\vl{z} = f(\vl{x}, \vl{y}) + \vl{R}$ is garbled instead. Boolean sharing of $\vl{z}$ is replaced with arithmetic, followed by computing $\shr{f(\vl{x}, \vl{y})} = \shr{\vl{z}} - \shr{\vl{R}}$.

\paragraph{Cases III \& IV: Input in Arithmetic Sharing} 
The function to be computed $f(\vl{x}, \vl{y})$, is modified as $f^{\prime}(\mk{\vl{x}}, \pad{\vl{x}}{1}, \pad{\vl{x}}{2}, \mk{\vl{y}}, \pad{\vl{y}}{1}, \pad{\vl{y}}{2}) = f(\mk{\vl{x}}-\pad{\vl{x}}{1}-\pad{\vl{x}}{2}, \mk{\vl{y}}-\pad{\vl{y}}{1}-\pad{\vl{y}}{2})$ where inputs $\vl{x}, \vl{y}$ are replaced by the sets $\{\mk{\vl{x}}, \pad{\vl{x}}{1}, \pad{\vl{x}}{2}, \pad{\vl{x}}{3}\}$, $\{\mk{\vl{y}}, \pad{\vl{y}}{1}, \pad{\vl{y}}{2}, \pad{\vl{y}}{3}\}$. The circuit to be garbled thus, corresponds to the function $f^{\prime}$. Parties generate $\shrG{\mk{\vl{x}}}, \shrG{\pad{\vl{x}}{1}}, \shrG{\pad{\vl{x}}{2}}, \shrG{\mk{\vl{y}}}, \shrG{\pad{\vl{y}}{1}}, \shrG{\pad{\vl{y}}{2}}$ via $\pigsh$, following which, parties proceed with the rest of the computation whose steps are similar to {\em Case I}, and {\em II}, depending on the requirement on the output sharing.
Function $f^{\prime}$ can be further optimized as $f(\av{\vl{x}}-\pad{\vl{x}}{2}, \av{\vl{y}}-\pad{\vl{y}}{2})$ with $\av{\vl{x}} = \mk{\vl{x}} - \pad{\vl{x}}{1}$ and $\av{\vl{y}} = \mk{\vl{y}} - \pad{\vl{y}}{1}$. Similar optimization can be done for the other garbling instance as well.

\subsection{Other Conversions} 
\label{sec:2pcSotherconv}

\paragraph{Arithmetic to Boolean} 
To convert arithmetic sharing of $\vl{v} \in \Z{\ell}$ to boolean sharing, observe that $\vl{v} = \vl{v}_1 + \vl{v}_2$ where $\vl{v}_1 = \mk{\vl{v}} - \pad{\vl{v}}{1}$ is possessed by $P_1$, while $\vl{v}_2 = - \pad{\vl{v}}{2}$ is possessed by $P_2$. Thus, $\shrB{\vl{v}}$  can be computed as $\shrB{\vl{v}} = \shrB{\vl{v}_1} + \shrB{\vl{v}_2}$. For this, $P_2$ can generate $\shrB{\vl{v}_2}$ in the preprocessing, and $\shrB{\vl{v}_1}$ can be generated in the online by $P_1$. The protocol appears in \boxref{fig:2pcSpiab}. Boolean addition, when instantiated using the adder of \cite{USENIX:PSSY21}, requires $\log_4(\ell)$ rounds.

\begin{protocolbox}{$\piab$}{Arithmetic to Boolean Conversion in $\TWthis$.}{fig:2pcSpiab}
	\justify
	\algoHead{Preprocessing:} $P_2$ generates $\shrB{\vl{v}_2}$ using $\prot{\Sh}$, where $\vl{v}_2 = -\pad{\vl{v}}{2}$.
	\justify
	\vspace{-2mm}
	\algoHead{Online:}
	\begin{enumerate} 
		\item $P_1$ generates $\shrB{\vl{v}_1}$ using $\prot{\Sh}$, where $\vl{v}_1 = \mk{\vl{v}} -\pad{\vl{v}}{1}$.
		\item Parties obtain $\shrB{\vl{v}} = \shrB{\vl{v}_1} + \shrB{\vl{v}_2}$ using a boolean adder circuit.
	\end{enumerate}
\end{protocolbox}

\paragraph{Boolean to Arithmetic} 
To convert a boolean sharing of $\vl{v} \in \Z{\ell}$ into an arithmetic sharing, note that 
\begin{small}
	\begin{align*}
		\vl{v} = \sum_{i=0}^{\ell - 1} 2^{i} \vl{v}[i] = \sum_{i=0}^{\ell - 1} 2^{i} (\pad{\vl{v}[i]}{} \xor \mk{\vl{v}[i]}) 
		= \sum_{i=0}^{\ell - 1} 2^{i} \left( \arval{\mk{\vl{v}[i]}} +  \padR{\vl{v}[i]} (1 - 2\arval{\mk{\vl{v}[i]}}) \right) 
	\end{align*}
\end{small}
where $\padR{\vl{v}[i]}, \arval{\mk{\vl{v}[i]}}$ denote the arithmetic value of bits $\pad{\vl{v}[i]}{}, \mk{\vl{v}[i]}$ over the ring $\Z{\ell}$. 
For each bit $\vl{v}[i]$ of $\vl{v}$, parties generate the $\sqr{\cdot}$-shares of $\padR{\vl{v}[i]}$ in the preprocessing, similar to $\prot{\bitA}$~(\boxref{fig:2pcSpiBitA}). During the online phase, additive shares for each bit $\vl{v}[i]$ are locally computed similar to $\prot{\bitA}$. Parties then multiply the $i$th share with $2^i$ and locally add up to obtain an additive sharing of $\vl{v}$. The rest of the steps are similar to $\prot{\bitA}$, and the formal protocol appears in \boxref{fig:2pcSpiba}.

\begin{protocolbox}{$\piba(\Partyset, \shrB{\vl{v}})$}{Boolean to Arithmetic Conversion in $\TWthis$.}{fig:2pcSpiba}
	Let $\vl{v}[i]$ denote the $i$th bit of $\vl{v}$. Let $\vl{p}_i = \arval{\mk{\vl{v}[i]}}$, and $\vl{q}_i = \padR{\vl{v}[i]}$. \\
	\justify 
	\vspace{-2mm}
	\algoHead{Preprocessing:} 
	\begin{enumerate} 
		\item For $i \in \{0, 1, \ldots, \ell-1 \} $, execute the preprocessing of $\prot{\bitA}$ (\boxref{fig:2pcSpiBitA}) for each bit $\vl{v}[i]$, to generate $\sqr{{\vl{q}_i}} = (\sqr{\vl{q}_i}_1 , \sqr{\vl{q}_i}_2 )$.
	\end{enumerate}
	\justify
	\vspace{-2mm}
	\algoHead{Online:} Let $\vl{y}_i = \arval{(\vl{v}[i])}$ and $\vl{y}$ denotes the arithmetic equivalent of $\vl{v}$.
	\begin{enumerate}
		\item Locally compute:
		\begin{align*}
			P_1: \vl{y}^1 = \sum_{i=0}^{\ell-1} 2^i \vl{y}_i^1  &=  \sum_{i=0}^{\ell-1} 2^i (\vl{p}_i + \sqr{\vl{q}_i}_1 (1 - 2\vl{p}_i)) \\
			P_2: \vl{y}^2 = \sum_{i=0}^{\ell-1} 2^i \vl{y}_i^2 &=  \sum_{i=0}^{\ell-1} 2^i ( \sqr{\vl{q}_i}_2 (1 - 2\vl{p}_i))  
		\end{align*}
		\item $P_j$ for $j \in \{1,2\}$ executes $\prot{\Sh}$ on $\vl{y}^j$ to generate the respective $\shr{\cdot}$-shares.
		\item Compute $\shr{\vl{y}} = \shr{\vl{y}^1} + \shr{\vl{y}^2}$.
	\end{enumerate}     
\end{protocolbox}

\part{Layer III: Applications}
\label{part:layer3}

\chapter*{Introduction to Layer III}
\label{chap:layer3_intro}
Solutions to privacy-preserving machine learning via MPC have been looked at in various works~\cite{SP:MohZha17, CCS:MohRin18, PoPETS:WagGupCha19, NDSS:PatSur20, CCSW:CCPS19, NDSS:ChaRacSur20, USENIX:KPPS21}. Our work considers PPML algorithms such as linear regression, logistic regression, deep neural networks~(NN) and support vector machines~(SVM) for benchmarking. We consider both the training and inference phases of all the algorithms except SVM. The training phase of SVM requires additional tools and primitives and is out of the scope of this work. We first give an overview of the ML algorithms, followed by the architectural details of the neural networks and support vector machine that we consider for benchmarking and the corresponding datasets.

\section*{Overview of ML algorithms}
Here we provide an overview of ML algorithms and the detailed benchmarking results. The training phase in most machine learning algorithms consists of two stages– i) forward propagation, where the model computes the output, and ii) backward propagation, where the model parameters are adjusted according to the computed output and the actual output. We define one {\em iteration} in the training phase as one forward propagation followed by a backward propagation. We refer readers to~\cite{SP:MohZha17,CCS:MohRin18,SP:DEFKSV19,NDSS:PatSur20,NDSS:ChaRacSur20,PoPETS:WTBKMR21} for formal details.

\subsection*{Linear Regression}
For linear regression, one iteration can be viewed as updating the weight vector $\vct{w}$ using the Gradient Descent algorithm (GD). The update function for $\vct{w}$ is given by

\begin{align*}
	\vct{w} = \vct{w} - \frac{\alpha}{B} \Mat{X}_i^{T}  \circ (\Mat{X}_i \circ \vct{w} -\Mat{Y}_i)
\end{align*}
where $\alpha$ denotes the learning rate and $\Mat{X}_i$ denotes a subset of batch size $B$, randomly selected from the entire dataset in the $i$th iteration. Here the forward propagation consists of computing $\Mat{X}_i \circ \vct{w}$, while the weight vector is updated in the backward propagation. The update function consists of a series of matrix multiplications, which can be achieved using dot product protocols. The operations of subtraction, as well as multiplication by a public constant, can be performed locally. We observe that the update function as mentioned above can be computed entirely in the arithmetic domain and can be viewed in the form of $\shr{\cdot}$-shares as
\begin{align*}
	\shr{\vct{w}} = \shr{\vct{w}} - \frac{\alpha}{B} \shr{\Mat{X}_j^{T}}  \circ (\shr{\Mat{X}_j} \circ \shr{\vct{w}}-\shr{\Mat{Y}_j})
\end{align*}

\subsection*{Logistic Regression}
The iteration for the case of logistic regression is similar to that of linear regression, apart from an activation function being applied on $\Mat{X}_i \circ \vct{w}$ in the forward propagation. We instantiate the activation function using the sigmoid function. The update function for $\vct{w}$ is given by
\begin{align*}
	\vct{w} = \vct{w} - \frac{\alpha}{B} \Mat{X}_i^{T}  \circ \left(\sig(\Mat{X}_i \circ \vct{w})-\Mat{Y}_i\right)
\end{align*}
One iteration of logistic regression incurs an additional cost for computing $\sig(\Mat{X}_j  \circ {\vct{w}})$ as compared with that for linear regression.

\subsection*{Neural Networks}
A neural network can be divided into various layers, where each layer contains a predefined number of nodes. These nodes are a linear function composed of a non-linear “activation” function. The nodes at the input layer or the first layer are evaluated on the input features to evaluate a neural network. The outputs from these nodes are fed as inputs to the nodes in the next layer. This process is repeated for all the layers to obtain the output. The underlying operation involved is the computation of activation matrices in all the layers. This constitutes the forward propagation phase. The backward propagation involves adjusting model parameters according to the difference between the computed and actual output and comprises computing error matrices. 

Concretely, each layer comprises matrix multiplications followed by an application of the ReLU function. The maxpool layer additionally follows convolutional layers after the ReLU layer. After evaluating the layers in a sequential manner, at the output layer, we use the MPC friendly variant of the softmax activation function, $\sftmx(u_i) = \frac{\relu(u_i)}{\sum_{j = 1}^{\vl{n}} \relu(u_j)}$, proposed by SecureML~\cite{SP:MohZha17}. To perform the division, we switch from arithmetic to garbled world and then use a division garbled circuit~\cite{FCW:PulSii15} followed by a switch back to the arithmetic world. 

The network is trained using the Gradient Descent, where the forward propagation comprises of computing activation matrices for all the layers in the network. Here, the activation matrix for all the layers except the output, is defined as $\Mat{A}_i = \relu(\Mat{U}_i)$, where $\Mat{U}_i = \Mat{A}_{i-1} \MatMul \Mat{W}_i$. $\Mat{A}_0$ is initialized to $\Mat{X}_j$, where $\Mat{X}_j$ is a subset of batch size B, randomly selected from the entire dataset for the j\textsuperscript{th} iteration. The activation matrix for the output layer is defined as $\Mat{A}_m = \sftmx(\Mat{U}_m)$.

During the backward propagation, error matrices are computed first. The error matrix for the output layer is defined as $\Mat{E}_m = (\Mat{A}_m - \Mat{T})$, while for the remaining layers it is defined as $\Mat{E}_i = (\Mat{E}_{i+1} \circ \Mat{W}_i^T) \otimes \drelu(\Mat{U}_i)$. Here the operation $\otimes$ denotes element wise multiplication and $\drelu$ denotes the derivative of ReLU. This is followed by updating the weights as $\Mat{W}_i = \Mat{W}_i - \frac{\alpha}{B} \Mat{A}_{i-1}^T  \circ \Mat{E}_i$.

\subsection*{Support Vector Machines (inference)}
We consider Support Vector Machines (SVM) which is a type of supervised learning algorithm used for classification. SVM is a function which takes as input an $n$-dimensional {\em feature vector}, $\vct{x}$, and outputs the {\em category} to which the feature vector belongs. SVM is implemented as a matrix $\Mat{F}$, of dimension $q \times n$ where each row of $\Mat{F}$ is called the support vector and a vector $\vct{b} = (b_1, \ldots, b_q)$, is called the {\em bias}. Each element of $\Mat{F}$ and $\vct{b}$ lies in $\Z{\ell}$. Each support vector along with a scalar from the bias can classify the input $\vct{x}$ into a specific category. More precisely, let $\Mat{F}_i$ denote the $i^{\text{th}}$ row of matrix $\Mat{F}$. Then, the value $\Mat{F}_i \cdot \vct{x} + b_i$ specifies how likely $\vct{x}$ is to be in category $i$. To find the most likely category, we compute argmax over these values, i.e. $\text{category}(\vct{x}) = \text{argmax}_{i \in \{1, \ldots, q\}} \Mat{F}_i \cdot \vct{x} + b_i$.

\section*{Network architectures}
We consider the following networks for benchmarking. These are chosen based on the different range of model parameters and layers used in the network. We refer readers to \cite{PoPETS:WTBKMR21} for a detailed architecture of the neural networks.

\begin{enumerate}
	\item {\em SVM:} This consists of 10 categories for classification \cite{SP:DEFKSV19}.
	
	\item {\em NN-1:} This is a fully connected network with 3 layers with $\relu$ activation after each layer. This network has around $118$K parameters and is chosen from ~\cite{CCS:MohRin18,NDSS:PatSur20}. 
	
	\item {\em NN-2:} This is a convolutional neural network comprising of 2 hidden layers, with 100 and 10 nodes \cite{ASIACCS:RWTSSK18,CCS:MohRin18,NDSS:ChaRacSur20}.
	
	\item {\em NN-3:} This network, called LeNet \cite{lenet}, comprises of $2$ convolutional layers and $2$ fully connected layers with ReLU activation after each layer, additionally followed by maxpool for convolutional layers. This network has approximately $431$K parameters.
	
	\item {\em NN-4:} This network, called VGG16 \cite{vgg16}, was the runner-up of ILSVRC-2014 competition. This network has $16$ layers in total and contains fully-connected, convolutional, ReLU activation and maxpool layers. This network has about $138$ million parameters. 
	
\end{enumerate}

\section*{Datasets.} 
To benchmark the machine learning algorithms, we use the following real-world datasets:

\begin{enumerate}
	\item[--] MNIST~\cite{MNIST10} is a collection of $28 \times 28$ pixel, handwritten digit images with a label between $0$ and $9$ for each. It has $60,000$ and respectively, $10,000$ images in training and test set. We evaluate Linear Regression, Logistic Regression, NN-1, NN-3 and SVM on this dataset.
	\item[--] CIFAR-10~\cite{CIFAR10} has $32 \times 32$ pixel images of $10$ different classes such as dogs, horses, etc. It has $50,000$ images for training and $10,000$ for testing, with $6000$ images in each class. We evaluate NN-2, NN-4 on this dataset.
\end{enumerate}

\section*{Benchmarking Environment Details} 
The protocols are benchmarked over a Wide Area Network (WAN), instantiated using n1-standard-64 instances of Google Cloud\footnote{https://cloud.google.com/}, with machines located in East Australia ($P_0$), South Asia ($P_1$), South East Asia ($P_2$), and West Europe ($P_3$). The machines are equipped with 2.0 GHz Intel (R) Xeon (R) (Skylake) processors supporting hyper-threading, with 64 vCPUs, and 240 GB of RAM Memory. Parties are connected by pairwise authenticated bidirectional synchronous channels (eg. instantiated via TLS over TCP/IP). We use a limited bandwidth of $40$ MBps between every pair of parties and the average round-trip time ($\rtt$)\footnote{Time for communicating 1 KB of data between a pair of parties} values among the parties are 

\begin{center} 
	\resizebox{0.6\textwidth}{!}
	{
		\begin{NiceTabular}{c c c c c c}
			\toprule
			$P_0$-$P_1$ & $P_0$-$P_2$ & $P_0$-$P_3$ & $P_1$-$P_2$ & $P_1$-$P_3$ & $P_2$-$P_3$\\
			\midrule
			$153.74 ms$ & $93.39 ms$ & $274.84 ms$ & $62.01 ms$  & $174.15 ms$  & $219.46 ms$\\
			\bottomrule 
		\end{NiceTabular}
	}
\end{center}

For a fair comparison, we implemented and benchmarked all the protocols, including the protocols of SecureML~\cite{SP:MohZha17} and ABY3~\cite{CCS:MohRin18}, building on the ENCRYPTO library~\cite{ENCRYPTO} in C++17. Primitives such as maxpool, which SecureML and ABY3 do not support, have been run using our building blocks. We would like to clarify that our code is developed for benchmarking, is not optimized for industry-grade use, and optimizations like GPU support can enhance performance. Our protocols are instantiated over a $64$-bit ring ($\Z{64}$), and the collision-resistant hash function is instantiated using SHA-256. We use multi-threading, and our machines are capable of handling a total of 64 threads. Each experiment is run 10 times, and the average values are reported. We use $1$ KB = $8192$ bits and use a batch size of $B = 128$ for training.

	\begin{table}[htb!]
		\centering
		\resizebox{0.7\textwidth}{!}{
			\begin{NiceTabular}{p{3cm} l }
				\toprule
				Notation & Description\\
				\midrule 
				${\sf T}_{\sf on,i}$        & Online runtime of party $P_i$.\\
				${\sf T}_{\sf tot,i}$        & Total runtime of party $P_i$.\\
				${\sf PT}_{\sf on}$        & Protocol online runtime; ${\sf max_i} \{ {\sf T}_{\sf on,i} \}$ .\\
				${\sf PT}_{\sf tot}$        & Protocol total runtime;  ${\sf max_i} \{ {\sf T}_{\sf tot,i} \}$ .\\
				${\sf CT}_{\sf on}$         & Cumulative online runtime; $\Sigma_i {\sf T}_{\sf on,i}$ .\\
				${\sf CT}_{\sf tot}$         & Cumulative total runtime; $\Sigma_i {\sf T}_{\sf tot,i}$ .\\
				${\sf Comm}_{\sf on}$    & Online communication.\\
				${\sf Comm}_{\sf tot}$    & Total communication.\\
				${\sf Cost}$                     & Total monetary cost.\\
				$\TP$                              & \Block[l]{}{Online throughput;~{\small higher = better}\\(\#iterations /  \#queries per minute in online)}\\
				\bottomrule
			\end{NiceTabular}
		}
		\caption{Benchmarking parameters (lower is better, except for $\TP$)\label{tab:notations}}
	\end{table}

\section*{Benchmarking Parameters} 
We evaluate the protocols across a variety of parameters as given in \tabref{notations}. In addition to parameters such as runtime, communication, and {\em online throughput} ($\TP$)~\cite{CCS:AFLNO16,SP:ABFLLN17,CCS:MohRin18,NDSS:ChaRacSur20}, the cumulative runtime (sum of the up-time of all the hired servers) is also reported. This is because when deployed over third-party cloud servers, one pays for them by the communication and the uptime of the hired servers. To analyze the cost of deployment of the framework, {\em monetary cost} ($\sf Cost$)~\cite{C:MPRSY20} is reported. This is done using the pricing of Google Cloud Platform\footnote{See https://cloud.google.com/vpc/network-pricing for network cost and  https://cloud.google.com/compute/vm-instance-pricing for computation cost.}, where for $1$ GB and $1$ hour of usage, the costs are USD $0.08$ and USD $3.04$, respectively. For protocols with an asymmetric communication graph, communication load is unevenly distributed among all the servers, leaving several communication channels underutilized. Load balancing improves the performance by running several execution threads in parallel, each with the roles of the servers changed. Load balancing has been performed in all the protocols benchmarked.

\subsection*{Discussion}
Broadly speaking, we consider two deployment scenarios -- optimized for time~({\sf T}), and for cost~({\sf C}). In the first one, participants want the result of the output as soon as possible while maximizing the online throughput. In the second one, they want the overall monetary cost of the system to be minimal and are willing to tolerate an overhead in the execution time. 
Using multi-input multiplication gates and the 2 GC variant of the garbled makes the online phase faster but incur an increase in monetary cost. This is because they cause an overhead in communication in the preprocessing phase, and communication affects monetary cost more than uptime (in our setting).

\chapter{$\TSthis$: 3PC Semi-honest Applications}
\label{chap:layer3_3pcsemi}
$\TSthisT$ uses multi-input multiplication gates and the 2 GC variant of the garbled world and is the fastest variant of the framework. On the other hand, $\TSthisC$ is the variant with a minimal monetary cost. We benchmark our protocols against the 3PC semi-honest framework of ABY3~\cite{CCS:MohRin18}. 

\section{ML Training}
\label{sec:bench_train_3pcS}

We begin with analyzing the benchmarks for linear and logistic regression. 
Starting with the time-optimized variant, $\TSthisT$ is $2.5 - 4\times$ faster than ABY3~\cite{CCS:MohRin18} in online runtime for training.
For linear regression, this reduction is observed due to the different $\rtt$s among the three parties. This difference vanishes if $\rtt$ between every pair of parties is the same. However, the reduction in the online run time for the case of logistic regression is primarily due to the round-optimized bit extraction circuit. Specifically, we use the depth-optimized bit extraction circuit while instantiating the sigmoid activation function using multi-input AND gates. 
We observe a reduction of up to $2\times$ in communication~(${\sf Comm}_{\sf tot}$) in $\TSthisT$ over ABY3. This is due to the extra cost required for performing truncation in ABY3. 
These reductions in communication and run time, coupled with the requirement of one less party in the online phase, directly impact the monetary cost of the system, where $\TSthisT$ brings in a saving of up to $78\%$ over ABY3. 
On the other hand, the cost-optimized variant $\TSthisC$ is around $1.5\times$ slower in the online phase than $\TSthisT$. However, it is still faster than ABY3 due to the reason discussed above. Further, this variant has $1.3\times$ lesser communication cost compared to $\TSthisT$.  

\begin{table}[htb!]
	\centering
	\resizebox{0.8\textwidth}{!}{
		\begin{NiceTabular}{r r | r r r | r r r}[notes/para]
			\toprule
			\Block{2-1}{Algorithm} & \Block{2-1}{Parameter\tabularnote{Time~(in seconds) and communication~(in KB) are reported.}}
			& \Block{1-3}{Training\tabularnote{For training, batch size is 128 and the monetary cost~(USD) is reported for $1000$ iterations.}} & & 
			& \Block{1-3}{Inference\tabularnote{ For inference, cost is reported for $1000$ queries.} } & & 
			\\ \cmidrule{3-8}
			& & ABY3& \TSthisT & \TSthisC & ABY3 & \TSthisT & \TSthisC \\
			\midrule 
            \Block{6-1}{Linear\\Regression} 
			& ${\sf PT}_{\sf on}$	&0.31	&0.12	&0.12	&0.15	&0.06	&0.06	\\
			
			& ${\sf PT}_{\sf tot}$	&0.32	&0.12	&0.12	&0.15	&0.06	&0.06	\\
			
			& ${\sf CT}_{\sf tot}$	&0.72	&0.25	&0.25	&0.34	&0.12	&0.12	\\
			
			& ${\sf Comm}_{\sf tot}$	&57.12	&27.5	&27.5	&0.05	&0.02	&0.02	\\
			
			& ${\sf Cost}$	&0.62	&0.21	&0.21	&0.29	&0.1	&0.1	\\
			
			& $\TP$ 	&24977.23	&37465.85	&37465.85	&49957.72	&74936.58	&74936.58	\\
			
			\midrule 
			\Block{6-1}{Logistic\\Regression} 
			& ${\sf PT}_{\sf on}$	&1.54	&0.37	&0.56	&1.38	&0.3	&0.48	\\
			
			& ${\sf PT}_{\sf tot}$	&1.55	&0.37	&0.56	&1.38	&0.3	&0.48	\\
			
			& ${\sf CT}_{\sf tot}$	&3.48	&0.74	&1.12	&3.08	&0.6	&0.96	\\
			
			& ${\sf Comm}_{\sf tot}$	&76.93	&63.5	&47.31	&0.2	&0.3	&0.18	\\
			
			& ${\sf Cost}$	&2.95	&0.64	&0.9	&2.61	&0.5	&0.81	\\
			
			& $\TP$ 	&4995.45	&12488.62	&8325.74	&5550.82	&14987.32	&9367.07	\\
			\bottomrule
		\end{NiceTabular}
	}
	\caption{Benchmarking of Linear Regression and Logistic Regression algorithms.\label{tab:reg3pcS}}
\end{table}

For neural networks, $\TSthisT$ is up to $3.6 \times$ faster than ABY3 in the online phase, similar to the observation in logistic regression. 
Concerning the communication, $\TSthisT$ has a slightly higher communication than ABY3 for smaller NNs. However, the gap closes for larger NNs. This phenomenon is observed because of the trade-off in the increase in communication due to the use of multi-input multiplication versus the reduction in communication due to the free truncation operation. However, the cost-optimized variant, $\TSthisC$, has a better communication cost than ABY3. Further, $\TSthisC$ is up to $1.4 \times$ slower than $\TSthisT$ in terms of online run time, while it is better than ABY3. Note that the requirement of one less party in the online phase coupled with the improvements in communication and run time results in saving up to $87\%$ in the monetary cost of $\TSthisC$ over ABY3, and up to $18\%$ over $\TSthisT$. As the depth increases, we observe that the gap in the monetary cost of $\TSthisC$ and $\TSthisT$ closes in.

These trends can be better captured with a pictorial representation as given in \figref{MLPlot3pcS}.

\begin{table}[htb!]
	\centering
	\resizebox{0.8\textwidth}{!}{
		\begin{NiceTabular}{r r | r r r | r r r}[notes/para]
			\toprule
			\Block{2-1}{Algorithm} & \Block{2-1}{Parameter\tabularnote{Time is reported in seconds}}
			& \Block{1-3}{Training\tabularnote{For training, communication is reported in GB. Monetary cost~(USD) is reported for $1000$ iterations and batch size is 128.}} & & 
			& \Block{1-3}{Inference\tabularnote{ For inference, communication is reported in MB and the cost is reported for $1000$ queries.} } & &
			\\ \cmidrule{3-8}
			& & ABY3 & \TSthisT & \TSthisC & ABY3 & \TSthisT & \TSthisC\\
			\midrule 
            \Block{6-1}{NN-1} 
            & ${\sf PT}_{\sf on}$	&5.66	&1.55	&2.17	&4.15	&0.93	&1.49	\\
            
            & ${\sf PT}_{\sf tot}$	&11.36	&4.11	&4.4	&4.16	&0.93	&1.49	\\
            
            & ${\sf CT}_{\sf tot}$	&26.97	&10.12	&8.81	&9.29	&1.86	&2.98	\\
            
            & ${\sf Comm}_{\sf tot}$	&0.15	&0.29	&0.15	&0.03	&0.04	&0.03	\\
            
            & ${\sf Cost}$	&48.49	&54.61	&30.9	&7.81	&1.54	&2.5	\\
            
            & $\TP$ 	&1160.7	&2844.48	&2139.75	&1850.17	&4995.45	&3122.15	\\
            
            \midrule 
            \Block{6-1}{NN-2} 
            & ${\sf PT}_{\sf on}$	&5.78	&1.64	&2.26	&4.15	&0.93	&1.49	\\
            
            & ${\sf PT}_{\sf tot}$	&30.64	&4.35	&4.98	&4.17	&0.93	&1.49	\\
            
            & ${\sf CT}_{\sf tot}$	&84.81	&11.26	&9.96	&9.33	&1.86	&2.98	\\
            
            & ${\sf Comm}_{\sf tot}$	&0.23	&0.34	&0.19	&0.13	&0.18	&0.12	\\
            
            & ${\sf Cost}$	&115.65	&63.85	&39.15	&7.85	&1.55	&2.51	\\
            
            & $\TP$ 	&225.59	&489.56	&483.51	&1850.17	&4995.45	&3122.15	\\
            
            \midrule 
            \Block{6-1}{NN-3} 
            & ${\sf PT}_{\sf on}$	&18.58	&5.42	&8.15	&10.45	&2.23	&3.72	\\
            
            & ${\sf PT}_{\sf tot}$	&157.39	&10.88	&13.61	&10.7	&2.24	&3.73	\\
            
            & ${\sf CT}_{\sf tot}$	&458	&24.32	&27.23	&24.12	&4.48	&7.46	\\
            
            & ${\sf Comm}_{\sf tot}$	&0.87	&1.11	&0.74	&2.72	&4.16	&2.53	\\
            
            & ${\sf Cost}$	&642.07	&198	&141.1	&21.11	&4.38	&6.69	\\
            
            & $\TP$ 	&14.03	&41.78	&40.6	&734.62	&2081.44	&1248.86	\\
            
            \midrule 
            \Block{6-1}{NN-4} 
            & ${\sf PT}_{\sf on}$	&134.63	&49.72	&66.54	&34.51	&7.45	&12.29	\\
            
            & ${\sf PT}_{\sf tot}$	&4753.2	&133.31	&150.12	&39.09	&7.59	&12.43	\\
            
            & ${\sf CT}_{\sf tot}$	&14201.97	&269.18	&300.25	&90.91	&15.18	&24.87	\\
            
            & ${\sf Comm}_{\sf tot}$	&18.23	&15.57	&12.27	&42.4	&61.53	&38.1	\\
            
            & ${\sf Cost}$	&17134.85	&2718	&2215.6	&88.41	&22.3	&26.9	\\
            
            & $\TP$ 	&0.79	&1.96	&1.92	&222.52	&623.45	&377.87	\\
			\bottomrule
		\end{NiceTabular}
	}
	\caption{Benchmarking of Neural Networks.\label{tab:nn3pcS}}
\end{table}

\begin{figure}[htb!]
	\centering
	\begin{subfigure}{.32\textwidth}
		\centering
		\resizebox{.95\textwidth}{!}{
			\begin{tikzpicture}[
				every axis/.style={ 
					ybar stacked,
					ymin=0,ymax=10,
					xtick={1,2,3,4}, xticklabels={NN-1,NN-2,NN-3,NN-4},
					enlarge x limits=0.2,
					cycle list name=exotic, 
					every axis plot/.append style={fill,draw=none,no markers},
					legend style = {anchor = south, legend columns = -1, draw=none, area legend},
					bar width=10pt},]
				
				\begin{axis}[bar shift=-12pt,hide axis, legend style = {at={(0.2, 0.825)}}]
					\addplot+ coordinates
					{(1,2.50) (2,2.53) (3,4.21) (4,7.07)}; 	
					\addlegendentry{{\footnotesize ABY3~~~~}}
				\end{axis}
				
				\begin{axis}[hide axis, legend style = {at={(0.2, 0.75)}}]
					\addplot+[fill=UniOrange] coordinates
					{(1,0.63) (2,0.71) (3,2.43) (4,5.63)}; 	
					\addlegendentry{{\footnotesize \TSthisT}}
				\end{axis}
				
				\begin{axis}[bar shift=12pt, legend style = {at={(0.2, 0.675)}}]
					\addplot+[fill=UniGruen] coordinates
					{(1,1.11) (2,1.17) (3,3.02) (4,6.05)};  
					\addlegendentry{\footnotesize \TSthisC}
				\end{axis}

			\end{tikzpicture}
		}
		\vspace{-1mm}
		\caption{\footnotesize Training: ${\sf PT}_{\sf on}$}\label{fig:TrainA3pcS}
	\end{subfigure}
	\begin{subfigure}{.32\textwidth}
		\centering
		\resizebox{.95\textwidth}{!}{
			\begin{tikzpicture}[
				every axis/.style={ 
					ybar stacked,
					ymin=0,ymax=16,
					xtick={1,2,3,4}, xticklabels={NN-1,NN-2,NN-3,NN-4},
					enlarge x limits=0.2,
					cycle list name=exotic, 
					every axis plot/.append style={fill,draw=none,no markers},
					legend style = {anchor = south, legend columns = -1, draw=none, area legend},
					bar width=10pt},]
				
				\begin{axis}[bar shift=-12pt,hide axis, legend style = {at={(0.2, 0.825)}}]
					\addplot+ coordinates
					{(1,5.59) (2,6.85) (3,9.32) (4,14.06)}; 	
					\addlegendentry{{\footnotesize ABY3~~~~}}
				\end{axis}
				
				\begin{axis}[hide axis, legend style = {at={(0.2, 0.75)}}]
					\addplot+[fill=UniOrange] coordinates
					{(1,5.77) (2,5.99) (3,7.62) (4,11.40)}; 	
					\addlegendentry{{\footnotesize \TSthisT}}
				\end{axis}
				
				\begin{axis}[bar shift=12pt, legend style = {at={(0.2, 0.675)}}]
					\addplot+[fill=UniGruen] coordinates
					{(1,4.94) (2,5.29) (3,7.14) (4,11.11)};  
					\addlegendentry{\footnotesize \TSthisC}
				\end{axis}

			\end{tikzpicture}
		}
		\vspace{-1mm}
		\caption{\footnotesize Training: ${\sf Cost}$}\label{fig:TrainB3pcS}
	\end{subfigure}
	\begin{subfigure}{.32\textwidth}
		\centering
		\resizebox{.95\textwidth}{!}{
			\begin{tikzpicture}[
				every axis/.style={ 
					ybar stacked,
					ymin=0,ymax=15,
					xtick={1,2,3}, xticklabels={SVM,NN-3,NN-4},
					enlarge x limits=0.28,
					cycle list name=exotic, 
					every axis plot/.append style={fill,draw=none,no markers},
					legend style = {anchor = south, legend columns = -1, draw=none, area legend},
					bar width=10pt},]
				
				\begin{axis}[bar shift=-13pt,hide axis, legend style = {at={(0.82, 0.85)}}]
					\addplot+ coordinates
					{(1,9.26) (2,9.52) (3,7.79)}; 	
					\addlegendentry{{\footnotesize ABY3~~~~}}
				\end{axis}
				
				\begin{axis}[hide axis, legend style = {at={(0.82, 0.775)}}]
					\addplot+[fill=UniOrange] coordinates
					{(1,10.83) (2,11.02) (3,9.28)}; 	
					\addlegendentry{{\footnotesize \TSthisT}}
				\end{axis}
				
				\begin{axis}[bar shift=13pt, legend style = {at={(0.82, 0.7)}}]
					\addplot+[fill=UniGruen] coordinates
					{(1,10.04) (2,10.28) (3,8.56)};  
					\addlegendentry{\footnotesize \TSthisC}
				\end{axis}

			\end{tikzpicture}
		}
		\vspace{-1mm}
		\caption{\footnotesize Inference: $\TP$}\label{fig:Inf3pcS}
	\end{subfigure}
	\caption{Analysis of protocols in terms of ${\sf PT}_{\sf on}$, ${\sf Cost}$ and $\TP$. All the values are reported in the $\log_2()$ scale.}\label{fig:MLPlot3pcS}
\end{figure}
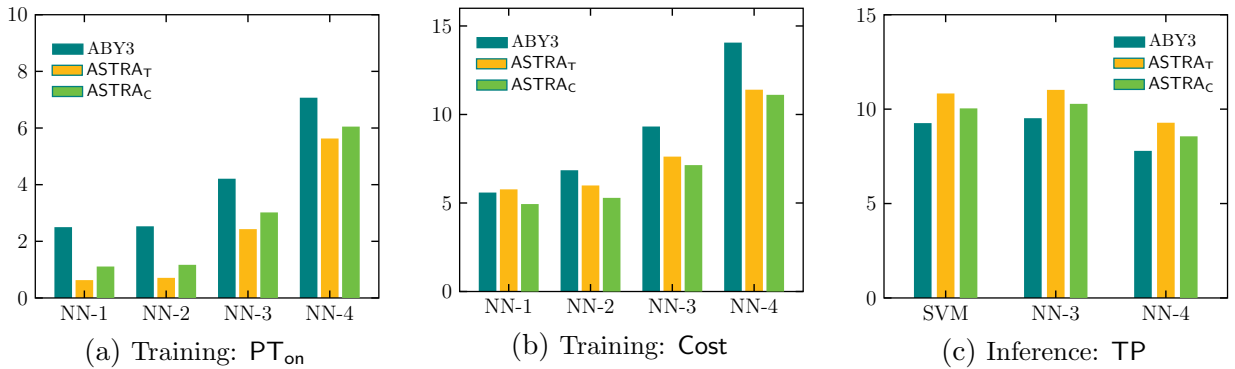

\section{ML Inference}
\label{sec:bench_inf_3pcS}
A similar trend for linear and logistic regression inference is observed for training, where both $\TSthisT$ and $\TSthisC$ outperform ABY3. The exception concerns the slightly higher communication of $\TSthisT$ compared to ABY3 due to the higher communication cost required for multi-input multiplication gates. This difference, however, vanishes for larger circuits, as will be evident from \tabref{nn3pcS}.  
For neural networks, the time-optimized variant $\TSthisT$ is faster when it comes to online run time (${\sf PT}_{\sf on}$), by $4.4\times$ over ABY3. This is also reflected in the $\TP$, where the improvement is up to $2.8 \times$, as evident from \figref{Inf3pcS}. For inference, the communication is in the order of a few megabytes, while run time is in the order of a few seconds. The key observation is that communication is well suited for the bandwidth used~(40 MBps). So unlike training, the monetary cost in inference depends more on run time rather than on communication. This is evident from \tabref{nn3pcS} which shows that $\TSthisT$ saves on monetary cost up to a factor of $4$ over ABY3. A similar trend is observed in the case of Support Vector Machines. 

\begin{table}[htb!]
	\centering
	\resizebox{0.55\textwidth}{!}{
		\begin{NiceTabular}{r r | r r r}[notes/para]
			\toprule
			\Block{2-1}{Algorithm} & \Block{2-1}{Parameter\tabularnote{Time~(in seconds) and communication~(in KB) are reported.}}
			& \Block{1-3}{Inference\tabularnote{ Cost is reported for $1000$ queries.} } & &
			\\ \cmidrule{3-5}
			& & ABY3 & \TSthisT & \TSthisC \\
			\midrule 
            \Block{6-1}{Support Vector\\Machines} 
            & ${\sf PT}_{\sf on}$	&12.45	&2.53	&4.39	\\
            
            & ${\sf PT}_{\sf tot}$	&12.45	&2.54	&4.39	\\
            
            & ${\sf CT}_{\sf tot}$	&27.86	&5.07	&8.78	\\
            
            & ${\sf Comm}_{\sf tot}$	&604.93	&1161.63	&666.46	\\
            
            & ${\sf Cost}$	&23.71	&4.43	&7.45	\\
            
            & $\TP$ 	&616.73	&1827.61	&1055.38	\\
			\bottomrule
		\end{NiceTabular}
	}
	\caption{Benchmarking of the inference phase of Support Vector Machines.\label{tab:svm3pcS}}
\end{table}

Note that the cost-optimized variant underperforms in terms of monetary cost compared to $\TSthisT$. This is because run time plays a more significant role in monetary cost than communication. Hence for inference, the time-optimized variant becomes the optimal choice.

\section{Additional Benchmarking}

\subsection{Varying batch sizes and feature sizes}
\tabref{nn13pcS} shows the online throughput~($\TP$) of neural network~(NN-1) training over varying batch sizes and feature sizes using synthetic datasets. 

\begin{table}[htb!]
	\centering
	\resizebox{0.48\textwidth}{!}{
		\begin{NiceTabular}{r r r r r}
			\toprule
			Batch Size & Features & ABY3 & \TSthisT & \TSthisC \\
			\midrule 
			\Block{3-1}{128} 
			&10	&1314.59	&2997.27	&2140.91	\\
			&100	&1314.59	&2997.27	&2140.91	\\
			&1000	&1104.37	&2625.67	&2139.75	\\
			\midrule 
			\Block{3-1}{256} 
			&10	&725.66	&2113.79	&2058.65	\\
			&100	&716.15	&2060.63	&2008.18	\\
			&1000	&633.13	&1646.47	&1612.81	\\
			\bottomrule
		\end{NiceTabular}
	}
	\caption{Online throughput~($\TP$) of NN-1 training~(iterations per minute) over various batch sizes and features.\label{tab:nn13pcS}}
\end{table}

We find that both $\TSthisT, \TSthisC$ are up to $2.9 \times$ higher in $\TP$. However, as the batch size and feature size increase, ABY3 and $\TSthis$ experience a bandwidth bottleneck. 

\subsection{Comparison operations}
\tabref{compb3pcS} compares the performance of the frameworks for circuits of varying depth. At each layer of the circuits, we perform 128 comparisons where the comparison results are generated in arithmetic shared form. The idea is that each layer emulates a comparison layer in an NN with a batch size of 128. 

	\begin{table}[htb!]
		\centering
		\resizebox{0.46\textwidth}{!}{
			\begin{NiceTabular}{r r r r r}
				\toprule
				Depth & Parameter & ABY3 & \TSthisT & \TSthisC\\
				\midrule 
				\Block{3-1}{128} 
				& ${\sf PT}_{\sf on}$	&2.62	&0.53	&0.93	\\
				
				& ${\sf CT}_{\sf tot}$	&5.87	&1.06	&1.85	\\
				
				& ${\sf Cost}$	&0.3	&0.05	&0.09	\\
				
				\midrule 
				\Block{3-1}{1024} 
				& ${\sf PT}_{\sf on}$	&20.99	&4.23	&7.41	\\
				
				& ${\sf CT}_{\sf tot}$	&46.99	&8.47	&14.82	\\
				
				& ${\sf Cost}$	&2.38	&0.43	&0.75	\\
				
				\midrule 
				\Block{3-1}{8192} 
				& ${\sf PT}_{\sf on}$	&167.93	&33.87	&59.27	\\
				
				& ${\sf CT}_{\sf tot}$	&375.89	&67.74	&118.54	\\
				
				& ${\sf Cost}$	&19.06	&3.45	&6.02	\\
				\bottomrule
			\end{NiceTabular}
		}
		\caption{Benchmarking of comparisons over various depths. Each of the layer has 128 comparisons. Time is reported in minutes, and monetary cost in USD.\label{tab:compb3pcS}}
	\end{table}

To summarise the experimental results, beyond a depth of roughly 100, the time-optimized variant~($\TSthisT$) starts outperforming in every metric, especially monetary cost, over the cost-optimized one~($\TSthisC$). This is because as the depth increases, runtime~({\sf CT}) grows at a much higher rate than the total communication. What we can infer from \tabref{compb3pcS} is that if one were to use a DNN with a depth of over 100, $\TSthisT$ becomes the optimal choice. 

\chapter{$\Tthis$: 3PC Fair and Robust Applications}
\label{chap:layer3_3pcmal}
$\TthisT$ uses multi-input multiplication gates and the 2 GC variant of the garbled world and is the fastest variant of the framework. On the other hand, $\TthisC$ is the variant with a minimal monetary cost. We report only the numbers for the fair variant of $\Tthis$ and not the robust variant since the overhead of robust over its fair counterpart is very minimal for the algorithms considered in this thesis.

\section{ML Training}
\label{sec:bench_train_3pcM}
We begin with analyzing the benchmarks for linear and logistic regression. The improvements observed in the three-party semi-honest case carry forward to $\Tthis$ as well. Both $\TthisT$ and $\TthisC$ showcase an improvement over ABY3 in terms of communication and run time. This also improves the monetary cost over ABY3, where the saving is up to $70\%$. One of the primary reasons for the improvement is the reduction in communication. This is attributed to an improved dot product protocol whose communication cost is independent of the vector dimension and a method for truncation which does not incur any overhead in the online phase. Moreover, our multiplication protocol has around $3.5\times$ improvement in terms of communication over ABY3.

\begin{table}[htb!]
	\centering
	\resizebox{0.8\textwidth}{!}{
		\begin{NiceTabular}{r r | r r r | r r r}[notes/para]
			\toprule
			\Block{2-1}{Algorithm} & \Block{2-1}{Parameter\tabularnote{Time~(in seconds) and communication~(in KB) are reported.}}
			& \Block{1-3}{Training\tabularnote{For training, batch size is 128 and the monetary cost~(USD) is reported for $1000$ iterations.}} & & 
			& \Block{1-3}{Inference\tabularnote{ For inference, cost is reported for $1000$ queries.} } & &
			\\ \cmidrule{3-8}
			& & ABY3 & \TthisT & \TthisC & ABY3 & \TthisT & \TthisC \\
			\midrule 
            \Block{6-1}{Linear\\Regression} 
            & ${\sf PT}_{\sf on}$	&1.09	&0.94	&0.94	&0.97	&0.88	&0.88	\\
            
            & ${\sf PT}_{\sf tot}$	&1.16	&1.61	&1.61	&1.43	&1.54	&1.54	\\
            
            & ${\sf CT}_{\sf tot}$	&1.69	&3.57	&3.57	&0.66	&3.41	&3.41	\\
            
            & ${\sf Comm}_{\sf tot}$	&33225.81	&97.91	&97.91	&128.94	&0.25	&0.25	\\
            
            & ${\sf Cost}$	&6.5	&3.03	&3.03	&0.58	&2.88	&2.88	\\
            
            & $\TP$ 	&521.32	&6111.54	&6111.54	&17497.49	&7404.84	&7404.84	\\
            
            \midrule 
            \Block{6-1}{Logistic\\Regression} 
            & ${\sf PT}_{\sf on}$	&2.64	&1.25	&1.44	&2.41	&1.18	&1.36	\\
            
            & ${\sf PT}_{\sf tot}$	&2.76	&1.92	&2.11	&2.49	&1.84	&2.02	\\
            
            & ${\sf CT}_{\sf tot}$	&8.28	&4.2	&4.57	&7.24	&4.01	&4.38	\\
            
            & ${\sf Comm}_{\sf tot}$	&33494.5	&204.13	&154.53	&131.04	&1.08	&0.69	\\
            
            & ${\sf Cost}$	&12.1	&3.58	&3.88	&6.14	&3.39	&3.69	\\
            
            & $\TP$ 	&517.39	&3262.62	&2549.53	&1590.75	&3598.18	&2749.97	\\
			\bottomrule
		\end{NiceTabular}
	}
	\caption{Benchmarking of Linear Regression and Logistic Regression algorithms.\label{tab:reg3pcM}}
\end{table}

The improvements are more evident in the case of neural networks. Here, $\TthisT$ is up to two orders of magnitude faster than ABY3 in the online phase. The same trend holds true for communication costs.
Like $\TSthis$, the cost-optimized variant, $\TthisC$, saves $15\%$ in monetary cost over $\TthisT$, while incurring the overhead of $1.2\times$ in the online run time. 

These trends can be better captured with a pictorial representation as given in \figref{MLPlot3pcM}.

\begin{figure}[htb!]
	\centering
	\begin{subfigure}{.32\textwidth}
		\centering
		\resizebox{.95\textwidth}{!}{
			\begin{tikzpicture}[
				every axis/.style={ 
					ybar stacked,
					ymin=0,ymax=15,
					xtick={1,2,3,4}, xticklabels={NN-1,NN-2,NN-3,NN-4},
					enlarge x limits=0.2,
					cycle list name=exotic, 
					every axis plot/.append style={fill,draw=none,no markers},
					legend style = {anchor = south, legend columns = -1, draw=none, area legend},
					bar width=10pt},]
				
				\begin{axis}[bar shift=-12pt,hide axis, legend style = {at={(0.2, 0.825)}}]
					\addplot+ coordinates
					{(1,3.621) (2,4.844) (3,6.63) (4,13.32)}; 	
					\addlegendentry{{\footnotesize ABY3~~~~}}
				\end{axis}
				
				\begin{axis}[hide axis, legend style = {at={(0.2, 0.75)}}]
					\addplot+[fill=UniOrange] coordinates
					{(1,1.356) (2,1.416) (3,2.935) (4,6.30)}; 	
					\addlegendentry{{\footnotesize \TthisT}}
				\end{axis}
				
				\begin{axis}[bar shift=12pt, legend style = {at={(0.2, 0.675)}}]
					\addplot+[fill=UniGruen] coordinates
					{(1,1.669) (2,1.722) (3,3.375) (4,6.585)};  
					\addlegendentry{\footnotesize \TthisC}
				\end{axis}

			\end{tikzpicture}
		}
		\vspace{-1mm}
		\caption{\footnotesize Training: ${\sf PT}_{\sf on}$}\label{fig:TrainA3pcM}
	\end{subfigure}
	\begin{subfigure}{.32\textwidth}
		\centering
		\resizebox{.95\textwidth}{!}{
			\begin{tikzpicture}[
				every axis/.style={ 
					ybar stacked,
					ymin=0,ymax=25,
					xtick={1,2,3,4}, xticklabels={NN-1,NN-2,NN-3,NN-4},
					enlarge x limits=0.2,
					cycle list name=exotic, 
					every axis plot/.append style={fill,draw=none,no markers},
					legend style = {anchor = south, legend columns = -1, draw=none, area legend},
					bar width=10pt},]
				
				\begin{axis}[bar shift=-12pt,hide axis, legend style = {at={(0.2, 0.825)}}]
					\addplot+ coordinates
					{(1,9.029) (2,12.08) (3,14.10) (4,21.10)}; 	
					\addlegendentry{{\footnotesize ABY3~~~~}}
				\end{axis}
				
				\begin{axis}[hide axis, legend style = {at={(0.2, 0.75)}}]
					\addplot+[fill=UniOrange] coordinates
					{(1,4.27) (2,6.451) (3,9.960) (4,14.33)}; 	
					\addlegendentry{{\footnotesize \TthisT}}
				\end{axis}
				
				\begin{axis}[bar shift=12pt, legend style = {at={(0.2, 0.675)}}]
					\addplot+[fill=UniGruen] coordinates
					{(1,4.19) (2,6.33) (3,9.59) (4,14.10)};  
					\addlegendentry{\footnotesize \TthisC}
				\end{axis}

			\end{tikzpicture}
		}
		\vspace{-1mm}
		\caption{\footnotesize Training: ${\sf Cost}$}\label{fig:TrainB3pcM}
	\end{subfigure}
	\begin{subfigure}{.32\textwidth}
		\centering
		\resizebox{.95\textwidth}{!}{
			\begin{tikzpicture}[
				every axis/.style={ 
					ybar stacked,
					ymin=0,ymax=12,
					xtick={1,2,3}, xticklabels={SVM,NN-3,NN-4},
					enlarge x limits=0.28,
					cycle list name=exotic, 
					every axis plot/.append style={fill,draw=none,no markers},
					legend style = {anchor = south, legend columns = -1, draw=none, area legend},
					bar width=10pt},]
				
				\begin{axis}[bar shift=-13pt,hide axis, legend style = {at={(0.82, 0.85)}}]
					\addplot+ coordinates
					{(1,7.343) (2,5.39) (3,0)}; 	
					\addlegendentry{{\footnotesize ABY3~~~~}}
				\end{axis}
				
				\begin{axis}[hide axis, legend style = {at={(0.82, 0.775)}}]
					\addplot+[fill=UniOrange] coordinates
					{(1,9.3) (2,9.52) (3,7.87)}; 	
					\addlegendentry{{\footnotesize \TthisT}}
				\end{axis}
				
				\begin{axis}[bar shift=13pt, legend style = {at={(0.82, 0.7)}}]
					\addplot+[fill=UniGruen] coordinates
					{(1,8.69) (2,8.93) (3,7.26)};  
					\addlegendentry{\footnotesize \TthisC}
				\end{axis}

			\end{tikzpicture}
		}
		\vspace{-1mm}
		\caption{\footnotesize Inference: $\TP$}\label{fig:Inf3pcM}
	\end{subfigure}
	\caption{Analysis of protocols in terms of ${\sf PT}_{\sf on}$, ${\sf Cost}$ and $\TP$. All the values are reported in the $\log_2()$ scale.}\label{fig:MLPlot3pcM}
\end{figure}
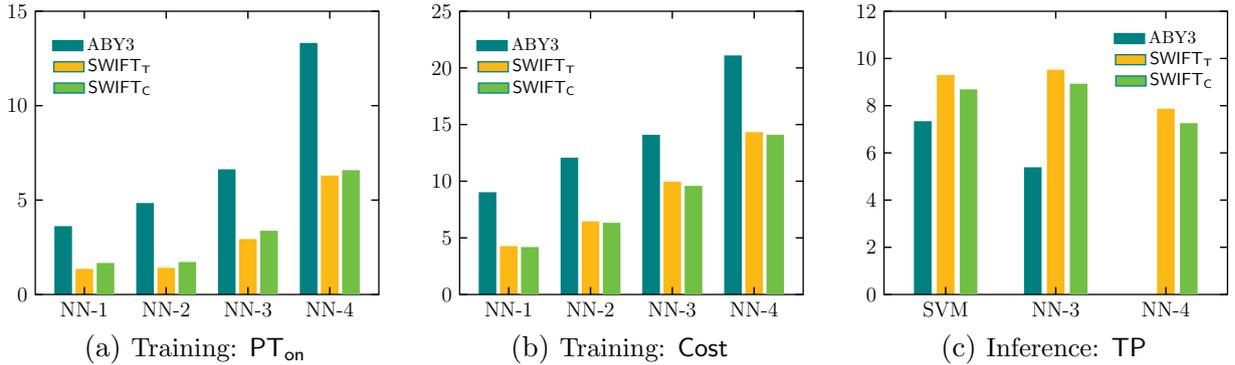

\begin{table}[htb!]
	\centering
	\resizebox{0.8\textwidth}{!}{
		\begin{NiceTabular}{r r | r r r | r r r}[notes/para]
			\toprule
			\Block{2-1}{Algorithm} & \Block{2-1}{Parameter\tabularnote{Time is reported in seconds}}
			& \Block{1-3}{Training\tabularnote{For training, communication is reported in GB. Monetary cost~(USD) is reported for $1000$ iterations and batch size is 128.}} & & 
			& \Block{1-3}{Inference\tabularnote{ For inference, communication is reported in MB and the cost is reported for $1000$ queries.} } & & 
			\\ \cmidrule{3-8}
			& & ABY3 & \TthisT & \TthisC & ABY3 & \TthisT & \TthisC \\
			\midrule 
            \Block{6-1}{NN-1} 
            & ${\sf PT}_{\sf on}$	&12.31	&2.56	&3.18	&7.24	&1.93	&2.49	\\
            
            & ${\sf PT}_{\sf tot}$	&28.89	&5.91	&6.53	&7.31	&2.6	&3.16	\\
            
            & ${\sf CT}_{\sf tot}$	&80.89	&14.82	&16.06	&21.92	&5.54	&6.66	\\
            
            & ${\sf Comm}_{\sf tot}$	&2.98	&0.31	&0.17	&11.85	&0.14	&0.09	\\
            
            & ${\sf Cost}$	&522.74	&19.31	&18.3	&21.44	&4.71	&5.64	\\
            
            & $\TP$ 	&5.86	&1102.34	&838.28	&530.13	&1610.74	&1138.94	\\
            
            \midrule 
            \Block{6-1}{NN-2} 
            & ${\sf PT}_{\sf on}$	&28.74	&2.67	&3.3	&7.27	&1.93	&2.49	\\
            
            & ${\sf PT}_{\sf tot}$	&146.24	&20.92	&21.55	&7.82	&2.62	&3.17	\\
            
            & ${\sf CT}_{\sf tot}$	&432.93	&59.5	&60.76	&23.47	&5.6	&6.71	\\
            
            & ${\sf Comm}_{\sf tot}$	&25.02	&0.49	&0.32	&178.09	&0.55	&0.35	\\
            
            & ${\sf Cost}$	&4341.79	&87.52	&80.22	&47.63	&4.85	&5.73	\\
            
            & $\TP$ 	&0.67	&304.35	&300.84	&94.52	&1610.74	&1138.94	\\
            
            \midrule 
            \Block{6-1}{NN-3} 
            & ${\sf PT}_{\sf on}$	&99.1	&7.65	&10.38	&18.51	&3.55	&5.03	\\
            
            & ${\sf PT}_{\sf tot}$	&643.69	&174.87	&177.6	&20.08	&4.65	&6.14	\\
            
            & ${\sf CT}_{\sf tot}$	&1925.28	&514.31	&519.77	&60.25	&10.07	&13.05	\\
            
            & ${\sf Comm}_{\sf tot}$	&100.52	&2.9	&2.05	&398.96	&12.34	&7.37	\\
            
            & ${\sf Cost}$	&17607.31	&996.09	&771.68	&112.93	&11.41	&12.36	\\
            
            & $\TP$ 	&0.17	&23.81	&23.23	&42.71	&733.25	&487.9	\\
            
            \midrule 
            \Block{6-1}{NN-4} 
            & ${\sf PT}_{\sf on}$	&10222	&79.22	&96.03	&96.52	&9.9	&14.75	\\
            
            & ${\sf PT}_{\sf tot}$	&58428.87	&4195.38	&4212.19	&264.58	&17.66	&22.5	\\
            
            & ${\sf CT}_{\sf tot}$	&175280.84	&12456.7	&12490.33	&793.73	&42.65	&52.34	\\
            
            & ${\sf Comm}_{\sf tot}$	&13225.44	&51.67	&41.82	&54335.94	&184.21	&112.44	\\
            
            & ${\sf Cost}$	&2262986.61	&20697.42	&17586.31	&9156.18	&78.92	&64.44	\\
            
            & $\TP$ 	&0	&1.18	&1.16	&0.31	&233.45	&153.53	\\
			\bottomrule
		\end{NiceTabular}
	}
	\caption{Benchmarking of Neural Networks.\label{tab:nn3pcM}}
\end{table}

\section{ML Inference}
\label{sec:bench_inf_3pcM}
For linear regression, logistic regression and support vector machines, we observe a similar trend as in the inference of $\TSthis$, where $\TthisT$ outperforms ABY3 and $\TthisC$ in terms of run time and monetary cost.
For neural networks, the time-optimized variant $\TthisT$ is faster when it comes to online run time (${\sf PT}_{\sf on}$), by $9.7\times$ over ABY3. This is also reflected in the $\TP$, where the improvement is up to $753 \times$, as evident from \tabref{nn3pcM}. Unlike the inference of $\TSthis$, the cost-optimized variant $\TthisC$ outperforms the rest (ABY3 and $\TthisT$) in terms of the monetary cost and communication as the network becomes deeper. The trends in throughput are captured in \figref{Inf3pcM}.

\begin{table}[htb!]
	\centering
	\resizebox{0.55\textwidth}{!}{
		\begin{NiceTabular}{r r | r r r}[notes/para]
			\toprule
			\Block{2-1}{Algorithm} & \Block{2-1}{Parameter\tabularnote{Time~(in seconds) and communication~(in KB) are reported.}}
			& \Block{1-3}{Inference\tabularnote{ Cost is reported for $1000$ queries.} } &  &
			\\ \cmidrule{3-5}
			& & ABY3 & \TthisT & \TthisC\\
			\midrule 
            \Block{6-1}{Support Vector\\Machines} 
            & ${\sf PT}_{\sf on}$	&22.17	&3.97	&5.83	\\
            
            & ${\sf PT}_{\sf tot}$	&22.18	&4.71	&6.57	\\
            
            & ${\sf CT}_{\sf tot}$	&66.54	&9.85	&13.56	\\
            
            & ${\sf Comm}_{\sf tot}$	&9497.64	&3339.88	&1822.93	\\
            
            & ${\sf Cost}$	&57.56	&9.12	&11.78	\\
            
            & $\TP$ 	&173	&639.46	&413.04	\\
			\bottomrule
		\end{NiceTabular}
	}
	\caption{Benchmarking of the inference phase of Support Vector Machines.\label{tab:svm3pcM}}
\end{table}

\section{Additional Benchmarking}

\subsection{Varying batch sizes and feature sizes}
\tabref{nn13pcM} shows the online throughput~($\TP$) of neural network~(NN-1) training over varying batch sizes and feature sizes using synthetic datasets. 

\begin{table}[htb!]
	\centering
	\resizebox{0.48\textwidth}{!}{
		\begin{NiceTabular}{r r r r r}
			\toprule
			Batch Size & Features & ABY3 & \TthisT & \TthisC \\
			\midrule 
			\Block{3-1}{128} 
			&10	&20.78	&1102.82	&838.56	\\
			&100	&16.04	&1102.82	&838.56	\\
			&1000	&4.88	&1102.09	&838.14	\\
			\midrule 
			\Block{3-1}{256} 
			&10	&10.41	&1102.56	&838.41	\\
			&100	&8.05	&1102.56	&838.41	\\
			&1000	&2.46	&981.09	&836.42	\\
			\bottomrule
		\end{NiceTabular}
	}
	\caption{Online throughput~($\TP$) of NN-1 training~(iterations per minute) over various batch sizes and features.\label{tab:nn13pcM}}
\end{table}

\subsection{Comparison operations}
\tabref{compb3pcM} compares the performance of the frameworks for circuits of varying depth. At each layer of the circuits, we perform 128 comparisons where the comparison results are generated in arithmetic shared form. The idea is that each layer emulates a comparison layer in an NN with a batch size of 128. 

	\begin{table}[htb!]
		\centering
		\resizebox{0.46\textwidth}{!}{
			\begin{NiceTabular}{r r r r r}
				\toprule
				Depth & Parameter & ABY3 & \TthisT & \TthisC\\
				\midrule 
				\Block{3-1}{128} 
				& ${\sf PT}_{\sf on}$	&4.21	&0.66	&1.06	\\
				
				& ${\sf CT}_{\sf tot}$	&12.64	&1.33	&2.12	\\
				
				& ${\sf Cost}$	&0.64	&0.07	&0.11	\\
				
					\midrule 
				\Block{3-1}{1024} 
				& ${\sf PT}_{\sf on}$	&33.71	&5.29	&8.47	\\
				
				& ${\sf CT}_{\sf tot}$	&101.13	&10.63	&16.98	\\
				
				& ${\sf Cost}$	&5.14	&0.55	&0.87	\\
				
					\midrule 
				\Block{3-1}{8192} 
				& ${\sf PT}_{\sf on}$	&269.67	&42.33	&67.73	\\
				
				& ${\sf CT}_{\sf tot}$	&809.07	&85.04	&135.84	\\
				
				& ${\sf Cost}$	&41.12	&4.41	&6.92	\\
				\bottomrule
			\end{NiceTabular}
		}
		\caption{Benchmarking of comparisons over various depths. Each of the layer has 128 comparisons. Time is reported in minutes, and monetary cost in USD.\label{tab:compb3pcM}}
	\end{table}

To summarize, $\Tthis$ improves over ABY3 up to two orders of magnitude in terms of monetary cost. As observed from the \tabref{nn3pcM}, $\TthisT$ provides the best online time while $\TthisC$ attains the best monetary cost, corroborating our claims. 

\chapter{$\Fthis$: 4PC Fair and Robust Applications}
\label{chap:layer3_4pc}
$\FthisT$ uses multi-input multiplication gates and the 2 GC variant of the garbled world and is the fastest variant of the framework. On the other hand, $\FthisC$ is the variant with a minimal monetary cost. We report only the numbers for the fair variant of $\Fthis$ and not the robust variant since the overhead of robust over its fair counterpart is very minimal for the algorithms considered in this thesis.

For training, we benchmark against the fair 4PC framework of Trident~\cite{NDSS:ChaRacSur20}. For inference, in addition to Trident, we also benchmark against the 4PC robust protocol of SWIFT~\cite{USENIX:KPPS21} since it supports NN inference. Note that the best case performance of Fantastic Four~\cite{EPRINT:DalEscKel20}, when cast in the preprocessing model, resembles that of SWIFT. In contrast, their worst-case execution (3PC malicious) is an order of magnitude slower (cf. \S\ref{pa:fantasticfour}), as demonstrated in their paper (cf. Table 2 of~\cite{EPRINT:DalEscKel20}).

\section{ML Training}
\label{sec:bench_train_4pcM}

Starting with the time-optimized variant, $\FthisT$ is $3 - 4\times$ faster than Trident in online runtime. 
The primary factor is the reduction in online rounds of our protocol due to multi-input gates. More precisely, we use the depth-optimized bit extraction circuit while instantiating the ReLU activation function using multi-input AND gates~(cf.~\secref{4pcBitExt}). Looking at the total communication~(${\sf Comm}_{\sf tot}$) in \tabref{nn4pcM}, we observe that the gap in ${\sf Comm}_{\sf tot}$ between $\FthisT$ vs. Trident decreases as the networks get deeper. This is justified as the improvement in communication of our dot product with truncation outpaces the overhead in communication caused by multi-input gates. The impact of this is more pronounced with NN-4, as observed by the lower monetary cost of $\FthisT$ over Trident.
Another reason is that there are two active parties ($P_1, P_2$) in our framework, whereas Trident has three. Given the allocation of servers, the best $\rtt$ Trident can get with three parties~$(P_0,P_1,P_2)$ is $153.74ms$, compared to $62.01ms$ of Tetrad, contributing to Tetrad being faster. However, if the $\rtt$ among all the parties were similar, this gap would be closed. Concretely, the online runtime (${\sf PT_{on}}$) of Trident will be similar to that of $\FthisC$.

\begin{table}[htb!]
	\centering
	\resizebox{0.9\textwidth}{!}{
		\begin{NiceTabular}{r r | r r r | r r r r}[notes/para]
			\toprule
			\Block{2-1}{Algorithm} & \Block{2-1}{Parameter\tabularnote{Time~(in seconds) and communication~(in KB) are reported.}}
			& \Block{1-3}{Training\tabularnote{For training, batch size is 128 and the monetary cost~(USD) is reported for $1000$ iterations.}} & & 
			& \Block{1-4}{Inference\tabularnote{ For inference, cost is reported for $1000$ queries.} } & &  &
			\\ \cmidrule{3-9}
			& & Trident & \FthisT & \FthisC & Trident & \FthisT & \FthisC & SWIFT\\
			\midrule 
            \Block{6-1}{Linear\\Regression} 
            & ${\sf PT}_{\sf on}$	&0.83	&0.5	&0.5	&0.44	&0.44	&0.44	&0.99	\\
            
            & ${\sf PT}_{\sf tot}$	&1.11	&0.78	&0.78	&0.71	&0.71	&0.71	&1.81	\\
            
            & ${\sf CT}_{\sf tot}$	&2.99	&2.15	&2.15	&2.02	&2.02	&2.02	&5.8	\\
            
            & ${\sf Comm}_{\sf tot}$	&76.5	&48.03	&48.03	&0.2	&0.2	&0.06	&0.21	\\
            
            & ${\sf Cost}$	&2.53	&1.83	&1.83	&1.71	&1.71	&1.71	&4.89	\\
            
            & $\TP$ 	&13971.76	&14780.03	&14780.03	&27944.03	&20094.71	&20094.71	&11688.96	\\
            
            \midrule 
            \Block{6-1}{Logistic\\Regression}
            & ${\sf PT}_{\sf on}$	&2.5	&0.75	&0.94	&2.1	&0.68	&0.86	&1.3	\\
            
            & ${\sf PT}_{\sf tot}$	&2.77	&1.03	&1.21	&2.38	&0.95	&1.13	&2.12	\\
            
            & ${\sf CT}_{\sf tot}$	&7.49	&2.65	&3.02	&6.52	&2.5	&2.86	&6.64	\\
            
            & ${\sf Comm}_{\sf tot}$	&119.16	&123.25	&86.75	&0.53	&0.78	&0.5	&0.54	\\
            
            & ${\sf Cost}$	&6.34	&2.26	&2.56	&5.5	&2.12	&2.42	&5.61	\\
            
            & $\TP$ 	&4299	&7182.26	&5183.72	&5080.81	&8241.66	&5713.88	&4743.86	\\
			\bottomrule
		\end{NiceTabular}
	}
	\caption{Benchmarking of Linear Regression and Logistic Regression algorithms.\label{tab:reg4pcM}}
\end{table}

\begin{figure}[htb!]
	\centering
	\begin{subfigure}{.32\textwidth}
		\centering
		\resizebox{.95\textwidth}{!}{
			\begin{tikzpicture}[
				every axis/.style={ 
					ybar stacked,
					ymin=0,ymax=8,
					xtick={1,2,3,4}, xticklabels={NN-1,NN-2,NN-3,NN-4},
					enlarge x limits=0.2,
					cycle list name=exotic, 
					every axis plot/.append style={fill,draw=none,no markers},
					legend style = {anchor = south, legend columns = -1, draw=none, area legend},
					bar width=10pt},]
				
				\begin{axis}[bar shift=-12pt,hide axis, legend style = {at={(0.2, 0.825)}}]
					\addplot+ coordinates
					{(1,3.01) (2,3.02) (3,4.46) (4,6.86)}; 	
					\addlegendentry{{\footnotesize Trident~}}
				\end{axis}
				
				\begin{axis}[hide axis, legend style = {at={(0.2, 0.75)}}]
					\addplot+[fill=UniOrange] coordinates
					{(1,0.94) (2,1.03) (3,2.66) (4,6.19)}; 	
					\addlegendentry{{\footnotesize \FthisT}}
				\end{axis}
				
				\begin{axis}[bar shift=12pt, legend style = {at={(0.2, 0.675)}}]
					\addplot+[fill=UniGruen] coordinates
					{(1,1.35) (2,1.41) (3,3.17) (4,6.49)};  
					\addlegendentry{\footnotesize \FthisC}
				\end{axis}

			\end{tikzpicture}
		}
		\vspace{-1mm}
		\caption{\footnotesize Training: ${\sf PT}_{\sf on}$}\label{fig:TrainA4pcM}
	\end{subfigure}
	\begin{subfigure}{.32\textwidth}
		\centering
		\resizebox{.95\textwidth}{!}{
			\begin{tikzpicture}[
				every axis/.style={ 
					ybar stacked,
					ymin=0,ymax=14,
					xtick={1,2,3,4}, xticklabels={NN-1,NN-2,NN-3,NN-4},
					enlarge x limits=0.2,
					cycle list name=exotic, 
					every axis plot/.append style={fill,draw=none,no markers},
					legend style = {anchor = south, legend columns = -1, draw=none, area legend},
					bar width=10pt},]
				
				\begin{axis}[bar shift=-12pt,hide axis, legend style = {at={(0.2, 0.825)}}]
					\addplot+ coordinates
					{(1,5.62) (2,6.14) (3,8.37) (4,12.52)}; 	
					\addlegendentry{{\footnotesize Trident~}}
				\end{axis}
				
				\begin{axis}[hide axis, legend style = {at={(0.2, 0.75)}}]
					\addplot+[fill=UniOrange] coordinates
					{(1,5.87) (2,6.24) (3,8.43) (4,12.35)}; 	
					\addlegendentry{{\footnotesize \FthisT}}
				\end{axis}
				
				\begin{axis}[bar shift=12pt, legend style = {at={(0.2, 0.675)}}]
					\addplot+[fill=UniGruen] coordinates
					{(1,5.09) (2,5.61) (3,7.91) (4,12.00)};  
					\addlegendentry{\footnotesize \FthisC}
				\end{axis}

			\end{tikzpicture}
		}
		\vspace{-1mm}
		\caption{\footnotesize Training: ${\sf Cost}$}\label{fig:TrainB4pcM}
	\end{subfigure}
	\begin{subfigure}{.32\textwidth}
		\centering
		\resizebox{.95\textwidth}{!}{
			\begin{tikzpicture}[
				every axis/.style={ 
					ybar stacked,
					ymin=0,ymax=12,
					xtick={1,2,3}, xticklabels={SVM,NN-3,NN-4},
					enlarge x limits=0.28,
					cycle list name=exotic, 
					every axis plot/.append style={fill,draw=none,no markers},
					legend style = {anchor = south, legend columns = -1, draw=none, area legend},
					bar width=10pt},]
				
				\begin{axis}[bar shift=-13pt,hide axis, legend style = {at={(0.82, 0.85)}}]
					\addplot+ coordinates
					{(1,9.24) (2,9.50) (3,7.79)}; 	
					\addlegendentry{{\footnotesize Trident~}}
				\end{axis}
				
				\begin{axis}[hide axis, legend style = {at={(0.82, 0.775)}}]
					\addplot+[fill=UniOrange] coordinates
					{(1,10.35) (2,10.53) (3,8.83)}; 	
					\addlegendentry{{\footnotesize \FthisT}}
				\end{axis}
				
				\begin{axis}[bar shift=13pt, legend style = {at={(0.82, 0.7)}}]
					\addplot+[fill=UniGruen] coordinates
					{(1,9.58) (2,9.82) (3,8.12)};  
					\addlegendentry{\footnotesize \FthisC}
				\end{axis}

			\end{tikzpicture}
		}
		\vspace{-1mm}
		\caption{\footnotesize Inference: $\TP$}\label{fig:Inf4pcM}
	\end{subfigure}
	\caption{Analysis of protocols in terms of ${\sf PT}_{\sf on}$, ${\sf Cost}$ and $\TP$. All the values are reported in the $\log_2()$ scale.}\label{fig:MLPlot4pcM}
\end{figure}
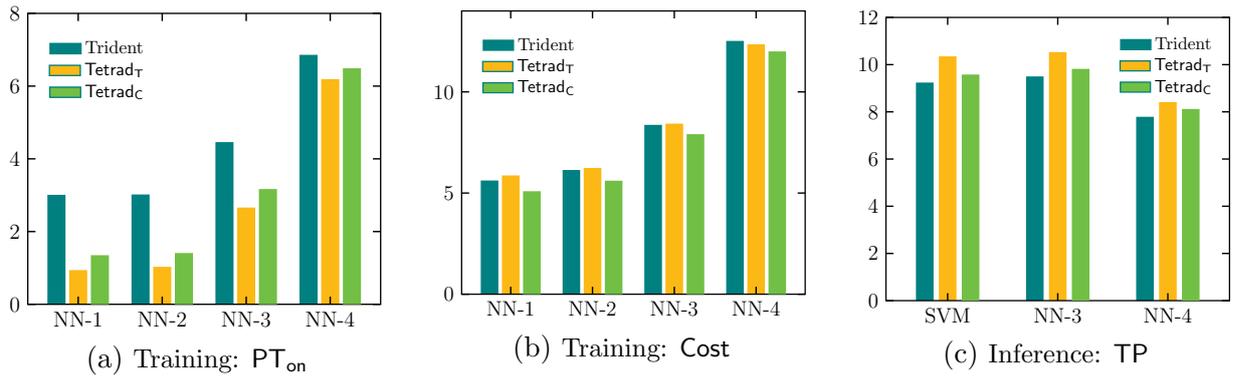

On the other hand, the cost-optimized variant $\FthisC$ is $1.5\times$ slower in the online phase than $\FthisT$. However, it is still faster than Trident owing to the $\rtt$ setup, as discussed above. When it comes to monetary cost, this variant is up to $20-40\%$ cheaper than its time-optimized counterpart and cheaper by around $30\%$ over Trident.

These trends can be better captured with a pictorial representation as given in \figref{MLPlot4pcM}.

\begin{table}[htb!]
	\centering
	\resizebox{0.9\textwidth}{!}{
		\begin{NiceTabular}{r r | r r r | r r r r}[notes/para]
			\toprule
			\Block{2-1}{Algorithm} & \Block{2-1}{Parameter\tabularnote{Time is reported in seconds}}
			& \Block{1-3}{Training\tabularnote{For training, communication is reported in GB. Monetary cost~(USD) is reported for $1000$ iterations and batch size is 128.}} & & 
			& \Block{1-4}{Inference\tabularnote{ For inference, communication is reported in MB and the cost is reported for $1000$ queries.} } & &  &
			\\ \cmidrule{3-9}
			& & Trident & \FthisT & \FthisC & Trident & \FthisT & \FthisC & SWIFT\\
			\midrule 
            \Block{6-1}{NN-1} 
            & ${\sf PT}_{\sf on}$	&8.06	&1.93	&2.55	&5.87	&1.31	&1.87	&2.31	\\
            
            & ${\sf PT}_{\sf tot}$	&10.76	&5.05	&5.27	&6.15	&1.58	&2.14	&3.13	\\
            
            & ${\sf CT}_{\sf tot}$	&27.9	&12.69	&11.22	&16.75	&3.76	&4.88	&8.65	\\
            
            & ${\sf Comm}_{\sf tot}$	&0.16	&0.3	&0.16	&0.06	&0.09	&0.05	&0.06	\\
            
            & ${\sf Cost}$	&49.33	&58.51	&34.29	&14.15	&3.19	&4.13	&7.32	\\
            
            & $\TP$ 	&118.75	&2083.68	&1517.79	&1802.8	&3330.33	&2167.73	&2011.68	\\
            
            \midrule 
            \Block{6-1}{NN-2} 
            & ${\sf PT}_{\sf on}$	&8.13	&2.05	&2.67	&5.87	&1.31	&1.87	&2.31	\\
            
            & ${\sf PT}_{\sf tot}$	&11.47	&5.79	&6.14	&6.15	&1.58	&2.14	&3.13	\\
            
            & ${\sf CT}_{\sf tot}$	&30.88	&14.86	&13.4	&16.75	&3.77	&4.88	&8.66	\\
            
            & ${\sf Comm}_{\sf tot}$	&0.28	&0.39	&0.24	&0.26	&0.37	&0.22	&0.25	\\
            
            & ${\sf Cost}$	&70.84	&75.67	&49.16	&14.19	&3.24	&4.16	&7.35	\\
            
            & $\TP$ 	&428.16	&652.75	&644.69	&1802.8	&3330.32	&2167.73	&2011.68	\\
            
            \midrule 
            \Block{6-1}{NN-3} 
            & ${\sf PT}_{\sf on}$	&22.04	&6.33	&9.06	&14.42	&2.61	&4.1	&4.54	\\
            
            & ${\sf PT}_{\sf tot}$	&30.91	&15.79	&18.53	&14.71	&2.91	&4.39	&5.39	\\
            
            & ${\sf CT}_{\sf tot}$	&92.37	&41.7	&44.45	&39.92	&6.43	&9.4	&13.18	\\
            
            & ${\sf Comm}_{\sf tot}$	&1.59	&1.94	&1.28	&5.62	&8.42	&4.76	&5.39	\\
            
            & ${\sf Cost}$	&331.76	&345.16	&241.83	&34.59	&6.74	&8.68	&11.97	\\
            
            & $\TP$ 	&53.62	&55.71	&54.13	&725.8	&1479.22	&904.6	&876.23	\\
            
            \midrule 
            \Block{6-1}{NN-4} 
            & ${\sf PT}_{\sf on}$	&116.32	&73.19	&90.01	&47.05	&7.85	&12.69	&13.13	\\
            
            & ${\sf PT}_{\sf tot}$	&328.2	&229.42	&246.23	&47.61	&8.44	&13.28	&14.33	\\
            
            & ${\sf CT}_{\sf tot}$	&983.74	&16866.48	&643.06	&129.41	&17.77	&27.46	&31.35	\\
            
            & ${\sf Comm}_{\sf tot}$	&31.59	&29.52	&22.24	&85.69	&124.09	&71.27	&81.33	\\
            
            & ${\sf Cost}$	&5884.81	&5240.81	&4101.26	&122.66	&34.32	&34.4	&39.18	\\
            
            & $\TP$ 	&2.54	&2.61	&2.56	&222.54	&458.25	&279.44	&276.67	\\
			\bottomrule
		\end{NiceTabular}
	}
	\caption{Benchmarking of Neural Networks.\label{tab:nn4pcM}}
\end{table}

\section{ML Inference}
\label{sec:bench_inf_4pcM}
Similar to training, the time-optimized variant for inference is faster when it comes to ${\sf PT}_{\sf on}$, by $4 - 6\times$ over Trident. This is also reflected in the $\TP$, where the improvement is about $2.8 - 5.5\times$, as evident from \figref{Inf4pcM}. In inference, the communication is in the order of megabytes, while run time is in the order of a few seconds. The key observation is that communication is well suited for the bandwidth used~(40 MBps). So unlike training, the monetary cost in inference depends more on run time rather than on communication. This is evident from \tabref{compb4pcM} which shows that $\FthisT$ saves on monetary cost up to a factor of $10$ over Trident.

\begin{table}[htb!]
	\centering
	\resizebox{0.6\textwidth}{!}{
		\begin{NiceTabular}{r r | r r r r}[notes/para]
			\toprule
			\Block{2-1}{Algorithm} & \Block{2-1}{Parameter\tabularnote{Time~(in seconds) and communication~(in KB) are reported.}}
			& \Block{1-4}{Inference\tabularnote{ Cost is reported for $1000$ queries.} } & &  &
			\\ \cmidrule{3-6}
			& & Trident & \FthisT & \FthisC & SWIFT\\
			\midrule 
            \Block{6-1}{Support Vector\\Machines} 
            & ${\sf PT}_{\sf on}$	&17.09	&2.91	&4.77	&5.21	\\
            
            & ${\sf PT}_{\sf tot}$	&17.37	&3.19	&5.05	&6.04	\\
            
            & ${\sf CT}_{\sf tot}$	&47.02	&6.99	&10.7	&14.47	\\
            
            & ${\sf Comm}_{\sf tot}$	&1395.72	&2391.47	&1275.01	&1395.59	\\
            
            & ${\sf Cost}$	&45.92	&6.26	&9.23	&12.43	\\
            
            & $\TP$ 	&607.47	&1306.34	&767.87	&747.34	\\
			\bottomrule
		\end{NiceTabular}
	}
	\caption{Benchmarking of the inference phase of Support Vector Machines.\label{tab:svm4pcM}}
\end{table}

Note that the cost-optimized variant underperforms in terms of monetary cost compared to $\FthisT$. This is because, as mentioned earlier, run time plays a more significant role in monetary cost than communication. Hence for inference, the time-optimized variant becomes the optimal choice.

\section{Additional Benchmarking}

\subsection{Varying batch sizes and feature sizes}
\tabref{nn14pcM} shows the online throughput~($\TP$) of neural network~(NN-1) training over varying batch sizes and feature sizes using synthetic datasets. 

\begin{table}[htb!]
	\centering
	\resizebox{0.48\textwidth}{!}{
		\begin{NiceTabular}{r r r r r}
			\toprule
			Batch Size & Features & Trident & \FthisT & \FthisC \\
			\midrule 
			\Block{3-1}{128} 
			&10	&1189.08	&2086.28	&1519.17	\\
			&100	&1189.08	&2086.28	&1519.17	\\
			&1000	&1188.75	&2083.68	&1517.79	\\
			\midrule 
			\Block{3-1}{256} 
			&10	&1189.08	&2084.19	&1518.06	\\
			&100	&1189.08	&2084.19	&1518.06	\\
			&1000	&1188.75	&2077.69	&1514.62	\\
			\bottomrule
		\end{NiceTabular}
	}
	\caption{Online throughput~($\TP$) of NN-1 training~(iterations per minute) over various batch sizes and features.\label{tab:nn14pcM}}
\end{table}

We find that both $\FthisT, \FthisC$ are up to $1.8 \times$ higher in $\TP$. However, as the batch size and feature size increase, Trident and $\Fthis$ experience a bandwidth bottleneck. The effect of the bandwidth limitation is higher for $\Fthis$; hence the gain in $\TP$ over Trident decreases a bit.

\subsection{Comparison operations}
\tabref{compb4pcM} compares the performance of the frameworks for circuits of varying depth. At each layer of the circuits, we perform 128 comparisons where the comparison results are generated in arithmetic shared form. The idea is that each layer emulates a comparison layer in an NN with a batch size of 128. 

	\begin{table}[htb!]
		\centering
		\resizebox{0.46\textwidth}{!}{
			\begin{NiceTabular}{r r r r r}
				\toprule
				Depth & Parameter & Trident & \FthisT & \FthisC\\
				\midrule 
				\Block{3-1}{128} 
				& ${\sf PT}_{\sf on}$	&3.55	&0.53	&0.93 	\\
				
				& ${\sf CT}_{\sf tot}$	&9.6	&1.06	&1.85	 \\
				
				& ${\sf Cost}$	&0.49	&0.05	&0.09 \\
				
				\midrule 
				\Block{3-1}{1024} 
				& ${\sf PT}_{\sf on}$	&28.42	&4.23	&7.41	\\
				
				& ${\sf CT}_{\sf tot}$	&76.79	&8.47	&14.82	\\
				
				& ${\sf Cost}$	&3.89	&0.43	&0.75	\\
				
				\midrule 
				\Block{3-1}{8192} 
				& ${\sf PT}_{\sf on}$	&227.34	&33.87	&59.27	\\
				
				& ${\sf CT}_{\sf tot}$	&614.3	&67.76	&118.56	\\
				
				& ${\sf Cost}$	&31.15	&3.48	&6.03	\\
				\bottomrule
			\end{NiceTabular}
		}
		\caption{Benchmarking of comparisons over various depths. Each of the layer has 128 comparisons. Time is reported in minutes, and monetary cost in USD.\label{tab:compb4pcM}}
	\end{table}

Interestingly, beyond a depth of roughly 100, the time-optimized variant~($\FthisT$) starts outperforming in every metric, especially monetary cost, over the cost-optimized one~($\FthisC$). This is because as the depth increases, runtime~({\sf CT}) grows at a much higher rate than the total communication. What we can infer from \tabref{compb4pcM} is that if one were to use a DNN with a depth of over 100, $\FthisT$ becomes the optimal choice. 

\chapter{$\TWthis$: 2PC Semi-honest Applications}
\label{chap:layer3_2pc}
$\TWthisT$ makes use of multi-input multiplication gates and is the fastest variant of the framework. On the other hand, $\TWthisC$ is the variant with a minimal monetary cost. We benchmark our protocols against the 2PC semi-honest framework of SecureML~\cite{SP:MohZha17}. The preprocessing phase of $\TWthis$ is similar to SecureML except for the use of multi-input multiplication in $\TWthisT$. The preprocessing can be performed either using oblivious transfer or via homomorphic encryption as discussed in Chapter \ref{chap:layer1_2pc}. Note that the benchmarking is performed only for the online phase. 

\section{ML Training}
\label{sec:bench_train_2pcS}
%
Starting with the time-optimized variant, $\TWthisT$ is up to two orders of magnitude faster than SecureML~\cite{SP:MohZha17} in run time as well as communication.
The reduction is primarily due to the following: (i) the improved dot product protocol whose online phase communication is independent of the dimension of the vector, and (ii) improvements in online rounds due to multi-input multiplication. 
These reductions in communication and run time directly impact the monetary cost of the system, where $\TWthisT$ brings in a saving of up to $342 \times$ over SecureML.  
On the other hand, the cost-optimized variant $\TWthisC$ is around $1.3\times$ slower than $\TWthisT$. However, it is still faster than SecureML due to the reasons discussed above. Further, this variant has a slightly higher communication than $\TWthisT$ due to the absence of multi-input multiplication. 

\begin{table}[htb!]
	\centering
	\resizebox{0.8\textwidth}{!}{
		\begin{NiceTabular}{r r | r r r | r r r}[notes/para]
			\toprule
			\Block{2-1}{Algorithm} & \Block{2-1}{Parameter\tabularnote{Time~(in seconds) and communication~(in KB) are reported.}}
			& \Block{1-3}{Training\tabularnote{For training, batch size is 128 and the monetary cost~(USD) is reported for $1000$ iterations.}} & & 
			& \Block{1-3}{Inference\tabularnote{ For inference, cost is reported for $1000$ queries.} } & & 
			\\ \cmidrule{3-8}
			& & SecureML & \TWthisT & \TWthisC & SecureML & \TWthisT & \TWthisC \\
			\midrule 
            \Block{5-1}{Linear\\Regression} 
            & ${\sf PT}_{\sf on}$	&0.14	&0.12	&0.12	&0.06	&0.06	&0.06	\\
            
            & ${\sf CT}_{\sf on}$	&0.28	&0.25	&0.25	&0.12	&0.12	&0.12	\\
            
            & ${\sf Comm}_{\sf on}$	&6272	&14.25	&14.25	&24.5	&0.02	&0.02	\\
            
            & ${\sf Cost}$	&1.2	&0.21	&0.21	&0.11	&0.1	&0.1	\\
            
            & $\TP$ 	&783.67	&30962.74	&30962.74	&61925.5	&63872.25	&63872.25	\\

            \midrule 
            \Block{5-1}{Logistic\\Regression} 
            & ${\sf PT}_{\sf on}$	&0.64	&0.37	&0.56	&0.55	&0.3	&0.48	\\
            
            & ${\sf CT}_{\sf on}$	&1.28	&0.74	&1.12	&1.11	&0.6	&0.96	\\
            
            & ${\sf Comm}_{\sf on}$	&6295.75	&23.44	&24.13	&24.68	&0.09	&0.09	\\
            
            & ${\sf Cost}$	&2.04	&0.63	&0.95	&0.94	&0.51	&0.81	\\
            
            & $\TP$ 	&779.73	&10320.91	&6880.61	&6927.53	&12774.45	&7984.05	\\
			\bottomrule
		\end{NiceTabular}
	}
	\caption{Benchmarking of Linear Regression and Logistic Regression algorithms.\label{tab:reg2pcS}}
\end{table}

These trends can be better captured with a pictorial representation as given in \figref{MLPlot2pcS}.

\begin{figure}[htb!]
	\centering
	\begin{subfigure}{.32\textwidth}
		\centering
		\resizebox{.95\textwidth}{!}{
			\begin{tikzpicture}[
				every axis/.style={ 
					ybar stacked,
					ymin=0,ymax=15,
					xtick={1,2,3,4}, xticklabels={NN-1,NN-2,NN-3,NN-4},
					enlarge x limits=0.2,
					cycle list name=exotic, 
					every axis plot/.append style={fill,draw=none,no markers},
					legend style = {anchor = south, legend columns = -1, draw=none, area legend},
					bar width=10pt},]
				
				\begin{axis}[bar shift=-12pt,hide axis, legend style = {at={(0.2, 0.825)}}]
					\addplot+ coordinates
					{(1,2.34) (2,4.73) (3,6.68) (4,13.69)}; 	
					\addlegendentry{{\footnotesize SecureML~~}}
				\end{axis}
				
				\begin{axis}[hide axis, legend style = {at={(0.2, 0.75)}}]
					\addplot+[fill=UniOrange] coordinates
					{(1,.687) (2,.76) (3,2.45) (4,5.63)}; 	
					\addlegendentry{{\footnotesize \TWthisT}}
				\end{axis}
				
				\begin{axis}[bar shift=12pt, legend style = {at={(0.2, 0.675)}}]
					\addplot+[fill=UniGruen] coordinates
					{(1,1.11) (2,1.17) (3,3.02) (4,6.05)};  
					\addlegendentry{\footnotesize \TWthisC}
				\end{axis}

			\end{tikzpicture}
		}
		\vspace{-1mm}
		\caption{\footnotesize Training: ${\sf PT}_{\sf on}$}\label{fig:TrainA2pcS}
	\end{subfigure}
	\begin{subfigure}{.32\textwidth}
		\centering
		\resizebox{.95\textwidth}{!}{
			\begin{tikzpicture}[
				every axis/.style={ 
					ybar stacked,
					ymin=0,ymax=20,
					xtick={1,2,3,4}, xticklabels={NN-1,NN-2,NN-3,NN-4},
					enlarge x limits=0.2,
					cycle list name=exotic, 
					every axis plot/.append style={fill,draw=none,no markers},
					legend style = {anchor = south, legend columns = -1, draw=none, area legend},
					bar width=10pt},]
				
				\begin{axis}[bar shift=-12pt,hide axis, legend style = {at={(0.2, 0.825)}}]
					\addplot+ coordinates
					{(1,6.54) (2,9.62) (3,11.62) (4,18.68)}; 	
					\addlegendentry{{\footnotesize SecureML~~}}
				\end{axis}
				
				\begin{axis}[hide axis, legend style = {at={(0.2, 0.75)}}]
					\addplot+[fill=UniOrange] coordinates
					{(1,1.80) (2,2.89) (3,5.97) (4,10.26)}; 	
					\addlegendentry{{\footnotesize \TWthisT}}
				\end{axis}
				
				\begin{axis}[bar shift=12pt, legend style = {at={(0.2, 0.675)}}]
					\addplot+[fill=UniGruen] coordinates
					{(1,2.68) (2,3.08) (3,6.11) (4,10.32)};  
					\addlegendentry{\footnotesize \TWthisC}
				\end{axis}

			\end{tikzpicture}
		}
		\vspace{-1mm}
		\caption{\footnotesize Training: ${\sf Cost}$}\label{fig:TrainB2pcS}
	\end{subfigure}
	\begin{subfigure}{.32\textwidth}
		\centering
		\resizebox{.95\textwidth}{!}{
			\begin{tikzpicture}[
				every axis/.style={ 
					ybar stacked,
					ymin=0,ymax=12,
					xtick={1,2,3}, xticklabels={SVM,NN-3,NN-4},
					enlarge x limits=0.28,
					cycle list name=exotic, 
					every axis plot/.append style={fill,draw=none,no markers},
					legend style = {anchor = south, legend columns = -1, draw=none, area legend},
					bar width=10pt},]
				
				\begin{axis}[bar shift=-13pt,hide axis, legend style = {at={(0.82, 0.85)}}]
					\addplot+ coordinates
					{(1,9.58) (2,6) (3,0)}; 	
					\addlegendentry{{\footnotesize SecureML~}}
				\end{axis}
				
				\begin{axis}[hide axis, legend style = {at={(0.82, 0.775)}}]
					\addplot+[fill=UniOrange] coordinates
					{(1,10.56) (2,10.74) (3,7.86)}; 	
					\addlegendentry{{\footnotesize \TWthisT}}
				\end{axis}
				
				\begin{axis}[bar shift=13pt, legend style = {at={(0.82, 0.7)}}]
					\addplot+[fill=UniGruen] coordinates
					{(1,9.77) (2,10.01) (3,7.79)};  
					\addlegendentry{\footnotesize \TWthisC}
				\end{axis}

			\end{tikzpicture}
		}
		\vspace{-1mm}
		\caption{\footnotesize Inference: $\TP$}\label{fig:Inf2pcS}
	\end{subfigure}
	\caption{Analysis of protocols in terms of ${\sf PT}_{\sf on}$, ${\sf Cost}$ and $\TP$. All the values are reported in the $\log_2()$ scale.}\label{fig:MLPlot2pcS}
\end{figure}

\begin{table}[htb!]
	\centering
	\resizebox{0.8\textwidth}{!}{
		\begin{NiceTabular}{r r | r r r | r r r}[notes/para]
			\toprule
			\Block{2-1}{Algorithm} & \Block{2-1}{Parameter\tabularnote{Time is reported in seconds and communication is reported in MB}}
			& \Block{1-3}{Training\tabularnote{For training, monetary cost~(USD) is reported for $1000$ iterations and batch size is 128.}} & & 
			& \Block{1-3}{Inference\tabularnote{ For inference, the cost is reported for $1000$ queries.} } & & 
			\\ \cmidrule{3-8}
			& & SecureML & \TWthisT & \TWthisC & SecureML & \TWthisT & \TWthisC \\
			\midrule 
            \Block{5-1}{NN-1} 
            & ${\sf PT}_{\sf on}$	&5.07	&1.61	&2.17	&1.68	&0.93	&1.49	\\
            
            & ${\sf CT}_{\sf on}$	&10.14	&3.23	&4.34	&3.37	&1.86	&2.98	\\
            
            & ${\sf Comm}_{\sf on}$	&540.22	&4.94	&5.03	&3.63	&0.02	&0.02	\\
            
            & ${\sf Cost}$	&93.62	&3.5	&6.45	&3.41	&1.57	&2.52	\\
            
            & $\TP$ 	&8.81	&947.79	&931.37	&1315.06	&4128.37	&2580.22	\\

            \midrule 
            \Block{5-1}{NN-2} 
            & ${\sf PT}_{\sf on}$	&26.7	&1.7	&2.26	&1.78	&0.93	&1.49	\\
            
            & ${\sf CT}_{\sf on}$	&53.41	&3.4	&4.53	&3.57	&1.86	&2.98	\\
            
            & ${\sf Comm}_{\sf on}$	&4752.91	&29.29	&29.66	&33.78	&0.06	&0.07	\\
            
            & ${\sf Cost}$	&790.4	&7.45	&8.46	&8.3	&1.58	&2.52	\\
            
            & $\TP$ 	&1.01	&163.17	&161.17	&141.82	&4128.37	&2580.22	\\

            \midrule 
            \Block{5-1}{NN-3} 
            & ${\sf PT}_{\sf on}$	&103.15	&5.48	&8.15	&4.46	&2.23	&3.72	\\
            
            & ${\sf CT}_{\sf on}$	&206.3	&10.97	&16.31	&8.92	&4.46	&7.44	\\
            
            & ${\sf Comm}_{\sf on}$	&18654	&344.54	&35.26	&73.16	&1.35	&1.42	\\
            
            & ${\sf Cost}$	&3157.52	&63.1	&69.18	&19.21	&3.98	&6.51	\\
            
            & $\TP$ 	&0.25	&13.93	&13.53	&64.2	&1720.15	&1032.09	\\

            \midrule 
            \Block{5-1}{NN-4} 
            & ${\sf PT}_{\sf on}$	&13254.9	&49.79	&66.54	&64.76	&7.45	&12.29	\\
            
            & ${\sf CT}_{\sf on}$	&26509.79	&99.57	&133.08	&129.51	&14.9	&24.59	\\
            
            & ${\sf Comm}_{\sf on}$	&2556821.6	&7364.1	&7501.93	&10304.95	&20.55	&21.54	\\
            
            & ${\sf Cost}$	&422846.24	&1234.73	&1284.55	&1723.13	&15.79	&24.13	\\
            
            & $\TP$ 	&0	&0.65	&0.64	&0.46	&233.58	&222.79	\\
			\bottomrule
		\end{NiceTabular}
	}
	\caption{Benchmarking of Neural Networks.\label{tab:nn2pcS}}
\end{table}

\section{ML Inference}
\label{sec:bench_inf_2pcS}
Like training, the time-optimized variant for inference is faster when it comes to the performance in the online phase. For shallow NNs such as NN-1, we observe a $4\times$ improvement in the online throughput over SecureML. However, the improvement changes drastically as the network grows bigger. Specifically, we observe a gain in throughput of up to $500\times$ over SecureML for the case of NN-4. The poor performance of SecureML is due to the huge increase in communication costs for deeper networks, which forms the bottleneck.

\begin{table}[htb!]
	\centering
	\resizebox{0.55\textwidth}{!}{
		\begin{NiceTabular}{r r | r r r}[notes/para]
			\toprule
			\Block{2-1}{Algorithm} & \Block{2-1}{Parameter\tabularnote{Time~(in seconds) and communication~(in KB) are reported.}}
			& \Block{1-3}{Inference\tabularnote{ Cost is reported for $1000$ queries.} } & &  
			\\ \cmidrule{3-5}
			& & SecureML & \TWthisT & \TWthisC \\
			\midrule 
            \Block{5-1}{Support Vector\\Machines} 
            & ${\sf PT}_{\sf on}$	&5.01	&2.53	&4.39	\\
            
            & ${\sf CT}_{\sf on}$	&10.02	&5.07	&8.78	\\
            
            & ${\sf Comm}_{\sf on}$	&1213.62	&341.46	&362.44	\\
            
            & ${\sf Cost}$	&8.72	&4.33	&7.47	\\
            
            & $\TP$ 	&766.82	&1514.88	&874.82	\\
			\bottomrule
		\end{NiceTabular}
	}
	\caption{Benchmarking of the inference phase of Support Vector Machines.\label{tab:svm2pcS}}
\end{table}

\section{Additional Benchmarking}

\subsection{Varying batch sizes and feature sizes}
\tabref{nn12pcS} shows the online throughput~($\TP$) of neural network~(NN-1) training over varying batch sizes and feature sizes using synthetic datasets. 

\begin{table}[htb!]
	\centering
	\resizebox{0.48\textwidth}{!}{
		\begin{NiceTabular}{r r r r r}
			\toprule
			Batch Size & Features & SecureML & \TWthisT & \TWthisC \\
			\midrule 
			\Block{3-1}{128} 
			&10	&31.02	&1351.09	&1317.96	\\
			&100	&23.99	&1287.39	&1257.28	\\
			&1000	&7.34	&874.91	&860.89	\\
			\midrule 
			\Block{3-1}{256}
			&10	&15.54	&704.19	&686.21	\\
			&100	&12.02	&686.49	&669.39	\\
			&1000	&3.68	&548.57	&537.6	\\
			\bottomrule
		\end{NiceTabular}
	}
	\caption{Online throughput~($\TP$) of NN-1 training~(iterations per minute) over various batch sizes and features.\label{tab:nn12pcS}}
\end{table}

\subsection{Comparison operations}
\tabref{compb2pcS} compares the performance of the frameworks for circuits of varying depth. At each layer of the circuits, we perform 128 comparisons where the comparison results are generated in arithmetic shared form. The idea is that each layer emulates a comparison layer in an NN with a batch size of 128. 

	\begin{table}[htb!]
		\centering
		\resizebox{0.46\textwidth}{!}{
			\begin{NiceTabular}{r r r r r}
				\toprule
				Depth & Parameter & SecureML & \TWthisT & \TWthisC\\
				\midrule 
				\Block{3-1}{128} 
				& ${\sf PT}_{\sf on}$	&0.93	&0.53	&0.93	\\
				
				& ${\sf CT}_{\sf on}$	&1.85	&1.06	&1.85	\\
				
				& ${\sf Cost}$	&0.09	&0.05	&0.09	\\
				
				\midrule 
				\Block{3-1}{1024} 
				& ${\sf PT}_{\sf on}$	&7.41	&4.23	&7.41	\\
				
				& ${\sf CT}_{\sf on}$	&14.82	&8.47	&14.82	\\
				
				& ${\sf Cost}$	&0.75	&0.43	&0.75	\\
				
				\midrule 
				\Block{3-1}{8192} 
				& ${\sf PT}_{\sf on}$	&59.27	&33.87	&59.27	\\
				
				& ${\sf CT}_{\sf on}$	&118.53	&67.73	&118.53	\\
				
				& ${\sf Cost}$	&6.03	&3.44	&6.01	\\
				\bottomrule
			\end{NiceTabular}
		}
		\caption{Benchmarking of comparisons over various depths. Each of the layer has 128 comparisons. Time is reported in minutes, and monetary cost in USD.\label{tab:compb2pcS}}
	\end{table}

Having benchmarked only the online phase, $\TWthisT$ is clearly the winner with respect to all the metrics. We believe a similar trend as observed in the prior frameworks will be followed here as well when considering the overall performance.  

\chapter{Conclusion and Open Problems}
\label{chap:conclusion}
This thesis designed $\thisP$, a robust MPC platform for privacy-preserving machine learning applications. The focus was on the small-party setting of two, three and four parties with at most one corruption under the control of a monolithic static adversary. A unified protocol design was presented, focusing on practical efficiency, which outperforms the state-of-the-art protocols by several orders of magnitude in the respective settings. On the way, several building blocks were identified for the PPML applications, and their efficient realizations were provided. Finally, the protocols were implemented by instantiating over Google Cloud instances and analyzed against various metrics such as run time, communication, throughput and monetary cost. The practicality of our platform was argued through improvements as observed in the benchmarks. 

\paragraph{Open Problems}
We leave the following problems open for further explorations. 
\begin{enumerate}
	\item {\em Applications:} The platform was designed for PPML applications such as linear regression, logistic regression, neural networks and support vector machines. However, other PPML applications such as graph neural networks, decision trees and random forests, quantized neural networks have not been explored much in the literature. Extending our platform to provide support for these advanced applications is an interesting direction. This may require support for new building blocks in layer II, such as square-root, exponentiation, batch normalization, to name a few.  While the platform discussed PPML applications, it is worthwhile to explore non-PPML applications such as private-set intersection, private-information retrieval, genome sequence matching. 
	\item {\em Adversarial setting:} The focus of the thesis was primarily on the honest majority setting. A step towards a dishonest majority was also taken, albeit in the semi-honest two-party setting. It is an interesting question to explore the dishonest majority setting in the presence of a malicious adversary. While protocols were designed in the synchronous network model with static corruptions, designing protocols in the asynchronous network model and against a stronger adaptive adversary is left open. 
	
	The recent notion of Friends-and-Foes~\cite{C:AloOmrPas20} (FaF) security resembles real-world corruption more closely, where the honest parties are instead considered to be semi-honest. Our protocols do not adhere to this security notion, and designing protocols for the same is an interesting future direction. 
	Finally, our protocols, together with the above-mentioned adversarial settings, can be explored for the general $n$-party case. 
	\item {\em Federated Learning:} The advancements in PPML have paved the way for federated learning which allows collaborative training while ensuring the training data resides only with the data owners. Since the data does not leave its owner, it increases the trust in the system and has gained a lot of attention recently. The traditional approach of realizing PPML via MPC does not extend naively to the federated setting. We leave open the question of realizing our architecture in the federated setting as an open problem.
\end{enumerate}

\bibliographystyle{plainnat}
\bibliography{others,cryptobib/abbrev0,cryptobib/crypto}
\end{document}